\documentclass{article}
\usepackage[utf8]{inputenc}
\usepackage{amsmath,amsthm,amssymb,bbm}
\usepackage{mathtools}
\usepackage{geometry}
\usepackage{graphicx}
\usepackage[dvipsnames]{xcolor}
\usepackage{natbib}
\usepackage{enumitem}
\usepackage{datetime}
\usepackage{changepage}
\usepackage{pgfplots}
\usepackage{multirow}
\usepackage[normalem]{ulem}
\usepackage{url}
\usepackage{datetime}
\usepackage{array}
\usepackage[symbol]{footmisc}
\pgfplotsset{compat=1.15}
\pgfplotsset{
    discard if not/.style 2 args={
        x filter/.code={
            \edef\tempa{\thisrow{#1}}
            \edef\tempb{#2}
            \ifx\tempa\tempb
            \else
                \def\pgfmathresult{inf}
            \fi
        }
    }
}

\usepackage[font={footnotesize}]{caption}
\usepackage{subcaption}
\usepackage{tikz}
\usetikzlibrary{calc,shapes,backgrounds,shapes,positioning,petri,arrows,external,intersections}
\usepgfplotslibrary{fillbetween}
\usetikzlibrary{patterns}

\newlength\figureheight
\newlength\figurewidth

\newdateformat{monthyeardate}{%
  \monthname[\THEMONTH] \THEYEAR}

\newcommand{\Prob}{\mathbbm{P}}
\newcommand{\E}{\mathbbm{E}}
\newcommand{\R}{\mathbb{R}}

\newcommand{\custtype}{\psi}

\newcommand{\bA}{\boldsymbol{\alpha}}
\newcommand{\bP}{\boldsymbol{\pi}}

\newcommand{\stkout}[1]{\ifmmode\text{\sout{\ensuremath{#1}}}\else\sout{#1}\fi}

\newtheorem{theorem}{Theorem}
\newtheorem{lemma}{Lemma}
\newtheorem{proposition}{Proposition}

\newtheorem{assumption}{Assumption}
\newtheorem{definition}{Definition}
\newtheorem{example}{Example}
\newtheorem{remark}{Remark}
\newtheorem*{lemma*}{Lemma}
\newtheorem*{proposition*}{Proposition}

\newtheorem{manualpropinner}{Proposition}
\newenvironment{manualprop}[1]{%
	\renewcommand\themanualpropinner{#1}%
	\manualpropinner
}{\endmanualpropinner}

\newtheorem{manualthminner}{Theorem}
\newenvironment{manualthm}[1]{%
	\renewcommand\themanualthminner{#1}%
	\manualthminner
}{\endmanualthminner}

\newtheorem{manualassumptioninner}{Assumption}
\newenvironment{manualassumption}[1]{%
	\renewcommand\themanualassumptioninner{#1}%
	\manualassumptioninner
}{\endmanualassumptioninner}

\usepackage{environ}
\makeatletter
\newsavebox{\measure@tikzpicture}
\NewEnviron{scaletikzpicturetowidth}[1]{%
  \def\tikz@width{#1}%
  \def\tikzscale{1}\begin{lrbox}{\measure@tikzpicture}%
  \BODY
  \end{lrbox}%
  \pgfmathparse{#1/\wd\measure@tikzpicture}%
  \edef\tikzscale{\pgfmathresult}%
  \BODY
}
\makeatother

\usepackage{setspace}
\usepackage[compact]{titlesec}
\titlespacing{\section}{0pt}{1.8ex}{1ex}
\titlespacing{\subsection}{0pt}{1.5ex}{1ex}
\titlespacing{\subsubsection}{0pt}{0.5ex}{0.5ex}

\begin{document}
\title{%
\vspace{-1.2cm}
Information Disclosure and Promotion\\\vspace{-0.1cm} Policy Design for Platforms\vspace{-0.1cm}
}
	

\author{
	{\sf Yonatan Gur}\thanks{Stanford University, {\tt $\{$ygur,ilanmor,dsaban$\}$@stanford.edu}}
	\and
	{\sf Gregory Macnamara}\thanks{Meta Platforms, Inc., {\tt gregory.macnamara@gmail.com}}
	\and
	{\sf Ilan Morgenstern}\footnotemark[1]
	\and
	{\sf Daniela Saban}\footnotemark[1]
}

\date{\monthyeardate\today}

\maketitle
\setstretch{1.05}
\vspace{-0.65cm}
\begin{abstract}\noindent
  We consider a platform facilitating trade between sellers and buyers with the objective of maximizing consumer surplus. Even though in many such marketplaces prices are set by revenue-maximizing sellers, platforms can influence prices through (i) price-dependent promotion policies that can increase demand for a product by featuring it in a prominent position on the webpage and (ii) the information revealed to sellers about the value of being promoted. Identifying effective joint information design and promotion policies is a challenging dynamic problem as sellers can sequentially learn the promotion value from sales observations and update prices accordingly. We introduce the notion of \emph{confounding} promotion policies, which are designed to prevent a Bayesian seller from learning the promotion value (at the expense of the short-run loss of diverting some consumers from the best product offering). Leveraging these policies, we characterize the maximum long-run average consumer surplus that is achievable through joint information design and promotion policies when the seller sets prices myopically. We then construct a Bayesian Nash equilibrium in which the seller's best response to the platform's  optimal policy is to price myopically in every period. Moreover, the equilibrium we identify is platform-optimal within the class of horizon-maximin equilibria, in which strategies are not predicated on precise knowledge of the horizon length, and are designed to maximize payoff over the worst-case horizon. Our analysis allows one to identify practical long-run average optimal platform policies in a broad range of demand models.

\end{abstract}
\setcounter{footnote}{0}
\renewcommand*{\thefootnote}{\arabic{footnote}}

\setstretch{1.5}
\section{Introduction}\vspace{-0.1cm}
Online marketplaces allow consumers to evaluate, compare, and purchase products while simultaneously providing a channel for third-party
sellers to reach a broader consumer base and increase demand for their products. 
In order to maintain a large consumer base, many platforms prioritize increasing consumer surplus by offering 
competitively priced products. At the same time, it is common practice in such marketplaces to let sellers set their own price, but such flexibility may result in higher prices that reduce consumer surplus. However, platforms retain the ability to impact consumer surplus by influencing sellers' pricing policies. One avenue for doing so is through designing the search and recommendation environment to incentivize sellers to post low prices. For example, a platform can choose to prominently feature sellers that set competitive prices, thereby increasing their visibility and boosting the demand they face. A second avenue for influencing prices involves strategically sharing information on how increased visibility impacts consumer demand. Platforms can typically observe and track consumer behavior across sellers and products and thus often have better information about consumer demand than sellers. Specifically, the additional demand that is associated with being promoted by the platform (e.g., being featured in a prominent position on the webpage) is typically a priori unknown to sellers. By strategically sharing this information, the platform can alter the seller's perceived value of being promoted and thereby impact the seller's posted prices.

In general, platforms may deploy various mechanisms for altering a given product's or seller's visibility throughout a consumer's interaction with the platform. A concrete example is provided by Amazon's \textmd{featured offer} (also known as the Buy Box), which is depicted in Figure~\ref{fig:BuyBox}. When a consumer reaches a product page on Amazon, 
she has the option to ``Buy Now" or ``Add to Cart" through links that are positioned in a designated, highly visible area of the webpage referred to as the Buy Box, or to consider ``Other Sellers on Amazon," an option that is positioned in a less visible area of the webpage and typically requires the consumer to scroll down the page.\footnote{Similar mechanisms are used by Walmart Marketplace (\citealt{SellerActive}) and eBay (\citealt{eCommerceEbay}).}

\begin{figure}[ht]
\begin{center}
    \includegraphics[width=0.85\textwidth]{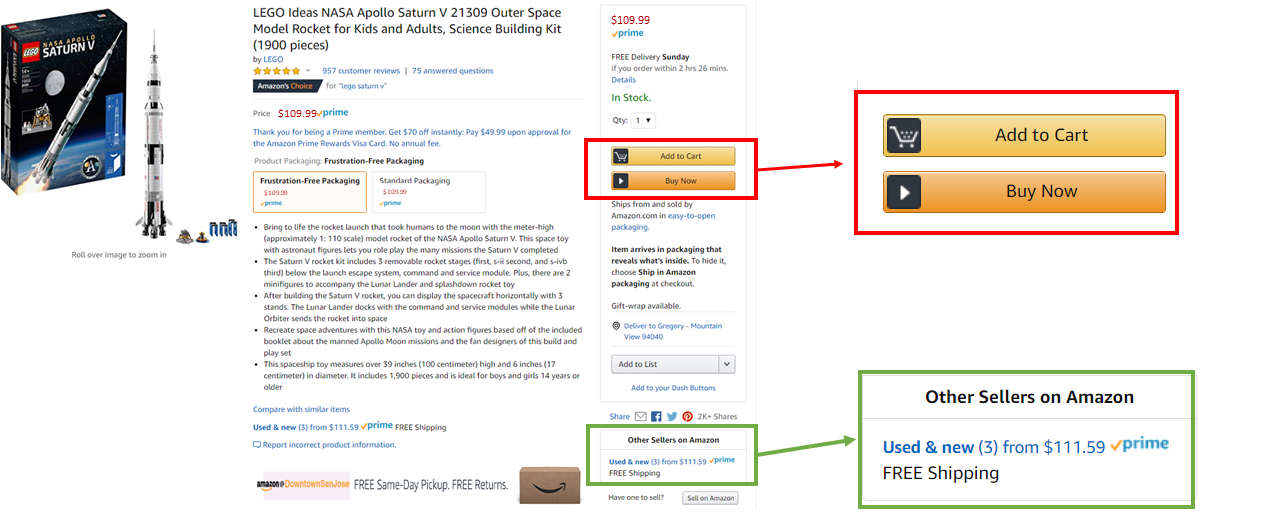}\vspace{-0.1cm}
    \caption{Example of Amazon's featured offer (Buy Box)}\vspace{-0.6cm}
    \label{fig:BuyBox}
    \end{center}
\end{figure}
If the consumer selects ``Buy Now" or ``Add to Cart," then the sale will be assigned to the seller that is featured in the Buy Box. Thus, by promoting a seller to the Buy Box, Amazon  \emph{effectively selects} the seller from which the consumer is purchasing; this valuable advantage allows the promoted seller to capitalize on demand from consumers that are ``impatient," or have a high cost of search. The rest of the sellers, which were not featured, are grouped under ``Other Sellers on Amazon," and will only be viewed by ``patient" consumers that scroll further down the page. Thus, the seller that is featured in the Buy Box faces limited competition for consumers who do not browse through all the available sellers and can expect to observe increased demand and a higher fraction of sales.

The impact of promotion on sales, which may 
vary across different 
product categories, has direct implications on the incentives of sellers and the pricing policies they deploy. If most consumers tend to buy using the Buy Box without considering other sellers, then being promoted generates a substantial increase in demand and sales. On the other hand, if the majority of consumers patiently scroll down to consider all candidate sellers, then the value of being promoted might be limited and even marginal. While it is known that being a featured merchant is valuable, individual sellers do not have access to consumer browsing information that would allow them to identify a priori the additional demand ``boost" associated with the platform's promotion decision.\footnote{Amazon's promotion decisions are based on a Featured Merchant Algorithm (FMA). While the factors accounted for by the FMA are not publicly declared, many resources suggest that the featured sellers are those who set low prices and have high consumer ratings. See, e.g., \cite{Chen:2016:EAA:2872427.2883089} and the blog post by \cite{informed.co}.} 
Thus, Amazon can influence pricing decisions not only through its Buy Box promotion policy, but also by leveraging the underlying information asymmetry through strategically disclosing information on the additional demand associated with being promoted.

A key challenge the platform faces in utilizing its private information is that sellers can, potentially, \emph{infer} the value of promotions over time from sales observations, and update their prices accordingly. Therefore, the platform needs to strike a balance between providing incentives for prices that maximize consumer surplus in the \emph{current} period, and controlling the information that is revealed by sales observations, which impacts consumer surplus in \emph{subsequent} periods. As the platform's information disclosure policy impacts the optimal promotion policy, which in turn impacts the seller's ability to collect information over time, the platform must consider the design of its promotion and information policies jointly. 
In this paper, we study how a platform can maximize consumer surplus through \emph{joint} information design and dynamic promotion policies that balance the aforementioned tradeoff.

We note that while the Buy Box example above describes a retail setting, our formulation and approach are relevant to similar ``promotions'' that are common in other online marketplaces and platforms where prices are set by sellers. Examples include lodging platforms (e.g., Airbnb), booking and travel fare platforms (e.g., Expedia, Booking.com, TripAdvisor), freelancing platforms (e.g., Upwork), and food ordering and delivery platforms (e.g., Uber Eats, Grubhub). While the structure of promotions and the criteria the platform uses to select promoted sellers may vary across these settings, they all share common features: promotions are valuable to sellers, though the exact value may be a priori unknown to sellers, and the platform may share information about this value to increase consumer surplus.

\subsection{Main Contributions}\vspace{-0.0cm}

Our contribution is threefold: (1) introducing a stylized model for studying the interaction between a platform and a seller who does not know the value of promotions; 
(2) characterizing the maximal long-run average expected consumer surplus that is achievable by the platform when the seller prices myopically; and (3) providing practical platform policies that achieve this maximal long-run average consumer surplus in equilibrium.
More specifically, our contribution is along the following dimensions.


\vspace{0.3cm}
\textbf{Modeling.} Our model considers a platform that can promote a single product to each arriving consumer, and a Bayesian seller that sequentially sets prices and has access to his own sales observations {(an extension to two Bayesian sellers is provided in \S\ref{sec: TwoSellers})}. Our formulation considers a broad class of demand and consumer choice models, and assumes that each arriving consumer is either impatient, and therefore considers only buying from the promoted seller (versus an outside option), or patient, and therefore considers all the relevant alternatives. The fraction of impatient consumers therefore captures the value of promotion for the seller.

The platform has private information about the true fraction of impatient consumers. At the beginning of the horizon, the platform provides an initial information signal regarding this fraction, and commits to a dynamic promotion policy (a dynamic sequence of functions) that at each possible history maps the price posted by the seller to a (possibly random) promotion decision. Subsequently, in each period the seller updates his belief about the fraction of impatient consumers and then posts a price. After the price is posted, the platform decides whether to promote the seller or an alternative. Then, a consumer arrives, forms a consideration set depending on her patience type, and makes a purchase decision according to an underlying demand model. The seller observes whether he made a sale or not.

Our baseline model is stylized, yet allows for tractability in a challenging dynamic problem that is relevant to many practical settings. Our model captures a fundamental tradeoff faced by the platform, between maximizing consumer surplus in the present, and controlling the demand information revealed to the seller, which may impact the achievable consumer surplus in the future.


\vspace{0.3cm} 
\textbf{Confounding Promotions: Design and Long-Run Average Optimality.} We observe that fully disclosing its private information can be detrimental to the platform. Moreover, even when the platform does not disclose the fraction of impatient consumers, the seller may be able to learn this value over time from sales observations. Thus, while optimizing the initial information signal can be valuable in the short term, it may be ineffective at increasing the long-run average consumer surplus, unless the platform's promotion policy carefully accounts for the information revealed by sales observations.

With this in mind, we propose a class of \emph{confounding} promotion policies, designed to control the information collected by the seller over time. Specifically, confounding policies ensure that if the seller sets prices myopically, his belief about the fraction of impatient consumers remains fixed throughout the problem horizon (after the initial information signal is sent). Such policies may require diverting some consumers away from the best product offering, thus weakening the seller's incentives to set low prices in the short run. Nevertheless, these policies can induce a long-run average consumer surplus that is higher than what is achievable by truthfully revealing the fraction of impatient consumers.

Furthermore, we characterize the maximum long-run average consumer surplus that is achievable by the platform when the seller prices myopically, and establish that confounding promotion policies (when coupled with a carefully selected initial information signal) are long-run average optimal.

We provide a simple procedure for constructing practical \emph{joint} information design and promotion policies that are long-run average optimal for a broad class of demand models. Our approach is based on reducing the platform's dynamic problem to one in which it needs to first identify the optimal confounding promotion policy for a given belief (which can be reduced to a static problem), and then identify the information signal that results in an optimal distribution of initial beliefs. This procedure allows one to study the impact that the underlying demand model and consumers' search patterns have on the design of effective promotion policies and the achievable consumer surplus.



\vspace{0.3cm} 
\textbf{Equilibrium Analysis.} We establish that myopic pricing is, in fact, a best response to the confounding platform strategy that results from our construction, thereby establishing a \emph{Bayesian Nash equilibrium} between the platform and the seller. In particular, the platform cannot benefit from deviating to any other joint information design and promotion policy, and the seller cannot gain from deviating to any other dynamic pricing policy at any stage of the game. {Furthermore, the equilibrium we identify is asymptotically platform-optimal within a class of horizon-maximin equilibria, in which the seller aims to maximize his payoff over the worst-case horizon length.} While the literature on dynamic pricing suggests that sellers should avoid confounding prices in order to learn the demand function, our characterization implies that, in the presence of a strategic platform, it might be optimal for the seller to set confounding prices, even though doing so leads to incomplete learning.

Moreover, we extend our formulation and analysis to incorporate price competition between two Bayesian sellers. 
 Our analysis demonstrates the value captured by confounding policies relative to truthful revelation in the presence of sellers that compete in prices in every period. We find that while confounding both sellers is typically infeasible or suboptimal, designing a joint signaling and promotion policy that confounds one of the sellers is often beneficial to the platform.

\subsection{Related Literature}\label{sec:relatedLiterature}

Our work relates to several strands of literature in operations and economics. First, the consideration of the seller's pricing decisions relates to the literature on dynamic pricing policies in settings characterized by demand uncertainty; see, e.g., \cite{araman2009dynamic}, \cite{Besbes2009}, \cite{farias2010dynamic},
\cite{Harrison2012},  \cite{DenBoer2014}, and \cite{keskinzeevi2014}; see also the surveys by \cite{araman2010revenue} and \cite{den2015dynamic} for an overview. {Some works (e.g., \citealt{conlisk1984cyclic}, \citealt{su2007intertemporal}, \citealt{besbes2015intertemporal}) study dynamic pricing with consumers that are heterogeneous in their willingness to wait to purchase a product. While we also consider patient and impatient consumers, we interpret patience as the propensity of consumers to spend time browsing for alternative products in the platform rather than their willingness to wait for lower prices in the future.
}

More broadly, the seller's problem relates to an extensive literature on sequential decision making under uncertainty in which a decision maker must balance a tradeoff between taking actions that generate high immediate payoffs and taking actions that generate information and therefore increase future payoffs. This tradeoff has  been studied in contexts including retail assortment selection (e.g., \citealt{Caro2007}, \citealt{saure2013optimal}) and inventory management (e.g., \citealt{huh2009nonparametric}, \citealt{Besbes2013}, \citealt{besbes2017exploration}). 
Our work departs from these models, which assume that, conditional on the decision maker's action, the payoff and information generated is \emph{exogenous}, by considering the pricing dynamics of a learning seller when demand is \emph{endogenously} affected by the platform. 

We analyze the aforementioned exploration/exploitation tradeoff from the perspective of the uninformed seller as well as the informed platform. See \cite{horner_skrzypacz_2017} for a survey of literature that considers the exploration/exploitation tradeoff in strategic settings. From the seller's perspective, most closely related is \cite{Harrison2012}, which considers a Bayesian seller who dynamically posts prices while learning about the underlying (exogenous) demand model. Their work demonstrates that while in many settings a myopic Bayesian pricing policy can be near optimal, it can suffer from incomplete learning if the seller reaches a \emph{confounding belief}. On the other hand, in our work the platform faces tradeoff between maximizing instantaneous consumer surplus and concealing demand information, and we show that, in many cases, effective platform strategies are designed to \emph{confound} the seller and prevent him from learning the underlying demand structure. 
Also related, \cite{gabaix2006shrouded} consider a setting where firms may choose to hide relevant product information to induce myopic consumers to buy an initial product at a low price and subsequently sell add-on products at higher prices. In this context, they show that information shrouding can be supported in equilibrium. While our results display some parallels, as we find that the platform often finds it optimal to prevent the seller from learning demand information, our setting differs in several ways. Particularly, our model features dynamic repeated interactions with sellers that learn both from a direct signal from the platform and from observing the evolution of sales over time.

In our formulation, the  platform and seller's interaction begins with a disclosure of information. In that sense, our work relates to the work on information design in the Bayesian persuasion framework originating in the work of \cite{Segal2010} and \cite{Kamenica2011}, and more broadly, to the work on repeated games of incomplete information in \cite{aumann1995repeated}. 
Thus, our work contributes to the growing field of information design in operational settings including queueing (\citealt{Lingenbrink}), networks (\citealt{Candogan2019}, \citealt{candogan2020optimal}), inventory (\citealt{drakopoulos2021persuading}), and exploration in platforms (\citealt{Papanastasiou},  \citealt{Bimpikis2019}, \citealt{Ozer2019}). The present paper departs from this line of work in terms of both the application domain and the setting; while the above studies typically consider a static formulation whereas in our setting the information signal is followed by a \emph{dynamic} interaction  through which further information may be revealed to the seller. For additional models of dynamic Bayesian persuasion see, e.g., \cite{ely2015suspense} and \cite{ely2017beeps}.

Other literature streams study the interaction between sellers, consumers, and platforms that facilitate their trade. In our model, the platform can impact purchase decisions through its selection of which product to promote, which has been studied empirically in,
e.g., \cite{kim2010online},  \cite{chen2016sequential}, 
and \cite{besbes2015optimization}. 
Within this stream, our work relates most closely to papers considering the design of platform recommendations and search environments such as \cite{hagiu2011intermediaries}, which studies 
how the design of consumer search may incentivize sellers to lower prices.
Our analysis identifies a new consideration for this design: preventing learning by a seller with incomplete demand information.
\cite{Dinerstein2018} empirically analyze a similar tradeoff between directing consumers to desired products and strengthening incentives for sellers to lower prices in the context of the eBay search environment. 



\vspace{-0.1cm}\section{Model}\label{sec:Model}\vspace{-0.1cm}
In this section, we introduce a stylized model of the dynamic interactions between a seller and a platform. We start by providing an overview of the model, followed by a detailed description of each of its components. We discuss our modeling assumptions and some extensions in \S\ref{sec:model-discussion}.

\textbf{Overview of Incentives.}  
We study how a platform, which facilitates trade between sellers and consumers, should design a joint promotion and information disclosure policy to maximize consumer surplus by leveraging information asymmetry about the value of promotion.
In our model, consumers arrive to the platform sequentially. Upon arrival, each consumer observes a promoted product and,
depending on her type, may consider additional products as well. We interpret a consumer's type as her willingness to spend time browsing for products in the platform. For simplicity, we assume that each consumer is either \textit{impatient} or \textit{patient}. Impatient consumers spend little time searching and only consider the promoted product whereas patient consumers spend enough time to consider all available products.\footnote{{This is not to be confused with other works that use this terminology to distinguish customers that are willing to wait to buy a product in the future for a lower price (e.g., \citealt{su2007intertemporal}, \citealt{cachon2009purchasing}, \citealt{besbes2015intertemporal}).}} 
%
Each arriving consumer observes 
 the product(s) in her consideration set and makes a purchase decision according to an underlying demand model. By selecting which product to promote, the policy directly affects the {consideration sets of impatient consumers (and thus their choices)}, which impacts their surplus and the seller's revenue. Moreover, the promotion policy can also influence the seller's pricing decisions: as impatient consumers only consider purchasing from the seller if he is promoted, a policy that promotes low-priced products {may incentivize} the seller to set low prices  to increase sales.

The platform's ability to influence pricing decisions, however, is determined by the fraction of impatient consumers. {If a large fraction of consumers are impatient, the promotion decision directly impacts a substantial part of the overall demand, which may incentivize 
the seller to set the platfom's target price.} 
On the other hand, if only a small fraction of consumers are impatient, promotion can generate little benefit for the seller and the extent to which the promotion policy can influence pricing is more limited. When this fraction is privately known by the platform, the seller's \emph{belief} about it impacts his pricing incentives. Thus, the platform can also influence the seller's pricing decisions by strategically sharing information about the fraction of impatient consumers.

A key challenge captured by this dynamic setting is that, as consumers arrive and make purchase decisions sequentially, the seller may progressively collect information about the fraction of impatient consumers and adjust his price accordingly. Importantly, as the platform's promotion policy impacts consumer demand, it also impacts the informativeness of sales observations. Therefore, the promotion policy impacts  not only the current price but also  future prices (by affecting the seller's beliefs). Thus, the platform's dynamic problem is inherently intertwined because it must \emph{jointly} optimize its information provision and promotion policies.

\textbf{The Dynamic Game.} We model the interaction between the platform and the seller as consisting of two stages. First, before consumers arrive, the platform publicly commits to $(i)$ a \textit{signaling mechanism}~$\sigma$, which may reveal information about the fraction of impatient consumers, and $(ii)$ a \textit{promotion policy}~$\boldsymbol{\alpha}$; both of these are described in detail below. The platform then \textit{privately} observes the true fraction of impatient consumers $\phi\in \{\phi_L,\phi_H\}$, where  $\phi = \phi_H$ with commonly known probability $\mu_0$ and where $0<\phi_L<\phi_H<1$. (Note that, in the tradition of the information design literature, we assume that the platform commits to a signaling mechanism before observing the true fraction of impatient consumers,~$\phi$.) Finally, the platform sends a signal $s$, which is drawn according to $\sigma$. See Figure~\ref{fig:preliminary_dynamics} for a summary of these dynamics.

\begin{figure}[htbp!]
	\centering
	\begin{scaletikzpicturetowidth}{.6\textwidth}
		\begin{tikzpicture}[scale=\tikzscale]
    \tikzstyle{fontbf} = [blue,text centered,font=\bfseries]
    \draw[thick,->] (0,0) -- (14.25,0);
    \draw[thick] (0,-0.25) -- (0,.25)node[anchor=north,yshift = -3mm, align = center, font=\footnotesize]{Period $0$} ;
    \draw[thick] (13.2,-0.25) -- (13.2,.25)node[anchor=north,yshift = -3mm, align = center, font=\footnotesize]{Period $1$  Begins} ;
    \draw[thick] (13.2/4-1,-0.1) -- (13.2/4-1,.1)node[anchor=south,yshift=1mm,align = center, font=\footnotesize]{Platform \\Commits to \\$\boldsymbol{\alpha},\sigma$ }; 
    \draw[thick] (13.2/2,-0.1) -- (13.2/2,.1)node[anchor=south,yshift=1mm,align = center, font=\footnotesize]{Platform \\Observes Fraction \\$\phi$} ;
    \draw[thick] (13.2*3/4+1,-0.1) -- (13.2*3/4+1,.1)node[anchor=south,yshift=1mm,align = center, font=\footnotesize]{Seller \\Observes Signal, \\ $s\sim \sigma(\phi)$ } ;
;    

    \end{tikzpicture}
	\end{scaletikzpicturetowidth}	
    \caption{Dynamics before horizon begins ($t=0$)} \vspace{-0.25cm}
    \label{fig:preliminary_dynamics}
\end{figure}

In the second stage, $T$ different consumers arrive sequentially. In each period $t=1,...,T$, before consumer $t$ arrives, the seller sets a price $p_t\in \mathcal{P}$ and then the platform uses its promotion policy $\bA$ to decide to promote the seller, $a_t=1$, or not, $a_t=0$. Consumer $t$ then arrives and observes the products in her consideration set and their prices. With probability $\phi$, she is impatient and only considers the promoted product. With probability $1-\phi$, she is patient and considers all products, regardless of the platform's promotion decision. The consumer then makes a purchase decision according to an underlying discrete choice model from a broad family of models described under ``consumer demand" below. 
Finally, the seller observes his own sales outcome, $y_t\in\left\{0,1\right\}$. See Figure~\ref{fig:within_period_dynamics} for a summary of these dynamics.

\begin{figure}[htbp!]
	\centering
	\begin{scaletikzpicturetowidth}{.85\textwidth}
		\begin{tikzpicture}[scale = \tikzscale]
    \draw[thick,->] (0,0) -- (14.25,0);
    \draw[thick] (0,-0.25) -- (0,.25)node[anchor=north,yshift = -5mm, align = center, font=\footnotesize]{Period $t$  Begins} ;
    \draw[thick] (13.2,-0.25) -- (13.2,.25)node[anchor=north,yshift = -5mm, align = center, font=\footnotesize]{Period $t+1$  Begins} ;
    \draw[thick] (4.5/4,-0.1) -- (4.5/4,.1)node[anchor=south,yshift=1mm,align = center, font=\footnotesize]{Seller \\ Updates Belief \\ $\mu_{t}$}; 
    \draw[thick] (18/4,-0.1) -- (18/4,.1)node[anchor=south,yshift=1mm,align = center, font=\footnotesize]{Seller \\ Sets Price \\ $p_t$} ;
    \draw[thick] (31.5/4,-0.1) -- (31.5/4,.1)node[anchor=south,yshift=1mm,align = center, font=\footnotesize]{Platform Makes \\ Promotion Decision \\ $a_t$} ;
    \draw[thick] (45/4,-0.1) -- (45/4,.1)node[anchor=south,yshift=1mm,align = center, font=\footnotesize]{Consumer Arrives \& \\ Purchases (or not) \\ $y_t$ } ;
;    
    \end{tikzpicture}
	\end{scaletikzpicturetowidth}	
    \caption{Dynamics in each period $t=1,...,T$}\vspace{-0.25cm}
    \label{fig:within_period_dynamics}
\end{figure}


\textbf{Consumer Demand}. In each period, a new consumer arrives with an independently drawn patience type, observes the products in her consideration set, and purchases according to a discrete choice model.
{The probability of purchasing from the seller is given by a commonly known function $\rho$ that depends on the consumer's type $\psi$, the seller's price $p$, competitors' prices $\vec{q}$, and a platform promotion decision~$a\in\{0,1\}$:
	
\begin{equation}\label{eq:ConsumerDemand:Rho}
    \rho(p,\vec{q},a, \psi) = \Prob(y = 1|p,\vec{q},a, \psi) =
\begin{cases}
    \bar{\rho}_c(p, \vec{q}), &\text{ if the consumer is patient} \\
    \bar{\rho}_0(p), &\text{ if the consumer is impatient} \text{ and } a = 1 \\
    0, &\text{ if the consumer is impatient} \text{ and } a =0.
\end{cases}
\end{equation}
The demand function $\rho$ captures the impact of the consumer's patience type: $\bar{\rho}_c$ denotes demand for the seller when {the consumer is} patient and thus considers purchasing from the seller {and his competitors};  $\bar{\rho}_0$, on the other hand, captures the demand when the consumer is impatient and the seller is promoted, and hence depends only on the seller's price.} Moreover, the probability of an impatient consumer purchasing from the seller equals $0$ unless the seller is promoted. We discuss a more general formulation that relaxes these assumptions in \S\ref{sec:model-discussion}. Note that if $\bar{\rho}_c(p,\vec{q}) \ne \bar{\rho}_0(p)$, then $\rho$ captures a setting where the probability that the consumer buys from the seller depends on whether the consumer considers other products. For the rest of the paper, we keep the competitor's prices fixed at some vector $\vec{q}$ and, slightly abusing notation, denote $\bar{\rho}_c(p):=\bar{\rho}_c(p,\vec{q})$. We relax this assumption by explicitly modeling price competition between sellers in~\S\ref{sec: TwoSellers}.

We assume a stationary arrival process where each consumer's patience type and purchase probability is independent of $t$. We make the following assumption on the demand function.
\begin{assumption}[Demand]\label{assump:Demand}
    $\bar{\rho}_c(p)$ and $\bar{\rho}_0(p)$ are decreasing and Lipschitz continuous in $p$; $p\bar{\rho}_c(p)$ and $p\bar{\rho}_0(p)$ are strictly concave in $p$; and $\bar{\rho}_0(p) \geq \bar{\rho}_c(p)$ for all $p \in \mathcal{P}$.
\end{assumption}\vspace{-0.2cm}
Assumption \ref{assump:Demand} is mild and satisfied by common demand models, including logit, mixed logit, and probit, among many others. The concavity of the seller's revenue function ensures that there is a unique revenue-maximizing price for each consumer type, and the ordering on purchase probability requires competitors' products to be substitutes for the seller's product. We illustrate a simple demand model that satisfies these conditions in Example~\ref{example:UniformDemand} (presented below).

\textbf{Payoffs.} Without loss of generality, we normalize the cost of the seller to be 0, and so the seller's payoff in period $t$ as a function of his price, $p \in \mathcal{P}$, and the consumer's purchase decision, $y \in \{0,1\}$, is \vspace{-0.2cm}
\begin{equation*}
    v(p,y) = py.
\vspace{-0.2cm}
\end{equation*}
The platform's payoff in each period equals the expected consumer surplus, which is captured by a commonly known function $W$ of the seller's price $p$, the platform's promotion decision $a$, and the consumer's type $\psi$:
\begin{equation}\label{eq:ConsumerWelfare:W}
    W(p,a,\psi) =
\begin{cases}
    \bar{W}_c(p),  &\text{ if the consumer is patient} \\
    \bar{W}_0(p),  &\text{ if the consumer is impatient and }  a = 1 \\
    \bar{W}_{\text{out}}, &\text{ if the consumer is impatient and } a =0.
\end{cases}
\end{equation}

We make the following mild assumption on consumer surplus.
\begin{assumption}[Consumer Surplus]\label{Assump:ConsSurplus}
    $\bar{W}_c(p)$ and $\bar{W}_0(p)$ are decreasing and Lipschitz continuous.\vspace{-0.2cm}
\end{assumption}

Our formulation does not specify any relationship between the demand function in \eqref{eq:ConsumerDemand:Rho} and the consumer welfare in \eqref{eq:ConsumerWelfare:W}, and our results do not depend on these two being related as long as they satisfy Assumptions~\ref{assump:Demand} and \ref{Assump:ConsSurplus}. However, in many practical settings 
the demand and consumer surplus functions are generated according to an underlying consumer utility model.
The following example illustrates such a relation, between the purchase probability in \eqref{eq:ConsumerDemand:Rho} and the consumer surplus in \eqref{eq:ConsumerWelfare:W}, when purchasing decisions correspond to uniformly distributed willingness to pay. We revisit this example throughout the paper for illustration.

\begin{example}[Uniform WtP]\label{example:UniformDemand}
    Suppose that there are two products on the platform, for which each customer~$t$ has willingness to pay that is independent and distributed uniformly over a unit square: $v^1_t \sim U[a-1,a]$ and $v^2_t \sim U[b-1,b],$ {where $a,b$ are parameters in $[0,1]$.} Suppose that each arriving customer maximizes her net utility (which is normalized to zero under the outside option) and that seller 2 sets a fixed price equal to~$0$ (or equivalently, that $v_t^2$ represents consumer $t$'s value relative to some fixed price).
    Then, given the first product's price $p\in [0,a]$, the demand function in~\eqref{eq:ConsumerDemand:Rho} is characterized by:
    \vspace{-0.25cm}
    $$\bar{\rho}_0(p)  = \Prob(v_1 - p\geq 0)
    = a-p, \quad \bar{\rho}_c(p)  =
    \begin{cases}
    (1-b)(a-p)+\frac{(a-p)^2}{2}, &\text{ if } p>a-b\\
    a-p - \frac{b^2}{2}, &\text{ if } p\leq a-b.
    \end{cases} $$
The consumer surplus function is then characterized by:
\begin{equation*}
	\bar{W}_0(p) = \int_{a-1}^{a} \max \{v_1-p,0\} \partial v_1 = \frac{(a-p)^2}{2}, \qquad \bar{W}_{\text{out}} = \int_{b-1}^{b} \max \{v_2,0\} \partial v_2 = \frac{b^2}{2},
\end{equation*}
\vspace{-0.5cm}
    \begin{equation*}
        \bar{W}_c(p) = \int_{b-1}^{b}\int_{a-1}^{a} \max \{v_1-p,v_2,0\} \partial v_1\partial v_2 =\begin{cases}
            \frac{1}{6}(3 b^2 + 3 (a-p)^2(1-b)  + (a-p)^3  ), & \text{ if } p>a-b\\
            \frac{1}{6}(3 (a-p)^2 + 3 b^2(1-a +p )  + b^3), &\text{ if } p\leq a-b.
        \end{cases}
    \end{equation*}
\end{example}

\textbf{Histories, Strategies and Beliefs.} Given a space $\mathcal{X}$, let $\Delta(\mathcal{X})$ be the space of probability measures on $\mathcal{X}$.
At the beginning of the horizon, before the observation of $\phi$, the platform commits to a joint promotion and information disclosure strategy $(\boldsymbol{\alpha}, \sigma)$.
Let the set of possible signals be denoted by $\mathcal{S}$; then, the platform's information disclosure strategy is a signaling mechanism, $\sigma:\{\phi_L,\phi_H\}\rightarrow \Delta(\mathcal{S})$, that maps the true fraction of impatient consumers to a distribution over signals. We denote the realized signal by $s \in \mathcal{S}$ and the space of signaling mechanisms by $\Sigma$. Let  $\bar{h}_t = \left\langle s, \left(p_{t'}, y_{t'}\right)_{t'=1}^{t-1} \right\rangle$ denote the signal and the sequence of seller's posted prices and sales realizations prior to the beginning of period $t$. Moreover, denote the set of these by~$\bar{H}_t = S\times  \left(\mathcal{P} \times \{0,1\}\right)^{t-1}$. The platform's promotion strategy, $\boldsymbol{\alpha} = \{\alpha_t\}_{t=1}^T$, is a vector of mappings,  where~$\alpha_t: \mathcal{P} \times \{\phi_L,\phi_H\} \times \bar{H}_t \rightarrow [0,1]$ specifies the probability that the seller is promoted in period $t$ as a function of the seller's current price, the value of $\phi$, and the previous prices and sales observations. We assume that the seller's pricing space $\mathcal{P}$ is a compact interval of the real line that contains zero.
We denote the realized promotion decision at time $t$ by~$a_t\in\left\{0,1\right\}$ and the set of dynamic promotion policies by~$\mathcal{A}$.

In addition to $\bar{h}_t$, the seller also observes the platform's announced strategy $(\bA,\sigma)$. Thus, denote the seller's information at the \emph{beginning} of period $t$ by $h_1 = \left\langle s,\bA,\sigma \right\rangle$ and by $h_t = \left\langle s,\bA,\sigma, \left(p_{t'}, y_{t'}\right)_{t'=1}^{t-1}\right\rangle$ for $t>1$. Moreover, we denote by $\{\mathcal{H}_t = \sigma(h_t), t=1,...,T\}$ the filtration associated with the process $\{h_t\}_{t=1}^T$, and the set of possible histories at the beginning of period $t$ by $H_t = S\times \mathcal{A}\times \Sigma \times \left(\mathcal{P} \times \{0,1\}\right)^{t-1}$. The seller's strategy is a vector of non-anticipating mappings $\boldsymbol{\pi} = \{\pi_t\}_{t=1}^T$, where each $\pi_t:H_t\rightarrow \Delta(\mathcal{P})$ maps the seller's information in period $t$ to a distribution from which the seller's period $t$ price is drawn. Denote the set of non-anticipating seller strategies by~$\Pi$.

In each period, based on the available history of information, the seller updates his belief about~$\phi$ according to Bayes' rule. We denote the seller's belief system by $\boldsymbol{\mu} = \{\mu_t\}_{t=1}^T$, where $\mu_t:{H}_t \rightarrow [0,1]$ is the probability that he assigns to $\{\phi=\phi_H\}$, the fraction of impatient consumers being high.

Given a platform strategy $\boldsymbol{\alpha},\sigma$ and a seller policy $\boldsymbol{\pi} \in \Pi$, denote the platform's expected payoff by
\begin{equation}
    W^{\boldsymbol{\alpha},\sigma,\boldsymbol{\pi}}_T(\mu_0) = \E\left(\sum_{t=1}^T  W(p_t,a_t,\psi_t) \middle|\boldsymbol{\alpha},\sigma,\boldsymbol{\pi}\right),
\end{equation}
where the expectation is with respect to ($\boldsymbol{p},\boldsymbol{a},\boldsymbol{ y},s,\phi$) and $\mu_0$ is the commonly known prior for $\phi$.  Moreover, denote the seller's expected continuation  payoff in the beginning of  period $t$ and given  history $h\in H_t$ by
\begin{equation}\label{eq:SellerPayoff}
    V_{t,T}^{\boldsymbol{\alpha},\sigma,\boldsymbol{\pi}} \left(h\right) = \E\left(\sum_{t'=t}^T v(p_t,y_t)\middle|h_t=h,\boldsymbol{\alpha},\sigma,\boldsymbol{\pi}\right).
\end{equation}

\subsection{Discussion of Model Assumptions}\label{sec:model-discussion}
\textbf{Platform Maximizes Consumer Surplus.} 
The platform objective of maximizing consumer surplus has been commonly considered in previous models of platform design (see, e.g., \cite{Dinerstein2018} and the  references therein). 
Moreover, we note that Assumption \ref{Assump:ConsSurplus} is quite general and satisfied in many instances by functions that, for example, 
seek to maximize the probability of a consumer purchase.

\textbf{One Learning Seller.} {For the sake of tractability, we first focus on a setting where there is a single seller who is learning and all the other sellers set the same price in each period. 
In \S\ref{sec: TwoSellers} we study an extension where two sellers compete in prices in every period and the platform selects which of them to promote. In that setting, we find that similar insights to the ones we derive for a single seller hold. 
}


\textbf{Patience Types and Search Costs.} To simplify exposition, we characterize consumers by a patience type, determining whether they consider only the promoted product or all available products. However, our main results and findings extend to a more general setting where (i) impatient consumers can buy from non-promoted sellers with nonzero probability (but are relatively less likely to do so than patient consumers), and (ii) the demand of patient consumers is also influenced by the platform's promotion decision, although to a lesser extent than for  impatient consumers. We provide this extension in Appendix~\ref{app:GeneralizedDemand}. Importantly, we remind the reader that we interpret patience in terms of search behavior, that is, the degree to which consumers are willing to spend time browsing for products in the platform.
{While we abstract away from explicitly modeling search costs, our formulation implies that search behavior is  independent of posted prices.}


\textbf{Platform Leads, Seller Follows.} In line with the information design literature (e.g., \citealt{Kamenica2011}) our model assumes that the platform commits to a dynamic promotion policy and a signaling mechanism upfront. Therefore, in each period, the seller knows the probability of being promoted as a function of the posted price, given the true value of $\phi$.

\textbf{Seller Information.} Our model and analysis are motivated by settings where many consumers arrive to the platform. In such settings, it would be difficult for a seller to track how prominently his product is featured to each consumer and the seller would have limited ability to know how many customers \emph{considered} his product without the platform sharing that information. We do, however, assume that the seller knows how many potential consumers have arrived to the product page (which relies on market characteristics) and how many of them purchased his product. 

\section{Preliminary Analysis} \label{sec: preliminary_analysis}\vspace{-0.1cm}

We begin our analysis by introducing and analyzing the setting in which the seller follows a myopic pricing policy that is designed to maximize 
his expected revenue in the current period. We show that, under myopic pricing, the platform may focus on a tractable class of policies without loss of optimality, which considerably simplifies its policy design problem. While we primarily focus on this setting henceforth, we 
consider general dynamic pricing policies in \S\ref{sec:EquilibriumAnalysis}, where we show that 
myopic pricing is in fact sustained in an equilibrium.

We then analyze a class of platform policies that are natural in the presence of myopic pricing: 
myopic promotions, designed to maximize consumer surplus in the current period, paired with an initial information signal. We observe that for such policies, optimizing the information signal may lead to a short-term gain in consumer surplus, but that the average gain per period diminishes as the time-average 
surplus they generate converges asymptotically to that achievable with a truthful signal. This motivates the approach advanced in~\S\ref{sec:LongRunAverageOptimalConsumerSurplus}, showing that a higher long-run average surplus can be generated by platform policies that might not maximize consumer surplus in each period, but control the information collected by the seller over time.

\subsection{Myopic Pricing}\label{sec: myopic_pricing}\vspace{-0.1cm}

We formally define the seller's \emph{myopic Bayesian pricing policy}, denoted by $\boldsymbol{\pi^*}$,  as follows.

\begin{definition}[Myopic Bayesian Pricing Policy]\label{def:MyopicPricingPolicy}
	In every period $t$ and at every history $h \in H_t$, a myopic Bayesian pricing policy $\boldsymbol{\pi^*}=\{\pi_t^*\}_{t=1}^T$ sets a price $p_t \in \mathcal{P}$ that maximizes the seller's expected revenue in the current period given history $h$ and $\alpha_t$. That is, $\pi_t^*$ satisfies
	\begin{equation}\label{eq:myopic}
		\Prob\left(p_{t} \in \arg \max_{p \in \mathcal{P}} ~~~\E_{a_t,y_{t},\psi}\left( v(p,y_{t}) |h_t ={h}\right)\middle|\pi_t^*\right)=1.
	\end{equation}
	If multiple prices satisfy \eqref{eq:myopic}, $\boldsymbol{\pi^*}$ selects one that  maximizes the current  consumer surplus.\footnote{This is akin to considering \emph{sender-preferred equilibria} which is standard in models of Bayesian persuasion; see related discussion in \cite{Kamenica2011} as well as in \cite{drakopoulos2021persuading}.}
\end{definition}
\vspace{-0.25cm}

The myopic pricing policy maximizes the seller's expected revenue in the current period given the platform's promotion policy and the information that the seller has about $\phi$. In particular, the promotion policy in future periods does not affect the myopic pricing policy. We note that considering myopic pricing reduces the complexity of the seller's decision, yet reflects a fair level of seller sophistication as it still requires recurrent updates of beliefs and prices.\footnote{Moreover, in many settings with uncertainty about demand, myopic pricing policies were shown to achieve good performance in terms of maximizing the seller's long-term payoffs; see, e.g., related discussion in \cite{Harrison2012}.}

In general, in each period the posted price may affect the seller's current   revenue as well as the platform's future promotion policy, 
and therefore the seller's pricing policy could potentially depend on the history in complex ways.
Nevertheless, from an analysis perspective, there is an advantage in focusing on policies that depend on the history in a simple way. For that purpose, we next define the set of promotion policies that depend on the history only through the seller's current belief.

\begin{definition}[Promotion Policies Based on Seller's Belief]\label{def:PromotionPolsSellerBelief}
	The set of promotion policies $\mathcal{A}^M \subset \mathcal{A}$ are those that are  constant across histories that generate the same belief. That is, $\boldsymbol{\alpha'} \in \mathcal{A}^M$, if and only if, for all $t=1,...,T$, $\sigma \in \Sigma$, and for any $\bar{h}',\bar{h}'' \in \bar{H}_t$ such that $\mu_t(\langle \boldsymbol{\alpha'},\sigma,\bar{h}'\rangle ) = \mu_t(\langle \boldsymbol{\alpha'},\sigma,\bar{h}''\rangle )$, one has~$\alpha_t'(p,\phi,\bar{h}') = \alpha_t'(p,\phi,\bar{h}'')$ for all $p \in \mathcal{P}$ and $\phi \in \{\phi_L,\phi_H\}$.
\end{definition}
\vspace{-0.25cm}

In the following lemma we establish that when the seller prices myopically, it is without loss of optimality for the platform to consider promotion policies in $\mathcal{A}^M$.



\begin{lemma}[Dependence on Histories through Beliefs]\label{lemma:PromotionAsFunctionOfBelief}
	Suppose that the seller follows the myopic Bayesian pricing policy $\boldsymbol{\pi^*}$. Then, for any $\bA \in \mathcal{A}$, $\sigma \in \Sigma$, there exists a promotion policy $\boldsymbol{\alpha'} \in \mathcal{A}^M$ such that \vspace{-0.2cm}
	$$W_T^{\bA,\sigma,\boldsymbol{\pi^*}} (\mu_0) \leq W_T^{\boldsymbol{\alpha'},\sigma,\boldsymbol{\pi^*}} (\mu_0).\vspace{-0.2cm}$$
\end{lemma}
\vspace{-0.25cm}

Formal proofs of Lemma~\ref{lemma:PromotionAsFunctionOfBelief} and subsequent results can be found in Appendix~\ref{sec:proofs}. The key idea of the proof is to observe that, in any given period, conditional on the belief $\mu_t$ and the promotion policy in that period~$\alpha_t$ (as a function of $\phi$ and $p$), the seller's expected revenue in that period is independent of the history. Therefore, at histories with the same belief and the present promotion policy, the set of optimal prices for a myopic seller is identical. As consumer surplus is a function of the posted price, the problem of designing an optimal promotion policy can be framed as a dynamic program in which the policy influences the price and the belief in each period and, importantly, 
the seller's belief is the only payoff-relevant state information. 
As a result, we can construct an optimal policy that  depends only on the history through the seller's belief. 

Moreover, the platform's policy design problem is further simplified when the seller sets prices myopically. Specifically, as we show below, it suffices to consider policies that promote the seller with positive probability only if his price matches a set target; such policies are practical to implement as the platform only needs to communicate a single price and the probability of promotion that corresponds to it.\footnote{In many cases these policies are equivalent to threshold policies where the platform communicates the maximum price that is promoted with positive probability and the corresponding probability.}

\begin{definition}[Single-Price Promotion Policies] \label{def: SinglePricePromoPolicies}
	Single-price promotion policies are defined as ones that, given any history $h \in \bar{H}_t$, promote at most one price with positive probability in each period. We denote the set of single-price promotion policies by $\mathcal{A}^P\subset\mathcal{A}$, formally defined as follows:\vspace{-0.25cm}
	\begin{equation*}
\mathcal{A}^{P}:=\left\{\boldsymbol{\alpha}\in\mathcal{A}:~ \forall t=1,...,T,~h\in \bar{H}_t, ~\exists \bar{p}_t(h)\in \mathcal{P}\text{~ s.t. } \alpha_t(p,\phi,h)=0, \forall p\ne \bar{p}_t(h) \right\}.
	\end{equation*}
\end{definition}
\vspace{-0.3cm}

In addition, let $\Sigma^S$ denote the set of \emph{simple signaling mechanisms} that are based on the reduced set of signals~$\mathcal{S}=\{\phi_L,\phi_H\}$. The next proposition establishes that considering single-price promotion policies with simple signaling mechanisms is without loss of optimality.

\begin{proposition}[Payoff Equivalence of Single-Price Promotion Policies with Reduced Signal Set]\label{lemma:RevenueEquivalence}
	Suppose that the seller adopts the myopic Bayesian pricing policy $\boldsymbol{\pi^*}$. Then, for any~$T\geq 1$, $\boldsymbol{\alpha} \in \mathcal{A}$, $\sigma\in\Sigma$, there exists a 
	promotion policy 
	$\boldsymbol{\alpha'} \in \mathcal{A}^{P}\cap \mathcal{A}^M$, and a simple signaling mechanism~$\sigma'\in\Sigma^S$, such that for all $\mu_0\in[0,1]$, \vspace{-0.25cm}
	\begin{equation*}
		W^{\boldsymbol{\alpha},\sigma,\boldsymbol{\pi^*}}_T(\mu_0) \leq W^{\boldsymbol{\alpha'},\sigma',\boldsymbol{\pi^*}}_T(\mu_0).
	\end{equation*}
\end{proposition}
\vspace{-0.5cm}

When the seller sets prices myopically, the above two results simplify the platform's policy design problem, implying that it suffices to track the 
seller's belief over time instead of the full history of actions, and the promotion policy can take a simple and tractable form: selecting a single target price and its associated promotion probabilities 
 at each history (where
the promotion probability is zero for all other prices). 



\vspace{-0.05cm}
\subsection{Myopic Promotions}\label{subsec:InsufficiencyOfTruthfulDisclosure}\vspace{-0.1cm}

We now study a class of platform policies that are natural to consider when the seller prices myopically: myopic promotions designed to maximize instantaneous consumer surplus, together with an initial signal that may reveal information about~$\phi$. Our analysis reveals that myopic promotion policies typically allow the seller to learn the true fraction of impatient consumers~$\phi$. As a result, with any initial information signal, the long-run average consumer surplus generated by myopic promotions converges to the one achievable by truthfully revealing $\phi$ to the seller. Therefore, while optimizing the initial signal may increase consumer welfare in the short run, the long-run average impact of optimizing this signal diminishes asymptotically.

We begin by introducing myopic promotion policies and discussing the incentives they create. By Lemma~\ref{lemma:PromotionAsFunctionOfBelief}, it is without loss to assume that myopic pricing decisions  depend only on the history through the promotion policy and the current belief. Therefore, in any period and at any history that corresponds to the same belief $\mu\in[0,1]$, a myopic promotion policy must generate the same expected consumer surplus, which we denote by $W^M(\mu)$. Thus, to characterize a myopic promotion policy,\footnote{There is not a unique myopic promotion policy because the probability of promotion at prices not selected by the seller does not affect the outcome or expected payoffs.} it suffices to set $T=1$ and solve the following optimization problem for each belief $\mu_1 \in [0,1]$:\vspace{-0.05cm}
\begin{equation} \label{eq:myopicPromotion}
	\begin{split}
		W^M(\mu_1)=
		\max_{\substack{ p \in \mathcal{P},\\ \alpha:~P\times \{\phi_L,\phi_H\}\times[0,1]\rightarrow [0,1]}}& ~~ \E_{\phi}\bigg(\phi(\bar{W}_{\text{out}} + \alpha(p,\phi,\mu_1) (\bar{W}_0(p) - \bar{W}_{\text{out}})) + (1-\phi)\bar{W}_c(p)\bigg|\mu_1\bigg)\\
		\text{s.t.}&~~ p\in\arg\max_{p'\in \mathcal{P}}~ \E_{\phi} \left(\phi \alpha(p',\phi,\mu_1) p'\bar{\rho}_0(p') + (1-\phi)p'\bar\rho_c(p') \middle|\mu_1 \right).
	\end{split} 
\end{equation}
{The constraint in (\ref{eq:myopicPromotion}) ensures that $p$ is myopically optimal for the seller, given the promotion policy, and letting $p$ be a variable ensures that $p$ maximizes consumer welfare among all myopically optimal prices (in line with the  myopic Bayesian pricing 
defined in Definition~\ref{def:MyopicPricingPolicy})}.

As consumer surplus is decreasing in the price set by the seller (Assumption~\ref{Assump:ConsSurplus}), the platform incentivizes a low price by solving \eqref{eq:myopicPromotion}. However, 
the seller can always choose to ignore the platform's promotion policy and 
set the revenue-maximizing price for patient consumers, denoted by $p^*$: 
\vspace{-0.1cm}
\begin{equation}\label{eq:pStar:OutsideOpt}
	p^* := \arg \max_{p\in \mathcal{P}} ~~ p\bar\rho_c(p). \vspace{-0.1cm}
\end{equation}
Note that $p^*$ is unique by Assumption~\ref{assump:Demand}, which requires $p\bar\rho_c(p)$ to be strictly concave. Denote the expected fraction of impatient customers as a function of the posterior belief $\mu_1$ by $\bar{\phi}(\mu_1):= \phi_L + (\phi_H-\phi_L)\mu_1$.
Given belief $\mu_1$,
the probability of a consumer being patient is $1-\bar{\phi}(\mu_1),$ and thus the seller's maximum expected payoff from selling 
only to patient consumers is $(1-\bar{\phi}(\mu_1))p^*\bar{\rho}_c(p^*)$.
Thus, to incentivize the seller to set a price lower than $p^*$, the platform must promote the seller with sufficiently high probability so that the seller's loss in revenue from patient consumers is, at least, made up for by revenue from impatient consumers.
To illustrate this further, we consider the setting of Example~\ref{example:UniformDemand} (put forth in~\S\ref{sec:Model}) and solve for an optimal myopic promotion policy below (see Appendix~\ref{App:AnalysisOfDemandModel} for a detailed analysis).

\begin{example}[Uniform WtP: Myopic Promotion Policy]\label{example:MyopicPromotion}
	Consider the demand structure in  Example~\ref{example:UniformDemand}, and suppose that $a> 2b\left(1-\frac{b}{4}\right)$. Then, $p^*= \frac{1}{4}(2a-b^2)$, and a myopic promotion policy has $\alpha_t = \alpha$, where:\vspace{-0.1cm}
	$$\alpha(p,\phi,\mu_t) = \begin{cases}
		1, &\text{ if } p \leq p(\mu_t),\\
		0, &\text{ otherwise, }
	\end{cases} \qquad \textup{and} \qquad  p(\mu_t) =\frac{1}{4}\left[ 2a-b^2(1-\bar{\phi}(\mu_t))-\sqrt{\bar{\phi}(\mu_t)}\sqrt{4a^2 -b^4(1-\bar{\phi}(\mu_t)) } \right].$$
\end{example}
\vspace{-0.25cm}

With this policy, the platform sets a target price $p(\mu_t)$ and, at time $t$, promotes the seller with probability~1 if his posted price does not exceed this target. As $p(\mu_t)<p^*$, the value of being promoted incentivizes the seller to price below $p^*$, and this incentive is stronger when the seller expects the fraction of impatient consumers 
to be higher, as reflected by the target price~$p(\mu_t)$ being a decreasing function of the belief $\mu_t$.

However, this policy results in sales observations that allow the seller to learn the true value of $\phi$ asymptotically. Thus, the prices posted by the seller eventually approach $p(1)$ if $\phi = \phi_H$, or $p(0)$ if $\phi = \phi_L$, which makes the long-run average consumer surplus converge to that achievable under truthful revelation.

\begin{proposition}[Payoff Equivalence of Myopic Promotions] \label{prop: LearningUnderMyopicPromo}
	Suppose that for any price $p\in(\inf(\mathcal{P}),p^*]$ one has $\bar{\rho}_0(p)>\bar{\rho}_c(p)$, and that $\bar{W}_0(p^*)~\geq~\bar{W}_{\text{out}}$. Suppose that the platform adopts a myopic promotion policy~$\hat{\bA}$. Then, for any signaling mechanism $\sigma \in \Sigma$, the resulting long-run average expected consumer surplus converges to the maximum surplus under truthful revelation, i.e., for all $\mu \in [0,1]$,\vspace{-0.1cm}
	\begin{equation*}
		\lim_{T\rightarrow \infty} ~~\frac{1}{T}W^{\hat{\bA},\sigma,\boldsymbol{\pi^*}}_T(\mu) =  W^{truth}(\mu): = \mu W^M(1)+(1-\mu)W^M(0).
	\end{equation*}
\end{proposition}
\vspace{-0.25cm}

Note that Proposition~\ref{prop: LearningUnderMyopicPromo} holds for any signaling mechanism. It implies that under natural conditions, if promotion decisions are myopic, efforts to optimize the initial information disclosed 
about the fraction~$\phi$ cannot increase the long-run average surplus, as the seller will learn the value of $\phi$ over time from sales~observations.

Nonetheless, even if the platform follows a myopic promotion policy, it may benefit from concealing information through the signaling mechanism if the time horizon is short. To illustrate this for $T=1$, Figure~\ref{Fig:Comp:OnePeriod} compares the consumer surplus generated by myopic promotions, along with
two natural signaling mechanisms: $(i)$ \emph{truthful}, that is, $\sigma^T(\phi) = \phi$, for~$\phi \in \{\phi_L,\phi_H\}$; and $(ii)$ \emph{uninformative}, such as setting $\sigma^U(\phi) = \phi_L$ for all $\phi \in \{\phi_L,\phi_H\}$. In general, the expected consumer surplus generated by a truthful signal may be larger or smaller than the one generated by an uninformative signal (depending on the form of $W$ and $\rho$); Figure~\ref{Fig:Comp:OnePeriod} depicts an instance in  the setting described in  Examples~\ref{example:UniformDemand} and~\ref{example:MyopicPromotion}, where concealing information is valuable, and
 revealing no information can generate short-term consumer surplus that is 5\% higher than that  
 generated by a truthful signal.
 \begin{figure}[!ht]
 	\centering
 	\setlength\figureheight{4.5cm}
 	\setlength\figurewidth{.5\textwidth}
 	\begin{tikzpicture}[trim axis left, trim axis right]
  \tikzstyle{every node}=[font=\small]
  \begin{axis}[
  clip=false,
  xlabel={$\mu_0$},
  xlabel style={align=center,at={(axis description cs:.5,-.17)}},
  ylabel={Average\\Consumer\\Surplus},
  ylabel style={rotate=-90,align = center},
  xmin=-0.05, xmax=1.05,
  width=\figurewidth,
  height=\figureheight,
  tick pos=left,
  x grid style={lightgray!92.026143790849673!black},
  y grid style={lightgray!92.026143790849673!black},
  yticklabel style={ /pgf/number format/fixed, /pgf/number format/precision=3},
  scaled y ticks=false
  ]
  \node[align=center] at (axis cs:0.5, .061){\scriptsize$(a = .4,~b=.2,~\phi_L = .02,~\phi_H = .3)$};


  \addplot[dashed,name path = NoInfo, thick,color= BrickRed,mark=none] table [x=belief, y=Myopic, col sep=comma] {Plots/uniformDemandExample.csv};
  \addplot[dashed,name path = NoInfo, thick,color=blue,mark=none] table [x=belief, y=Truth, col sep=comma] {Plots/uniformDemandExample.csv};


  \node[anchor=south,align=center] (sourceMax) at (axis  cs:.24,.054){\footnotesize Uninformative: $\sigma^U$ \\{\scriptsize $T=1$}};
  \node (destinationMax) at (axis cs:0.25, .051){};
  \draw[->](sourceMax)--(destinationMax);

  \node[anchor=south,align=center] (sourceFull) at (axis  cs:.7,.045){\footnotesize Truthful: $\sigma^T$};
  \node (destinationFull) at (axis cs:0.5, .052){};
  \draw[->](sourceFull)--(destinationFull);

 
  \end{axis}
  \end{tikzpicture}\vspace{-0.2cm}
 	\caption{{Comparison of the expected consumer surplus associated with a myopic promotion policy and either an informative or uninformative signaling mechanism, in the setting described in Examples~\ref{example:UniformDemand} and~\ref{example:MyopicPromotion}, with $T=1$}.}\vspace{-0.5cm}
 	\label{Fig:Comp:OnePeriod}
 \end{figure}



To conclude, when designing its promotion policy, the platform faces a tradeoff between increasing the consumer surplus in the current period and limiting the information contained in sales observations, which in turn impacts consumer surplus in future periods. When the promotion policy is myopic, the seller is able to learn the value of $\phi$ over time. Thus, while optimizing the initial signal provided about $\phi$ might be valuable in the short term, it is ineffective at increasing the long-run average consumer surplus.

In the next section we introduce a new class of platform policies designed to control the information collected by the seller over time. We will show that, when carefully designed, these policies guarantee long-run average consumer surplus that is higher than what is achievable by truthfully revealing the fraction $\phi$, and that, in fact, they are long-run average optimal.

\section{Confounding Promotions: Performance and Design}\label{sec:LongRunAverageOptimalConsumerSurplus}\vspace{-0.1cm}

In the previous section we  identified a tradeoff between increasing consumer surplus in the current  period (through incentivizing low prices) and limiting the information contained in sales observations, which may drive high consumer surplus in future periods. 
A class of policies that is key to  balancing this tradeoff consists of policies that \emph{confound} the seller's belief, in the sense of preventing him from learning the value of $\phi$ through sales observations. We next define these policies, 
leverage them to 
 characterize the achievable long-run average consumer surplus when the seller makes myopic pricing decisions, and, in addition,  establish that they are long-run average optimal in that setting. We then turn to develop a class of simple confounding policies that are tractable, practical to implement, and 
 still 
  achieve long-run average optimality.

\subsection{Confounding Promotion Policies}\label{sec:ConfoundingPromotionPolicies}\vspace{-0.1cm}



Formally, we define confounding promotion policies as follows.

\begin{definition}[Confounding Promotion Policies]\label{def:confounding}
    Suppose that the seller uses the myopic pricing policy,~$\boldsymbol{\pi^*}$. For each belief $\mu \in [0,1]$, define the set of confounding promotion policies $\mathcal{A}^C(\mu) \subset \mathcal{A}^M$ as those that keep the seller's belief constant 
    throughout periods $t=1,\ldots,T$. That is, $\bA \in \mathcal{A}^C(\mu)$, if and only if for all~$t=1,...,T$, one has $\Prob(\mu_{t+1}=\mu|\mu_t=\mu,\boldsymbol{\pi^*},\bA) = 1.$
\end{definition}
\vspace{-0.25cm}

Definition~\ref{def:confounding} encompasses two ways for sales observations to contain no new information about the true fraction of impatient consumers, $\phi$. The first one is trivial: if the seller knows the true value of $\phi$ with certainty, that is,~$\mu~\in \{0,1\}$, then sales observations do not affect his belief, and one has $\mathcal{A}^C(\mu) = \mathcal{A}$. On the other hand, if $\mu \in (0,1)$, then the platform may prevent sales observations  from conveying information,  by ensuring that the probability of a sale is independent of $\phi$.
To do so, the platform must design $\bA$ so that in each period $t$ and for every belief $\mu\in[0,1]$, at the price $p$ set by the seller given $\mu_t=\mu$ and $\bA$, one has  \vspace{-0.2cm}
\begin{equation}\label{eq:EqualityOfSalesProb}
    \begin{split}
        \Prob\left(y_t=1|\phi=\phi_H,\mu_t=\mu,p_t = p,\alpha_t\right)  &= \Prob\left(y_t=1|\phi=\phi_L,\mu_t=\mu,p_t = p,\alpha_t\right).\vspace{-0.2cm}
    \end{split}
\end{equation}
In any given period, a patient or impatient consumer may arrive. A patient consumer arrives and purchases from the seller with probability $(1-\phi)\bar\rho_c(p)$; an impatient consumer arrives and purchases from the seller with probability~$\phi\alpha(p,\phi,h)\bar\rho_0(p)$. Therefore, one may rewrite~\eqref{eq:EqualityOfSalesProb} as \vspace{-0.05cm}
\begin{equation}\label{eq:EqualityOfSalesProb_demand}
	\phi_H\alpha(p,\phi_H,h)\bar\rho_0(p) +(1-\phi_H)\bar\rho_c(p) = \phi_L\alpha(p,\phi_L,h)\bar\rho_0(p) +(1-\phi_L)\bar\rho_c(p). \vspace{-0.05cm}
\end{equation}
An essential feature of our model is that regardless of the seller's belief, the platform can design a confounding policy that eliminates the informational content of sales observations. This is formalized next in Proposition~\ref{prop: ConfoundingPolicyExistence}, which establishes that the platform can always design a promotion policy that satisfies~\eqref{eq:EqualityOfSalesProb}--\eqref{eq:EqualityOfSalesProb_demand} in all periods and for any belief $\mu$, and thus implies that the platform has the ability to keep the seller's belief constant throughout the horizon.
\begin{proposition}[Existence of Confounding Promotion Policies] \label{prop: ConfoundingPolicyExistence}
		$A^C(\mu)$ is non-empty for all $\mu\in[0,1]$. That is, for any seller's belief the platform can construct a confounding promotion policy.
	\end{proposition}
\vspace{-0.25cm}

In contrast to the myopic promotions described in \S\ref{subsec:InsufficiencyOfTruthfulDisclosure}, a confounding promotion policy may not maximize instantaneous consumer surplus, as it may weaken incentives for a low price and/or commit to divert impatient consumers from the product that generates the largest expected consumer surplus. For instance, the myopic promotion policy in Example~\ref{example:MyopicPromotion} is not confounding as \eqref{eq:EqualityOfSalesProb_demand} does not hold at the myopically optimal price $p=p(\mu)$, since $\alpha(p,\phi_L,\mu)=\alpha(p,\phi_H,\mu) = 1$ and $\bar\rho_c(p) < \bar\rho_0(p)$. To confound the seller, one must decrease the sales probability when $\phi = \phi_H$ by decreasing $\alpha(p,\phi_H,\mu)$ and/or incentivize a price where the difference between $\bar\rho_c(p)$ and $ \bar\rho_0(p)$ is smaller. To illustrate this observation, we refer back to the demand and welfare structure in  Examples~\ref{example:UniformDemand} and~\ref{example:MyopicPromotion}, and provide a 
confounding promotion policy in this setting.

\begin{example}[Uniform WtP: 
Confounding Promotion Policy]\label{example:ConfoundingPromotion}
	Consider the setting in Example~\ref{example:UniformDemand}, and let $a> 2b\left(1-\frac{b}{4}\right)$. For $\mu \in (0,1)$, an optimal confounding promotion policy has $\alpha_t = \alpha^C$,  where \vspace{-0.1cm}
	\begin{equation*}
		\alpha^C(p,\phi_L,\mu) = \begin{cases}
			1, &\text{ if } p = p^C(\mu)\\
			0, &\text{ otherwise, }
		\end{cases}, \qquad \alpha^C(p,\phi_H,\mu) = \begin{cases}
			\frac{\alpha-p^C(\mu) - b^2/2}{ \alpha-p^C(\mu)}\left(\frac{\phi_H-\phi_L}{\phi_H}\right) + \frac{\phi_L}{\phi_H}, &\text{ if } p = p^C(\mu)\\
			0, &\text{ otherwise, }
		\end{cases}\vspace{-0.1cm}
	\end{equation*}
	where $p^C(\mu)$, defined by
	$p^C(\mu) =\frac{1}{4}\left(2a-b^2(1- \phi_L) -\sqrt{(2a-b^2)^2(\phi_H-\phi_L)\mu+\phi_L(4a^2-b^4)+b^4\phi_L^2} \right)$, is the myopically optimal price set by the seller given belief $\mu$ and the promotion policy~$\alpha^C$. For $\mu\in\{0,1\}$, one has $\alpha^C(p,\phi,\mu)= \alpha(p,\phi,\mu)$ and $p^C(\mu) = p(\mu)$,  where these are described as in Example~\ref{example:MyopicPromotion}.
\end{example}
\vspace{-0.25cm}

Recall that the myopic promotion policy in Example~\ref{example:MyopicPromotion} promotes the seller with probability 1 provided that he matches the target price. By  contrast, in the confounding policy given in Example~\ref{example:ConfoundingPromotion}, the promotion probability of the seller 
is reduced when $\phi=\phi_H$. Thus, some impatient consumers are not shown the seller's product and are instead shown an alternative with lower expected consumer surplus. Moreover, the target price $p^C(\mu)$ of the above  confounding policy is higher
 than the one in Example~\ref{example:MyopicPromotion}, which lowers the expected surplus for all consumers. Thus, in order to confound the seller, the platform weakens his incentive to set low prices, and diverts a fraction of consumers away from their preferred product. Finally, note that at belief~$\mu_1 \in \{0,1\}$ the optimal confounding policy coincides with the myopic promotion policy since, as the seller's belief never updates, the confounding constraint is satisfied.  


\subsection{Long-Run Average Optimal Consumer Surplus}\label{sec:longrunoptimalsurplus}\vspace{-0.1cm}
In this section, we leverage the notion of confounding policies to characterize the long-run average optimal consumer surplus. For that purpose, we define the maximum consumer surplus generated by a confounding promotion policy given a posterior belief $\mu_1 \in [0,1]$ as
\begin{equation}\label{PlatformProblem:SinglePeriodConfounding}
	\begin{split}
		W^C(\mu_1):= \max_{\boldsymbol{\alpha} \in \mathcal{A}^C(\mu_1)}& ~~
		\frac{1}{T}\E\left(\sum_{t=1}^T  W(p_t,a_t,\psi_t)\middle| \boldsymbol{\alpha},\boldsymbol{\pi^*},\mu_1 \right).
	\end{split}
\end{equation}
Note that $W^C(\mu_1)$ is independent of $T$, as under any confounding promotion policy the seller's belief is constant by construction and, as the seller prices myopically, the maximum consumer surplus is the same in each period. 
We note that the confounding policy given in Example~\ref{example:ConfoundingPromotion} is derived by solving the single-period version of (\ref{PlatformProblem:SinglePeriodConfounding}), and in fact maximizes consumer surplus within the set of  confounding promotion policies.

In addition, for any function $f: \R\rightarrow \R$, define $co(f)$ as the \emph{concavification} of\footnote{This function appears often in the information design literature when characterizing the optimal signaling mechanism and corresponding payoff (see e.g. \cite{aumann1995repeated} and \cite{Kamenica2011}).} $f$: \vspace{-0.2cm}
\begin{equation*}
	co(f)(\mu) := \sup \{z|(\mu,z) \in Conv(f)\},
	\vspace{-0.2cm}
\end{equation*}
where $Conv(f)$ denotes the convex hull of the set $\{ (x,t):t\leq f(x)\}$. The following key result characterizes the maximum long-run average consumer surplus that can be generated jointly by a signaling mechanism and a dynamic promotion policy, and establishes that it can be attained by a confounding promotion policy.

\begin{theorem}[Characterization of Long-Run Average Optimal Consumer Surplus] \label{thm:LongRunAverageOptimalConsumerSurplus} Let $W^C(\mu)$ be defined as in \eqref{PlatformProblem:SinglePeriodConfounding}. For all $\mu \in [0,1]$,
    $$\lim_{T\rightarrow \infty}\sup_{\substack{\alpha \in\mathcal{A},\\\sigma \in \Sigma}} ~~\frac{1}{T}W^{\bA,\sigma,\boldsymbol{\pi^*}}_T(\mu) = co(W^C)(\mu).$$
{Furthermore, for any fixed $T$, there exists a signaling mechanism $\sigma$ and a confounding promotion policy $\bA$ that generate an expected average consumer surplus of $co(W^C)(\mu_0)$. 
}
\end{theorem}

The characterization in Theorem~\ref{thm:LongRunAverageOptimalConsumerSurplus} follows since for any promotion policy, the seller's belief $\mu_t$ asymptotically converges to a limit belief as the number of periods grows large, and the long-run average consumer surplus is effectively determined by the expected consumer surplus at that limit belief. Thus the platform's challenge is to design a policy that ensures that the distribution of the seller's limit beliefs is optimal. Figure~\ref{Fig:Comp:LongRunAverage} shows  that for an instance of the demand model in Example~\ref{example:UniformDemand} (and used in Figure~\ref{Fig:Comp:OnePeriod}), the maximum 
long-run average consumer surplus is distinctively higher (up to~3.5\%) than what is achieved by truthful revelation.

\begin{figure}[!htbp]
	\centering
	\setlength\figureheight{4.8cm}
	\setlength\figurewidth{.5\textwidth}
	\begin{tikzpicture}[trim axis left, trim axis right]
    \tikzstyle{every node}=[font=\small]
    \begin{axis}[
    clip=false,
    xlabel={$\mu_0$},
    xlabel style={align=center,at={(axis description cs:.5,-.16)}},
    ylabel={Average\\Consumer\\Surplus},
    ylabel style={rotate=-90,align = center},
    xmin=-0.01, xmax=1.01,
    width=\figurewidth,
    height=\figureheight,
    tick pos=left,
    yticklabel style={ /pgf/number format/fixed, /pgf/number format/precision=3},
scaled y ticks=false
    ]
    \node[align=center] at (axis cs:0.5, .061){\scriptsize$(a = .4,~b=.2,~\phi_L = .02,~\phi_H = .3)$};
  
    \node[anchor=south,align=center] (sourceFull) at (axis  cs:.8,.048){\footnotesize Truthful: $\sigma^T$};
    \node (destinationFull) at (axis cs:0.67, .055){};
    \draw[->](sourceFull)--(destinationFull);

    \node[anchor=south,align=center] (sourceCoWC) at (axis  cs:.2,.0555){\footnotesize $co(W^C)(\mu)$};
    \node (destinationCoWC) at (axis cs:0.35, .051){};
    \draw[->](sourceCoWC)--(destinationCoWC);

    
    \addplot [line width=0.4mm,black,thick,name path = ConcaveConfounding]table [x=belief, y=ConcaveConfounding, col sep=comma] {Plots/uniformDemandExample.csv};
  
    \addplot [line width=0.4mm,blue,dashed,thick,name path = FullInfo] table [x=belief, y=Truth, col sep=comma] {Plots/uniformDemandExample.csv};
   
    \end{axis}
    \end{tikzpicture}\vspace{-0.2cm}
	\caption{Long-run average optimal policy in the setting described in Example~\ref{example:UniformDemand}.}\vspace{-0.3cm}
	\label{Fig:Comp:LongRunAverage}
\end{figure}

The characterization in Theorem~\ref{thm:LongRunAverageOptimalConsumerSurplus} can also be understood as follows. Myopic pricing decisions are based on the current belief and the platform's promotion policy, and therefore the platform has two available levers to control the seller's belief process $\{\mu_t\}$, in order to induce lower prices: (i) the signaling mechanism $\sigma$, which determines the \textit{starting point} of the belief process (i.e., the value of $\mu_1$); and (ii) the promotion policy $\bA$, which controls the \textit{drift} of the seller's belief process (determined by Bayes' rule given the platform's promotion policy). From this perspective, the platform's problem of jointly designing its signaling mechanism and promotion policy to maximize the long-run average expected consumer surplus can be viewed as selecting the starting point and the drift of the belief process, subject to incentive compatibility and feasibility constraints. 
Theorem~\ref{thm:LongRunAverageOptimalConsumerSurplus} establishes that in the long run, it is optimal for the platform to deploy a confounding promotion policy that induces zero drift on the belief process, thus keeping the seller's belief constant. Taking this into account,
 the platform's problem  reduces to setting the distribution of the starting point of the belief process via its signaling mechanism. Following the information design literature (e.g., \citealt{Kamenica2011}), the platform can optimally design such a mechanism to achieve an expected payoff of $co(W^C)(\mu_0)$.


The key ideas of the proof of Theorem~\ref{thm:LongRunAverageOptimalConsumerSurplus} are as follows.
To establish an upper bound on the long-run average consumer surplus, we first show that in periods in which the platform's payoff is high, sales observations carry significant information value to the seller. Formally, for any $\epsilon>0$, in periods where the platform's payoff exceeds $co(W^C)(\mu_t) + \epsilon$, one has $|\Prob(y_t=1|\phi=\phi_H,\alpha_t) - \Prob(y_t=1|\phi=\phi_L,\alpha_t)|>\delta$ for some $\delta>0$ (Lemma~\ref{lemma:delta_epsilonRelationCoU}, Appendix~\ref{sec:proofs}).
Next, we show that as the number of such periods increases, the seller's belief must converge to the true value of $\phi$ exponentially fast (Lemma~\ref{lemma:generalizedHarrison}, Appendix~\ref{sec:proofs}), 
resulting in a long-run average payoff that approaches the one under truthful revelation of $\phi$, in line with the dynamics established in Proposition~\ref{prop: LearningUnderMyopicPromo}.
We then establish the upper bound by showing that the platform-optimal outcome 
when the seller asymptotically learns the value of $\phi$ is at most $co(W^C)(\mu_0)$. 
Finally, we prove the result by constructing a policy $(\bA,\sigma)$ for which $\frac{1}{T}W^{\bA,\sigma,\boldsymbol{\pi^*}}_T(\mu_0) = co(W^C)(\mu_0)$ for all $T\geq 1$. This construction will be described in more detail in the following subsection.


\subsection{Designing Optimal Simple Confounding Policies} \label{sec: OptSimplePolicies}

%
Theorem~\ref{thm:LongRunAverageOptimalConsumerSurplus} implies that in order to achieve the maximal long-run average consumer surplus, it suffices to construct an optimal confounding promotion policy and then determine an optimal signaling mechanism. We next develop a simple procedure to do so by following three steps: (i) constructing an optimal confounding policy that generates consumer surplus of $W^C(\mu)$ for a given belief~$\mu$, as defined by~\eqref{PlatformProblem:SinglePeriodConfounding} in \S\ref{sec:longrunoptimalsurplus}; (ii) determining the form of $co(W^C(\mu))$ based on the characterization of $W^C(\mu)$; and (iii) obtaining an optimal signaling mechanism given the prior belief and the form of $co(W^C(\mu))$. To construct an optimal confounding policy in step (i), we first introduce the notion of simple confounding policies. 

\begin{definition}[Simple Confounding Promotion Policies]\label{def:simplePromotionPolicies}
The set of simple confounding promotion policies~$\mathcal{A}^S \subset \mathcal{A}^M$ consists of all policies that are static, single price, and confounding, that is, policies where~$\boldsymbol{\alpha} \in \mathcal{A}^C(\mu)\cap\mathcal{A}^P$ for all $\mu \in [0,1]$ and $\boldsymbol{\alpha}$ is static  (that is,  $\alpha_1(p,\phi,\mu) = \alpha_t(p,\phi,\mu)$ for all $t =2,\dots,T$, $p\in \mathcal{P},$ $\phi \in \{\phi_L,\phi_H\}$, $\mu\in[0,1]$).
\end{definition}
\vspace{-0.25cm}

The proof of Proposition~\ref{prop: ConfoundingPolicyExistence} establishes that $\mathcal{A}^S$ is non-empty, and in particular that the promotion policy in Example~\ref{example:ConfoundingPromotion} is a simple confounding policy. Considering only simple confounding promotion policies reduces the design problem to a subclass of policies that are intuitive, tractable, and, as the proof of Theorem~\ref{thm:LongRunAverageOptimalConsumerSurplus} shows, sufficient for achieving long-run average optimality.

Furthermore, under the class $\mathcal{A}^S$ the platform's design problem is static after the seller updates his belief based on the initial signal sent by the platform. The belief remains constant thereafter and the promotion policy does not change, as the policy is confounding and static. The seller's posted price does not change across time periods, since the policy is also single price. Thus, to construct an optimal confounding policy that generates a  consumer surplus of $W^C(\mu)$ for a given belief $\mu$, one only needs to consider the problem with~$T=1$ and simple confounding promotion policies. That is, for any belief $\mu \in (0,1)$, the platform solves:
\begin{equation} \label{eq: ConfoundingProblemSimple}
\begin{split}
	W^C(\mu):= \max_{\substack{\alpha_{\phi_H},\alpha_{\phi_L} \in [0,1],\\p \in \mathcal{P}}}& ~~ \E_{\phi}\left(\phi \alpha_{\phi} \bar{W}_0(p) +\phi(1- \alpha_{\phi} )\bar{W}_{\text{out}}  + (1-\phi)\bar{W}_c(p)|\mu\right) \\
	\text{s.t.}&~~ p\bar{\rho}_0(p)(\phi_L\alpha_{\phi_L}(1-\mu) +\phi_H\alpha_{\phi_H}\mu) +  p\bar{\rho}_c(p)(1-\phi_L-\mu(\phi_H-\phi_L)) \geq \\
	&\quad \quad \quad p^*\bar{\rho}_c(p^*)(1-\phi_L-\mu(\phi_H-\phi_L)),\\
	&~~ \phi_H\alpha_{\phi_H}\bar{\rho}_0(p) + (1-\phi_H)\bar{\rho}_c(p) = \phi_L\alpha_{\phi_L}\bar{\rho}_0(p) + (1-\phi_L)\bar{\rho}_c(p).
\end{split}
\end{equation}
Note that focusing on simple confounding promotion policies simplifies the design problem considerably since, as these are single-price policies, one only needs to optimize over the target price $p$ and the probability of promotion at that price for each realization of $\phi$ (that is, $\alpha_{\phi_H}$ and $\alpha_{\phi_L}$), instead of over a general function~$\alpha$ (e.g., as in \eqref{eq:myopicPromotion} in \S\ref{subsec:InsufficiencyOfTruthfulDisclosure}).
The first constraint in~\eqref{eq: ConfoundingProblemSimple} ensures that the selected price is myopically optimal for the seller, and is simplified compared to \eqref{eq:myopicPromotion} as one only needs to compare the target price with $p^*$ in terms of expected revenue. The second constraint, which ensures that the policy is confounding, fully defines~$\alpha_{\phi_H}$ given $\alpha_{\phi_L}$ and $p$, and therefore one only needs to optimize over these two variables. For many demand models, one may further establish that $\alpha_{\phi_L}$ is fully defined given a price $p$, which allows one to optimize only over the price~$p$. In summary, the formulation in~\eqref{eq: ConfoundingProblemSimple} reduces the construction of an optimal confounding policy to solving a static optimization problem with  three decision variables.

Having solved for the optimal confounding policy at every belief (step (i)), it remains to find the optimal signaling mechanism by executing steps (ii) and (iii), namely, determining the form of $co(W^C(\mu))$ based on the characterization of $W^C(\mu)$, and obtaining an optimal signaling mechanism given the prior $\mu_0$ and the form of $co(W^C(\mu))$. Whenever it is possible to solve~\eqref{eq: ConfoundingProblemSimple} analytically, as in Example~\ref{example:ConfoundingPromotion}, finding the optimal signal is typically straightforward as $W^C(\mu)$ is continuous on $(0,1)$ and oftentimes concave. However, when it is not possible to solve this problem analytically, one may still solve for the optimal promotion policy numerically (e.g., for a grid of beliefs in $[0,1]$), approximate $co(W^C(\mu))$ based on these values, and then solve for the optimal signal. In addition, by Proposition~\ref{lemma:RevenueEquivalence}, it suffices to consider a reduced signal set~$\mathcal{S}=\{\phi_L,\phi_H\}$ to perform this final step. For further details, see the computation of policies in Appendix~\ref{app:DesigningSimplePolicies}.

\subsubsection{Example: Evaluation of Simple Confounding Platform Policies}\label{sec: NumericalEvals}\vspace{-0.1cm}


{We now evaluate numerically the optimal simple confounding policies we described, in the context of the demand model in  Example~\ref{example:UniformDemand}.} We note that we focus on this concrete setting here for the sake of consistency with previous analysis; while the precise outcomes clearly depend on the specific demand and consumer welfare that are assumed, the phenomena that we illustrate next are broad and hold across many structures.

Define the \emph{relative gain} in consumer surplus from using the optimal simple confounding policy compared to the optimal truthful policy:
$$
RG(\mu):=\frac{co(W^C)(\mu)-W^{truth}(\mu)}{W^{truth}(\mu)}.
$$
The upper panels of Figure~\ref{Fig:ValueGain} depict the relative gain compared to the optimal truthful policy for two different parametric specifications of the demand model in  Example~\ref{example:UniformDemand}. For each of these combinations, we calculate the relative gain at a grid of beliefs $\mu\in[0,1]$ and plot the maximum, average, and minimum values. One may observe that the gain that is captured by the optimal confounding policy relative to truthful revelation can be significant, and is larger when the seller's product is superior to the outside option.  


\begin{figure}[t!]
\centering
\begin{subfigure}{.45\linewidth}
	\centering
	\setlength\figureheight{4.5cm}
	\setlength\figurewidth{.9\textwidth}
	    \begin{tikzpicture}[ trim axis left, trim axis right]
        \begin{axis}
            [
        clip=false,
        xlabel={$\phi_H$ },
        xlabel style={align=center},
        ylabel = {$RG(\mu)$ },
        ylabel style={rotate=-90,align = center, font=\small,at={(axis description cs:-.15,.5)}},
        yticklabel={\pgfmathparse{100*\tick}\pgfmathprintnumber{\pgfmathresult}\%},,
        xmin=-0.05, xmax=1.05,
        ymin=-0.005, ymax=.1,
        width=\figurewidth,
        height=\figureheight,
        tick align=outside,
        tick pos=left
        ]

        \node[align=center] at (axis  cs:.5,.093){\scriptsize $(a=.4,~b=.3,~\phi_L = .01)$};
        \node[align=center] (source01Max) at (axis  cs:.35,.068){\scriptsize Maximum};
        \node (destination01Max) at (axis cs:0.4, .038){};
        \draw[->](source01Max)--(destination01Max);
        
        \node[align=center] (source01Mean) at (axis  cs:.8,.02){\scriptsize Average};
        \node (destination01Mean) at (axis cs:0.6, .028){};
        \draw[->](source01Mean)--(destination01Mean);
        
        \node[align=left] (source01Minimum) at (axis  cs:.5,.013){\scriptsize Minimum};
        \node (destination01Minimum) at (axis cs:0.55, .0){};
        \draw[->](source01Minimum)--(destination01Minimum);

        
        
        \addplot+[ discard if not={phiL}{0.01},blue,mark size = 1,thick] table [x=phiH, y=amax, col sep=comma] {python/NumericAnalysis/valueOfConfoundingalpha4.csv};
        \addplot+[ discard if not={phiL}{0.01}, blue,dashed, no marks,thick] table [x=phiH, y=mean, col sep=comma] {python/NumericAnalysis/valueOfConfoundingalpha4.csv};

        \addplot+[ discard if not={phiL}{0.01}, blue,dotted, no marks,thick] table [x=phiH, y=amin, col sep=comma] {python/NumericAnalysis/valueOfConfoundingalpha4.csv};


        \end{axis};
        \end{tikzpicture}

     \vspace{-0.0cm}
	\label{Fig:ValueGain:A4}
\end{subfigure}
\begin{subfigure}{.45\linewidth}
	\centering
	\setlength\figureheight{4.5cm}
	\setlength\figurewidth{.9\textwidth}
	    \begin{tikzpicture}[ trim axis left, trim axis right]
        \begin{axis}
            [
        clip=false,
        xlabel={$\phi_H$ },
        xlabel style={align=center},
        ylabel={$RG(\mu)$},
        ylabel style={rotate=-90,align = center, font=\small,at={(axis description cs:-.15,.5)}},
        yticklabel={\pgfmathparse{100*\tick}\pgfmathprintnumber{\pgfmathresult}\%},,
        xmin=-0.05, xmax=1.05,
        ymin=-0.005, ymax=.1,
        width=\figurewidth,
        height=\figureheight,
        tick align=outside,
        tick pos=left
        ]

        \node[align=center] at (axis  cs:.5,.093){\scriptsize $(a=.6,~b=.3,~\phi_L = .01)$};
        \node[align=center] (source01Max) at (axis  cs:.35,.07){\scriptsize Maximum};
        \node (destination01Max) at (axis cs:0.4, .051){};
        \draw[->](source01Max)--(destination01Max);
        
        \node[align=center] (source01Mean) at (axis  cs:.8,.02){\scriptsize Average};
        \node (destination01Mean) at (axis cs:0.6, .038){};
        \draw[->](source01Mean)--(destination01Mean);
        
        \node[align=left] (source01Minimum) at (axis  cs:.5,.013){\scriptsize Minimum};
        \node (destination01Minimum) at (axis cs:0.55, .0){};
        \draw[->](source01Minimum)--(destination01Minimum);

        \addplot+[ discard if not={phiL}{0.01},blue,mark size = 1,thick] table [x=phiH, y=amax, col sep=comma] {python/NumericAnalysis/valueOfConfoundingalpha6.csv};
        \addplot+[ discard if not={phiL}{0.01}, blue,dashed,no marks,thick] table [x=phiH, y=mean, col sep=comma] {python/NumericAnalysis/valueOfConfoundingalpha6.csv};
        \addplot+[ discard if not={phiL}{0.01}, blue,dotted, no marks,thick] table [x=phiH, y=amin, col sep=comma] {python/NumericAnalysis/valueOfConfoundingalpha6.csv};


        \end{axis};
        \end{tikzpicture}

     
	\label{Fig:ValueGain:A6}
\end{subfigure}
\par\smallskip
\begin{subfigure}{.45\linewidth}
	\centering
	\setlength\figureheight{4.5cm}
	\setlength\figurewidth{.9\textwidth}
	    \begin{tikzpicture}[ trim axis left, trim axis right]
        \begin{axis}
            [
        clip=false,
        xlabel={$\phi_H$ },
        xlabel style={align=center},
        ylabel={$CCS(\mu)$},
        ylabel style={rotate=-90,align = center,at={(axis description cs:-.18,.5)},font=\small },
        yticklabel={\pgfmathparse{100*\tick}\pgfmathprintnumber{\pgfmathresult}\%},,
        xmin=-0.05, xmax=1.05,
        ymin=.96, ymax=1.007,
        width=\figurewidth,
        height=\figureheight,
        tick align=outside,
        tick pos=left
        ]
        
        \node[align=center] at (axis  cs:.5,1.003){\scriptsize {$(a=.4,b=.3,~\phi_L=.01)$}};;

        \node[align=center] (sourceMax) at (axis  cs:.35,.994){\scriptsize Maximum};
        \node (destinationMax) at (axis cs:0.5, 1){};
        \draw[->](sourceMax)--(destinationMax);

        \node[align=center] (sourceMin) at (axis  cs:.2,.97){\scriptsize Minimum};
        \node (destinationMin) at (axis cs:0.4, .98){};
        \draw[->](sourceMin)--(destinationMin);
        
        \node[align=center] (sourceAvg) at (axis  cs:.8,.992){\scriptsize Average};
        \node (destinationAvg) at (axis cs:0.7, .985){};
        \draw[->](sourceAvg)--(destinationAvg);
        
        \addplot+[ discard if not={phiL}{0.01}, blue,mark size = 1,thick] table [x=phiH, y=amin, col sep=comma] {python/NumericAnalysis/shortTermLossalpha4.csv};

        \addplot+[ discard if not={phiL}{0.01}, blue,no marks,dashed] table [x=phiH, y=mean, col sep=comma] {python/NumericAnalysis/shortTermLossalpha4.csv};

        \addplot+[ discard if not={phiL}{0.01}, blue,no marks,dotted] table [x=phiH, y=amax, col sep=comma] {python/NumericAnalysis/shortTermLossalpha4.csv};

        \end{axis};
        \end{tikzpicture}

     \vspace{-0.0cm}
	\label{Fig:PercentOfMyopic:A4}
\end{subfigure}
\begin{subfigure}{.45\linewidth}
	\centering
	\setlength\figureheight{4.5cm}
	\setlength\figurewidth{.9\textwidth}
	    \begin{tikzpicture}[ trim axis left, trim axis right]
        \begin{axis}
            [
        clip=false,
        xlabel={$\phi_H$ },
        xlabel style={align=center},
        ylabel={$CCS(\mu)$},
        ylabel style={rotate=-90,align = center,at={(axis description cs:-.18,.5)},font=\small },
        yticklabel={\pgfmathparse{100*\tick}\pgfmathprintnumber{\pgfmathresult}\%},,
        xmin=-0.05, xmax=1.05,
        ymin=.96, ymax=1.007,
        width=\figurewidth,
        height=\figureheight,
        tick align=outside,
        tick pos=left
        ]

        \node[align=center] at (axis  cs:.5,1.003){\scriptsize {$(a=.6,b=.3,~\phi_L=.01)$}};;

        \node[align=center] (sourceMax) at (axis  cs:.35,.994){\scriptsize Maximum};
        \node (destinationMax) at (axis cs:0.5, 1){};
        \draw[->](sourceMax)--(destinationMax);
        
        \node[align=center] (sourceMin) at (axis  cs:.2,.97){\scriptsize Minimum};
        \node (destinationMin) at (axis cs:0.4, .975){};
        \draw[->](sourceMin)--(destinationMin);
        
        \node[align=center] (sourceAvg) at (axis  cs:.8,.992){\scriptsize Average};
        \node (destinationAvg) at (axis cs:0.7, .982){};
        \draw[->](sourceAvg)--(destinationAvg);

        \addplot+[ discard if not={phiL}{0.01}, blue,mark size = 1,thick] table [x=phiH, y=amin, col sep=comma] {python/NumericAnalysis/shortTermLossalpha6.csv};

        \addplot+[ discard if not={phiL}{0.01}, blue,no marks,dashed] table [x=phiH, y=mean, col sep=comma] {python/NumericAnalysis/shortTermLossalpha6.csv};

        \addplot+[ discard if not={phiL}{0.01}, blue,no marks,dotted] table [x=phiH, y=amax, col sep=comma] {python/NumericAnalysis/shortTermLossalpha6.csv};


        \end{axis};
        \end{tikzpicture}

     
	\label{Fig:PercentOfMyopic:A6}
\end{subfigure}\vspace{-0.2cm}
\caption{Long-run value and short-term loss from confounding. The plots in the upper panels depict three measures of the relative gain $RG(\mu)$ 
for a range of parametric specifications in  Example~\ref{example:UniformDemand}. The plots in the lower panels do so for the captured consumer surplus $CCS(\mu)$. 
For each specification, we show the maximum, average, and minimum (over a grid of $\mu \in [0,1]$) values of $RG(\mu)$ and $CCS(\mu)$. 
The parameters for the left-hand plots reflect a demand model where valuations for the seller's product and its alternative are comparable. In the right-hand plots the seller's product is more often preferred by consumers.}\vspace{-0.2cm}
\label{Fig:ValueGain}
\end{figure}

Moreover, the optimal simple confounding promotion policy can be nearly optimal even in the short run. Define the fraction of \emph{one-period consumer surplus} that can be \emph{captured} by a confounding policy out of the maximum one-period consumer surplus as\vspace{-0.1cm}
$$
CCS(\mu):=\frac{co(W^C)(\mu)}{W^{\max}(\mu)}.\vspace{-0.1cm}
$$
The lower panels of Figure~\ref{Fig:ValueGain} 
depict the maximum, average, and minimum values of the \emph{captured consumer surplus} $CCS(\mu)$ over a grid of beliefs $\mu \in [0,1]$ for two parametric specifications of the demand model in  Example~\ref{example:UniformDemand}. One may observe that the optimal simple confounding policy, even in the worst case, captures nearly $97\%$ of the maximum one-period surplus, and oftentimes achieves an even better performance. 


\subsection{The Value of Signaling} \label{sec: no_signaling}\vspace{-0.1cm}

The characterization of Theorem~\ref{thm:LongRunAverageOptimalConsumerSurplus} highlights that confounding promotion policies are key to maximizing the long-run average optimal consumer surplus, as they control the information collected by the seller. However, on their own, these policies do not guarantee long-run average optimality. Instead, they must be paired with a signaling mechanism that adjusts the seller's prior belief to one at which confounding is optimal (in expectation). Thus, for a range of prior beliefs, the platform can generate the optimal long-run average consumer surplus by revealing (some) information and then confounding the seller.

However, to do so, the platform must carefully design a signaling mechanism~$\sigma$ that impacts the seller's belief in a certain way, which might be infeasible in some contexts (e.g., it might  require communicating with the seller directly and in a randomized fashion). We next consider the setting where the platform has no ability to send an initial informative signal to the seller and may only optimize over its promotion policy. We find that even in the absence of an informative signaling mechanism, the platform may still achieve good performance in the long run. Formally, the next result establishes that the platform can still approach the maximum long-run average consumer welfare of $co(W^C)(\mu_0)$ by carefully designing its promotion policy.
	
	\begin{theorem}\label{prop: noinfodesign}
		Fix $\epsilon>0$, and suppose that the platform's signaling mechanism is uninformative (i.e., $\sigma = \sigma^U$). Then, given $\mu_0\in [0,1]$, there exists a promotion policy $\bA$ such that\vspace{-0.1cm}
		\begin{equation*}
			\lim_{T\rightarrow \infty} ~~\frac{1}{T}W_T^{\bA,\sigma^U,\boldsymbol{\pi^*}}(\mu_0) \geq co(W^C)(\mu_0) -\epsilon.
		\end{equation*}
	\end{theorem}
	\vspace{-0.25cm}

The proof of Theorem~\ref{prop: noinfodesign} is provided in Appendix~\ref{app: NoInfoDesign}. Importantly, the proof shows that without initial information signaling, the platform cannot rely anymore on confounding promotion policies, and must instead deploy a more complex promotion policy to approximate the long-run average payoff of $co(W^C)(\mu_0)$. Intuitively, if one removes the platform's ability to influence the starting point of the seller's belief process, setting the drift of the belief process to zero (as in Theorem~\ref{thm:LongRunAverageOptimalConsumerSurplus}) is no longer optimal. Instead, the platform's policy needs to dynamically control the drift of the seller's belief process so that it approaches a desirable limit. As a result, the platform may require a dynamic promotion policy that is complex to calculate and implement relative to the static policy presented in \S\ref{sec: OptSimplePolicies}. Particularly, that dynamic policy might not be confounding for an initial number of periods, which requires the platform to track the evolution of the seller's belief process over time.
	
In summary, the ability to send an initial signal reduces the complexity of the platform's design problem considerably, as it allows it to achieve long-run average optimality using simple confounding policies, which can be tractably designed by following the approach described in \S\ref{sec: OptSimplePolicies}. In addition, signaling ability enables the platform to achieve an expected payoff of $co(W^C)(\mu_0)$ in every period, whereas without such ability the platform might only be able to approach $co(W^C)(\mu_0)$ asymptotically.

\section{Competition between Sellers}\label{sec: TwoSellers}\vspace{-0.1cm}

So far, our analysis has focused on a single seller's pricing decisions, while  keeping his competitors' prices fixed throughout the horizon. In the presence of price competition, sellers may react to prices set by their competitors, and learn not only from their own sales 
but also from their competitors' actions. To capture these interactions we extend the model presented in~\S\ref{sec:Model} to a setting where two sellers compete in prices in every period.  We study three platform policies that naturally extend our analysis of \S\ref{sec: preliminary_analysis}--\ref{sec:LongRunAverageOptimalConsumerSurplus}: truthful revelation of $\phi$ to both sellers, confounding only one seller, and confounding both sellers. We compare the performance of these policies numerically for a demand structure that extends  Example~\ref{example:UniformDemand}. We find that while confounding both sellers is typically infeasible or suboptimal, the platform may often benefit from confounding one of the sellers.
We next discuss the extension and the key insights derived; a summary is presented in Table~\ref{tab: competiton_policy_comparison}, and complete details of the model and results are deferred to Appendix~\ref{app: TwoSellers}.

\begin{table}[!hbtp]
	\centering\small
	\begin{tabular}{|m{3.5cm} m{9cm}|}
		\hline
		\textbf{Platform Policy} & \textbf{Main Findings} \\ \hline 
		Truthful revelation to both sellers     & - Always feasible; outperforms confounding in $54\%$ of the  parameter combinations\\
\hline
\multirow{3}{*}{Confounding one seller}
    & 
   		- Always feasible; outperforms truthful revelation in $45\%$ of the parameter combinations\\
   		& - Suffices to consider confounding the seller with relatively higher demand and revealing~$\phi$ to the seller with lower demand\\ 
\hline
		Confounding both sellers   & - Usually infeasible (see Proposition~\ref{prop: impossbility_confounding_two_sellers}); 
		suboptimal in more than 99\% of the  parameter combinations\\
\hline
	\end{tabular}
	\caption{Summary of the main findings 
with competition between two sellers.} \vspace{-0.25cm}
	\label{tab: competiton_policy_comparison}
\end{table}\normalsize 

We extend the setting described in \S\ref{sec:Model} to two sellers as follows. At the beginning of the horizon, the platform adopts a promotion policy determining which seller will be promoted in each period given both sellers'  prices, and a signaling mechanism that provides each seller with his own private signal about the value of $\phi$. Then, in every period, the sellers set their prices taking into account their initial signal, the platform's promotion policy, and the history so far. {Then, the  sellers observe their own sales realization and their  competitor's price, and use both pieces of information to update their beliefs about the fraction of impatient consumers.}


Using similar notation as in \S\ref{sec:Model}, we assume that the probability of purchasing from seller  $i\in \{1,2\}$ is captured by a function $\rho_i$ that depends on both sellers' prices $p_1$ and $p_2$, the consumer's type $\psi\in\{I,P\}$ where $\psi = P$ ($\psi=I$) denotes that the consumer is patient (impatient), and the platform's promotion decision~$a\in\{0,1,2\}$ (where $a=0$ represents not promoting either seller):
\begin{equation}\label{eq: two_seller_demand_rho_main}
	\rho_i(p_1,p_2,a, \psi) =
	\begin{cases}
		\bar{\rho}_c^i(p_1,p_2), &\text{ if $\psi=P$,} \\
		\bar{\rho}_0^i(p_i), &\text{ if $\psi=I$} \text{ and } a = i, \\
		0, &\text{ if $\psi=I$} \text{ and } a \neq i.
	\end{cases}
\end{equation}
Thus, the functions $\rho_i$ extend the demand function defined in \eqref{eq:ConsumerDemand:Rho} to the case of two sellers. We assume that the demand model satisfies the following conditions, which are analogous to Assumption~\ref{assump:Demand}:

\begin{assumption}[Two-seller Demand]\label{assump:TwoSellerDemand}
	For each $i=1,2$, the function $\bar{\rho}_0^i(p_i)$ is decreasing and Lipschitz continuous in $p_i$; and the function $\bar{\rho}_c^i(p_i,p_{-i})$ is decreasing in $p_i$, increasing in $p_{-i}$, and Lipschitz continuous in $(p_i,p_{-i})$. In addition,  $p_i\bar{\rho}_c^i(p_i,p_{-i})$ and $p_i\bar{\rho}_0^i(p_i)$ are strictly concave in $p_i$ for any fixed value of $p_{-i}$. Finally, $\bar{\rho}_0^i(p_i) \geq \bar{\rho}_c^i(p_i,p_{-i})$ for all $p_i\in\mathcal{P}_i$, $p_{-i}\in\mathcal{P}_{-i}$.
\end{assumption} \vspace{-0.25cm}

Similarly, the consumer welfare function in this setting depends on both sellers' prices, and we denote the platform's per-period payoff by a function $W = W(p_1,p_2,a, \psi)$ that extends \eqref{eq:ConsumerWelfare:W} to the case of two sellers.  

\subsection{Confounding and Truthful Policies with Two Sellers}\label{sec: two_sellers_policy_classes}

In line with the previous analysis we focus on a setting where sellers price myopically given their beliefs, but now sellers also take into account the price set by their competitor when making their pricing decisions; we therefore assume that in each period sellers price according to a (myopic) Bayesian Nash equilibrium.

We consider three natural classes of platform policies. The first one is \textit{truthful revelation,} under which the platform discloses the value of $\phi$ to both sellers at the beginning of the horizon and employs an optimal promotion policy thereafter. The second family is designed to \textit{confound one seller only,} which extends the notion of confounding policies defined in \S\ref{sec:ConfoundingPromotionPolicies} to the setting with two sellers. These policies are designed to disclose the value of $\phi$ to one seller while preventing the other (uninformed) seller from updating his belief  throughout the horizon.
 The third class of policies involves \textit{confounding both sellers;} these policies also extend the notion of confounding policies, but they are instead designed to keep the beliefs of both sellers constant by keeping them uninformed of the value of $\phi$. 
Within each of these classes, building on \S\ref{sec: OptSimplePolicies}, we focus on static and single-price promotion policies. This allows us to simplify the dynamics of the game to a case where the equilibrium prices set by the sellers and the platform's promotion policy are the same in every period, and makes the analysis of consumer  welfare induced by these policies tractable.

We compare the effectiveness of these classes of policies by: $(i)$ characterizing the optimization problems that correspond to optimal static and single-price promotion policies that are designed to confound one or both sellers, when sellers price according to a myopic Bayesian Nash equilibrium in every period; and then $(ii)$ solving for these optimal policies and comparing their performances numerically. This approach directly extends the one described in \S\ref{sec: OptSimplePolicies} to the case of two sellers. It is based on solving a static optimization problem that is similar to \eqref{eq: ConfoundingProblemSimple}, where the constraints of the problem now depend on the subset of sellers that the platform seeks to confound. This procedure highlights a key difference in the design of policies that confound one seller with respect to the single-seller setting. Namely, as both sellers observe each other's pricing decisions (as well as their own sales), the platform must take into account both sellers' actions and incentives when designing its policy, even if the goal is to confound only one of the sellers.

To illustrate this, suppose that the platform chooses to confound seller 1 and reveals the value of $\phi$ to seller~2. To keep seller 1's belief constant, the platform must design its policy in a way that not only eliminates the informational content of seller 1's demand realizations (e.g.,  as in \eqref{eq:EqualityOfSalesProb}), but also induces seller 2 to select prices that are uninformative about the true value of $\phi$. Concretely, if in a given time period seller 2's price is $p_2^H$ if $\phi = \phi_H$ and $p_2^L$ otherwise, the following two conditions must hold in order to confound seller 1:
\begin{equation*}
	\phi_H\alpha_{\phi_H}^1\bar{\rho}^1_0(p_1) + (1-\phi_H)\bar{\rho}^1_c(p_1,p_2^H) = \phi_L\alpha_{\phi_L}^1\bar{\rho}^1_0(p_1) + (1-\phi_L)\bar{\rho}^1_c(p_1,p_2^L), \quad \text{and} \quad
	p_2^H  = p_2^L,
\end{equation*}
where $\alpha_{\phi_H}^1,\alpha_{\phi_L}^1\in[0,1]$ denote the state-dependent probabilities of promoting seller 1. In particular, the platform's policy must incentivize seller 2 to set the same price regardless of the value of $\phi$. 

Furthermore, we find that there is a fundamental difference between the two classes of confounding policies that we consider: while it is always feasible to confound one of the sellers (under Assumption~\ref{assump:TwoSellerDemand}), there is a wide range of demand models under which confounding both sellers is infeasible. Formally, we establish the following existence result for policies that confound one seller.

\begin{proposition}\label{prop: existence_confound_one_policy}
Suppose that seller 2 observes the value of $\phi$ and that Assumption~\ref{assump:TwoSellerDemand} holds. 
Then, there exists a simple promotion policy that confounds seller 1 given any belief $\mu\in[0,1]$ of seller 1.
\end{proposition}
\vspace{-0.25cm}

By contrast, we show that for a broad range of demand models (i.e., that satisfy \eqref{eq: impossibility_confounding_maintext} below), there exist model parameters for which confounding both sellers is infeasible.

\begin{proposition} \label{prop: impossbility_confounding_two_sellers}
	Suppose that both sellers have an initial belief of $\mu\in(0,1)$ and that the demand model is such that for any prices $p_1,p_2$ one has that\vspace{-0.1cm}
	\begin{equation} \label{eq: impossibility_confounding_maintext}
		\frac{\bar{\rho}^1_c(p_1,p_2)}{\bar{\rho}^1_0(p_1)}+\frac{\bar{\rho}^2_c(p_1,p_2)}{\bar{\rho}^2_0(p_2)} > 1.\vspace{-0.1cm}
	\end{equation}
	Then, for all $\phi_L$ small enough, there exists no promotion policy that confounds both sellers.
\end{proposition}
\vspace{-0.25cm}

The condition in \eqref{eq: impossibility_confounding_maintext} essentially requires the aggregate demand when consumers face two sellers to be larger than when consumers have only one seller to buy from. This condition is natural in most settings of interest, and is satisfied by a broad range of demand models, including the uniform willingness to pay model from Example~\ref{example:UniformDemand}, logit demand models, and many others.
In the next section, we illustrate that policies that confound one seller are also typically superior in terms of the consumer surplus they generate.

\subsection{Evaluation of Confounding Policies} \label{sec: two_sellers_policy_comparison}\vspace{-0.1cm}

To analyze the value captured by confounding promotion policies in the setting with two competing sellers, we extend the demand model of Example~\ref{example:UniformDemand} to allow seller 2's price to vary (instead of fixing it at zero). Concretely, given parameters $a,b\in[0,1]$ and sellers' prices  $p_1\in[0,a],p_2\in[0,b]$, the probability of purchasing from seller 1 (see $\rho_1$ in \eqref{eq: two_seller_demand_rho_main}) is characterized by the functions $\bar{\rho}_c^1$, $\bar{\rho}_0^1$ defined as
\begin{equation} \label{eq: rho_1_uniform_two_sellers_main}
		\bar{\rho}_0^1(p_1) = a-p_1, \qquad \bar{\rho}_c^1(p_1,p_2) = 	\begin{cases}
			(a-p_1)(1-b+p_2)+\frac{1}{2}(a-p_1)^2, &\text{ if $p_1-p_2 \geq a-b$,} \\
			a-p_1 -\frac{1}{2}(b-p_2)^2, &\text{ if $p_1-p_2 < a-b$,}
		\end{cases}
\end{equation}
where the corresponding functions for seller 2 are defined by switching $a$ and $b$, as well as $p_1$ and $p_2$ in \eqref{eq: rho_1_uniform_two_sellers_main}. Under this demand model, we numerically solve the optimization problems that correspond to the optimal policy from each class (truthful, confound one, confound both) when sellers price in each period according to a  Bayesian Nash equilibrium (see equations~\eqref{eq: max_welfare_cond_state}--\eqref{eq: max_welfare_confound_both} in Appendix~\ref{sec: policies_definitions_two_sellers}) and compare the maximum consumer welfare that is achievable in each of these classes across a broad range of demand parameters.



First, 
we find that in the vast majority of cases (in 99.75\% of the parameter combinations), confounding both sellers is either infeasible or outperformed by truthful revelation; see the example in the left panel of Figure~\ref{fig:competition_charts_main}. While, given Proposition~\ref{prop: impossbility_confounding_two_sellers}, the frequent infeasibility of such policies might not be surprising, the typical suboptimality of confounding both sellers, even when feasible, illustrates that such policies considerably restrict the platform's actions.

By contrast, we observe that confounding one seller often outperforms truthfully revealing the value of $\phi$ to both sellers (specifically, in 45.46\% of the parameter combinations). That is, by employing a signaling mechanism that discloses the true state to one seller, together with a promotion policy that ensures that the belief of the other seller remains constant throughout the horizon, the platform can often generate larger surplus than under truthful revelation to both sellers. This finding illustrates that designing the promotion policy to eliminate the information provided by sales observations can generate surplus gains even with competition between  sellers, and shows that the intuition behind Theorem~\ref{thm:LongRunAverageOptimalConsumerSurplus} applies here as well.

Together, these findings showcase that typically the platform should consider either confounding one seller or neither of them. However, these findings also illustrate that, whether or not confounding one seller is beneficial relative to truthful revelation might depend on the particular demand parameters at hand. Therefore, these results also highlight the practical benefit of reducing the policy design problem to one that focuses on simple confounding policies (in particular, static and single price)  and the procedure of identifying such policies (as illustrated in \S\ref{sec: OptSimplePolicies} and here).

%

Finally, when confounding one seller, we also consider the platform's choice of which seller to confound. Our analysis shows that when confounding one of the sellers outperforms truthful revelation, it is optimal to confound the seller with relatively higher demand. Concretely, we observe that if the parameters of the demand model (see \eqref{eq: rho_1_uniform_two_sellers_main}) satisfy $a>b$, a policy that optimally confounds seller 1 outperforms confounding seller 2; see the example in the right panel of Figure~\ref{fig:competition_charts_main}. Intuitively, by revealing the true state to the seller with lower demand, the platform is able to incentivize the seller with high demand (who carries  more weight in the consumer welfare function) to keep his price low, thereby generating a consumer surplus gain that outweighs the cost of a relatively higher price by the low-demand seller.

\begin{figure}[!hbtp]
	\centering
	\includegraphics[width=0.95\linewidth]{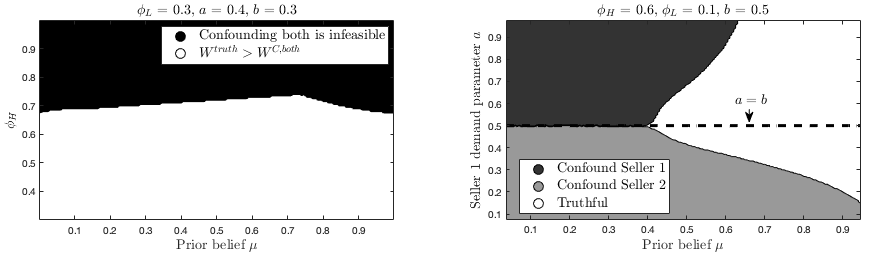}\vspace{-0.1cm}
	\caption{\textbf{Left:} Comparison of the consumer surplus achievable by truthful revelation vs. a policy that confounds both sellers, as a function of $\mu$ and $\phi_H$ with the rest of the  parameters fixed. Typically (for 99.75\% of the combinations we tested), confounding both sellers is either infeasible (black) or outperformed by truthful revelation (white). \textbf{Right:} Split of parameter space by the optimal seller to confound, as a function of $\mu$ and $a$ with the rest of the parameters fixed. When $a>b$, the platform's optimal choice is either to confound seller 1 or to reveal the state $\phi$ to both sellers, while the opposite holds when $b>a$. This indicates that to maximize consumer surplus, the platform should consider only whether to confound the seller with relatively high demand while revealing the value of $\phi$ to the seller with lower demand, or to disclose the value of $\phi$ to both sellers.	
	}
	\label{fig:competition_charts_main}
\end{figure}



\section{Equilibrium Analysis}\label{sec:EquilibriumAnalysis}

So far we have analyzed joint information and promotion policy design under the assumption that the seller makes myopic pricing decisions. 
In this section we consider the general set of non-anticipating pricing policies~$\Pi$ in the single-seller model (defined in~\S\ref{sec:Model}). We characterize a Bayesian Nash equilibrium by establishing that pricing myopically is in fact the seller's best response to a long-run average optimal platform policy. 
Moreover, this equilibrium maximizes the long-run average consumer surplus (i.e., it is platform-optimal) within the set of equilibria where players use strategies that maximize their worst-case payoff.

\subsection{Bayesian Nash Equilibrium}\label{subsec:equilibrium}

In what follows, we show that the outcome in which the platform employs a confounding policy that generates an expected average consumer surplus of $co(W^C)(\mu_0)$, and where the seller makes myopic pricing decisions, can be supported as a Bayesian Nash equilibrium, which is defined as follows.

\begin{definition}\label{def: BayesianNashEqm}
	Fix the time horizon $T\geq 1$ and the prior belief $\mu_0$. A Bayesian Nash equilibrium $(\bA, \sigma, \bP)$ consists of a platform's strategy $(\bA, \sigma)$ and a seller's pricing policy $\bP$ such that: \vspace{-0.25cm}
	\begin{enumerate}[label=(\roman*)]
		\item The platform's strategy $(\bA, \sigma)$ is a best response to the seller's pricing policy $\bP$; i.e., for any platform's strategy $(\bA', \sigma')$ we have that 
		\begin{equation}\label{def:equilibrium:platform}
			W^{\boldsymbol{\alpha},\sigma,\boldsymbol{\pi}}_T(\mu_0) \geq W^{\boldsymbol{\alpha'},\sigma',\boldsymbol{\pi}}_T(\mu_0). 
		\end{equation}
	\item At the beginning of the first period, the seller's pricing policy $\bP$ is a best response to the platform's strategy $(\bA, \sigma)$; i.e., given any signal realization $s\in S$ and any pricing policy $\bP'$, we have that 
	\begin{equation}\label{def:equilibrium:seller:firstperiod}
		V_{1,T}^{\boldsymbol{\alpha},\sigma,\boldsymbol{\pi}} \left(\langle s,\boldsymbol{\alpha},\sigma \rangle\right) \geq V_{1,T}^{\boldsymbol{\alpha},\sigma,\boldsymbol{\pi}'} \left(\langle s,\boldsymbol{\alpha},\sigma \rangle\right). 
\end{equation}
	\end{enumerate}
\end{definition}
\vspace{-0.25cm}

The following theorem establishes that for any $T\geq 1$, there exists a Bayesian Nash equilibrium $(\boldsymbol{\alpha},\sigma,\boldsymbol{\pi})$ where the seller's best response to the platform policy is to price myopically each period, and the platform {follows a confounding policy that} generates an  average consumer surplus equal to $co(W^C)(\mu_0)$.

{
\begin{theorem}[Bayesian Nash Equilibrium]\label{thm:epsilonEq}
Fix $T\geq 1$ and a prior belief $\mu_0$. Then, there exists a Bayesian Nash equilibrium $(\bA,\sigma,\bP)\in\mathcal{A} \times \Sigma \times \Pi$ such that (i) $\bP$ induces the seller to price myopically in 
	every period (i.e.,  according to \eqref{eq:myopic}), and (ii) the platform uses  a simple confounding
	promotion policy on the equilibrium path, which results in an average consumer surplus of $co(W^C)(\mu_0)$, i.e.,
	$
	 \frac{1}{T}W^{\boldsymbol{\alpha},\sigma,\boldsymbol{\pi}}_T(\mu_0) = 	co(W^C)(\mu_0).
	$
\end{theorem}}
\vspace{-0.2cm}
Although the definition of Bayesian Nash equilibrium only requires the seller's pricing policy to be a best response to the platform's policy at the beginning of the first period (condition~\eqref{def:equilibrium:seller:firstperiod}), the equilibrium that we construct to prove Theorem~\ref{thm:epsilonEq} satisfies a stronger dynamic consistency property, in the sense that the seller's pricing policy is a best response to the platform's strategy at any point in time. Formally, it satisfies that for any time $t=1,\dots T$, history $\bar{h}\in \bar{H}_t$, and pricing policy $\bP'\in \Pi$, we have that  \vspace{-0.2cm}
	\begin{equation}\label{def:equilibrium:seller}
		V_{t,T}^{\boldsymbol{\alpha},\sigma,\boldsymbol{\pi}} \left(\langle\boldsymbol{\alpha},\sigma,\bar{h} \rangle\right) \geq V_{t,T}^{\boldsymbol{\alpha},\sigma,\boldsymbol{\pi}'} \left(\langle\boldsymbol{\alpha},\sigma,\bar{h} \rangle\right). \vspace{-0.2cm}
	\end{equation}
Note that while Theorem~\ref{thm:LongRunAverageOptimalConsumerSurplus} assumes that the seller sets prices myopically (according to $\boldsymbol{\pi^*}$) and establishes the optimality of confounding promotion policies, Theorem~\ref{thm:epsilonEq} establishes a finite-time equilibrium without such an assumption. However, confounding promotion policies are once again key to the result, as it is established by showing that equilibria are maintained by a confounding policy.

To prove Theorem~\ref{thm:epsilonEq}, we leverage the procedure described in \S\ref{sec: OptSimplePolicies} to design an optimal confounding policy. Specifically, given any 
	confounding promotion policy $\bA^C$ that is constructed by solving problem \eqref{eq: ConfoundingProblemSimple}, we construct a promotion policy $\bA \in \mathcal{A}$ that is equal to $\bA^C$ along the equilibrium path and that guarantees, for all~$\mu\in[0,1]$, an expected consumer surplus of $W^C(\mu)$ in  each period, given that the seller follows a myopic pricing policy. We then prove that an optimal signal exists by adapting standard analysis in information design to our setting, which establishes that the platform policy generates an expected consumer surplus of  $co(W^C)(\mu)$ in each period. As we do not require sequential rationality from the seller in response to non-equilibrium platform policies, we construct a pricing policy $\bP$ that is myopic at all histories $\langle \bA,\sigma, \bar{h}\rangle$, but prices at~$p^*$ (see~\eqref{eq:pStar:OutsideOpt}) at all other histories. Thus, the platform cannot improve consumer surplus by deviating from~$(\bA,\sigma)$. We show that this pricing policy is a \emph{best-response} to the platform's policy in all time periods by showing that the seller's expected continuation payoff weakly decreases in any unilateral deviation. 
{As a result, we establish that for any optimal simple confounding policy that is constructed using the approach described in \S\ref{sec: OptSimplePolicies}, there exists a Bayesian Nash equilibrium in which the platform's promotion decisions follow this policy, and where the seller best responds by pricing myopically in every period.}



Moreover, Theorem~\ref{thm:epsilonEq} implies that ``semi-myopic"  policies that price myopically unless the resulting price is confounding given the seller's belief (see, e.g., \citealt{Harrison2012}), might not be effective 
for the seller. When the myopic price is confounding, semi-myopic  policies select an alternate price that generates information; in that sense these policies are designed to avoid precisely the prices that the platform incentivizes. 
Interestingly, in the equilibrium we construct, the promotion policy is designed so that the seller is incentivized to set the confounding price and does not deviate as the information gained by such deviation generates no net value.

\subsection{Robustness of Equilibrium} \label{sec: EfficiencyEqm}

The equilibrium characterization in Theorem~\ref{thm:epsilonEq} does not rule out the existence of other equilibria, including ones that may generate even higher consumer surplus.\footnote{For example, there may exist equilibria where the platform incentivizes the seller to set low prices for a fixed number of periods (which benefits early consumers), and then commits to promote the seller at high prices in later periods in order to increase the seller's payoff. Depending on the structure of $\rho$ and $W$, such alternation between high and low prices might generate higher  average consumer welfare in equilibrium, relative to the equilibrium characterized in Theorem~\ref{thm:epsilonEq}, while maintaining the seller's total expected revenue. Nonetheless, it can be shown that if the optimal simple confounding policy $\alpha^C$ that is constructed by solving problem \eqref{eq: ConfoundingProblemSimple} and the optimal signaling mechanism $\sigma$ associated with  this policy are uniquely determined, then the on-the-equilibrium path induced by the strategy profile of Theorem~\ref{thm:epsilonEq} is the only such path that (i) corresponds to a horizon maximin equilibrium that is asymptotically platform-optimal (see Definition~\ref{def: HorizonMaximinEqm}), (ii) consists of a simple confounding promotion policy, and (iii) assumes that the seller breaks ties in favor of the platform in case of indifference.
}
{While it is not necessarily unique, the equilibrium characterized in Theorem~\ref{thm:epsilonEq} satisfies desirable robustness properties. 
Concretely, in what follows we consider a setting where agents have no precise knowledge of the game's time horizon, and aim to maintain a high time-average payoff while taking into account this uncertainty. We establish that the equilibrium characterized by Theorem~\ref{thm:epsilonEq} is also an equilibrium in this setting.}

{Formally, we establish that the equilibrium characterized by Theorem~\ref{thm:epsilonEq} is {asymptotically} platform-optimal within a class of maximin equilibria in which strategies are not predicated on precise knowledge of the horizon length, but rather designed to maximize payoff over the worst-case realized horizon length (up to time $T$). For $\mu \in [0,1]$ and $T \geq 1$, and given the strategies $(\bA,\sigma,\bP)$, define the minimal time-average payoff obtained by the platform by some period $T$ as 
	\vspace{-0.2cm}
	\begin{equation*}
		RW_T^{\bA,\sigma,\bP} (\mu ) := \min_{\bar{t}\leq T}\frac{1}{\bar{t}}W_{\bar t}^{\bA,\sigma,\bP} (\mu ),\vspace{-0.2cm}
	\end{equation*}
	and for a fixed period $t$, define the seller's minimal time-average continuation payoff given history $h \in H_t$ as:\vspace{-0.2cm}
	\begin{equation*}
		RV_{t,T}^{\bA,\sigma,\bP} (h ) := \min_{t\leq \bar{t}\leq T}\left(\frac{1}{\bar{t}-t+1}\right)V_{t,\bar t}^{\bA,\sigma,\bP} (h). \vspace{-0.2cm}
	\end{equation*}
	Based on these payoff functions, we introduce equilibria in which the platform and the seller maximize their minimal time-average payoffs at every history.



	\begin{definition}[Horizon-Maximin Equilibrium]\label{def: HorizonMaximinEqm}
		Fix the time horizon $T\geq 1$ and the prior belief $\mu_0$. We say that a strategy profile $(\bA,\sigma,\bP)$ is a \emph{horizon-maximin equilibrium} if:
		\vspace{-0.25cm}
		\begin{enumerate}[label=(\roman*)]
			\item For any platform's strategy $(\bA', \sigma')\in \mathcal{A} \times \Sigma$ we have that \vspace{-0.2cm}
			\begin{equation}\label{def:equilibrium:platform:robust}
				RW_T^{\bA,\sigma,\bP} (\mu_0) \geq RW_{T}^{\bA',\sigma',\bP} (\mu_0).\vspace{-0.2cm}
			\end{equation}
		\item In any period $t=1,\dots T$, at  any history $\bar{h}\in \bar{H}_t$, and for any pricing policy $\bP'\in \Pi$ we have that \vspace{-0.2cm}
		\begin{equation}\label{def:equilibrium:seller:robust}
			RV_{t,T}^{\boldsymbol{\alpha},\sigma,\boldsymbol{\pi}} \left(\langle\boldsymbol{\alpha},\sigma,\bar{h} \rangle\right) \geq RV_{t,T}^{\boldsymbol{\alpha},\sigma,\boldsymbol{\pi}'} \left(\langle\boldsymbol{\alpha},\sigma,\bar{h} \rangle\right). \vspace{-0.2cm}
		\end{equation}
		\end{enumerate}
	\end{definition}

	Let $\mathcal{E}(T)\subset \mathcal{A}\times \Sigma\times \Pi$ denote the set of horizon-maximin equilibria with maximal horizon length $T$. 
Note that, by requiring sequential rationality from the seller on the equilibrium path, horizon-maximin equilibria still capture that the seller utilizes collected information to dynamically improve performance. Our next result shows that the Bayesian Nash equilibrium characterized in Theorem~\ref{thm:epsilonEq} is a horizon-maximin equilibrium, and is in fact long-run optimal for the platform in the set of horizon-maximin equilibria.
	
{\begin{theorem}[Optimal Horizon-Maximin Equilibria]\label{thm:OptimalRobustEquilibria}
Fix $T\geq 1$, $\mu_0 \in [0,1]$ and let $(\bA,\sigma,\bP)$ be the Bayesian Nash equilibrium constructed in Theorem~\ref{thm:epsilonEq}. Then, $(\bA,\sigma,\bP)$ is a horizon-maximin equilibrium. Moreover, this equilibrium generates the maximum long-run worst-case average payoff for the platform, namely 
		\begin{equation}\label{eq: long_run_platform_optimal_eqm}
			\lim_{T'\rightarrow \infty}\sup_{(\bA',\sigma',\bP') \in \mathcal{E}(T')}  RW^{\bA',\sigma',\bP'}_{T'}(\mu_0) =   RW^{\bA,\sigma,\bP}_T(\mu_0)= co(W^C)(\mu_0).
	\end{equation}
	\end{theorem}}
\vspace{-0.25cm}

In the proof of Theorem~\ref{thm:OptimalRobustEquilibria}, we establish that the equilibrium described in Theorem~\ref{thm:epsilonEq} is a horizon-maximin equilibrium for all $T$. {As in  Theorem~\ref{thm:epsilonEq}, in the equilibrium we construct, the seller prices at~$p^*$ following deviations by the platform; thus, it is optimal for the platform to use the prescribed policy on the equilibrium path.} On the other hand, given any platform policy, the seller's minimal time-average continuation payoff is, at most, the maximum expected revenue in the first period. As the platform's policy is static and confounding, the seller can achieve this revenue in every period. Thus, by pricing myopically in response to the confounding policy, the seller also maximizes his worst-case time-average payoff. To establish long-run optimality with respect to the platform's robust payoff, we prove that in the \emph{platform-optimal} horizon-maximin equilibrium, it is without loss to assume that the seller follows a myopic pricing policy.
} 
\section{Concluding Remarks}\label{sec:Conclusion}

\textbf{Summary.} In this paper, we propose a model of a platform that can impact the price of a seller through a promotion policy, and by disclosing information on the additional demand associated with being promoted.  
We introduce the notion of \emph{confounding} promotion policies, which are designed to prevent the seller from learning the fraction of impatient consumers (while incurring a short-term cost from diverting some impatient consumers from the best product offerings), and leverage these policies to characterize the maximum long-run average consumer surplus that is achievable by the platform for a broad class of demand models. In addition, we provide a procedure to construct such optimal policies in a tractable manner, and showcase that the ability to send an initial information signal considerably simplifies the platform's policy design problem. Moreover, we extend our setting to incorporate competition between two sellers and illustrate that confounding promotion policies continue to be a valuable tool for the platform in this setting. Finally, we construct a Bayesian Nash equilibrium by showing that, in response to the platform's optimal policy, the seller's best response in every period is to use a Bayesian myopic pricing policy. We further establish that the equilibrium we identify is platform-optimal within the set of horizon-maximin equilibria, in which sellers do not have precise knowledge of the horizon length, and thus  maximize the  worst-case average payoff.

\textbf{Design Implication: Observable vs. Unobservable Promotions.} Our model emphasizes that a platform should not only carefully design the information that it shares with sellers, but also consider how design features may impact a seller's ability to procure information. For example, one feature a platform may consider is whether to reveal promotion decisions to sellers. 
However, confounding the seller at $\mu \in (0,1)$ is not possible when promotion decisions are observed, as in that case, the seller can learn based on sales observations and/or promotion decisions. In Appendix~\ref{sec:ObservedPromotions}, we consider this question formally and establish, consistent with the insight of our baseline model, that it is optimal to conceal promotion decisions rather than revealing them in terms of the achievable long-run average consumer surplus.
In general, a similar analysis holds when the platform has private information along more dimensions. The more aspects that it allows a seller to observe, the harder it is to prevent learning, because confounding policies must satisfy constraints along all dimensions and not only in expectation.



\textbf{Future Research Directions.} There are several interesting extensions of our model. First,
in many cases the platform may observe relevant information about each arriving consumer; e.g., the patience type can be learned from the browsing and purchase history. While additional information may allow the platform to confound the seller more effectively, understanding the impact of additional information and identifying settings where it can increase consumer surplus is an interesting avenue of research. 
Second, understanding how the platform can design confounding policies in a setting where the seller has private information, is an interesting and challenging direction. Finally, while we expect our findings in the case of two competing sellers to hold for more competitors as well, it would be interesting to study settings with many competitors to better understand how the number of sellers may affect the value of confounding policies.


\small
\setstretch{1}
\bibliographystyle{chicago}
\bibliography{papers}

\begin{thebibliography}{}

\bibitem[\protect\citeauthoryear{Araman and Caldentey}{Araman and
  Caldentey}{2009}]{araman2009dynamic}
Araman, V.~F. and R.~Caldentey (2009).
\newblock Dynamic pricing for nonperishable products with demand learning.
\newblock {\em Operations research\/}~{\em 57\/}(5), 1169--1188.

\bibitem[\protect\citeauthoryear{Araman and Caldentey}{Araman and
  Caldentey}{2011}]{araman2010revenue}
Araman, V.~F. and R.~Caldentey (2011).
\newblock Revenue management with incomplete demand information.
\newblock In {\em Wiley Encyclopedia of Operations Research and Management
  Science}. John Wiley \& Sons.

\bibitem[\protect\citeauthoryear{Aumann and Maschler}{Aumann and
  Maschler}{1995}]{aumann1995repeated}
Aumann, R.~J. and M.~Maschler (1995).
\newblock {\em Repeated games with incomplete information}.
\newblock MIT press.

\bibitem[\protect\citeauthoryear{Besbes, Chaneton, and Moallemi}{Besbes
  et~al.}{2022}]{besbes2017exploration}
Besbes, O., J.~M. Chaneton, and C.~C. Moallemi (2022).
\newblock The exploration-exploitation trade-off in the newsvendor problem.
\newblock {\em Stochastic Systems\/}~{\em 12\/}(4), 319--339.

\bibitem[\protect\citeauthoryear{Besbes, Gur, and Zeevi}{Besbes
  et~al.}{2016}]{besbes2015optimization}
Besbes, O., Y.~Gur, and A.~Zeevi (2016).
\newblock Optimization in online content recommendation services: Beyond
  click-through rates.
\newblock {\em Manufacturing \& Service Operations Management\/}~{\em 18\/}(1),
  15--33.

\bibitem[\protect\citeauthoryear{Besbes and Lobel}{Besbes and
  Lobel}{2015}]{besbes2015intertemporal}
Besbes, O. and I.~Lobel (2015).
\newblock Intertemporal price discrimination: Structure and computation of
  optimal policies.
\newblock {\em Management Science\/}~{\em 61\/}(1), 92--110.

\bibitem[\protect\citeauthoryear{Besbes and Muharremoglu}{Besbes and
  Muharremoglu}{2013}]{Besbes2013}
Besbes, O. and A.~Muharremoglu (2013).
\newblock {On Implications of Demand Censoring in the Newsvendor Problem}.
\newblock {\em Management Science\/}~{\em 59\/}(6), 1407--1424.

\bibitem[\protect\citeauthoryear{Besbes and Zeevi}{Besbes and
  Zeevi}{2009}]{Besbes2009}
Besbes, O. and A.~Zeevi (2009).
\newblock Dynamic pricing without knowing the demand function: Risk bounds and
  near-optimal algorithms.
\newblock {\em Operations Research\/}~{\em 57\/}(6), 1407--1420.

\bibitem[\protect\citeauthoryear{Bimpikis and Papanastasiou}{Bimpikis and
  Papanastasiou}{2019}]{Bimpikis2019}
Bimpikis, K. and Y.~Papanastasiou (2019).
\newblock {Inducing Exploration in Service Platforms}.
\newblock In M.~Hu (Ed.), {\em Sharing Economy: Making Supply Meet Demand,},
  pp.\  193--216. Springer Series in Supply Chain Management.

\bibitem[\protect\citeauthoryear{Cachon and Swinney}{Cachon and
  Swinney}{2009}]{cachon2009purchasing}
Cachon, G.~P. and R.~Swinney (2009).
\newblock Purchasing, pricing, and quick response in the presence of strategic
  consumers.
\newblock {\em Management Science\/}~{\em 55\/}(3), 497--511.

\bibitem[\protect\citeauthoryear{Candogan}{Candogan}{2019}]{Candogan2019}
Candogan, O. (2019).
\newblock {Optimality of Double Intervals in Persuasion: A Convex Programming
  Framework}.
\newblock Working Paper.

\bibitem[\protect\citeauthoryear{Candogan and Drakopoulos}{Candogan and
  Drakopoulos}{2020}]{candogan2020optimal}
Candogan, O. and K.~Drakopoulos (2020).
\newblock Optimal signaling of content accuracy: Engagement vs. misinformation.
\newblock {\em Operations Research\/}~{\em 68\/}(2), 497--515.

\bibitem[\protect\citeauthoryear{Caro and Gallien}{Caro and
  Gallien}{2007}]{Caro2007}
Caro, F. and J.~Gallien (2007).
\newblock {Dynamic Assortment with Demand Learning for Seasonal Consumer
  Goods}.
\newblock {\em Management Science\/}~{\em 53\/}(2), 276--292.

\bibitem[\protect\citeauthoryear{Chen, Mislove, and Wilson}{Chen
  et~al.}{2016}]{Chen:2016:EAA:2872427.2883089}
Chen, L., A.~Mislove, and C.~Wilson (2016).
\newblock An empirical analysis of algorithmic pricing on amazon marketplace.
\newblock In {\em Proceedings of the 25th International Conference on World
  Wide Web}, WWW '16, Republic and Canton of Geneva, Switzerland, pp.\
  1339--1349. International World Wide Web Conferences Steering Committee.

\bibitem[\protect\citeauthoryear{Chen and Yao}{Chen and
  Yao}{2016}]{chen2016sequential}
Chen, Y. and S.~Yao (2016).
\newblock Sequential search with refinement: Model and application with
  click-stream data.
\newblock {\em Management Science\/}~{\em 63\/}(12), 4345--4365.

\bibitem[\protect\citeauthoryear{Conlisk, Gerstner, and Sobel}{Conlisk
  et~al.}{1984}]{conlisk1984cyclic}
Conlisk, J., E.~Gerstner, and J.~Sobel (1984).
\newblock Cyclic pricing by a durable goods monopolist.
\newblock {\em The Quarterly Journal of Economics\/}~{\em 99\/}(3), 489--505.

\bibitem[\protect\citeauthoryear{den Boer}{den Boer}{2015}]{den2015dynamic}
den Boer, A.~V. (2015).
\newblock Dynamic pricing and learning: historical origins, current research,
  and new directions.
\newblock {\em Surveys in operations research and management science\/}~{\em
  20\/}(1), 1--18.

\bibitem[\protect\citeauthoryear{den Boer and Zwart}{den Boer and
  Zwart}{2014}]{DenBoer2014}
den Boer, A.~V. and B.~Zwart (2014).
\newblock {Simultaneously Learning and Optimizing Using Controlled Variance
  Pricing}.
\newblock {\em Management Science\/}~{\em 60\/}(3), 770--783.

\bibitem[\protect\citeauthoryear{Dinerstein, Einav, Levin, and
  Sundaresan}{Dinerstein et~al.}{2018}]{Dinerstein2018}
Dinerstein, M., L.~Einav, J.~Levin, and N.~Sundaresan (2018).
\newblock {Consumer Price Search and Platform Design in Internet Commerce}.
\newblock {\em American Economic Review\/}~{\em 108\/}(7), 1820--1859.

\bibitem[\protect\citeauthoryear{Drakopoulos, Jain, and Randhawa}{Drakopoulos
  et~al.}{2021}]{drakopoulos2021persuading}
Drakopoulos, K., S.~Jain, and R.~Randhawa (2021).
\newblock Persuading customers to buy early: The value of personalized
  information provisioning.
\newblock {\em Management Science\/}~{\em 67\/}(2), 828--853.

\bibitem[\protect\citeauthoryear{Ely, Frankel, and Kamenica}{Ely
  et~al.}{2015}]{ely2015suspense}
Ely, J., A.~Frankel, and E.~Kamenica (2015).
\newblock Suspense and surprise.
\newblock {\em Journal of Political Economy\/}~{\em 123\/}(1), 215--260.

\bibitem[\protect\citeauthoryear{Ely}{Ely}{2017}]{ely2017beeps}
Ely, J.~C. (2017).
\newblock Beeps.
\newblock {\em American Economic Review\/}~{\em 107\/}(1), 31--53.

\bibitem[\protect\citeauthoryear{Farias and Van~Roy}{Farias and
  Van~Roy}{2010}]{farias2010dynamic}
Farias, V.~F. and B.~Van~Roy (2010).
\newblock Dynamic pricing with a prior on market response.
\newblock {\em Operations Research\/}~{\em 58\/}(1), 16--29.

\bibitem[\protect\citeauthoryear{Foster}{Foster}{2017}]{SellerActive}
Foster, K. (2017).
\newblock Amazon buy box and walmart buy box: What you need to know.
\newblock
  \url{https://selleractive.com/e-commerce-blog/amazon-buy-box-and-walmart-buy-box-what-you-need-to-know}.
\newblock Accessed: 2021-12-23.

\bibitem[\protect\citeauthoryear{Gabaix and Laibson}{Gabaix and
  Laibson}{2006}]{gabaix2006shrouded}
Gabaix, X. and D.~Laibson (2006).
\newblock Shrouded attributes, consumer myopia, and information suppression in
  competitive markets.
\newblock {\em The Quarterly Journal of Economics\/}~{\em 121\/}(2), 505--540.

\bibitem[\protect\citeauthoryear{Hagiu and Jullien}{Hagiu and
  Jullien}{2011}]{hagiu2011intermediaries}
Hagiu, A. and B.~Jullien (2011).
\newblock Why do intermediaries divert search?
\newblock {\em The RAND Journal of Economics\/}~{\em 42\/}(2), 337--362.

\bibitem[\protect\citeauthoryear{Harrison, Keskin, and Zeevi}{Harrison
  et~al.}{2012}]{Harrison2012}
Harrison, J.~M., N.~B. Keskin, and A.~Zeevi (2012).
\newblock {Bayesian dynamic pricing policies: Learning and earning under a
  binary prior distribution}.
\newblock {\em Management Science\/}~{\em 58\/}(3), 570--586.

\bibitem[\protect\citeauthoryear{H{\"o}rner and Skrzypacz}{H{\"o}rner and
  Skrzypacz}{2016}]{horner_skrzypacz_2017}
H{\"o}rner, J. and A.~Skrzypacz (2016).
\newblock Learning, experimentation and information design.
\newblock In {\em Advances in Economics and Econometrics: Eleventh World
  Congress}, Volume~1, pp.\  63--98.

\bibitem[\protect\citeauthoryear{Huh and Rusmevichientong}{Huh and
  Rusmevichientong}{2009}]{huh2009nonparametric}
Huh, W.~T. and P.~Rusmevichientong (2009).
\newblock A nonparametric asymptotic analysis of inventory planning with
  censored demand.
\newblock {\em Mathematics of Operations Research\/}~{\em 34\/}(1), 103--123.

\bibitem[\protect\citeauthoryear{Informed.co}{Informed.co}{2018}]{informed.co}
Informed.co (2018).
\newblock Everything you need to know about amazon featured merchant status.
\newblock
  \url{https://medium.com/informed/amazon-featured-merchant-status-e8276f5e1479}.
\newblock Accessed: 2021-12-23.

\bibitem[\protect\citeauthoryear{Kamenica and Gentzkow}{Kamenica and
  Gentzkow}{2011}]{Kamenica2011}
Kamenica, E. and M.~Gentzkow (2011).
\newblock Bayesian persuasion.
\newblock {\em American Economic Review\/}~{\em 101\/}(6), 2590--2615.

\bibitem[\protect\citeauthoryear{Keskin and Zeevi}{Keskin and
  Zeevi}{2014}]{keskinzeevi2014}
Keskin, N.~B. and A.~Zeevi (2014).
\newblock Dynamic pricing with an unknown demand model: Asymptotically optimal
  semi-myopic policies.
\newblock {\em Operations research\/}~{\em 62\/}(5), 1142--1167.

\bibitem[\protect\citeauthoryear{Kim, Albuquerque, and Bronnenberg}{Kim
  et~al.}{2010}]{kim2010online}
Kim, J.~B., P.~Albuquerque, and B.~J. Bronnenberg (2010).
\newblock Online demand under limited consumer search.
\newblock {\em Marketing science\/}~{\em 29\/}(6), 1001--1023.

\bibitem[\protect\citeauthoryear{K{\"u}{\c{c}}{\"u}kg{\"u}l, {\"O}zer, and
  Wang}{K{\"u}{\c{c}}{\"u}kg{\"u}l et~al.}{2022}]{Ozer2019}
K{\"u}{\c{c}}{\"u}kg{\"u}l, C., {\"O}.~{\"O}zer, and S.~Wang (2022).
\newblock Engineering social learning: Information design of time-locked sales
  campaigns for online platforms.
\newblock {\em Management Science\/}~{\em 68\/}(7), 4899--4918.

\bibitem[\protect\citeauthoryear{Lingenbrink and Iyer}{Lingenbrink and
  Iyer}{2019}]{Lingenbrink}
Lingenbrink, D. and K.~Iyer (2019).
\newblock Optimal signaling mechanisms in unobservable queues.
\newblock {\em Operations research\/}~{\em 67\/}(5), 1397--1416.

\bibitem[\protect\citeauthoryear{Papanastasiou, Bimpikis, and
  Savva}{Papanastasiou et~al.}{2017}]{Papanastasiou}
Papanastasiou, Y., K.~Bimpikis, and N.~Savva (2017).
\newblock {Crowdsourcing Exploration}.
\newblock {\em Management Science\/}~{\em 64\/}(4), 1727--1746.

\bibitem[\protect\citeauthoryear{Saur{\'e} and Zeevi}{Saur{\'e} and
  Zeevi}{2013}]{saure2013optimal}
Saur{\'e}, D. and A.~Zeevi (2013).
\newblock Optimal dynamic assortment planning with demand learning.
\newblock {\em Manufacturing \& Service Operations Management\/}~{\em 15\/}(3),
  387--404.

\bibitem[\protect\citeauthoryear{Segal and Rayo}{Segal and
  Rayo}{2010}]{Segal2010}
Segal, I. and L.~Rayo (2010).
\newblock {Optimal Information Disclosure}.
\newblock {\em Journal of Political Economy\/}~{\em 118\/}(5), 949--987.

\bibitem[\protect\citeauthoryear{Steiner}{Steiner}{2017}]{eCommerceEbay}
Steiner, I. (2017).
\newblock ebay implements amazon style buy box.
\newblock
  \url{https://www.ecommercebytes.com/C/blog/blog.pl?/pl/2017/8/1504061806.html}.
\newblock Accessed: 2021-12-23.

\bibitem[\protect\citeauthoryear{Su}{Su}{2007}]{su2007intertemporal}
Su, X. (2007).
\newblock Intertemporal pricing with strategic customer behavior.
\newblock {\em Management Science\/}~{\em 53\/}(5), 726--741.

\end{thebibliography}
\normalsize
\setstretch{1.45}
\newpage
\appendix

\setcounter{page}{1}
\renewcommand*{\thepage}{OA--\arabic{page}}

\setcounter{footnote}{0}
\renewcommand*{\thefootnote}{\fnsymbol{footnote}}

\begin{center}
	\Large Online Appendix:\\
	\vspace{-0.1cm} Information Disclosure and Promotion\\
	\vspace{-0.1cm} Policy Design for Platforms
\end{center}

\vspace{-0.5cm}

\begin{center}
	\large
		{\sf Yonatan Gur}\footnote{Stanford University, {\tt $\{$ygur,ilanmor,dsaban$\}$@stanford.edu}}
		\qquad \,
		{\sf Gregory Macnamara}\footnote{Meta Platforms, Inc., {\tt gregory.macnamara@gmail.com}}
		\qquad \,
		{\sf Ilan Morgenstern}\footnotemark[1]
		\qquad \,
		{\sf Daniela Saban}\footnotemark[1]
\end{center}

\vspace{-0.5cm}

\begin{center}
	\large \date{\monthyeardate\today}
\end{center}

\vspace{0.5cm}


\setcounter{footnote}{9}
\renewcommand*{\thefootnote}{\arabic{footnote}}

	The appendix is organized as follows: Appendix~\ref{app: PolicyDesign} describes the procedure to compute optimal simple policies that we introduced in \S\ref{sec: OptSimplePolicies}. Appendix~\ref{sec:proofs} provides the proofs of all our main results. Appendix~\ref{app: TwoSellers} develops the  model with two sellers that compete in prices, which we discussed in \S\ref{sec: TwoSellers}.  Appendix~\ref{app:AdditionalExtensions} presents two more extensions: the first one considers the case where the seller observes the platform's promotion decisions (briefly described in \S\ref{sec:Conclusion}), and the 
	second one extends our main results to a more general demand structure (as mentioned in \S\ref{sec:model-discussion}).

\section{Computation of Optimal Simple Confounding Policies} \label{app: PolicyDesign}

In this appendix we describe in more detail the procedure to compute simple confounding promotion policies introduced in \S\ref{sec: OptSimplePolicies}. We detail this procedure in \S\ref{app:DesigningSimplePolicies} and discuss some structural properties of the optimization problem defined to design simple confounding promotion policies that further simplify their construction. Then, in \S\ref{App:AnalysisOfDemandModel} we illustrate this procedure for the concrete demand model presented in Example~\ref{example:UniformDemand} and construct a myopic promotion policy for comparison with the optimal simple confounding policy.

\subsection{Designing Simple Confounding Policies}\label{app:DesigningSimplePolicies}
In this section, we detail a recipe for how to design optimal simple confounding platform policies. In general, given any concrete demand structure that satisfies 
Assumption~\ref{assump:Demand}, one can design the optimal simple  policy in three steps; $(i)$ characterize a simple confounding promotion policy which generates value $W^C(\mu)$ for all $\mu$ given that the seller makes myopic pricing decisions; $(ii)$ determine $co(W^C(\mu))$ based on the characterization of $W^C(\mu)$; and $(iii)$ determine an optimal simple signal given the prior $\mu_0$,  $W^C(\mu)$, and $co(W^C)(\mu)$.

Working backwards, once the first two steps have been completed, determining the optimal signaling mechanism is straightforward.  As the platform has set an optimal confounding policy, the optimal long-run average consumer surplus is $co(W^C)(\mu_0)$, and the seller's belief does not change based on sales observations. Therefore, an optimal signaling mechanism ensures the seller's posterior belief distribution (in period 1) is optimal given that the expected continuation value will be $W^C(\mu_1)$ in every period. Thus, following \cite{Kamenica2011}, an optimal simple signaling mechanism, $\sigma'$, takes the form:

\begin{align*}
    \sigma'(\phi_L) &= \begin{cases}
        \phi_L, &~w.p.~\left(\frac{1-\mu'}{1-\mu_0}\right)\left(\frac{\mu''-\mu_0}{\mu''-\mu'}\right)\\
        \phi_H, &~w.p.~1-\left(\frac{1-\mu'}{1-\mu_0}\right)\left(\frac{\mu''-\mu_0}{\mu''-\mu'}\right)\\
    \end{cases}&
    \sigma'(\phi_H) &= \begin{cases}
        \phi_L, &~w.p.~\left(\frac{\mu'}{\mu_0}\right)\left(\frac{\mu''-\mu_0}{\mu''-\mu'}\right)\\
        \phi_H, &~w.p.~1-\left(\frac{\mu'}{\mu_0}\right)\left(\frac{\mu''-\mu_0}{\mu''-\mu'}\right)\\
    \end{cases}
\end{align*}
where $\mu' = \sup\{\mu \leq \mu_0: co(W^C(\mu)) = W^C(\mu) \}$ and $\mu'' = \inf\{\mu \geq \mu_0: co(W^C(\mu)) = W^C(\mu) \}$.

In the second step, one then solves for the concavification of $W^C(\mu)$. In many cases, for example, the demand model of Example \ref{example:UniformDemand}, $W^C(\mu)$ can be described analytically and is concave on the interior $(0,1)$, which simplifies the computation of $co(W^C)$. 
However, if $W^C(\mu)$ cannot be described analytically, then $W^C(\mu)$ can be determined numerically over a grid of beliefs. In this case,  $co(W^C)(\mu)$ can be approximated by a simple numerical procedure. Finally, in completing the first step, the platform must solve the following optimization problem for each $\mu \in [0,1]$:
\begin{equation*}
	\begin{split}
		W^C(\mu):= \max_{\boldsymbol{\alpha} \in \mathcal{A}^C(\mu)}& ~~
		\frac{1}{T}\E\left(\sum_{t=1}^T  W(p_t,a_t,\psi_t)\middle| \boldsymbol{\alpha},\boldsymbol{\pi^*},\mu \right).
	\end{split}
\end{equation*}
Again, while $\mathcal{A}^C(\mu)$ remains a large space of policies, we can simplify the problem in several steps. First, by considering simple confounding promotion policies, it suffices to consider the analysis with $T=1$. Second, by Proposition~\ref{lemma:RevenueEquivalence}, for a fixed $\mu$, the promotion and pricing policies can be characterized using only three variables: the target price $p \in \mathcal{P}$, and the probability of promotion for each realized value of $\phi$, $\alpha_{\phi} \in [0,1]$ for $\phi\in\{\phi_L,\phi_H\}$. Here, we have two cases; if $\mu \in \{0,1\}$, then the seller's belief remain constant regardless of the platform's policy so we have that $\mathcal{A}^C(\mu) = \mathcal{A}$. Thus, for $\mu\in\{0,1\}$, choosing a myopically optimal promotion policy (e.g., as in \eqref{eq: myopicpromosingleprice} in the proof of Proposition~\ref{prop: LearningUnderMyopicPromo}; see \S\ref{app: proof_prop_LearningUnderMyopicPromo}) solves $W^C(\mu)$.

On the other hand, if $\mu\in(0,1)$, the design of an optimal simple confounding policy is reduced to solving problem~\eqref{eq: ConfoundingProblemSimple} in \S\ref{sec: OptSimplePolicies}. Moreover, we observe that one may remove the dependence on one of the two latter quantities through the confounding constraint. If $\mu \in (0,1)$, then given a promotion probability $\alpha_{\phi_L}$ and price $p$, the confounding constraint fully defines~$\alpha_{\phi_H}$:
\begin{equation*}
    \alpha_{\phi_H} = \left(\frac{\phi_H-\phi_L}{\phi_H}\right)\left(\frac{\bar{\rho}_c(p)}{\bar{\rho}_0(p)}\right)+\alpha_{\phi_L}\left(\frac{\phi_L}{\phi_H}\right).
\end{equation*}
Plugging in this constraint into the objective function and the first constraint of \eqref{eq: ConfoundingProblemSimple} in \S\ref{sec: OptSimplePolicies} results in the following optimization problem, where we denote $\bar{\phi}(\mu) = \phi_L +\mu(\phi_H-\phi_L)$:
\begin{equation}\label{eq:simplifiedOptimizationProblem}
    \begin{split}
        W^C(\mu):= \max_{\substack{\alpha_{\phi_L} \in [0,1],\\p \in \mathcal{P}}}& ~~ (\bar {W}_0(p)-\bar{W}_{\text{out}}) \bigg( \mu(\phi_H-\phi_L)\frac{\bar{\rho}_c(p) }{\bar{\rho}_0(p)} + \phi_L\alpha_{\phi_L}\bigg) +\bar{W}_c(p)(1-\bar{\phi}(\mu)) + \bar{\phi}(\mu)\bar{W}_{\text{out}}\\
        \text{s.t.}&~~(1-\phi_L)p\bar{\rho}_c(p)+ \phi_L\alpha_{\phi_L}p\bar{\rho}_0(p)\geq p^*\bar{\rho}_c(p^*)(1-\phi_L - (\phi_H-\phi_L)\mu)
    \end{split}
\end{equation}
Given a price $p \in \mathcal{P}$, the objective is linear in $\alpha_{\phi_L}$, so at the optimal solution, at least one of the constraints that involve $\alpha_{\phi_L}$ will bind. That is,
$$\alpha_{\phi_L} \in \left\{0,\frac{p^*\bar{\rho}_c(p^*)(1-\phi_L - (\phi_H-\phi_L)\mu)- (1-\phi_L)p\bar{\rho}_c(p)} {\phi_L p\bar{\rho}_0(p) }, 1\right\}.$$
In either of these three cases, the only optimization variable that remains is the price, so it is often possible to characterize $W^C(\mu)$ and the associated promotion policy analytically in closed form as shown in Example \ref{example:UniformDemand}. However, even when the analytical characterization is not possible, with a concrete demand model specified, one may  characterize $W^C(\mu)$ numerically and then follow steps $(ii)$ and $(iii)$. 

The structure of the optimal $(p,\alpha_{\phi_L})$ depends on the underlying demand model and reflects the tension that the platform faces in achieving three goals: incenvitizing the seller to set low prices, promoting the best product offering to impatient consumers, and confounding the seller. In some cases, these goals are aligned. For example, if the seller's product generates more expected consumer surplus than the outside option and confounding the seller is easier at lower prices (reflected by a decreasing ratio $\frac{\bar{\rho}_c(p)}{\bar{\rho}_0(p)}$), then setting~$\alpha_{\phi_L} = 1$ and selecting $p$ as the smallest price that satisfies the constraint is optimal. Depending on the demand model, however, this may not be the case. If the outside option generates more consumer surplus in expectation, then increasing $\alpha_{\phi_L}$ means more impatient consumers see an inferior product but also that the platform may incentivize a lower price which benefits patient consumers. In this case, the platform must balance these competing objectives. Similarly, depending on the structure of $\frac{\bar{\rho}_c(p)}{\bar{\rho}_0(p)}$, the goal of confounding may or may not be aligned with the other two goals because the ratio can increase or decrease in $p$.

\subsection{Consumers with Uniformly Distributed Valuations}\label{App:AnalysisOfDemandModel}

 We now compute optimal myopic and confounding promotion policies for the demand model presented in Example~\ref{example:UniformDemand}, which is associated with consumers that have independent valuations for two products, which are uniformly distributed in $[a-1,a]$ and $[b-1,b]$, respectively, where $a,b\in[0,1]$ are fixed parameters. Recall that we 
 assume that the price of product 2 is fixed to zero. Then, the demand functions for this model in terms of the price of product 1 $p\in[0,a]$  are (see \eqref{eq:ConsumerDemand:Rho}):
$$\bar{\rho}_0(p)  = \Prob(v_1 - p\geq 0) = a-p, \qquad \bar{\rho}_c(p)  = 
\begin{cases} 
(1-b)(a-p)+\frac{(a-p)^2}{2}, &\text{ if } p>a-b\\
a-p - \frac{b^2}{2}, &\text{ if } p\leq a-b
\end{cases} $$
In addition, the consumer surplus functions are given by (see \eqref{eq:ConsumerWelfare:W}):
\begin{equation*}
	\bar{W}_0(p) = \int_{a-1}^{a} \max \{v_1-p,0\} \partial v_1 = \frac{(a-p)^2}{2}, \qquad \bar{W}_{\text{out}} = \int_{b-1}^{b} \max \{v_2,0\} \partial v_2 = \frac{b^2}{2} 
\end{equation*}
\vspace{-0.5cm}
\begin{equation*}
    \bar{W}_c(p) =  \int_{b-1}^{b} \int_{a-1}^{a} \max \{v_1-p,v_2,0\} \partial v_1\partial v_2
    =\begin{cases} 
    \frac{1}{6}(3 b^2 + 3 (a-p)^2(1-b)  + (a-p)^3  )     , &\text{ if } p>a-b\\
\frac{1}{6}(3 (a-p)^2 + 3 b^2(1-a +p )  + b^3  ), &\text{ if } p\leq a-b.
\end{cases}
\end{equation*}
In what follows, we derive optimal myopic and confounding promotion policies for this  demand model, which results in the expressions defined in Examples~\ref{example:MyopicPromotion} and \ref{example:ConfoundingPromotion}. For simplicity, we focus on the case where  most consumers prefer the product offered by our focal seller. In particular, we assume that $a > 2b\left(1-\frac{b}{4}\right)$  which ensures that at all prices that may arise in equilibrium, the seller's product generates more value than the  outside option. This condition is equivalent to $(a-p^*>b)$ where $p^*:=\arg \max_{p} p\bar{\rho}_c(p)$. In this case, the platform is incentivized to promote the seller with high probability. 

Throughout this section, we will use the notation corresponding to simple promotion policies (see Definition~\ref{def:simplePromotionPolicies}). Thus, we only specify the price that is promoted with positive probability and the associated probabilities of promotion as a function of the seller's belief.
        
    \begin{proposition}[Optimal Policies with High Quality Seller]
    Fix $\mu \in [0,1]$ and suppose $a> 2b\left(1-\frac{b}{4}\right)$. An optimal myopic promotion policy has $\alpha_{\phi_L}(\mu)=\alpha_{\phi_H}(\mu)=1$, and:
    \begin{align*}
        p(\mu) &=\frac{1}{4}(2a-b^2(1-\bar{\phi}(\mu))-\sqrt{\bar{\phi}(\mu)}\sqrt{4a^2 -b^4(1-\bar{\phi}(\mu)) }.
    \end{align*}
    Moreover, for $\mu\in(0,1)$, an optimal confounding promotion policy has $\alpha_L^C(\mu)=1$ and:
    \begin{align*}
        \alpha_H^C(\mu) &= \frac{a-p^C(\mu) - b^2/2}{ \alpha-p^C(\mu)}\left(\frac{\phi_H-\phi_L}{\phi_H}\right) + \frac{\phi_L}{\phi_H},\\
        ~p^C(\mu) &=\frac{1}{4}\left(2a-b^2(1- \phi_L) -\sqrt{(2a-b^2)^2(\phi_H-\phi_L)\mu+\phi_L(4a^2-b^4)+b^4\phi_L^2} \right).
    \end{align*}
    \end{proposition}
    
    \begin{proof}
    We first compute the optimal price set by the seller when focusing exclusively on patient consumers, i.e., $p^*:=\arg \max_{p} p\bar{\rho}_c(p)$. To do so, we solve
    $$\max_{p\in[0,a]} p\bar{\rho}_c(p) = \max_{p\in[0,a]} p\left(a-p-\frac{b^2}{2}\right) $$
    The objective is strictly concave in $p$, so from first order conditions:
    $$ p^* = \frac{1}{4}(2a-b^2).$$
    Simple algebra and the assumption that $a> 2b\left(1-\frac{b}{4}\right)$ show that indeed we have that $p^* \leq a-b$. Moreover, if the seller sets price $p^*$ and is not promoted by the platform, his revenue is
    $$\pi^O(\mu) := p^*\bar{\rho}_c(p^*)(1-\phi_L - (\phi_H-\phi_L)\mu) = (1-\bar{\phi}(\mu))\left(\frac{2a-b^2}{4}\right)^2.$$   
    \textbf{Myopic policy.} To design a myopic promotion policy, we consider the optimization problem defined in \eqref{eq:myopicPromotion} (see \S\ref{subsec:InsufficiencyOfTruthfulDisclosure}). By focusing on single-price promotion policies, this problem reduces to:
    \begin{equation*}
    	\begin{split}
    		\max_{\substack{\alpha_{\phi_H},\alpha_{\phi_L} \in [0,1],\\p \in \mathcal{P}}}& ~~ \E_{\phi}\left(\phi \alpha_{\phi} \bar{W}_0(p) +\phi(1- \alpha_{\phi} )\bar{W}_{\text{out}}  + (1-\phi)\bar{W}_c(p)|\mu\right) \\
    		\text{s.t.}&~~ p\bar{\rho}_0(p)(\phi_L\alpha_{\phi_L}(1-\mu) +\phi_H\alpha_{\phi_H}\mu) +  p\bar{\rho}_c(p)(1-\phi_L-\mu(\phi_H-\phi_L)) \geq \pi^O(\mu).
    	\end{split}
    \end{equation*}
    The objective is decreasing in $p$, by Assumption \ref{Assump:ConsSurplus}. Moreover, since a feasible solution is $\alpha_{\phi_L}=\alpha_{\phi_H} = 1$ and~$p=p^*$ and the objective at this solution dominates the consumer surplus at any $p>p^*$ (with any promotion probabilities), the optimal price must be less than $p^*$. Finally, at any price $p<p^*$, the objective is increasing in $\alpha_{\phi_L},\,\alpha_{\phi_H}$, as we assumed that $a> 2b\left(1-\frac{b}{4}\right)$. Thus, at the optimal solution we have $\alpha_{\phi_L} = \alpha_{\phi_H}=1$. Moreover, the objective is decreasing in $p$ so the optimal price is the smallest one that makes the constraint binding, i.e.
    \begin{equation*}
    	p\bar{\rho}_0(p)(\phi_L\alpha_{\phi_L}(1-\mu) +\phi_H\alpha_{\phi_H}\mu) +  p\bar{\rho}_c(p)(1-\phi_L-\mu(\phi_H-\phi_L)) = \pi^O(\mu).
    \end{equation*}
	As argued above, the optimal solution has $\alpha_{\phi_L} = 1,\alpha_{\phi_H}=1$, so this expression becomes
	\begin{equation*}
		\bar{\phi}(\mu)p(a-p) + (1-\bar{\phi}(\mu))p\left(a-p-\frac{b^2}{2}\right) = (1-\bar{\phi}(\mu))\left(\frac{2a-b^2}{4}\right)^2.
	\end{equation*}
	Finally, by algebra we have
	\begin{equation*}
		p(\mu) = \frac{1}{4} \left[ 2a-b^2(1-\bar{\phi}(\mu))-\sqrt{\bar{\phi}(\mu)}\sqrt{4a^2 -b^4(1-\bar{\phi}(\mu)) } \right].
	\end{equation*}
    \textbf{Confounding Policy.} By our analysis of \S\ref{app:DesigningSimplePolicies}, it suffices to solve problem \eqref{eq:simplifiedOptimizationProblem} in terms of $\alpha_{\phi_L}$ and $p$, which is:
    \begin{equation*}
    	\begin{split}
    		W^C(\mu):= \max_{\substack{\alpha_{\phi_L} \in [0,1],\\p \in \mathcal{P}}}& ~~ (\bar {W}_0(p)-\bar{W}_{\text{out}}) \bigg( \mu(\phi_H-\phi_L)\frac{\bar{\rho}_c(p) }{\bar{\rho}_0(p)} + \phi_L\alpha_{\phi_L}\bigg) +\bar{W}_c(p)(1-\bar{\phi}(\mu)) + \bar{\phi}(\mu)\bar{W}_{\text{out}}\\
    		\text{s.t.}&~~(1-\phi_L)p\bar{\rho}_c(p)+ \phi_L\alpha_{\phi_L}p\bar{\rho}_0(p)\geq \pi^O(\mu).
    	\end{split}
    \end{equation*}   
%
From the assumption that  $a> 2b\left(1-\frac{b}{4}\right)$, it follows that $\frac{\bar{\rho}_c(p)}{\bar{\rho}_0(p)}$ is decreasing in $p$, and therefore the objective of the problem above is itself decreasing in $p$. It follows that the optimal solution satisfies $\alpha_{\phi_L}^C = 1$ and $p^C$ is the lowest price that makes the incentive-compatibility constraint binding:
\begin{equation*}
	(1-\phi_L)p\bar{\rho}_c(p)+ \phi_L p\bar{\rho}_0(p)= \pi^O(\mu).
\end{equation*}
By plugging in the expressions for the demand functions we have:
\begin{equation*}
	p\left[\left(a-p\right) - (1-\phi_L)\frac{b^2}{2}\right] = (1-\bar{\phi}(\mu))\left(\frac{2a-b^2}{4}\right)^2.
\end{equation*}
    Solving for $p$ and simplifying the resulting expression yields
\begin{equation*}
	p^C(\mu) = \frac{1}{4}\left(2a-b^2(1- \phi_L) -\sqrt{\mu(\phi_H-\phi_L)(2a-b^2)^2+\phi_L(4a^2-b^4)+b^4\phi_L^2} \right).
\end{equation*}
Finally, the expression for $\alpha_{\phi_H}^C(\mu)$ can be derived from the confounding constraint:
\begin{equation*}
	\begin{split}
		\alpha_{\phi_H}^C(\mu) = \left(\frac{\phi_H-\phi_L}{\phi_H}\right)\left(\frac{\bar{\rho}_c(p^C(\mu))}{\bar{\rho}_0(p^C(\mu))}\right)+\alpha_{\phi_L}^C(\mu)\left(\frac{\phi_L}{\phi_H}\right)  = \left(\frac{a-p^C(\mu) - b^2/2}{a-p^C(\mu) }\right)\left(\frac{\phi_H-\phi_L}{\phi_H}\right)+\left(\frac{\phi_L}{\phi_H}\right)
	\end{split}
\end{equation*}
\vspace{-0.1cm}
    \end{proof}

\section{Proofs}\label{sec:proofs}
	In this appendix we provide the proofs of our results. We start by proving Lemma~\ref{lemma:PromotionAsFunctionOfBelief} in \S\ref{proof:lemma:PromotionAsFunctionOfBelief}. Then, we prove Propositions~\ref{lemma:RevenueEquivalence}, \ref{prop: LearningUnderMyopicPromo},  and~\ref{prop: ConfoundingPolicyExistence}  in \S\ref{app: proof_RevenueEquivalence}, \S\ref{app: proof_prop_LearningUnderMyopicPromo}, and  \S\ref{proof:prop: ConfoundingPolicyExistence} respectively. Finally, we prove Theorems~\ref{thm:LongRunAverageOptimalConsumerSurplus}, \ref{prop: noinfodesign},  \ref{thm:epsilonEq} and \ref{thm:OptimalRobustEquilibria} in \S\ref{Sec:Proof:TheoremLongRun}, \S\ref{app: NoInfoDesign},  \S\ref{app: proof_BayesianNashEqm} and \S\ref{app: proof_RobustEqm}, respectively.
\subsection{Proof of Lemma \ref{lemma:PromotionAsFunctionOfBelief}} \label{proof:lemma:PromotionAsFunctionOfBelief}


Fix $\bA \in \mathcal{A}$ and $\sigma \in \Sigma$. 
We will construct a promotion policy $\bA'\in \mathcal{A}^M$ that generates at least the same total consumer surplus as $\bA$, by  inductively altering $\bA$ backwards over periods $t=T,...,1$. The proof follows a dynamic programming argument, so that the promotion decisions from a period $t$ onwards maximize the continuation consumer surplus  given that the seller prices according to $\bP^*$. Since the seller's myopic pricing decisions at any point in time only depend on his present belief and the platform's promotion policy, we are able to construct such a policy in a way that it depends on the history only  through the seller's belief. 

 For~$\bar{h}\in \bar{H}_t$, define the expected consumer surplus that is generated in the remaining periods from history $\bar{h}$ onwards:
$$W_{t,T}^{\bA,\sigma,\boldsymbol{\pi^*}}(\bar{h}) = \E_Z\left(\sum_{t'=t}^T  W(p_{t'},a_{t'},\psi_{t'}) \middle|\boldsymbol{\alpha},\sigma,\boldsymbol{\pi}^*,\bar{h}_t = \bar{h}\right),$$

where $\E_Z$ denotes taking the expectation with respect to any randomness in the pricing policy, promotion decisions, customer types, purchase decisions, and the true value of $\phi$. 


\textbf{Period T}: We will define an optimization problem to design a promotion policy $\alpha_T'$ that maximizes the remaining welfare at period $T$, and therefore generates at least the same consumer surplus in period $T$ as~$\alpha_T$. For fixed $\mu\in[0,1]$, define:
\begin{equation}\label{eq: W_cont_problem}
	\begin{split}
		W^{cont}_T(\mu)&:= \max_{\substack{\alpha_{\phi_H},\alpha_{\phi_L} \in [0,1],\\p \in \mathcal{P}}} ~~ \E_{\phi}\left(\phi \alpha_{\phi} \bar{W}_0(p) +\phi(1- \alpha_{\phi} )\bar{W}_{\text{out}}  + (1-\phi)\bar{W}_c(p)|\mu\right) \\
		&\text{s.t.}~~ p\bar{\rho}_0(p)(\phi_L\alpha_{\phi_L}(1-\mu) +\phi_H\alpha_{\phi_H}\mu) +  p\bar{\rho}_c(p)(1-\bar{\phi}(\mu)) \geq  p^*\bar{\rho}_c(p^*)(1-\bar{\phi}(\mu)),
	\end{split}
\end{equation}
where we denote $\bar{\phi}(\mu) = \phi_L +\mu(\phi_H-\phi_L)$. Note that the feasible region of problem \eqref{eq: W_cont_problem} is non-empty (take $p=p^*$, $\alpha_\phi=0$ for $\phi = \phi_L,\phi_H$ as a feasible solution). Moreover, problem \eqref{eq: W_cont_problem} admits a solution since its objective is continuous and its feasible set is compact. Let $(\alpha_{\phi_L}(\mu), \alpha_{\phi_H}(\mu),p_T(\mu))$ be an optimal solution for \eqref{eq: W_cont_problem} and define the promotion policy for time $T$ as follows. For a history $\bar{h} \in \bar{H}_T$ such that $\mu = \mu_T(\langle\bA,\sigma,\bar h \rangle)$, define:
\begin{equation}\label{eq: aux_myopic_promo_policy_T}
	\alpha_T'(p,\phi,\bar{h}) = \begin{cases}
		\alpha_{\phi_L}(\mu), &\text{ if $p=p_T(\mu)$ and $\phi=\phi_L$} \\
		\alpha_{\phi_H}(\mu), &\text{ if $p=p_T(\mu)$ and $\phi=\phi_H$} \\
		0, &\text{otherwise.}
	\end{cases}
\end{equation}

We now define a new promotion policy, $\bA'_T$, as follows. At time $T$, use the policy defined by $\alpha_T'$ in \eqref{eq: aux_myopic_promo_policy_T}. At time periods $t<T$, use the same policy as $\bA$ (i.e., $\{\alpha_t\}_{t=1}^{T-1}$). Note that the promotion policy and pricing decisions in periods that precede $T$ are unaffected by changes in period $T$; thus, the beliefs associated with any history remain the same under $(\bA'_T,\sigma)$ and $(\bA,\sigma)$, i.e., for any $\bar{h}\in \bar{H}_T$:
\begin{equation*}
	\mu_T \left( \langle \bA, \sigma, \bar{h}  \rangle \right) = \mu_T \left( \langle \bA'_T, \sigma, \bar{h}  \rangle \right).
\end{equation*}
As a result, $\bA'_T$ satifies the condition of Definition~\ref{def:PromotionPolsSellerBelief} at time $T$, as \eqref{eq: aux_myopic_promo_policy_T} depends on the history $\bar{h}$ only through its corresponding belief, i.e., if $\bar{h}',\bar{h}''$ are such that $\mu_T \left( \langle \bA'_T, \sigma, \bar{h}'  \rangle \right) = \mu_T \left( \langle \bA'_T, \sigma, \bar{h}''  \rangle \right)$, we have that for all~$p\in \mathcal{P}$ and $\phi\in \{\phi_L,\phi_H\}$:
\begin{equation*}
	\alpha_T'(p,\phi,\bar{h}') = \alpha_T'(p,\phi,\bar{h}'').
\end{equation*}

Furtheremore, note that the constraint in \eqref{eq: W_cont_problem} ensures that setting price $p=p_T(\mu)$ is myopically optimal for a seller with belief $\mu$ at time $T$, given the promotion policy $\alpha_T'$ and independently of the previous history. In addition, since the objective of \eqref{eq: W_cont_problem} is to maximize the expected consumer surplus at time $T$, the promotion policy $\bA_T$ generates at least the same consumer surplus as $\bA$ on time $T$, i.e., for any $\bar{h}\in \bar{H}_T$:
\begin{equation*}
	W_{T,T}^{\bA'_T,\sigma,\boldsymbol{\pi^*}}(\bar{h})  \geq W_{T,T}^{\bA,\sigma,\boldsymbol{\pi^*}}(\bar{h}).
\end{equation*}
Again, since  the promotion policy and pricing decisions in periods preceding $T$ are unaffected by changes in period $T$,  the expected consumer surplus generated in periods $t<T$ remains the same under both policies, from where we have that for any $t=1,\dots,T$ and any history $\bar{h}_t\in \bar{H}_t$,
\begin{equation*}
	W_{t,T}^{\bA'_T,\sigma,\boldsymbol{\pi^*}}(\bar{h})  \geq W_{t,T}^{\bA,\sigma,\boldsymbol{\pi^*}}(\bar{h}).
\end{equation*}
In particular, for $t=1$ we have that
$W_{T}^{\bA'_T,\sigma,\boldsymbol{\pi^*}}(\mu_0) \geq W_{T}^{\bA,\sigma,\boldsymbol{\pi^*}}(\mu_0)$. Finally, note that given a history $\bar{h}\in \bar{H}_T$ with corresponding seller belief $\mu$, by construction of \eqref{eq: W_cont_problem} we can write the remaining consumer surplus as
\begin{equation*}
	W_{T,T}^{\bA'_T,\sigma,\boldsymbol{\pi^*}}(\bar{h}) = W^{cont}_T(\mu),
\end{equation*}
where $W^{cont}_T$ is a continuous function\footnote{This follows from Berge's Maximum Theorem, which applies since the objective of \eqref{eq: W_cont_problem} is continuous on $(\alpha_{\phi_L}, \alpha_{\phi_H},p,\mu)$ and the constraint is linear in $\mu$.} of $\mu \in[0,1]$. We continue this procedure iteratively.

\textbf{Induction Hypothesis}: Fix $1\leq t\leq T$. Assume that we have constructed $\boldsymbol{\alpha}_{t+1}'$ such that the policy for the first $t$ periods is the same as $\bA$: that is, $\alpha_{t'}' = \alpha_{t'}$ for all $t'\leq t$. Moreover, for all $t'>t$, we can express the promotion policy as a function of the seller's belief. Moreover, assume that for any $t' = 1,\dots,T$ and any history $\bar{h} \in \bar{H}_{t'}$ we have:
$$W_{t',T}^{\bA'_{t+1},\sigma,\boldsymbol{\pi^*}}(\bar{h}) \geq W_{t',T}^{\bA,\sigma,\boldsymbol{\pi^*}}(\bar{h}),$$
and, in particular,
$W_{T}^{\bA'_{t+1},\sigma,\boldsymbol{\pi^*}}(\mu_0) \geq W_{T}^{\bA,\sigma,\boldsymbol{\pi^*}}(\mu_0)$.

In addition, suppose given a history $\bar{h}\in \bar{H}_{t+1}$ we can write the remaining consumer surplus at time $t+1$ as a function of the seller's belief, i.e.,
\begin{equation*}
	W_{t+1,T}^{\bA'_{t+1},\sigma,\boldsymbol{\pi^*}}(\bar{h}) = W^{cont}_{t+1}(\mu)
\end{equation*}
where $\mu = \mu_{t+1}\left( \langle \bA'_{t+1}, \sigma, \bar{h}  \rangle \right)$, and $W^{cont}_{t+1}$ is a continuous function of $\mu \in [0,1]$.

\textbf{Induction Step}: Fix $t$ and a promotion policy $\bA'_{t+1}$ that satisfies the induction hypothesis. Following the same reasoning as in the construction for period $T$, 
we will define an optimization problem to design a promotion policy for time $t$ that only depends on the history through the seller's belief, and that generates at least the same consumer surplus as $\bA'_{t+1}$ from period $t$ onwards. Here, the objective will not only consider the consumer surplus at time $t$, but also the continuation value from time $t+1$ onwards. 

We define some auxiliary notation to write this objective. First, given $p\in \mathcal{P}$ and $\phi\in \{\phi_L,\phi_H\}$, let $\rho^\phi(p,\alpha_\phi)$ be the probability of purchase if the seller sets price $p$ and is promoted with probability $\alpha_\phi$, i.e.
\begin{equation*}
	\rho^\phi(p,\alpha_\phi) = \phi \alpha_\phi \bar{\rho}_0(p)+(1-\phi)\bar{\rho}_c(p).
\end{equation*}
Note that these probabilities determine the seller's belief updating process. Concretely, if the seller's belief at time $t$ is $\mu_t$ and sets a price $p$, and is promoted with probabilities $\alpha_{\phi_L}$, $\alpha_{\phi_H}$ (that depend on the true value of $\phi$), by applying Bayes' rule we have that the seller's belief at time $t+1$ updates to:
\begin{equation*}
 \mu_{t+1} = \begin{cases}
 	\frac{\mu_t\rho^{\phi_H}\left(p,\alpha_{\phi_H}\right)}{\mu_t\rho^{\phi_H}\left(p,\alpha_{\phi_H}\right) + (1-\mu_t)\rho^{\phi_L}\left(p,\alpha_{\phi_L}\right)}, & \text{if the consumer buys from the seller,}\\
 	\frac{\mu_t\left(1-\rho^{\phi_H}\left(p,\alpha_{\phi_H}\right)\right)}{\mu_t\left(1-\rho^{\phi_H}\left(p,\alpha_{\phi_H}\right)\right) + (1-\mu_t)\left(1-\rho^{\phi_L}\left(p,\alpha_{\phi_L}\right)\right)}, & \text{otherwise.}
 \end{cases}	
\end{equation*}
Then, we define $\Gamma_{t+1}^\phi(\alpha_{\phi_L},\alpha_{\phi_H},p,\mu)$ as the expected consumer surplus generated from time $t+1$ onwards if the true state if $\phi$, the seller's belief at time $t$ is $\mu$, and assuming the seller sets a price of $p$ at time $t$ and is promoted with probabilities $\alpha_{\phi_L}$, $\alpha_{\phi_H}$. By the induction hypothesis, we can write:
\begin{equation}\label{eq: aux_Gamma_continuation_value}
	\begin{split}
		\Gamma_{t+1}^\phi(\alpha_{\phi_L},\alpha_{\phi_H},p,\mu)& =\rho^\phi(p,\alpha_\phi)W^{cont}_{t+1}\left(	\frac{\mu\rho^{\phi_H}\left(p,\alpha_{\phi_H}\right)}{\mu\rho^{\phi_H}\left(p,\alpha_{\phi_H}\right) + (1-\mu)\rho^{\phi_L}\left(p,\alpha_{\phi_L}\right)} \right)\\
		&+(1-\rho^\phi(p,\alpha_\phi))W^{cont}_{t+1}\left(\frac{\mu\left(1-\rho^{\phi_H}\left(p,\alpha_{\phi_H}\right)\right)}{\mu\left(1-\rho^{\phi_H}\left(p,\alpha_{\phi_H}\right)\right) + (1-\mu)\left(1-\rho^{\phi_L}\left(p,\alpha_{\phi_L}\right)\right)} \right)
	\end{split}
\end{equation}

Notice that by the induction hypothesis and Assumption~\ref{assump:Demand}, $\Gamma_{t+1}^\phi$ is a continuous function of $(\alpha_{\phi_L},\alpha_{\phi_H},p,\mu)$ for fixed $\phi\in \{\phi_L,\phi_H\}$. We now define an optimization problem to design the promotion policy at time $t$ in a way that maximizes the remaining consumer surplus from time $t$ onwards:
\begin{equation}\label{eq: W_cont_problem_t}
	\begin{split}
	W^{cont}_{t}(\mu)&:= \max_{\substack{p \in \mathcal{P},\\\alpha_{\phi_H},\alpha_{\phi_L} \in [0,1]}}  \E_{\phi}\left(\phi \alpha_{\phi} \bar{W}_0(p) +\phi(1- \alpha_{\phi} )\bar{W}_{\text{out}}  + (1-\phi)\bar{W}_c(p)+\Gamma_{t+1}^\phi(\alpha_{\phi_L},\alpha_{\phi_H},p,\mu) \middle|\mu \right) \\
	~&\text{s.t.}~~ p\bar{\rho}_0(p)(\phi_L\alpha_{\phi_L}(1-\mu) +\phi_H\alpha_{\phi_H}\mu) +  p\bar{\rho}_c(p)(1-\bar{\phi}(\mu)) \geq  p^*\bar{\rho}_c(p^*)(1-\bar{\phi}(\mu)).
\end{split}
\end{equation}

Note that the objective of \eqref{eq: W_cont_problem_t} consists of the consumer surplus at period $t$ plus the expected continuation value from time $t+1$ onwards which is defined by $\Gamma_{t+1}^\phi$ in \eqref{eq: aux_Gamma_continuation_value}. Moreover, the constraint remains as in \eqref{eq: W_cont_problem}, which ensures that the price associated with the designed policy will be myopically optimal for the seller.

Proceeding as in the first step, observe that \eqref{eq: W_cont_problem_t} admits an optimal solution since its objective is continuous (by the induction hypothesis and Assumption~\ref{Assump:ConsSurplus}) and its feasible set is compact. Let $(\alpha_{\phi_L}^t(\mu), \alpha_{\phi_H}^t(\mu),p_t(\mu))$ be an optimal solution in \eqref{eq: W_cont_problem_t} and define the promotion policy for time $t$ as follows. For a history $\bar{h} \in \bar{H}_t$ such that $\mu = \mu_t(\langle\bA_{t+1}',\sigma,\bar h \rangle)$, define:
\begin{equation}\label{eq: aux_myopic_promo_policy_t}
	\alpha_t'(p,\phi,\bar{h}) = \begin{cases}
		\alpha_{\phi_L}^t(\mu), &\text{ if $p=p_t(\mu)$ and $\phi=\phi_L$} \\
		\alpha_{\phi_H}^t(\mu), &\text{ if $p=p_t(\mu)$ and $\phi=\phi_H$} \\
		0, &\text{otherwise.}
	\end{cases}
\end{equation}
Define the policy $\bA_t'$ as follows. At period $t$, use the policy $\alpha_t'$ defined in \eqref{eq: aux_myopic_promo_policy_t}. At periods $t+1,\dots,T$, use the policy $\bA_{t+1}'$ from the induction hypothesis, and at periods before $t$ use the policy $\bA$. We now verify that~$\bA_t'$ satifies the confitions of the induction hypothesis.

By the construction in \eqref{eq: aux_myopic_promo_policy_t} and the induction hypothesis, $\bA_t'$ only depends on the history through the seller's belief for periods $t,\dots,T$, and is equal to $\bA$ in periods before $t$. By construction of the objective in \eqref{eq: W_cont_problem_t} and the induction hypothesis, we have that for any $\bar{h}\in \bar{H}_t$:
\begin{equation*}
	W_{t,T}^{\bA'_t,\sigma,\boldsymbol{\pi^*}}( \bar{h})\geq W_{t,T}^{\bA'_{t+1},\sigma,\boldsymbol{\pi^*}}(\bar{h})\geq W_{t,T}^{\bA,\sigma,\boldsymbol{\pi^*}}(\bar{h}).
\end{equation*}
As in the first step, we have that as the promotion policy and pricing decisions in periods preceding $t$ are unaffected by changes in period $t$, the beliefs associated with any history $\bar{h}\in \bar{H}_t$ remain the same under~$(\bA'_t,\sigma)$ and $(\bA'_{t+1},\sigma)$, i.e., for any $\bar{h}\in \bar{H}_t$:
\begin{equation*}
	\mu_t \left( \langle \bA'_{t+1}, \sigma, \bar{h}  \rangle \right) = \mu_t \left( \langle \bA'_t, \sigma, \bar{h}  \rangle \right),
\end{equation*}
and, for the same reason, the consumer surplus generated in periods earlier than $t$ remains the same under both policies. In addition, in periods that follow $t$, both $\bA_t$ and $\bA_{t+1}$ are identical, so from the induction hypothesis we have that:
\begin{equation*}
	W_{t',T}^{\bA'_t,\sigma,\boldsymbol{\pi^*}}( \bar{h})\geq W_{t',T}^{\bA'_{t+1},\sigma,\boldsymbol{\pi^*}}(\bar{h})\geq W_{t',T}^{\bA,\sigma,\boldsymbol{\pi^*}}(\bar{h}),
\end{equation*}
for any $t'=1,\dots,T$ and any history $\bar{h}\in \bar{H}_{t'}$. In particular for the first period we have:
$$W_{T}^{\bA'_t,\sigma,\boldsymbol{\pi^*}}(\mu_0) \geq W_{T}^{\bA'_{t+1},\sigma,\boldsymbol{\pi^*}}(\mu_0) \geq W_{T}^{\bA,\sigma,\boldsymbol{\pi^*}}(\mu_0).$$

Finally, since the objective of \eqref{eq: W_cont_problem_t} is continuous, we have as in the first step that $W^{cont}_t$ is a continuous function of $\mu \in[0,1]$. By construction,  given a history $\bar{h}\in \bar{H}_t$ with corresponding seller belief $\mu = \mu_{t}\left( \langle \bA'_{t}, \sigma, \bar{h}  \rangle \right)$, we can write
\begin{equation*}
	W_{t,T}^{\bA'_t,\sigma,\boldsymbol{\pi^*}}(\bar{h}) = W^{cont}_t(\mu).
\end{equation*}
Therefore, $\bA'_t$ satisfies the conditions of the induction hypothesis.  By continuing this process iteratively until period $1$, we obtain $\bA'_1\in\mathcal{A}^M$ such that 
\begin{equation*}
	W_T^{\boldsymbol{\alpha}'_1,\sigma,\boldsymbol{\pi^*}} (\mu_0) \geq W_T^{\bA,\sigma,\boldsymbol{\pi^*}} (\mu_0).
\end{equation*}
 \subsection{Proof of Proposition \ref{lemma:RevenueEquivalence}}\label{app: proof_RevenueEquivalence}
 
By Lemma~\ref{lemma:PromotionAsFunctionOfBelief}, we can assume without loss that $\boldsymbol{\alpha} \in \mathcal{A}^M$. The proof then consists of two steps. First, we establish that given $T\geq 1, \boldsymbol{\alpha} \in \mathcal{A},\sigma\in\Sigma$, we can construct a single-price promotion policy, $\boldsymbol{\alpha'} \in \mathcal{A}^{P}$ that generates at least the same consumer surplus as $\bA$, i.e., that
$$W^{\bA,\sigma,\boldsymbol{\pi^*}}_T(\mu_0) = W^{\boldsymbol{\alpha'},\sigma,\boldsymbol{\pi^*}}_T(\mu_0).$$
Then, the second step establishes that we can consider simple signaling mechanisms without welfare loss. Formally, we show that there exists a signaling mechanism $\sigma'\in\Sigma^S$ such that:
$$W^{\bA,\sigma,\boldsymbol{\pi^*}}_T(\mu_0) \leq W^{\bA,\sigma',\boldsymbol{\pi^*}}_T(\mu_0).$$ 
\textbf{Step 1 (Promotion Policy).}
Fix $\bA \in \mathcal{A}^M$ and $\sigma \in \Sigma$. Recall that $\boldsymbol{\pi}^*$ is the Bayesian myopic pricing policy defined in \S\ref{sec:Model}.
We will show that there exists a single-price  promotion policy  $\bA'$  that generates the same consumer surplus as $\bA$. 
Given that $\boldsymbol{\pi}^*$ is the myopic Bayesian pricing policy, for each period~$t=1,\dots,T$ and 
history $\bar{h}\in \bar{H}_t$,
there exists a (deterministic) price $p_t^*(\bar{h})$ such that $\pi^*_t\left( \langle \bA,\sigma,\bar{h} \rangle \right) = p_t^*(\bar{h})~w.p.~1$.  Using these prices, define:
\begin{equation}
	\alpha_t'(p,\phi,\bar{h}) = \begin{cases}
		\alpha_t(p,\phi,\bar{h}), &\text{ if } p=p_t^*(\bar{h})\\
		0, &\text{otherwise.}
	\end{cases}
\end{equation}
Letting $\boldsymbol{\alpha}' = \{\alpha_t'\}_{t=1}^T$, we have that
$\boldsymbol{\alpha}'\in\mathcal{A}^P$. Moreover, it is straighforward to verify that the seller's myopic pricing decisions (induced by $\boldsymbol{\pi^*}$) in response to the platform's promotion policies $\boldsymbol{\alpha}$ and $\boldsymbol{\alpha}'$ are the same. Thus, both $\boldsymbol{\alpha}$ and~$\boldsymbol{\alpha}'$ induce the same outcomes and therefore generate the same total consumer surplus. 

\textbf{Step 2 (Signaling Mechanism).} Fix $\boldsymbol{\alpha} \in \mathcal{A}$ and $\sigma \in \Sigma$. The signaling mechanism $\sigma$ induces a probability distribution over posteriors $\mu_1 \in [0,1]$. Conditioned on $\mu_1$, by Lemma~\ref{prop: LearningUnderMyopicPromo}, we can assume without loss that the platform's expected value is independent of the realized signal $s$. Thus, given $\bA$, we can write the expected consumer surplus conditional on the belief in the first period: $W^{\boldsymbol{\alpha},\boldsymbol{\pi^*}}(\mu_1) := \E_Z\left(\sum_{t=1}^T W(p_t,a_t,\custtype_t)|\boldsymbol{\alpha},\boldsymbol{\pi^*},\mu_1\right)$.

If $W^{\boldsymbol{\alpha},\boldsymbol{\pi^*}}(\mu_0) \geq W^{\bA,\sigma,\boldsymbol{\pi^*}}_T(\mu_0)=\E_s \left(W^{\boldsymbol{\alpha},\boldsymbol{\pi^*}}(\mu_1)|\sigma \right) $, then we have the result by defining an uninformative simple signal. Namely, let $S=\{\phi_L,\phi_H\}$ and $\sigma'(\phi) = L ~w.p.~1, $ for $\phi \in \{\phi_L,\phi_H\}$.

Otherwise, following the same reasoning as in Proposition 1 in \cite{Kamenica2011}, we have that since $W^{\bA,\sigma,\boldsymbol{\pi^*}}_T(\mu_0)=\E_s \left(W^{\boldsymbol{\alpha},\boldsymbol{\pi^*}}(\mu_1)|\sigma \right)$ is a convex combination of points in the set $\{W^{\boldsymbol{\alpha},\pi}(\mu):\, \mu \in [0,1] \}$, 
there exist points $0\leq \mu'<\mu_0<\mu''\leq 1$ and $\gamma \in (0,1)$ such that 

\begin{equation*}
	W^{\bA,\sigma,\boldsymbol{\pi^*}}_T(\mu_0) = \left[\frac{\mu'' -\mu_0}{\mu''-\mu'}\right] W^{\boldsymbol{\alpha},\boldsymbol{\pi^*}}(\mu') + \left[\frac{\mu_0 -\mu'}{\mu''-\mu'}\right]W^{\boldsymbol{\alpha},\boldsymbol{\pi^*}}(\mu'').
\end{equation*}

Letting $S=\{\phi_L,\phi_H\}$ and  defining the signaling mechanism $\sigma'$ by
$$\sigma'(\phi_L) = \begin{cases}
	\phi_L, &~w.p.~\left(\frac{1-\mu'}{1-\mu_0}\right)\left(\frac{\mu''-\mu_0}{\mu''-\mu'}\right)\\
	\phi_H, &~w.p.~1-\left(\frac{1-\mu'}{1-\mu_0}\right)\left(\frac{\mu''-\mu_0}{\mu''-\mu'}\right)\\
\end{cases} \quad \text{and} \quad \sigma'(\phi_H) = \begin{cases}
	\phi_L, &~w.p.~\left(\frac{\mu'}{\mu_0}\right)\left(\frac{\mu''-\mu_0}{\mu''-\mu'}\right)\\
	\phi_H, &~w.p.~1-\left(\frac{\mu'}{\mu_0}\right)\left(\frac{\mu''-\mu_0}{\mu''-\mu'}\right)\\
\end{cases}$$
completes the result.
\qed
\subsection{Proof of Proposition~\ref{prop: LearningUnderMyopicPromo}}\label{app: proof_prop_LearningUnderMyopicPromo}

Let $W^{M}(\mu)$ denote the consumer surplus associated with a myopic promotion policy that maximizes the present period consumer surplus (i.e., the value of problem~\eqref{eq:myopicPromotion}). 
 By  Proposition~\ref{lemma:RevenueEquivalence},  it is without loss of optimality to consider single-price promotion policies. 
 Thus,  we can write:
\begin{equation}\label{eq: myopicpromosingleprice}
	    \begin{split}
		W^M(\mu):= \max_{\substack{\alpha_{\phi_H},\alpha_{\phi_L} \in [0,1],\\p \in \mathcal{P}}}& ~~ \E_{\phi}\left(\phi \alpha_{\phi} \bar{W}_0(p) +\phi(1- \alpha_{\phi} )\bar{W}_{\text{out}}  + (1-\phi)\bar{W}_c(p)|\mu\right) \\
		\text{s.t.}&~~ p\bar{\rho}_0(p)(\phi_L\alpha_{\phi_L}(1-\mu) +\phi_H\alpha_{\phi_H}\mu) +  p\bar{\rho}_c(p)(1-\phi_L-\mu(\phi_H-\phi_L)) \geq \\
		&\quad \quad \quad p^*\bar{\rho}_c(p^*)(1-\phi_L-\mu(\phi_H-\phi_L)).
	\end{split}
\end{equation}

Denote the myopic promotion policy that results from solving problem \eqref{eq: myopicpromosingleprice} by $\hat{\bA}$. That is, if for each belief $\mu\in[0,1]$, $\left(p(\mu),\alpha_{\phi_L}(\mu),\alpha_{\phi_H}(\mu)\right)$ is an optimal solution to \eqref{eq: myopicpromosingleprice}, $\hat{\bA}$ is defined as the following static policy:

\begin{equation*}
	\hat{\alpha}(p,\phi,\mu) = \begin{cases}
		\alpha_{\phi_L}(\mu), &\text{ if $p=p(\mu)$ and $\phi=\phi_L$} \\
		\alpha_{\phi_H}(\mu), &\text{ if $p=p(\mu)$ and $\phi=\phi_H$} \\
		0, &\text{otherwise.}
	\end{cases}
\end{equation*}

Then, we can write the average consumer surplus that results from an optimal myopic promotion policy as 
\begin{equation*} 
	\begin{split}
		\frac{1}{T}W_T^{\hat{\bA},\sigma,\boldsymbol{\pi^*}}(\mu_0)&  = \frac{1}{T}\E\left(\sum_{t=1}^T  W^M(\mu_{t}) \middle|\sigma\right) = \frac{1}{T}\left[\mu_0\E^H\left(\sum_{t=1}^T  W^M(\mu_{t}) \middle|\sigma\right) +(1-\mu_0)\E^L\left(\sum_{t=1}^T  W^M(\mu_{t}) \middle|\sigma\right)\right],
	\end{split}
\end{equation*}
where $\E^H, \E^L$ denote the conditional expectations given $\phi =\phi_H,~\phi_L$ respectively. As a result, it suffices to show that 
\begin{equation*}
	\lim_{T\rightarrow \infty} \frac{1}{T} \E^H\left(\sum_{t=1}^T  W^M(\mu_{t}) \middle|\sigma \right) = W^M(1), \quad \text{and} \quad \lim_{T\rightarrow \infty} \frac{1}{T} \E^L\left(\sum_{t=1}^T  W^M(\mu_{t}) \middle|\sigma \right) = W^M(0),
\end{equation*}

We will show that the limit holds for $\phi = \phi_L$  while the corresponding result for $\phi = \phi_H$ follows analogously. First, we claim that the optimal solution to \eqref{eq: myopicpromosingleprice} is to set $\alpha_{\phi_L}=\alpha_{\phi_H}=1$ and $p$ as the smallest price that makes the constraint binding. To see this, note that setting a price $p>p^*$ is never optimal in~\eqref{eq: myopicpromosingleprice}, since the platform can incentivize the seller to set a price that does not exceed $p^*$ for any choice of $(\alpha_{\phi_H},\alpha_{\phi_L})$. Thus, it is without loss of optimality to consider only $p\in[\inf(\mathcal{P}),p^*]$ in~\eqref{eq: myopicpromosingleprice}. Recall that we assumed that~$\bar{W}_0(p^*) \geq \bar{W}_{\text{out}}$. By Assumption~\ref{Assump:ConsSurplus}, $\bar{W}_0(p)$ is decreasing in $p$ and therefore $\bar{W}_0(p) \geq \bar{W}_{\text{out}}$ for all~$p\in[\inf(\mathcal{P}),p^*]$. This condition implies that the objective of~\eqref{eq: myopicpromosingleprice} is decreasing in $p$ and increasing in $\alpha_{\phi_H}$ and $\alpha_{\phi_L}$, and therefore an optimal solution is given by $\alpha_{\phi_L}=\alpha_{\phi_H}=1$ and $p=\hat{p}(\mu)$, where $p=\hat{p}(\mu)$ is the smallest solution to\footnote{A solution is guaranteed to exist since $\mathcal{P}$ contains zero.}
\begin{equation*}
	p\bar{\rho}_0(p)(\phi_L(1-\mu) +\phi_H\mu) +  p\bar{\rho}_c(p)(1-\phi_L-\mu(\phi_H-\phi_L)) = p^*\bar{\rho}_c(p^*)(1-\phi_L-\mu(\phi_H-\phi_L)).
\end{equation*}
We now claim that the demand path induced by this promotion policy, together with the seller's myopic pricing decisions results in sales observations that are informative to the seller. Indeed, we have that at any time $t$:
 \begin{equation*}
 	\begin{split}
 		&|\phi_H\hat{\alpha}(\hat{p}(\mu_t),\phi_H,\mu_t)\bar{\rho}_0(\hat{p}(\mu_t)) + (1-\phi_H)\bar{\rho}_c(\hat{p}(\mu_t)) - \phi_L\hat{\alpha}(\hat{p}(\mu_t),\phi_L,\mu_t)\bar{\rho}_0(\hat{p}(\mu_t)) - (1-\phi_L)\bar{\rho}_c(\hat{p}(\mu_t))|\\
 		= & |\phi_H\bar{\rho}_0(\hat{p}(\mu_t)) + (1-\phi_H)\bar{\rho}_c(\hat{p}(\mu_t)) - \phi_L\bar{\rho}_0(\hat{p}(\mu_t)) - (1-\phi_L)\bar{\rho}_c(\hat{p}(\mu_t))| \\
 		= & (\phi_H - \phi_L)\left(\bar{\rho}_0(\hat{p}(\mu_t))-\bar{\rho}_c(\hat{p}(\mu_t)\right).
 	\end{split}
\end{equation*}
Moreover, we claim that there exists $\gamma>0$ such that $(\phi_H - \phi_L)\left(\bar{\rho}_0(\hat{p}(\mu_t))-\bar{\rho}_c(\hat{p}(\mu_t)\right)>\gamma$ at all times $t$. To see that this holds, note first that $\hat{p}(\mu_t)< p^*$. Moreover, since the constraint in~\eqref{eq: myopicpromosingleprice} is binding in optimality, it follows that $\hat{p}(\mu_t)>0\geq inf(\mathcal{P})$. Thus, we have that $(\phi_H - \phi_L)\left(\bar{\rho}_0(\hat{p}(\mu_t))-\bar{\rho}_c(\hat{p}(\mu_t)\right)\geq\gamma$ where 
\begin{equation*}
	\gamma =~ \min_{p\in[inf(\mathcal{P})+\epsilon,p^*-\epsilon]}(\phi_H - \phi_L)\left(\bar{\rho}_0(p)-\bar{\rho}_c(p)\right)>0,
\end{equation*}
for some $\epsilon>0$ small enough. Then, by Proposition 7 in \cite{Harrison2012}, there exist constants $\chi,\beta>0$ such that for all $t\geq 1$ we have
\begin{align*}
	\E(\mu_{t} |\phi=\phi_L )&\leq \chi \exp(-\beta t),&\E(1-\mu_{t} |\phi=\phi_H )&\leq \chi \exp(-\beta t).
\end{align*}
To finalize the proof, take $\epsilon>0.$ Notice that since both the objective and the constraints of~\eqref{eq: myopicpromosingleprice} are linear in $\mu$, it follows from Berge's maximum theorem that $W^M(\mu)$ is continuous in $\mu$. Thus, there exists $\delta>0$ such that if $|\mu|\leq\delta$, then $|W^M(\mu) - W^M(0)|\leq\epsilon$. We then have that
\begin{equation*}
	\begin{split}
		\left|\E^L\left(\sum_{t=1}^T  W^M(\mu_{t}) \middle|\sigma \right) - T W^M(0)\right|& \leq \sum_{t=1}^T  E^L\left( |W^M(\mu_{t})-W^M(0)| \middle|\sigma\right)\\
		& \leq \sum_{t=1}^T  E^L\left( |W^M(\mu_{t})-W^M(0)|I(|\mu_t|\leq \delta)\middle|\sigma\right)\\
		&+\sum_{t=1}^T  E^L\left( |W^M(\mu_{t})-W^M(0)|I(|\mu_t|> \delta)\middle|\sigma\right)\\
		& \leq T\epsilon +C\sum_{t=1}^T\Prob\left[\mu_t> \delta | \phi=\phi_L \right]\\
		& \leq T\epsilon +\frac{C}{\delta}\sum_{t=1}^{\infty}\E\left[\mu_t | \phi=\phi_L \right]\leq T\epsilon +\frac{C}{\delta}\sum_{t=1}^{\infty}\chi e^{-\beta t},		
	\end{split}
\end{equation*}
where the last two inequalities follow by taking $C = \max_{\mu \in [0,1]} \{W^M(\mu)-W^M(0)\}$ and by Markov's inequality. Finally, dividing by $T$ yields
\begin{equation*}
\left|\frac{1}{T}\E^L\left(\sum_{t=1}^T  W^M(\mu_{t}) \middle|\sigma \right) - W^M(0)\right| \leq \epsilon + \frac{C\chi e^{-\beta}}{T\delta(1-e^{-\beta})}.
\end{equation*}
Taking $T\to\infty$ results in the desired limit. By repeating the same argument for $\phi = \phi_H$ we have that
$$\lim_{T\rightarrow \infty} ~~\frac{1}{T}W^{\hat{\bA},\sigma,\boldsymbol{\pi^*}}_T(\mu)  = \mu W^M(1)+(1-\mu)W^M(0) = W^{\text{truth}}(\mu),$$
as desired.

\subsection{Proof of Proposition \ref{prop: ConfoundingPolicyExistence}}\label{proof:prop: ConfoundingPolicyExistence}
		Recall that $p^*$ is defined as the unique price that maximizes the seller's revenue  if he only sold to patient consumers (see \eqref{eq:pStar:OutsideOpt}) and note that $p^*$  does not depend on the seller's belief about $\phi$. Define $\boldsymbol{\bar{\alpha}} = \{\bar{\alpha}_t\}_{t=1}^T$ where for all~$t$ and $\mu \in [0,1]$:
		\begin{equation} \label{eq:DefinitionOfAlphaBar}
			\begin{split}
				\bar{\alpha}_t(p,\phi,\mu) =~ \begin{cases}
					\left(\frac{\phi_H-\phi_L}{\phi_H}\right)\left(\frac{\bar{\rho}_c(p^*)}{\bar{\rho}_0(p^*)}\right), &\text{ if } p = p^* \text{ and } \phi = \phi_H\\
					0,& \text{otherwise}
				\end{cases}
			\end{split}
		\end{equation}
		This promotion policy is well-defined as $0<\frac{\phi_H-\phi_L}{\phi_H}<1$ by definition, and $0<\frac{\bar{\rho}_c(p^*)}{\bar{\rho}_0(p^*)}<1$ by Assumption~\ref{assump:Demand}. One may observe that  $p^*$ is the unique myopically optimal price to set in response to $\bar\alpha_t$ at each period $t$ and for all $\mu\in[0,1]$. Moreover, by construction:
		$$\bar{\alpha}_t(p,\phi_H,\mu)\phi_H\bar{\rho}_0(p^*) + (1-\phi_H)\bar{\rho}_c(p^*) = (1-\phi_L)\bar{\rho}_c(p^*) + \bar{\alpha}_t(p,\phi_H,\mu)\phi_H\bar{\rho}_0(p^*), $$
		so the probability of a sale at price $p^*$ is independent of the true value of $\phi$. Thus, the seller's posterior belief will not update throughout the horizon, and thus $\boldsymbol{\bar{\alpha}} \in \mathcal{A}^C(\mu)$ for all $\mu \in [0,1].$

\subsection{Proof of Theorem \ref{thm:LongRunAverageOptimalConsumerSurplus}}\label{Sec:Proof:TheoremLongRun}
The proof is divided into three sections. We first state preliminaries and auxiliary results in \S\ref{sec:Theorem1Proof:PrelAux}. Using these results, we prove the statement of Theorem \ref{thm:LongRunAverageOptimalConsumerSurplus} in \S\ref{sec:Theorem1Proof:Main}. Finally, we prove the auxiliary results in \S\ref{sec:Theorem1Proof:AuxProofs}.

\subsubsection{Preliminaries and Auxiliary Results}\label{sec:Theorem1Proof:PrelAux}
First, recall that $\boldsymbol{\pi^*}$ is the Bayesian myopic pricing policy (see Definition \ref{def:MyopicPricingPolicy}). 
By Lemma \ref{lemma:PromotionAsFunctionOfBelief}, since the seller makes myopic pricing decisions, it is without loss to specify the platform promotion strategy as a function of the seller belief instead of the entire history. Thus, throughout the proof we focus our analysis on $\mathcal{A}^M\subset \mathcal{A}$, the space of promotion policies based on the seller's belief. In addition, by Proposition~\ref{lemma:RevenueEquivalence}, we can restrict out attention to the class of single-price policies $\mathcal{A}^P$ without loss of optimality.


For a fixed $\epsilon>0$ and promotion policy $\bA \in \mathcal{A}^P$, define the sets of beliefs, $M^{\alpha_t}(\epsilon) \subset [0,1]$, for $t=1,...,T$ where the expected consumer surplus is at least $\epsilon$ more than the corresponding value $co(W^{C})(\mu)$:
$$M^{\alpha_t}(\epsilon):=\{\mu\in[0,1]: \E_{a_t,p_t,\custtype_t,\phi} \left(W(p_t,a_t,\custtype_t)\middle|\alpha_t,\boldsymbol{\pi}^*,\mu\right)>co(W^{C})(\mu)+\epsilon\}.$$
The following result, that is later proved in \S\ref{sec:Theorem1Proof:AuxProofs}, establishes that if the platform uses a promotion policy that generates expected consumer surplus greater than $co(W^C)(\mu)+\epsilon$ in a given period, the sales observation is informative for the seller.
\begin{lemma}[Separation of Purchase Probabilities]\label{lemma:delta_epsilonRelationCoU}
    Fix $\epsilon>0$. {Then, there exists a $\delta=\delta(\epsilon)>0$} such that for all  $\boldsymbol{\alpha} \in \mathcal{A}^P$, if $\mu \in M^{\alpha_t}(\epsilon)$ and $p_t = \pi^*_t(\mu)$, then:
    $$|\phi_H\alpha_t(p_t,\phi_H,\mu)\bar{\rho}_0(p_t) + (1-\phi_H)\bar{\rho}_c(p_t) - \phi_L\alpha_t(p_t,\phi_L,\mu)\bar{\rho}_0(p_t) - (1-\phi_L)\bar{\rho}_c(p_t)|>\delta.$$
\end{lemma}

Next we show that beliefs converge to the truth exponentially fast in the number of periods in which $\mu_t \in M^{\alpha_t}(\epsilon)$, which is closely related to Lemma A.1. in \cite{Harrison2012}. Given a fixed $\epsilon>0$, define

\begin{equation}\label{eq:stoppingTimes}
    t_n = \min\left\{t:\sum_{t'=1}^t \mathbbm{1}\{\mu_{t'} \in M^{\alpha_{t'}}(\epsilon)\} \geq n\right\}
\end{equation}

where $t_n = T+1$ if $n>\sum_{t'=1}^T \mathbbm{1}\{\mu_{t'} \in M^{\alpha_{t'}}(\epsilon)\}$ and for convenience, we define history $h_{T+1}$ to include $\phi$ so that $\mu_{T+1}=0$ if $\phi=\phi_L$ and $\mu_{T+1}=1$ if $\phi=\phi_H$.
\begin{lemma}[Convergence of Seller Beliefs]\label{lemma:generalizedHarrison}
    Fix $\mu_0 \in [0,1]$ and let $\{t_n\}$ be defined according to \eqref{eq:stoppingTimes}. There exist constants $\chi,\beta>0$ such that for all $n\geq 1$,
    \begin{align*}
        \E(\mu_{t_n} |\phi=\phi_L )&\leq \chi \exp(-\beta n),&\E(1-\mu_{t_n} |\phi=\phi_H )&\leq \chi \exp(-\beta n)
    \end{align*}
\end{lemma}
Finally, define $W^{\max}(\mu)$ as the maximum consumer surplus achievable by any  promotion policy when $T=1$ and the seller has belief $\mu$. 
\begin{lemma}[$W^C(\mu)$ Bounded by Linear Functions]\label{lemma:ConvergenceOfWC} 
    Fix $\epsilon>0$. There exists $\bar{C}\geq 0$ such that for all $\mu \in [0,1]$:
\begin{align*}
    co(W^{\max})(\mu) - co(W^{C})(\mu) &< \frac{\epsilon}{2} + \bar{C}\mu,\text{ and} & co(W^{\max})(\mu) - co(W^{C})(\mu) &< \frac{\epsilon}{2} +\bar{C}(1-\mu).&
\end{align*}
\end{lemma}

\subsubsection{Proof of Theorem \ref{thm:LongRunAverageOptimalConsumerSurplus}}\label{sec:Theorem1Proof:Main}
Fix $\epsilon>0$, $T \geq 1$, and the platform strategies $\boldsymbol{\alpha} \in \mathcal{A}^M,\sigma \in \Sigma$. By Proposition~\ref{lemma:RevenueEquivalence}, we can take $\boldsymbol{\alpha} \in \mathcal{A}^P$ without loss of optimality. Let $\E_Z$ indicate expectation with respect to any randomness in the pricing policy, promotion policy, customer types, purchase decisions, and the true value of $\phi$: that is, ($\boldsymbol{p},\boldsymbol{a},\boldsymbol{\custtype},\boldsymbol{y},\phi$). Morover, let us denote $\Delta(\mu) = W^{\boldsymbol{\alpha},\sigma,\boldsymbol{\pi^*}}_T(\mu)- T co(W^{C})(\mu)$. Then,
    \begin{equation}\label{eq:simplifyingPayoffComparison}
        \begin{split}
            \Delta(\mu) =&\E_{Z}\left(\sum_{t=1}^T W(p_t,a_t,\custtype_t)\mathbbm{1}\{\mu_t \in M^{\alpha_t}(\epsilon/2) \}+ W(p_t,a_t,\custtype_t)\mathbbm{1}\{\mu_t \not\in M^{\alpha_t}(\epsilon/2) \}\right)  - T co(W^{C})(\mu)\\
        \stackrel{(a)}{\leq}& \E_Z\left(\sum_{t=1}^T W(p_t,a_t,\custtype_t)\mathbbm{1}\{\mu_t \in M^{\alpha_t}(\epsilon/2) \}+ W(p_t,a_t,\custtype_t)\mathbbm{1}\{\mu_t \not\in M^{\alpha_1}(\epsilon/2) \} - co(W^{C})(\mu_t)\right) \\
        \stackrel{(b)}{\leq} & \sum_{t=1}^T \E_Z \left((W(p_t,a_t,\custtype_t) -  co(W^{C})(\mu_t))\mathbbm{1}\{\mu_t \in M^{\alpha_t}(\epsilon/2) \}\right) + \frac{\epsilon}{2} \sum_{t=1}^T\E \left(\mathbbm{1}\{\mu_t \not\in M^{\alpha_t}(\epsilon/2)  \} \right).
    \end{split}
    \end{equation}
    These inequalities follow since:
\begin{itemize}
    \item[(a)] $co(W^{C})$ is concave by construction and, for all periods $t$, $\E_Z [\mu_t] =\mu_0$ because Bayesian beliefs are a martingale. Thus, by Jensen's inequality: $\E_Z[co(W^C(\mu_t))] \leq co(W^C(\E_Z [\mu_t])) = co(W^C(\mu_0)).$
    \item[(b)] Splitting $co(W^{C})(\mu_t)$ across the outcomes $\mu_t \in M^{\alpha_t}(\epsilon/2)$ and $\mu_t \not\in M^{\alpha_t}(\epsilon/2)$ and applying the definition of $M^{\alpha_t}(\epsilon/2)$
\end{itemize}

Consider the last expression in \eqref{eq:simplifyingPayoffComparison}. The second term is less than $\frac{T\epsilon}{2}$ for all $T$. We complete the proof by showing the existence of $\bar{T}$ such that for all $T>\bar{T}$ the first term is less than  $KT\epsilon$ for some fixed $K>0$. That is:
\begin{equation}
    \label{eq:ValueInPeriodsWithSeparation}
    \sum_{t=1}^T \E_Z \left((W(p_t,a_t,\custtype_t) -  co(W^{C})(\mu_t))\mathbbm{1}\{\mu_t \in M^{\alpha}(\epsilon/2) \}\right) \leq KT\epsilon. 
\end{equation}


Consider the case when $\phi=\phi_L$ (denoting the conditional expectation by $\E^L$) and select $\chi,\beta>0$ according to Lemma~\ref{lemma:generalizedHarrison}:
\begin{align*}
	\sum_{t=1}^T \E^L_Z \left((W(p_t,a_t,\custtype_t) -  co(W^{C})(\mu_t))\mathbbm{1}\{\mu_t \in M^{\alpha}(\epsilon/2) \}\right) & \leq \sum_{n=1}^T \E_{Z}^L \left((W(p_{t_n},a_{t_n},\custtype_{t_n}) -  co(W^{C})(\mu_{t_n}))\right)\\
	 &\leq \sum_{n=1}^T \E_{Z}^L \left(co(W^{\max})(\mu_{t_n}) -  co(W^{C})(\mu_{t_n})\right)\\
    &\leq \sum_{n=1}^T\left(\frac{\epsilon}{2}+ \bar{C}\E_{Z}^L \left(\mu_{t_n}\right)\right), \quad \text{[by Lemma \ref{lemma:ConvergenceOfWC}]}\\
    &\leq \frac{T \epsilon}{2} + \bar{C}\sum_{n=1}^T \chi \exp(-\beta n), \quad  \text{[by Lemma \ref{lemma:generalizedHarrison}]}\\
    &\leq \frac{T \epsilon}{2} +\bar{C}\sum_{n=1}^{\infty} \chi \exp(-\beta n)\\
    &= \frac{T \epsilon}{2} +\bar{C} \frac{\chi}{e^{\beta} - 1},
\end{align*}
from where we have that
\begin{equation*}
	\frac{1}{T}\sum_{t=1}^T \E \left((W(p_t,a_t,\custtype_t) -  co(W^{C})(\mu_t))\mathbbm{1}\{\mu_t \in M^{\alpha}(\epsilon/2) \}\middle|\phi =\phi_L\right) \leq \frac{ \epsilon}{2} +\frac{\bar{C}}{T} \frac{\chi}{e^{\beta} - 1}.
\end{equation*}

This inequality follows analagously  for $\phi=\phi_H$. Thus we can select $\bar{T}>\bar{C}\frac{ \chi}{\epsilon} \frac{4}{e^{\beta} - 1}$, and for any $T>\bar{T}$:
$$\frac{1}{T}W^{\boldsymbol{\alpha},\sigma,\boldsymbol{\pi^*}}_T(\mu) < co(W^C)(\mu)+\frac{5}{4}{\epsilon}.$$

By taking the limit as $T\rightarrow\infty$ and then letting $\epsilon \rightarrow 0$, we have:
$$\lim_{T\rightarrow \infty}\sup_{\bA \in\mathcal{A},\sigma \in \Sigma} ~~\frac{1}{T}W^{\bA,\sigma,\boldsymbol{\pi^*}}(\mu) \leq co(W^C)(\mu).$$

To establish that this bound is tight, we prove the existence of $\bA\in \mathcal{A}$, $\sigma\in \Sigma$ that generate an expected surplus of $co(W^C)(\mu)$ for all $\mu$. First, we show that there exists a promotion policy that generates $W^C(\mu)$ for $\mu \in [0,1]$ and any $T\geq 1$. Second, there exists a signaling mechanism that, coupled with the promotion policy, generates an expected payoff $co(W^C)(\mu)$ of in every period.

\textbf{Existence of Optimal Confounding Promotion Policy}. Consider the optimization problem \eqref{PlatformProblem:relaxedOptimizationProblem} from the proof of Lemma \ref{lemma:delta_epsilonRelationCoU}. Fix $\mu \in [0,1],\delta= 0$. The feasible set of \eqref{PlatformProblem:relaxedOptimizationProblem}: $F(\mu,\delta) \subset [0,1]^2\times \mathcal{P}$, is compact since it is closed and bounded for any fixed $\mu,\delta$. Thus, there exists an optimal solution to \eqref{PlatformProblem:relaxedOptimizationProblem} by the extreme value theorem.

Let $\alpha^C \in \mathcal{A}^M$ correspond to the simple confounding promotion policy where the one-period solution to \eqref{PlatformProblem:relaxedOptimizationProblem} with $\delta = 0$ is repeated $T$ times for every $\mu$ (which is equivalent to the policy derived by solving problem \eqref{eq: ConfoundingProblemSimple} in \S\ref{sec: OptSimplePolicies}). By construction the seller's belief, myopically optimal price, and the expected welfare are the same in each period. Thus, the payoff generated by this policy given posterior belief $\mu_1 \in [0,1]$ is $ T\cdot W^C(\mu_1)$.

\textbf{Existence of Optimal Signaling Mechanism.} We now show that an optimal signalling mechanism achieves $co(W^C)(\mu)$. Note that $W^C(\mu)$ is upper semicontinuous in $\mu$ by the proof of Lemma \ref{lemma:ConvergenceOfWC}. Therefore, an optimal signal $\sigma \in \Sigma$ exists (see \cite{Kamenica2011} Corollaries 1 and 2, and related discussion). It follows that given a prior of $\mu_0$, the optimal signal generates expected consumer surplus of $co(W^C)(\mu_0)$ in each period.

Thus, for any fixed $T\geq 1$ and $\mu\in[0,1]$, there exists $\bA$, $\sigma$ such that:
$$\frac{1}{T}W^{\bA,\sigma,\boldsymbol{\pi^*}}_T(\mu) = co(W^C)(\mu).$$
This concludes the proof of Theorem~\ref{thm:LongRunAverageOptimalConsumerSurplus}.

 \subsubsection{Proofs of Auxiliary Results}\label{sec:Theorem1Proof:AuxProofs}
 
 We now provide the proofs of Lemmas~\ref{lemma:delta_epsilonRelationCoU}, \ref{lemma:generalizedHarrison} and \ref{lemma:ConvergenceOfWC}, which were used as auxiliary results in the proof of Theorem~\ref{thm:LongRunAverageOptimalConsumerSurplus}.

 \subsubsection*{Proof of Lemma \ref{lemma:delta_epsilonRelationCoU}}
 
            
	Fix $\epsilon>0$. We will show that there exists $\delta>0$ such that for all $\boldsymbol{\alpha} \in \mathcal{A}^P$, 
	if $\mu \in M^{\alpha_t}(\epsilon)$, then:
	\begin{equation}\label{eq: ConstraintViolationLemma2}
		|\phi_H\alpha_t(p_t,\phi_H,\mu)\bar{\rho}_0(p_t) + (1-\phi_H)\bar{\rho}_c(p_t) - \phi_L\alpha_t(p_t,\phi_L,\mu)\bar{\rho}_0(p_t) - (1-\phi_L)\bar{\rho}_c(p_t)|>\delta.
	\end{equation}

		First note that $M^{\alpha_t}(\epsilon)\subseteq (0,1)$. This follows since, for $\mu\in \{0,1\}$, any policy is confounding and thus $W^C(\mu)$ corresponds to the maximum surplus achievable by any policy in a single time period; so no policy can generate higher welfare than $W^C(\mu)$.
            
        Thus, for the rest of the proof it suffices to consider $\mu\in(0,1)$. By Lemma \ref{lemma:PromotionAsFunctionOfBelief}, it is without loss to specify the promotion as a function of $\mu$ instead of the entire history. Fix a time period $t$; we then simplify notation by letting, for $\phi \in \{\phi_L,\phi_H\}$, $\alpha_{\phi} = \alpha_t(p,\phi,\mu)$ where $\mu$ will be left implicit and $p$ is the single price where $\alpha_{\phi}$ may be greater than 0. 
        

        With this notation, define the following  optimization problem given $\mu\in [0,1],\delta\geq 0$
     \begin{equation}\label{PlatformProblem:relaxedOptimizationProblem}
        \begin{split}
            W^C(\mu,\delta):= \max_{\substack{\alpha_{\phi_H},\alpha_{\phi_L} \in [0,1],\\p \in \mathcal{P}}}& ~~ \E_{\phi}\left(\phi \alpha_{\phi} \bar{W}_0(p) +\phi(1- \alpha_{\phi} )\bar{W}_{\text{out}}  + (1-\phi)\bar{W}_c(p)|\mu\right) \\
            \text{s.t.}&~~ p\bar{\rho}_0(p)(\phi_L\alpha_{\phi_L}(1-\mu) +\phi_H\alpha_{\phi_H}\mu) +  p\bar{\rho}_c(p)(1-\phi_L-\mu(\phi_H-\phi_L)) \geq \\
            &\quad \quad \quad p^*\bar{\rho}_c(p^*)(1-\phi_L-\mu(\phi_H-\phi_L))\\
            &~~ |\phi_H\alpha_{\phi_H}\bar{\rho}_0(p) + (1-\phi_H)\bar{\rho}_c(p) - \phi_L\alpha_{\phi_L}\bar{\rho}_0(p) - (1-\phi_L)\bar{\rho}_c(p)| \leq  \delta.
        \end{split}
    \end{equation}
    Note that when $\delta=0$, this problem is exactly the platform's promotion design problem when considering simple confounding policies (see \eqref{eq: ConfoundingProblemSimple} in \S\ref{sec: OptSimplePolicies}). Thus, \eqref{PlatformProblem:relaxedOptimizationProblem} is a relaxed version of that problem, where we require the promotion policy to be only ``$\delta$-confounding''. As in problem \eqref{eq: ConfoundingProblemSimple}, the first constraint requires that the pricing policy be myopically optimal. Setting $p$ as a decision variable ensures that it is the price that maximizes consumer surplus. Note that we also define the problem for $\mu\in\{0,1\}$ to be able to work with the compact interval $[0,1]$.
    
    In what follows, we denote the feasible set of problem \eqref{PlatformProblem:relaxedOptimizationProblem} by $F(\mu, \delta)$. That is, we say that $ (\alpha_{\phi_L},\alpha_{\phi_H},p) \in F(\mu,\delta)$ if $(\alpha_{\phi_L},\alpha_{\phi_H},p)\in[0,1]\times[0,1]\times \mathcal{P}$ satisfies the two constrains in \eqref{PlatformProblem:relaxedOptimizationProblem}. With this in mind, we prove the statement of the Lemma in a series of claims.
    
   	\textbf{Step 1.} The objective of \eqref{PlatformProblem:relaxedOptimizationProblem} is Lipschitz continuous in $(\alpha_{\phi_L},\alpha_{\phi_H},p)$.
   	\begin{proof}
   		Follows since the objective of \eqref{PlatformProblem:relaxedOptimizationProblem} is linear in $\alpha_{\phi_L},\alpha_{\phi_H}$ (with bounded coefficients) and is Lipschitz continuous in $p$ by Assumption \ref{Assump:ConsSurplus}.
   	\end{proof}
   
   \textbf{Step 2.} Fix $\epsilon>0$. There exists $\bar{\delta}>0$ such that for all $\mu \in [0,1]$, we have that if $\delta<\bar{\delta}$ and $x \in F(\mu,\delta)$, there exists $y \in F(\mu,0)$ such that $||x-y||<\epsilon$.
   \begin{proof}
   	First note that the second constraint in \eqref{PlatformProblem:relaxedOptimizationProblem} is never binding when $\delta>1$, and the feasible sets $F(\mu,\delta)$ are increasing in $\delta$, so we can focus on $(\mu,\delta)\in [0,1]\times [0,1]$ without loss of generality.
   	
   	For $\mu \in[0,1]$ and $x\in[0,1]\times[0,1]\times \mathcal{P}$ let us define
   	\begin{equation*}
   		G(x,\mu) = \min \left \lbrace ||x-y|| :\,y \in F(\mu,0) \right \rbrace.
   	\end{equation*}
   	This function is well-defined since $F(\mu,0)$ is a compact set. In addition, note that $G(x,\mu)=0$ if~$x\in F(\mu,0)$ and that $G(x,\mu)$ is continuous in $x$ for fixed  $\mu$, by continuity of the norm. Now, given~$\delta>0$, define
   	\begin{equation*}
   		H(\mu,\delta) = \max \left \lbrace G(x,\mu) :\,x \in F(\mu,\delta) \right \rbrace.
   	\end{equation*}
   	We claim that given fixed $\mu>0$, $H(\mu,\delta)$ is continuous in $\delta$. This follows by Berge's maximum theorem, which can be applied by noting $G(x,\mu)$ is continuous in $x$ (per our previous argument) and that the correspondence $F(\mu,\delta)$ is both upper and lowerhemicontinuous in $\delta$ for any fixed $\mu$ (which we show below). 
   	
   	Furthermore, note that $H(\mu,0)=0$ by definition. Therefore, given $\epsilon>0$, there exists $\delta_\mu>0$ such that $H(\mu,\delta)<\epsilon$ for all $\delta<\delta_{\mu}$. To find a $\bar{\delta}$ that establishes this condition uniformly for $\mu$, define the collection of sets given by $A_\delta = \left \lbrace \mu \in [0,1]:\, H(\mu,\delta)<\epsilon \right \rbrace$. By the previous argument, and since~$H(\mu,\delta)$ is a continuous function\footnote{This can be established in a similar fashion as we did for $\delta$. The proof of Lemma~\ref{lemma:ConvergenceOfWC} establishes continuity of $F(\mu,\delta)$ in $\mu$.} of $\mu$ for fixed $\delta$, we have that $\{A_\delta:\,\delta>0\}$ is an open cover of $[0,1]$. By compactness, this cover admits a finite subcover, say $\{A_{\delta_1},A_{\delta_2},\dots,A_{\delta_k}\}$. Letting $\bar{\delta} = \min_{1\leq j\leq k} \delta_j$ proves the claim.
   	
   	Finally, we show that $F(\mu,\delta)$ is both upper and lower hemicontinuous in $\delta$ for fixed $\mu$. This is established in two steps: 
   	\begin{enumerate}
   		\item Upper hemicontinuity is established by observing that $F$ satisfies the closed graph property, which follows from continuity of the functions that define the constraints in \eqref{PlatformProblem:relaxedOptimizationProblem}.
   		\item To establish lower hemicontinuity, let $\left(\delta_n\right)_{n\in \mathbb{N}} \subseteq [0,1]$ be a sequence that converges to some $\delta\in[0,1]$, and let $(\alpha_{\phi_L},\alpha_{\phi_H},p) \in F(\mu,\delta)$. We need to construct a sequence $(\alpha_{\phi_L},\alpha_{\phi_H},p)_n \rightarrow (\alpha_{\phi_L},\alpha_{\phi_H},p)$ such that $(\alpha_{\phi_L},\alpha_{\phi_H},p)_n \in F(\mu,\delta_n)$ for all large enough $n$. If the second constraint in \eqref{PlatformProblem:relaxedOptimizationProblem} is not binding, we can simply let $(\alpha_{\phi_L},\alpha_{\phi_H},p)_n = (\alpha_{\phi_L},\alpha_{\phi_H},p)$, and the result follows by continuity of the function that defines the second constraint in \eqref{PlatformProblem:relaxedOptimizationProblem}, and since the first constraint is independent of $\delta$. If the second constraint  in \eqref{PlatformProblem:relaxedOptimizationProblem} is binding, we have two cases:
   		\begin{enumerate}
   			\item If $\delta_n\geq \delta$, since the feasible sets are increasing in $\delta$, we simply let $(\alpha_{\phi_L},\alpha_{\phi_H},p)_n = (\alpha_{\phi_L},\alpha_{\phi_H},p)$.
   			\item If $\delta_n < \delta$, we consider the case where the expression inside the absolute value in \eqref{PlatformProblem:relaxedOptimizationProblem} is positive (the opposite case can be established by the same reasoning), which implies that:
   			\begin{equation*}
   				\phi_H\alpha_{\phi_H}\bar{\rho}_0(p) - \phi_L\alpha_{\phi_L}\bar{\rho}_0(p)  = \delta +\left(\phi_H-\phi_L\right)\bar{\rho}_c(p)
   			\end{equation*}
   		\end{enumerate}
   		We consider two cases. First, if $\alpha_{\phi_L}<1$, we  define:
   		\begin{equation*}
   			\left(\alpha_{\phi_L},\alpha_{\phi_H},p\right)_n = \left(\alpha_{\phi_L}+\frac{\delta-\delta_n}{\phi_L\bar{\rho}_0(p)},\alpha_{\phi_H},p\right).
   		\end{equation*}
   		For all large enough $n$, we have that $\left(\alpha_{\phi_L},\alpha_{\phi_H},p\right)_n\in F(\mu,\delta_n)$ by construction, as the first constraint is relaxed at this point.
   		
   		If $\alpha_{\phi_L}=1$, then we must have that $\alpha_{\phi_H}>0$. Suppose, in addition, that the first constraint in \eqref{PlatformProblem:relaxedOptimizationProblem} is not binding. Then, we  define:
   		\begin{equation*}
   			\left(\alpha_{\phi_L},\alpha_{\phi_H},p\right)_n = \left(\alpha_{\phi_L},\alpha_{\phi_H}-\frac{\delta-\delta_n}{\phi_H\bar{\rho}_0(p)},p\right).
   		\end{equation*}
   		One can verify that $\left(\alpha_{\phi_L},\alpha_{\phi_H},p\right)_n\in F(\mu,\delta_n)$ for all $n$ large enough by construction.
   		
   		Finally, if $\alpha_{\phi_L}=1$ and both constraints are binding, we must have that  $p\neq p^*$. Let us then define
   		\begin{equation*}
   			\left(\alpha_{\phi_L},\alpha_{\phi_H},p\right)_n = \left(1,\left(\frac{\phi_H-\phi_L}{\phi_H}\right)\frac{\bar{\rho}_c(p_n)}{\bar{\rho}_0(p_n)} + \frac{\phi_L}{\phi_H} + \frac{\delta_n }{\bar{\rho}_0(p_n)\phi_H},(1-\gamma_n)p+\gamma_np^*\right),
   		\end{equation*}
   		where the sequence $\gamma_n\to 0$ will be chosen later. By construction, the second constraint of \eqref{PlatformProblem:relaxedOptimizationProblem} is binding at $\left(\alpha_{\phi_L},\alpha_{\phi_H},p\right)_n$, and $\left(\alpha_{\phi_L},\alpha_{\phi_H},p\right)_n \to \left(\alpha_{\phi_L},\alpha_{\phi_H},p\right)$. Thus, it remains to check that we can choose $\gamma_n$ to satisfy the first constraint in \eqref{PlatformProblem:relaxedOptimizationProblem}. Let us denote the left-hand side of this constraint as $L(\alpha_{\phi_L},\alpha_{\phi_H},p),$ which, by  Assumption~\ref{assump:Demand},  is strictly concave in $p$. After  some algebra, we have that
   		\begin{equation*}
   			\begin{split}
   				L((\alpha_{\phi_L},\alpha_{\phi_H},p)_n) & = \phi_L p_n \bar{\rho}_0(p_n)+(1-\phi_L)p_n\bar{\rho}_c(p_n)+\mu p_n \delta_n\\
   				& > (1-\gamma_n) \left[\phi_L p \bar{\rho}_0(p)+(1-\phi_L)p\bar{\rho}_c(p)+\mu p \delta_n \right] \\
   				&+ \gamma_n \left[\phi_L p^* \bar{\rho}_0(p^*)+(1-\phi_L)p^*\bar{\rho}_c(p^*)+\mu p^* \delta_n\right]\\
   				& >p^*\bar{\rho}_c(p^*)(1-\phi_L-\mu(\phi_H-\phi_L)) + \gamma_np^*\left(\phi_L  \bar{\rho}_0(p^*)+ \mu(\phi_H-\phi_L)\bar{\rho}_c(p^*) \right)  \\
   				&+\gamma_n\mu p^* \delta_n-(1-\gamma_n) \mu p(\delta-\delta_n), 
   			\end{split}
   		\end{equation*}
   		where the first inequality follows from strict concavity  of $p\bar{\rho}_c(p)$ and $p\bar{\rho}_0(p)$. Finally, we choose the speed of convergence of $p_n$ by taking $\gamma_n$ that satisfies:
   		\begin{equation*}
   			\gamma_n>\frac{\mu p (\delta-\delta_n)}{\phi_L  \bar{\rho}_0(p^*)+ \mu(\phi_H-\phi_L)\bar{\rho}_c(p^*)},
   		\end{equation*}
   		which allows us to conclude that $L((\alpha_{\phi_L},\alpha_{\phi_H},p)_n)>p^*\bar{\rho}_c(p^*)(1-\phi_L-\mu(\phi_H-\phi_L))$ for all $n$ large enough such that $\delta_n<\delta$.
   	\end{enumerate}
   	It follows that $(\alpha_{\phi_L},\alpha_{\phi_H},p)_n \rightarrow (\alpha_{\phi_L},\alpha_{\phi_H},p)$ such that $(\alpha_{\phi_L},\alpha_{\phi_H},p)_n \in F(\mu,\delta_n)$ for all large enough~$n$, and we thus have lower hemicontinuity.
   \end{proof}

\textbf{Step 3.} Fix $\epsilon>0$. There exists $\bar{\delta}>0$ such that for all $\mu \in (0,1)$ one has $W^C(\mu,\bar{\delta}) - W^C(\mu,0) < \epsilon$.

\begin{proof}
	Fix $\epsilon>0$ and denote the objective of \eqref{PlatformProblem:relaxedOptimizationProblem} by
	\begin{equation*}
		Obj(\alpha_{\phi_L},\alpha_{\phi_H},p,\mu) = \E_{\phi}\left(\phi \alpha_{\phi} \bar{W}_0(p) +\phi(1- \alpha_{\phi} )\bar{W}_{\text{out}}  + (1-\phi)\bar{W}_c(p)|\mu \right).
	\end{equation*}
	By Assumption~\ref{Assump:ConsSurplus}, the objective is Lipschitz continuous. Moreover, denote an optimal solution to $W^C(\mu,\delta)$ by
		$\left(\alpha_{\phi_L}(\mu,\delta),\alpha_{\phi_H}(\mu,\delta),p(\mu,\delta)\right)$. 
	By step 2 and Lipschitz continuity of the objective, there exists $\bar{\delta}>0$ such that for $\mu\in[0,1]$, we can choose $\left(\tilde{\alpha}_{\phi_L},\tilde{\alpha}_{\phi_H},\tilde{p}\right)\in F(\mu,0)$ such that:
	\begin{equation*}
		Obj(\alpha_{\phi_L}(\mu,\bar{\delta}),\alpha_{\phi_H}(\mu,\bar{\delta}),p(\mu,\bar{\delta}),\mu)-Obj(\tilde{\alpha}_{\phi_L},\tilde{\alpha}_{\phi_H},\tilde{p},\mu)<\epsilon,
	\end{equation*}
	which implies that
	\begin{equation*}
		W^C(\mu,\bar{\delta})< Obj(\tilde{\alpha}_{\phi_L},\tilde{\alpha}_{\phi_H},\tilde{p},\mu) +\epsilon < W^C(\mu,0) + \epsilon,
	\end{equation*}
for all $\mu\in(0,1)$, as desired.
\end{proof}


To conclude the proof, note that by step 3, we have that if $\mu \in M^{\alpha_t}(\epsilon)$, then there exists $\bar{\delta}>0$ such that
\begin{equation*}
	\E_{a_t,p_t,\custtype_t,\phi} \left(W(p_t,a_t,\custtype_t)\middle|\alpha_t,\boldsymbol{\pi}^*,\mu\right)>co(W^{C})(\mu)+\epsilon>W^C(\mu,0) +\epsilon>W^C(\mu,\bar{\delta}).
\end{equation*}
This implies that at time $t$, the tuple $\left(\alpha_{\phi_L},\alpha_{\phi_H},p_t\right)$ is infeasible in $W^C(\mu,\bar{\delta})$, so one of the constraints in \eqref{PlatformProblem:relaxedOptimizationProblem} must be violated. Since $p_t$ is chosen according to $\pi^*_t$, it maximizes the seller's present revenue and therefore the first constraint in \eqref{PlatformProblem:relaxedOptimizationProblem} is satisfied. It follows then that the second constraint is not satisfied, which implies \eqref{eq: ConstraintViolationLemma2}. This completes the proof of the Lemma. 

\subsubsection*{Proof of Lemma \ref{lemma:generalizedHarrison}} 
 
 Let us denote for $i=L,H$:
 	\begin{equation*}
 		\rho_t^i := \phi_i\alpha_t(p_t,\phi_i,\mu_t)\bar{\rho}_0(p_t) + (1-\phi_i)\bar{\rho}_c(p_t),
 	\end{equation*}
 	where $p_t = \pi_t^*(\mu_t)$.

    Consider the first inequality in the statement of the lemma (i.e., conditioned on $\phi=\phi_L$). The proof of the second follows nearly verbatim. Let $\E^L$ indicate that we are taking expectation conditional on $\phi=\phi_L$. Assume that $\sigma$ is uninformative (we will incorporate this adjustment at the end) and consider the evolution of the seller's belief from the first period onwards.
    
    Let us define $\{\tilde{\mu}_t\}_{t=1}^T$ as an alternative belief process that only updates at the periods where $\mu_t \in M^{\alpha_t}(\epsilon)$. Formally, let $\tilde{\mu}_1 = \mu_1$ and for $t\geq 2$ define:
    \begin{equation*}
    	\tilde{\mu}_t = \begin{cases}
    		\tilde{\mu}_{t-1}, & \text{ if $\mu_t \notin M^{\alpha_t}(\epsilon)$,}\\
    		\frac{\tilde{\mu}_{t-1}\left(\frac{\rho_t^H}{\rho_t^L}\right)^{y_t}}{\tilde{\mu}_{t-1}\left(\frac{\rho_t^H}{\rho_t^L}\right)^{y_t} + (1-\tilde{\mu}_{t-1})\left(\frac{1-\rho_t^L}{1-\rho_t^H}\right)^{1-y_t}},& \text{ if $\mu_t \in M^{\alpha_t}(\epsilon)$.}
    	\end{cases}
    \end{equation*}

	We claim that there exist constants $\xi,\beta>0$ such that
	\begin{equation*}
		 \E^L(\mu_{t_n} )\leq\E^L(\tilde{\mu}_{t_n} )\leq \chi \exp(-\beta n).
	\end{equation*}
	The first inequality follows by noting that $\E^L(\mu_t )\leq\E^L(\tilde{\mu}_t )$ for all $t\geq 1$. This follows from noting that the Bayesian updated belief process $\{\mu_t\}$ is a supermartingale with respect to $\E^L$ and by observing that this process is updated at all periods so it considers more information than the process $\{\tilde{\mu}_t\}$ at any given time (including the information that is considered to update $\{\tilde{\mu}_t\}$).
    
    To establish the second inequality, we draw from the proof of Lemma A.1 in \cite{Harrison2012}  (see equation (A4) in their paper and the following equation). We have that:
    $$\E^L(\tilde{\mu}_{t_n}) = \E^L \left(\frac{1}{1+\left(\frac{1-\mu_0}{\mu_0}\right)\exp(L_{n})}\right),$$
    where
    \begin{equation}
        \label{def:L_tn}
        L_{n} = \sum_{j=1}^{n}(y_{t_j} - \rho_{t_j}^L)\log\left(\frac{\rho_{t_j}^L(1-\rho_{t_j}^H)}{\rho_{t_j}^H(1-\rho_{t_j}^L)}\right)+\sum_{j=1}^{n}\left( (1-\rho_{t_j}^L)\log\left(\frac{1-\rho_{t_j}^L}{1-\rho_{t_j}^H}\right)+\rho_{t_j}^L\log\left(\frac{\rho_{t_j}^L}{\rho_{t_j}^H}\right) \right).
    \end{equation}
	Notice that, by Lemma \ref{lemma:delta_epsilonRelationCoU}, one may fix $\delta>0$ such that if $\mu_t \in M^{\alpha_t}(\epsilon)$, then: $ |\rho_t^H - \rho_t^L|>\delta.$ Thus, we have that $ |\rho_{t_j}^H - \rho_{t_j}^L|>\delta$ for all $j\geq 1$. Therefore, we can directly invoke Lemma A.1. in \cite{Harrison2012}, so that there exist constants $\xi,\beta>0$ such that $\E^L(\tilde{\mu}_{t_n} )\leq \chi \exp(-\beta n)$.

Now we consider the evolution of the seller's belief accounting for the platform's opportunity to use a signal. Thus, we take expectation over the signal $s$ which is chosen according to some signaling mechanism $\sigma$. Fix~$\sigma \in \Sigma^S$, which by Proposition \ref{lemma:RevenueEquivalence}, is without loss of optimality. Thus $\mu_1$ can take two values which we denote: $\underline{\mu}= \mu_1\left(\langle s= L \rangle\right)\leq \mu_0 \leq \mu_1\left(\langle s=H \rangle\right) = \overline{\mu} $ which, using Bayes' rule and algebra, implies that:
\begin{align*}
    \Prob(s=L|\phi = \phi_L) = \frac{(1-\underline{\mu})(\overline{\mu}-\mu_0)}{(1-\mu_0)(\overline{\mu}-\underline{\mu})},& & \Prob(s=L|\phi=\phi_H) = \frac{\underline{\mu}(\overline{\mu}-\mu_0)}{\mu_0(\overline{\mu}-\underline{\mu})},
\end{align*}
where we assume\footnote{Otherwise, the belief converges immediately to the true state and the last step of the proof follows without to adjust the constants in the bound to depend on the posteriors induced by the signal.} that the mechanism $\sigma$ induces posterior beliefs $0<\underline{\mu}<\overline{\mu}<1$.

From the first part of the proof, we have that there exist constants $\bar{\xi},\underline{\xi},\bar{\beta},\underline{\beta}$ such that, by taking expectation over the signal $s$ we have
\begin{align*}
    \E_{s}\left(\mu_{t_n}|\phi=\phi_L\right) &= \Prob(s=H|\phi=\phi_L) \E\left(\mu_{t_n}|\phi=\phi_L,\mu_1 = \overline{\mu}\right) + \Prob(s=L|\phi=\phi_L) \E\left(\mu_{t_n}|\phi=\phi_L,\mu_1 = \underline{\mu}\right) \\
    &\leq  \Prob(s=H|\phi=\phi_L) \bar{\xi}e^{-\bar{\beta} n} + \Prob(s=L|\phi=\phi_L) \underline{\xi}e^{-\underline{\beta} n} \\
    & \leq \max\left\{\bar{\xi},\underline{\xi} \right\}e^{-\min \left\{\bar{\beta},\underline{\beta} \right\} n}
\end{align*}
Defining $\xi = \max\left\{\bar{\xi},\underline{\xi} \right\}$, $\beta = \min \left\{\bar{\beta},\underline{\beta} \right\}$ completes the proof for the case of $\phi=\phi_L$.
The same proof holds when $\phi=\phi_H$, where although the constants $(\chi,\beta)$ may be different, taking the respective maxima among them establishes the result.

\subsubsection*{Proof of Lemma \ref{lemma:ConvergenceOfWC}}



        We first note that for all $\delta \geq 0$, $W^C(\mu,\delta)$ is upper semi-continuous in $\mu$, for $\mu \in (0,1)$ (we defer this argument to the end of the proof). In addition, since the confounding constraint in \eqref{eq:simplifiedOptimizationProblem} need not hold for $\mu \in \{0,1\}$ (see related discussion in Appendix~\ref{app:DesigningSimplePolicies}), we have for all $\delta \geq 0$:
        $$\lim_{\mu\rightarrow 0^+} W^C(\mu,\delta) \leq W^C(0), \quad \text{and} \quad \lim_{\mu\rightarrow 1^-} W^C(\mu,\delta) \leq W^C(1),$$
        where, for $\mu \in \{0,1\}$, $W^C(\mu)$ is as in Appenxix~\ref{app:DesigningSimplePolicies} (i.e., the maximum consumer surplus achieved by a myopic policy in a single period, given $\mu\in\{0,1\}$).
        
        Abusing notation, let us redefine $W^C(\mu,\delta)$ for $\mu\in\{0,1\}$ as $W^C(\mu,\delta) := W^C(\mu)$. By extending the definition of $W^C(\mu,\delta)$ to consider $\mu \in \{0,1\}$, we have that $W^C(\mu,\delta)$ is upper semi-continuous 
        in $\mu$ for all fixed $\delta\geq 0$. Therefore, its concavification $co(W^C)(\mu,\delta)$ is bounded and continuous in $\mu$, given~$\delta\geq0$. In particular, this is the case for $co(W^C)(\mu) = co(W^C)(\mu,0)$
        
        Moreover, it follows from \eqref{eq:myopicPromotion} in \S\ref{subsec:InsufficiencyOfTruthfulDisclosure} and Proposition~\ref{lemma:RevenueEquivalence} that $W^{max}(\mu) = W^C(\mu,\delta)$, for any fixed $\delta>1$, as the second constraint in \eqref{PlatformProblem:relaxedOptimizationProblem} is never binding for $\delta>1$. In particular, by the same reasoning as above we have that $W^{max}(\mu)$ is upper semi-continuous for $\mu \in [0,1]$. In addition, $W^{max}(0) = W^C(0)$ and $W^{max}(1) = W^C(1)$, from where it follows that $co(W^{max})(\mu)$ is bounded and continuous in $\mu$.
        
        Finally, by the previous arguments observe that 	$co(W^{max})(\mu) = co(W^C)(\mu)$ for $\mu \in \{0,1\}$. Since both $co(W^{max})$ and $co(W^C)$ are bounded and continuous functions of $\mu\in [0,1]$ that are equal in the endpoints of the interval, given $\epsilon>0$, we can construct $C_1,C_2>0$ such that
        \begin{equation*}
        co(W^{\max})(\mu) - co(W^{C})(\mu) < \frac{\epsilon}{2} + C_1\mu,\text{ and }  co(W^{\max})(\mu) - co(W^{C})(\mu) < \frac{\epsilon}{2} +C_2(1-\mu).
        \end{equation*}
        Taking $\bar{C} = \max \{C_1,C_2\}$ yields the result.        	        

        We complete the proof by establishing that $W^C(\mu,\delta)$ is upper semi-continuous in $\mu$ for all $\delta\geq 0$ by showing the conditions for Berge's Maximum Theorem hold, i.e. that the objective of \eqref{PlatformProblem:relaxedOptimizationProblem} is  continuous in the problem's decision variables, and that the feasible set correspondence $F(\mu,\delta)$ is non-empty and compact-valued and non-empty, as well as upper hemicontinuous in $\mu$, for a fixed $\delta\geq 0.$

        \textbf{\emph{Proof of Continuity of Objective Function.}}
        The objective of \eqref{PlatformProblem:relaxedOptimizationProblem} is  continuous in $(\mu,\alpha_{\phi_L},\alpha_{\phi_H},p)$ as it is linear in $\mu,\alpha_{\phi_L},\alpha_{\phi_H}$ and is  continuous in $p$ by Assumption \ref{Assump:ConsSurplus}.
        
        \textbf{\emph{Proof of Compactness and Non-emptyness of Feasible Set.}} For any $\mu\in(0,1)$ and $\delta\geq 0$, $F(\mu,\delta)$ is bounded since it is a subset of $[0,1]^2 \times P$, and it is closed since the constraints of \eqref{PlatformProblem:relaxedOptimizationProblem} are defined by continuous functions. Moreover $F(\mu,\delta)\neq \emptyset$ since we always have
        	$\left( \alpha_{\phi_L},\alpha_{\phi_H},p \right) = \left( 0,\left(\frac{\phi_H-\phi_L}{\phi_H}\right) \frac{\bar{\rho}_c(p^*)}{\bar{\rho}_0(p^*)}  ,p^* \right) \in F(\mu,\delta)$.

        \textbf{\emph{Proof of Upper Hemicontinuity.}} As in the proof of Lemma~\ref{lemma:delta_epsilonRelationCoU}, this follows by noting that $F$ satisfies the closed graph property.
                This concludes the proof of the Lemma.

\subsection{Proof of Theorem~\ref{prop: noinfodesign}}\label{app: NoInfoDesign}


We now consider the setting where the platform may not employ a signaling mechanism to communicate information about the underlying state $\phi$ to the seller, and therefore the platform's only lever for optimization is its promotion policy.
Our analysis here therefore quantifies the value of the platform's signaling ability by comparing the maximum achievable long-run average consumer surplus in this setting with the corresponding one for the baseline model presented in~\S\ref{sec:Model}, where the platform may send an informative signal about the value of~$\phi$ at the beginning of the horizon.

As established in Theorem~\ref{thm:LongRunAverageOptimalConsumerSurplus}, joint optimization of the signaling mechanism and the promotion policy allows the platform to generate an optimal long-run average consumer surplus of $co(W^C)(\mu_0)$. This is achieved by
optimally designing the initial information signal so that the seller's prior belief $\mu_0$ is updated to some (possibly different) value $\mu_1$ at the beginning of the first period, and employing an optimal confounding promotion policy thereafter. In contrast, removing the platform's signaling ability constraints the seller's belief at the beginning of the first period to be equal to the prior belief $\mu_0$.

We establish that, interestingly, even without the ability to send an initial signal to the seller, the platform can still design a promotion policy that approximately generates the maximal long-run average consumer surplus of $co(W^C)(\mu_0)$. This finding is formalized in Theorem~\ref{prop: noinfodesign}, which was presented in~\S\ref{sec: no_signaling}. 
 Nevertheless, removing the ability to send an initial informative signal to the seller generates three main differences with respect to the baseline setting.

First, the platform must design a more complex promotion policy to approximate the long-run average payoff of $co(W^C)(\mu_0)$ that it would if it had signaling ability. As we establish in Theorem~\ref{thm:LongRunAverageOptimalConsumerSurplus} and discuss in \S\ref{sec: OptSimplePolicies}, when the platform has the ability to optimize over signaling mechanism as well, it can construct a \emph{static} policy that guarantees the maximum achievable long-run average payoff. Without signaling ability this is not true anymore, and the platform may require a dynamic promotion policy instead (the proof of Theorem~\ref{prop: noinfodesign} details the design of such policy).

Second, the dynamic promotion policy that we construct to approximate the long-run average payoff of $co(W^C)(\mu_0)$ need not be confounding in all periods. In fact, under this policy the seller's belief may change until it approaches a certain target, after which the promotion policy becomes confounding and the seller's belief remains constant. Thus, as without an initial signal the platform has no ability to control the starting point of the seller's belief process, it may need to control the drift of this process over time. This is in contrast with the baseline setting, in which the platform designs the initial point of the seller's belief process via the signaling mechanism, and thereafter ensures the belief does not change (see related discussion in  \S\ref{sec:longrunoptimalsurplus}). Importantly, to be able to implement the proposed (dynamic and non-confounding) policy, the platform must track the seller's belief for a number of periods that depends on the sample path of sales, resulting in an additional nontrivial burden relative to the setting with signaling ability.

Finally, as the dynamic promotion policy that may be non-confounding and the seller's belief may change over time, the per-period consumer welfare generated by the platform may be non-stationary. This is in contrast with the baseline setting where the seller's belief stays constant after the initial signal is realized and the expected consumer surplus is fixed thereafter. 

\subsubsection{Near-Optimal Policy Design Without Signaling Ability} \label{sec: NoSignalingOptimalWelfare}

We next establish that even without the ability to send an initial informative signal to the seller, the platform can construct a promotion policy that approximately generates a long-run average consumer surplus of $co(W^C)(\mu_0)$ which, by Theorem~\ref{thm:LongRunAverageOptimalConsumerSurplus}, is a tight upper bound on the maximal long-run average consumer surplus achievable by the platform. Note that removing the ability to send an initial signal is outcome-equivalent to assuming that the platform employs an uninformative signaling mechanism in the baseline setting.

\begin{manualthm}{\ref{prop: noinfodesign}}
	Fix $\epsilon>0$, and suppose that the platform's signaling mechanism is uninformative (i.e., $\sigma = \sigma^U$). Then, given $\mu_0\in [0,1]$, there exists a promotion policy $\bA$ such that
	\begin{equation*}
		\lim_{T\rightarrow \infty} ~~\frac{1}{T}W_T^{\bA,\sigma^U,\boldsymbol{\pi^*}}(\mu_0) \geq co(W^C)(\mu_0) -\epsilon.
	\end{equation*}
\end{manualthm}

\begin{proof}[Proof of Theorem~\ref{prop: noinfodesign}]
	
To prove the theorem, we construct a promotion policy that generates an expected long-run average consumer surplus of at least $co(W^C)(\mu_0)-\epsilon$, provided a fixed $\epsilon>0$.

Let $\epsilon>0$. If $W^C(\mu_0) = co(W^C)(\mu_0)$, let $\bA$ be the constant policy where in each period, $\alpha_t$ is a confounding policy that maximizes consumer surplus (e.g., the simple confounding policy defined as the solution to~\eqref{eq: ConfoundingProblemSimple} in \S\ref{sec: OptSimplePolicies}). By definition of a confounding policy, we have that $\mu_t = \mu_0$ for all $t=1,\dots,T$ and therefore
	\begin{equation*}
		\lim_{T\rightarrow \infty} ~~\frac{1}{T}W_T^{\bA,\sigma^U,\boldsymbol{\pi^*}}(\mu_0) = co(W^C)(\mu_0).
	\end{equation*}
	Suppose instead that $W^C(\mu_0) < co(W^C)(\mu_0)$. In particular, since by definition of the concavification operator we have that $co(W^C)(0) = W^C(0)$ and $co(W^C)(1) = W^C(1)$, this assumption implies that $0<\mu_0<1$. The rest of the proof consists of six steps.
	
	\textbf{Step 1 (Continuity conditions).} In this step we establish that there exist constants $\mu_L,\mu_H\in[0,1]$ such that $\mu_L<\mu_0<\mu_H$, $W^C(\mu_L) = co(W^C)(\mu_L)$ and $W^C(\mu_H) = co(W^C)(\mu_H)$. In addition, there exists $\delta>0$ such that if $|\mu-\mu_L|\leq\delta$ or $|\mu-\mu_H|\leq\delta$, one has that $|W^C(\mu)-W^C(\mu_j)|\leq\epsilon/2$
	
To show this, we propose the following constants:
	\begin{equation} \label{eq: mu_L_and_mu_H}
		\begin{split}
			\mu_L = \sup \left \lbrace \mu\leq \mu_0 :\,W^C(\mu) = co(W^C)(\mu) \right \rbrace,\\
			\mu_H = \inf \left \lbrace \mu\geq \mu_0 :\,W^C(\mu) = co(W^C)(\mu) \right \rbrace.
		\end{split}
	\end{equation}

Since we know that $co(W^C)(0) = W^C(0)$ and $co(W^C)(1) = W^C(1)$ (as explained above), it follows that $\mu_L \geq 0$ and $\mu_H \leq 1$. To show that these constant satisfy the remaining properties required by the claim, we rely of the following lemma that establishes continuity of $W^C(\mu)$ for $\mu \in (0,1)$. We defer the proof of the Lemma for the end of this section.
	
	\begin{lemma} \label{lem: continuity_W^C}
		Define $W^C(\mu)$ as in~\eqref{PlatformProblem:SinglePeriodConfounding}. Then, $W^C(\mu)$ is a continuous function of $\mu \in (0,1)$.
	\end{lemma}

In addition, since $co(W^C)$ is a concave function of $\mu\in [0,1]$, it is continuous in $(0,1)$. Thus, by Lemma~\ref{lem: continuity_W^C}, the expression $co(W^C)(\mu)-W^C(\mu)$ is continuous for $\mu\in(0,1)$. Therefore, it follows from~\eqref{eq: mu_L_and_mu_H} that $\mu_L<\mu_0<\mu_H$, $W^C(\mu_L) = co(W^C)(\mu_L)$ and $W^C(\mu_H) = co(W^C)(\mu_H)$.
	
	
	
	Moreover, by continuity of $W^C$ (Lemma~\ref{lem: continuity_W^C}), we can choose $\delta>0$ such that for $j\in\{L,H\}$, if $|\mu-\mu_j|\leq\delta$, we have that $|W^C(\mu)-W^C(\mu_j)|\leq\epsilon/2$. 
	
	\textbf{Step 2 (Construction of promotion policy with drift on seller's belief).} 

In this step we construct a policy $\tilde{\bA}$ such that the induced process of seller's beliefs $\{\mu_t\}$ satisfies $0<|\mu_{t+1}-\mu_{t}|\leq 2\delta$ for all periods $t=1,\dots,T$.
	
We construct the proposed policy by solving the following optimization problem. Given $\mu \in(0,1)$, $0<b<1$, and $\xi>0$, consider the problem given by
	\begin{equation} \label{eq: PolicyConstructionProblem}
		\begin{split}
			W^{aux}(\mu,\xi,b)&:=\max_{\substack{\alpha_{\phi_H},\alpha_{\phi_L} \in [0,1],\\p \in P}} ~~ \E_{\phi}\left(\phi \alpha_{\phi} \bar{W}_0(p) +\phi(1- \alpha_{\phi} )\bar{W}_c  + (1-\phi)\bar{W}_c(p)|\mu\right) \\
			\text{s.t.}&~~ p\bar{\rho}_0(p)(\phi_L\alpha_{\phi_L}(1-\mu) +\phi_H\alpha_{\phi_H}\mu) +  p\bar{\rho}_c(p)(1-\phi_L-\mu(\phi_H-\phi_L)) \geq \\
			&\quad \quad \quad p^*\bar{\rho}_c(p^*)(1-\phi_L-\mu(\phi_H-\phi_L)),\\
			&~~ b\xi\leq |\phi_H\alpha_{\phi_H}\bar{\rho}_0(p) + (1-\phi_H)\bar{\rho}_c(p) - \phi_L\alpha_{\phi_L}\bar{\rho}_0(p) - (1-\phi_L)\bar{\rho}_c(p)| \leq  \xi.
		\end{split}
	\end{equation}
	Note that for $b=0$, the problem is identical to $W^C(\mu,\xi)$ as defined by~\eqref{PlatformProblem:relaxedOptimizationProblem} in \S\ref{sec:Theorem1Proof:AuxProofs}. Moreover, we claim that problem \eqref{eq: PolicyConstructionProblem} is feasible for all small enough $b>0$. This can be seen by taking $p=p^*$, $\alpha_{\phi_L} = 0$ and
	\begin{equation*}
		\alpha_{\phi_H} = \left(\frac{\phi_H-\phi_L}{\phi_H}\right)\left(\frac{\bar{\rho}_c(p^*)}{\bar{\rho}_0(p^*)}\right) + \frac{b\xi}{\phi_H\bar{\rho}_0(p^*)},
	\end{equation*}

which satisfy the constraints of \eqref{eq: PolicyConstructionProblem} provided that $b$ is small enough.

	
	Moreover, for a single-price promotion policy, we can write the seller's beliefs at period $t$ as follows (see e.g., the proof of Lemma A.1. in \cite{Harrison2012}):
	\begin{equation*}
		\mu_{t} =  \frac{1}{ 1+\left(\frac{1-\mu_0}{\mu_0}\right)\exp(L_{t})},
	\end{equation*}
	where $ L_t = \sum_{j=1}^{t} \gamma_t$, and
	\begin{equation*}
		\gamma_t = (1-y_t)\log \left(\frac{1-\rho_t^L}{1-\rho_t^H}\right) + y_t \log \left(\frac{\rho_t^H}{\rho_t^L}\right),
	\end{equation*}
	where we denote for $i=L,H$:
	\begin{equation*}
		\rho_t^i := \phi_i\alpha_t(p_t,\phi_i,\mu_t)\bar{\rho}_0(p_t) + (1-\phi_i)\bar{\rho}_c(p_t),
	\end{equation*}
with $p_t = \pi_t^*(\mu_t)$. Note that $\mu_t$ is a continuous function of $L_t$, and thus there exists $\xi_0>0$ such that if~$|\gamma_{t+1}|<\xi_0$, one has that $|\mu_{t+1}-\mu_t|\leq 2\delta$. Furthermore, since $y_t\in\{0,1\}$, ensuring that $|\gamma_{t+1}|\leq \xi_0$ is equivalent to requiring that $|\rho_t^L-\rho_t^H|\leq \xi_1$ for some $\xi_1>0$.
	
Thus, by defining $\tilde{\bA}$ as a static and single-price policy induced by solving $W^{aux}(\mu,\xi_1,b)$ for some small $b>0$, one has that the belief process induced by $\tilde{\bA}$ satisfies: (i) $|\mu_{t+1}-\mu_{t}|\leq 2\delta$; and (ii) $b\xi\leq |\rho_t^L-\rho_t^H|$ at all times $t=1,\dots,T$.

\textbf{Step 3 (Construction of stopping times).} Fix the promotion policy defined in Step 2. For $j\in\{L,H\}$, define the stopping times $\tau_j = \min\{t:\mu_t \in [\mu_j-\delta,\mu_j+\delta]\}$, and $\tau=~\min \{\tau_L,\tau_H\}$, where $\delta$ is as in step 1. Then, we next show that $\E[\tau]<\infty$, and $\tau$ is finite almost surely.
	
To prove this, note that $\E[\tau] = \mu_0\E[\tau | \phi = \phi_H]+ (1-\mu_0)\E[\tau | \phi = \phi_L]$. We will show that $\E[\tau | \phi = \phi_L]<\infty$ (a similar argument proves that $\E[\tau | \phi = \phi_H]<\infty$).
	
Fix $a\in(0,1)$ such that $a<\mu_L+\delta$, and define $\tau_a$ as the stopping time representing the first period such that $\mu_t\in[0,a]$, i.e., $\tau_a = \inf \left \lbrace t\geq 1:\, \mu_t\in[0,a] \right \rbrace$. Since $|\mu_{t+1}-\mu_t|\leq 2\delta$ (by step 1) and $\mu_0>\mu_L$, it follows that $\tau_a\geq \tau_L$. Therefore,
	\begin{equation*}
		\begin{split}
			\E[\tau | \phi = \phi_L] &\leq \E[\tau_L | \phi = \phi_L] \\
			& \leq \E[\tau_a | \phi = \phi_L] = \sum_{k=1}^{\infty}\Prob\left[\tau_a\geq k | \phi=\phi_L \right] \stackrel{(i)}{\leq} \frac{1}{a}\sum_{k=1}^{\infty}\E\left[\mu_k | \phi=\phi_L \right]\stackrel{(ii)}{\leq} \frac{1}{a}\sum_{k=1}^{\infty} \chi e^{-\lambda k}<\infty,
		\end{split}
	\end{equation*}
	for some constants $\chi,\lambda>0$. Here, (i) follows by applying Markov inequality and (ii) follows from Lemma A.1 in \cite{Harrison2012}, where the constants $\chi,\lambda>0$ exist since, in the notation of step 2, we have that $0<b\xi\leq |\rho_t^L-\rho_t^H|$ for all $t$ and therefore Lemma A.1 in \cite{Harrison2012} applies directly to the belief process induced by $\tilde{\bA}$. We conclude that $\E\left[\tau\right]<\infty$ and, as a result, $\tau<\infty$ almost surely.

\textbf{Step 4 (Bounds for stopping times' hitting probabilities).} In this step we establish the following bounds:
	\begin{equation*}
		\frac{\mu_H-\mu_0-\delta}{\mu_H-\mu_L}\leq \Prob \left[\tau = \tau_L\right] \leq \frac{\mu_H-\mu_0+\delta}{\mu_H-\mu_L}.
	\end{equation*}

To prove these bounds, note that the belief process $\{\mu_t\}$ induced by the policy $\tilde{\bA}$ defined in Step 2 is a martingale with bounded increments, since it is defined by Bayesian updating. Moreover, the stopping time $\tau$ has finite expectation by step 3. Thus, we can apply the optional stopping theorem to derive that~$\E[\mu_\tau] = \mu_0$. This implies that
	\begin{equation*}
		\begin{split}
			\mu_0&=\Prob \left[\tau = \tau_L\right]\E[\mu_\tau | \tau = \tau_L]+\left(1-\Prob \left[\tau = \tau_L\right]\right)\E[\mu_\tau | \tau = \tau_H]\\
			& \geq \Prob \left[\tau = \tau_L\right](\mu_L-\delta)+\left(1-\Prob \left[\tau = \tau_L\right]\right)(\mu_H-\delta),
		\end{split}
	\end{equation*}
	which implies that
	\begin{equation*}
		\frac{\mu_H-\mu_0-\delta}{\mu_H-\mu_L}\leq \Prob \left[\tau = \tau_L\right].
	\end{equation*}
	Similarly, we have that $\mu_0\leq \Prob \left[\tau = \tau_L\right](\mu_L+\delta)+\left(1-\Prob \left[\tau = \tau_L\right]\right)(\mu_H+\delta)$, which implies that
	\begin{equation*}
		\Prob \left[\tau = \tau_L\right] \leq \frac{\mu_H-\mu_0+\delta}{\mu_H-\mu_L}.
	\end{equation*}
	
\textbf{Step 5 (Welfare at stopping time).} In this step we show that $\E\left[W^C(\mu_\tau)\right] \geq co(W^C)(\mu_0)-\epsilon$. Note that by definition of $\mu_L$ and $\mu_H$ (see~\eqref{eq: mu_L_and_mu_H}), we have that $co(W^C)(\mu)>W^C(\mu)$ for all $\mu \in (\mu_L,\mu_H)$. Thus,~$co(W^C)(\mu)$ is an affine function for $\mu \in [\mu_L,\mu_H]$, from where it follows that
	\begin{equation*}
		co(W^C)(\mu_0) = \frac{\mu_H-\mu_0}{\mu_H-\mu_L}W^C(\mu_L) + \frac{\mu_0-\mu_L}{\mu_H-\mu_L}W^C(\mu_H).
	\end{equation*}
	Therefore,
	\begin{equation*}
		\begin{split}
			\E\left[W^C(\mu_\tau)\right]& = \Prob \left[\tau = \tau_L\right]\E\left[W^C(\mu_\tau)|\tau=\tau_L\right] + \Prob \left[\tau = \tau_H\right]\E\left[W^C(\mu_\tau)|\tau=\tau_H\right]\\
			& \stackrel{(a)}{\geq} \Prob \left[\tau = \tau_L\right] \left(W^C(\mu_L)-\epsilon/2\right)+\Prob \left[\tau = \tau_H\right] \left(W^C(\mu_H)-\epsilon/2\right)\\
			& \stackrel{(b)}{\geq} co(W^C)(\mu_0)-\frac{\epsilon}{2} - \frac{\delta}{\mu_H-\mu_L}\left(W^C(\mu_L)+W^C(\mu_H)\right),
		\end{split}
	\end{equation*}
where (a) follows by steps 1 and 3, and (b) follows by step 4. Without loss of generality, we can choose $\delta>0$ to be small enough so that the last term above is larger than $-\epsilon/2$. Thus, we have that $\E\left[W^C(\mu_\tau)\right] \geq co(W^C)(\mu_0)-\epsilon$.

\textbf{Step 6 (Definition of final policy).} In this step we construct a promotion policy $\bA$ that satisfies Theorem~\ref{prop: noinfodesign}. Define the policy $\bA$ by:
	\begin{equation*}
		\alpha_t = \begin{cases}
			\tilde{\alpha}_t, & \text{ if $t<\tau$}\\
			\alpha^C, & \text{ if $t\geq \tau$,}
		\end{cases}
	\end{equation*}
	where $\alpha^C$ is the simple confounding promotion policy constructed in the proof of Theorem~\ref{thm:LongRunAverageOptimalConsumerSurplus} (i.e., as in \S\ref{sec: OptSimplePolicies}) that, in particular, generates an expected per-period welfare of $W^C(\mu_t)$.
	Then, we have that for all $K\leq T$,
	
	\begin{equation} \label{eq: split_W_benchmark}
		\begin{split}
			\frac{1}{T}W_T^{\bA,\sigma^U,\boldsymbol{\pi^*}}(\mu_0)&  = \frac{1}{T}\E\left(\sum_{t=1}^T  W(p_{t},a_{t},\psi_{t}) \middle|\boldsymbol{\alpha},\sigma^U,\boldsymbol{\pi}^*,\mu_0\right)\\
			& = \frac{1}{T}\E\left(\sum_{t=1}^T  W(p_{t},a_{t},\psi_{t}) \middle|\boldsymbol{\alpha},\sigma^U,\boldsymbol{\pi}^*,\mu_0,\tau\leq K\right) \Prob\left(\tau\leq K\right) \\
			&+\frac{1}{T}\E\left(\sum_{t=1}^T  W(p_{t},a_{t},\psi_{t}) \middle|\boldsymbol{\alpha},\sigma^U,\boldsymbol{\pi}^*,\mu_0,\tau>K\right) \Prob\left(\tau > K\right).
		\end{split}
	\end{equation}
	To bound the second term, note that
	\begin{equation*}
		\frac{1}{T}\E\left(\sum_{t=1}^T  W(p_{t},a_{t},\psi_{t}) \middle|\boldsymbol{\alpha},\sigma^U,\boldsymbol{\pi}^*,\mu_0,\tau\geq K\right) \Prob\left(\tau > K\right)\geq W_{min}\Prob\left(\tau> K\right),
	\end{equation*}
	where we denote $W_{min} = \inf \left \lbrace W^{aux}(\mu,\xi_1,b):\,\mu\in[m_L+\delta,\mu_H-\delta] \right \rbrace$, with $\xi_1$ and $b$ as in step 1.
	
	For the first term, by definition of $\bA$ we have that
	\begin{equation*}
		\begin{split}
			\frac{1}{T}\E\left(\sum_{t=1}^T  W(p_{t},a_{t},\psi_{t}) \middle|\boldsymbol{\alpha},\sigma^U,\boldsymbol{\pi}^*,\mu_0,\tau\leq K\right)& = \frac{1}{T}\E\left(\sum_{t=1}^{\tau-1}  W(p_{t},a_{t},\psi_{t}) + \sum_{t=\tau}^{T}  W(p_{t},a_{t},\psi_{t}) \middle|\boldsymbol{\alpha},\sigma^U,\boldsymbol{\pi}^*,\mu_0,\tau\leq K\right)\\
			& \geq  \frac{(K-1)W_{min}}{T} + \frac{T-K}{T}\E\left(W^C(\mu_\tau)\mid \mu_0\right)\\
			& \geq \frac{(K-1)W_{min}}{T} + \frac{T-K}{T} \left(co(W^C)(\mu_0)-\epsilon\right),
		\end{split}
	\end{equation*}
	where the last inequality follows by step 5. Plugging back into \eqref{eq: split_W_benchmark} and taking the limit as $T\to \infty$ yields
	\begin{equation*}
		\lim_{T\rightarrow \infty} ~~\frac{1}{T}W_T^{\bA,\sigma^U,\boldsymbol{\pi^*}}(\mu_0) \geq \Prob\left(\tau\leq K\right) \left( co(W^C)(\mu_0)-\epsilon \right) + \Prob\left(\tau> K\right)W_{min},
	\end{equation*}
	for all $K>0$. Since $\tau<\infty$ almost surely, by taking $K\to \infty$ on the right hand side, we conclude that
	\begin{equation*}
		\lim_{T\rightarrow \infty} ~~\frac{1}{T}W_T^{\bA,\sigma^U,\boldsymbol{\pi^*}}(\mu_0) \geq  co(W^C)(\mu_0)-\epsilon,
	\end{equation*}
which concludes the proof.\end{proof}

To conclude this section, we prove Lemma~\ref{lem: continuity_W^C}, which was an used as an auxiliary result in step 1 above.

\begin{proof}[Proof of Lemma~\ref{lem: continuity_W^C}]
	Note that we can write
	\begin{equation}\label{eq: W^C(mu)}
		W^C(\mu) = \max \left \lbrace \E_{\phi}\left(\phi \alpha_{\phi} \bar{W}_0(p) +\phi(1- \alpha_{\phi} )\bar{W}_c  + (1-\phi)\bar{W}_c(p)|\mu\right):\, \left(\alpha_{\phi_H},\alpha_{\phi_L},p\right) \in F(\mu) \right \rbrace,
	\end{equation}
	where the feasible set correspondence $F(\mu)$ is defined as the set of vectors $\left(\alpha_{\phi_H},\alpha_{\phi_L},p\right) \in[0,1]\times [0,1]\times \mathcal{P}$ that satisfy $G_1\left(\alpha_{\phi_H},\alpha_{\phi_L},p,\mu\right) \geq 0,$ and $G_2\left(\alpha_{\phi_H},\alpha_{\phi_L},p\right) = 0,$  where
	
	\begin{equation*}
		\begin{split}
			G_1\left(\alpha_{\phi_H},\alpha_{\phi_L},p,\mu\right)&= p\bar{\rho}_0(p)(\phi_L\alpha_{\phi_L}(1-\mu) +\phi_H\alpha_{\phi_H}\mu) +  \left(p\bar{\rho}_c(p) -p^*\bar{\rho}_c(p^*) \right)(1-\phi_L-\mu(\phi_H-\phi_L)),\\
			G_2\left(\alpha_{\phi_H},\alpha_{\phi_L},p\right)&=
			\left(\phi_H\alpha_{\phi_H}-\phi_L\alpha_{\phi_L}\right)\bar{\rho}_0(p) - (\phi_H-\phi_L)\bar{\rho}_c(p).
		\end{split}
	\end{equation*}
	To show that $W^C(\mu)$ is continuous in $\mu$, it suffices to show that the conditions for Berge's Maximum Theorem are satisfied for \eqref{eq: W^C(mu)}. We have shown that the objective function of \eqref{eq: W^C(mu)} is continuous, and that the feasible set correspondence is non-empty, compact-valued and upper hemicontinuous in the proof of Lemma~\ref{lemma:ConvergenceOfWC}, so it remains to show that $F(\mu)$ is lower hemicontinuous in $\mu$.
	
	To establish lower hemicontinuity, let $\left(\mu_n\right)_{n\in \mathbb{N}} \subseteq (0,1)$ be a sequence that converges to some $\mu\in(0,1)$, and let $(\alpha_{\phi_L},\alpha_{\phi_H},p) \in F(\mu)$. We need to construct a sequence $(\alpha_{\phi_L},\alpha_{\phi_H},p)_n \rightarrow (\alpha_{\phi_L},\alpha_{\phi_H},p)$ such that~$(\alpha_{\phi_L},\alpha_{\phi_H},p)_n \in F(\mu_n)$ for all large enough $n$.
	
	If $G_1\left(\alpha_{\phi_H},\alpha_{\phi_L},p,\mu\right)> 0$, then we can simply take $(\alpha_{\phi_L},\alpha_{\phi_H},p)_n = (\alpha_{\phi_L},\alpha_{\phi_H},p)$ for all $n$. By continuity of $G_1$, this sequence satsifies the desired property.
	
	Suppose then that the first constraint is binding, i.e., $G_1\left(\alpha_{\phi_H},\alpha_{\phi_L},p,\mu\right)=0$. Note that since $p\bar{\rho}_c(p) \leq p^*\bar{\rho}_c(p^*)$ (by definition of $p^*$), this implies that $\bar{\rho}_0(p)>0$. Then, from the second constraint we have that
	\begin{equation*}
		\alpha_{\phi_H} = \frac{\phi_L}{\phi_H}\alpha_{\phi_L} + \left(1-\frac{\phi_L}{\phi_H}\right) \frac{\bar{\rho}_c(p)}{\bar{\rho}_0(p)}.
	\end{equation*}
	Plugging this expression into $G_1$ yields
	\begin{equation*}
		G_1\left(\alpha_{\phi_H},\alpha_{\phi_L},p,\mu\right) = p\bar{\rho}_0(p)\phi_L \alpha_{\phi_L} + (\phi_H-\phi_L) p \bar{\rho}_c(p)+\left(p\bar{\rho}_c(p) -p^*\bar{\rho}_c(p^*) \right)(1-\phi_L-\mu(\phi_H-\phi_L)).
	\end{equation*}
	We propose the following sequence:
	\begin{equation*}
		\left(\alpha_{\phi_H},\alpha_{\phi_L},p\right)_n = \left( \frac{\phi_L}{\phi_H}\alpha_{\phi_L} + \left(1-\frac{\phi_L}{\phi_H}\right) \frac{\bar{\rho}_c(p_n)}{\bar{\rho}_0(p_n)},\alpha_{\phi_L},(1-\gamma_n)p+\gamma_np^*\right),
	\end{equation*}
	where the sequence $\gamma_n$ is given by
	\begin{equation*}
		\gamma_n = \min \left \lbrace 1,\max \left \lbrace 0,\frac{(\phi_H-\phi_L)(\mu-\mu_n)(p^*\bar{\rho}_c(p^*)-p\bar{\rho}_c(p))}{p^*\bar{\rho}_0(p^*)\phi_L\alpha_{\phi_L}  +(\phi_H-\phi_L)p^*\bar{\rho}_c(p^*)}\right \rbrace\right \rbrace.
	\end{equation*}
	Note that $\gamma_n \to 0$, and $\gamma_n \in [0,1]$, so that $p_n\in \mathcal{P}$ for all $n$. By construction, we have that $G_2\left(\left(\alpha_{\phi_H},\alpha_{\phi_L},p\right)_n\right) = 0$, so it remains to show that our choice of the sequence $\left(\gamma_n\right)$ ensures that $G_1\left(\left(\alpha_{\phi_H},\alpha_{\phi_L},p,\mu\right)_n\right) \geq 0$ for all large enough $n$.
	
	By Assumption~\ref{assump:Demand}, we have that $G_1$ is strictly concave in $p$. Thus, we have that
	
	\begin{equation*}
		\begin{split}
			G_1\left(\left(\alpha_{\phi_H},\alpha_{\phi_L},p,\mu\right)_n\right) & = p_n\bar{\rho}_0(p_n)\phi_L \alpha_{\phi_L} + (\phi_H-\phi_L) p_n \bar{\rho}_c(p_n) \\
			& +\left(p_n\bar{\rho}_c(p_n) -p^*\bar{\rho}_c(p^*) \right)(1-\phi_L-\mu_n(\phi_H-\phi_L))\\
			& \geq \gamma_n \left[p^*\bar{\rho}_0(p^*)\phi_L \alpha_{\phi_L} + (\phi_H-\phi_L) p^*\bar{\rho}_c(p^*)\right]\\
			&+ (1-\gamma_n)\left[G_1\left(\alpha_{\phi_H},\alpha_{\phi_L},p,\mu\right) -(\mu-\mu_n)(\phi_H-\phi_L)(p^*\bar{\rho}_c(p^*)-p\bar{\rho}_c(p))\right]\\
			& = \gamma_n \left[p^*\bar{\rho}_0(p^*)\phi_L \alpha_{\phi_L} + (\phi_H-\phi_L) p^*\bar{\rho}_c(p^*)\right]\\
			& - (1-\gamma_n)(\mu-\mu_n)(\phi_H-\phi_L)(p^*\bar{\rho}_c(p^*)-p\bar{\rho}_c(p)),
		\end{split}
	\end{equation*}
	where the last step follows since we assumed that $G_1\left(\alpha_{\phi_H},\alpha_{\phi_L},p,\mu\right)=0$. Note that if $\mu_n\geq \mu$, the last term above is non-negative and we have that $G_1\left(\left(\alpha_{\phi_H},\alpha_{\phi_L},p,\mu\right)_n\right)\geq 0$. On the contrary, if $\mu_n< \mu$, we have that for all large $n$, $0<\gamma_n<1$, and by construction
	\begin{equation*}
		\begin{split}
			G_1\left(\left(\alpha_{\phi_H},\alpha_{\phi_L},p,\mu\right)_n\right)& \geq \gamma_n \left[p^*\bar{\rho}_0(p^*)\phi_L \alpha_{\phi_L} + (\phi_H-\phi_L) p^*\bar{\rho}_c(p^*)\right]\\
			& - (\mu-\mu_n)(\phi_H-\phi_L)(p^*\bar{\rho}_c(p^*)-p\bar{\rho}_c(p)) \geq 0.
		\end{split}
	\end{equation*}
	It follows that $(\alpha_{\phi_L},\alpha_{\phi_H},p)_n \in F(\mu_n)$ for all large enough $n$. Finally, by construction we have that $(\alpha_{\phi_L},\alpha_{\phi_H},p)_n \rightarrow (\alpha_{\phi_L},\alpha_{\phi_H},p)$. This proves that $F(\mu)$ is lower hemicontinuous and therefore $W^C(\mu)$ is a continuous function of $\mu$.
\end{proof}



{\subsection{Proof of Theorem \ref{thm:epsilonEq}} \label{app: proof_BayesianNashEqm}

Let $\bA^C \in \mathcal{A}^S$ be a promotion policy defined by solving problem \eqref{eq: ConfoundingProblemSimple} for all beliefs $\mu\in[0,1]$ as we describe in \S\ref{sec: OptSimplePolicies} and Appendix~\ref{app:DesigningSimplePolicies}. By the construction of Appendix~\ref{app:DesigningSimplePolicies}, there exists a  signaling mechanism $\sigma^C$ that, together with the promotion policy $\bA^C$, generate average consumer surplus of $co(W^C)(\mu_0)$ given any~$\mu_0\in[0,1]$ and $T\geq 1$, given that the seller makes myopic pricing decisions. Moreover, as described in Appendix~\ref{app:DesigningSimplePolicies} and formally shown in Proposition \ref{lemma:RevenueEquivalence}, we can take $\sigma^C$ as a simple signaling mechanism (i.e., that has a signal space of cardinality 2, e.g. $S = \{s',s''\}$). We will adjust the promotion policy at off-path beliefs below to establish an equilibrium.

\textbf{Construction of Equilibrium.} We construct a strategy profile $(\tilde\bA,\tilde\sigma,\tilde \bP)$ that satisfies conditions~\eqref{def:equilibrium:platform} and \eqref{def:equilibrium:seller}, which in particular implies condition~\eqref{def:equilibrium:seller:firstperiod} and therefore constitutes a Bayesian Nash Equilibrium. Fix~$\mu_0$ and set $\tilde\sigma = \sigma^C$. Since $\tilde\sigma$ has two outcomes, the seller's belief at the start of period 1, $\mu_1$ can take two possible values; denote these as $\mu'$, $\mu''$ and assume they correspond to signals $s',s''$, respectively. Define, for all $t$ and $\bar{h} \in \bar{H}_t$, the promotion policy $\tilde \bA =  \{\tilde \alpha_t\} _{t=1}^T$ as follows:

\begin{equation}\label{eq: PromoPolicyThm2}
	\tilde\alpha_t(p,\phi,\bar{h}) = \begin{cases}
		\alpha^C(p,\phi,\mu'), &\text{ if } s = s' \text{ and } \mu_t\left(\langle \bA^C,\sigma^C,\bar{h} \rangle \right) = \mu' \\
		\alpha^C(p,\phi,\mu''), &\text{ if } s = s'' \text{ and } \mu_t\left(\langle \bA^C,\sigma^C,\bar{h} \rangle \right) = \mu'' \\
		0, &\text{otherwise}.
	\end{cases}
\end{equation}

Letting $\bP^*$ be the Bayesian myopic pricing policy, define $\tilde\bP = \{\tilde\pi_t\}_{t=1}^T$ where, for each $t=1,...,T$:
\begin{equation}\label{eq: SellerPolicyThm2}
\tilde \pi_t\left(\langle \bA,\sigma, \bar{h}\rangle\right) = \begin{cases}
    \pi_t^*\left(\langle \bA,\sigma, \bar{h}\rangle\right), &\text{ if } \bA = \tilde\bA, \text{ and } \sigma = \tilde\sigma\\
    p^* & \text{ otherwise}.
\end{cases}
\end{equation}

To establish that $(\tilde\bA,\tilde\sigma,\tilde \bP)$ is a Bayesian Nash Equilibrium, we first show that $\tilde \bP$ is a best response for the seller given the platform's strategy $(\tilde\bA,\tilde\sigma)$ at any point in time, i.e., that condition~\eqref{def:equilibrium:seller} is satisfied. To do so, fix a period $t\in \{1,\dots, T\}$, a history $\bar{h}\in \bar{H}_t$ and suppose that the seller's belief is $\mu_t$. By construction, the seller's immediate payoff at time $t$ is maximized by choosing $p_t = \tilde \pi_t(\langle \tilde\bA,\tilde\sigma, \bar{h}\rangle)$, and the resulting payoff is at least $p^*\bar{\rho}_c(p^*)(1-\phi_L - (\phi_H-\phi_L)\mu_t)$. By contrast, if the seller deviates to some other price $p_t'$, by construction of $\tilde\bA$ we have that his time $t$ expected payoff is bounded above by $p^*\bar{\rho}_c(p^*)(1-\phi_L - (\phi_H-\phi_L)\mu_t)$. Thus, deviating to $p_t'$ does not result in a time $t$ payoff gain for the seller.
	
In subsequent periods, as choosing $p_t'$ at time $t$ might result in a shift in the seller's belief, we have by construction of $\tilde\bA$, that at any period $t'>t$ where the seller's belief is $\mu_{t'}\neq \mu_t$ is bounded above by~$p^*\bar{\rho}_c(p^*)(1-\phi_L - (\phi_H-\phi_L)\mu_{t'})$, as the seller is promoted with probability zero in such periods. As this expression is linear in the seller's belief which is constructed by Bayesian updating, it follows that the expected value at any time $t'$ of such deviation is bounded above by $p^*\bar{\rho}_c(p^*)(1-\phi_L - (\phi_H-\phi_L)\mu_{t})$, which is again, by definition, a lower bound for the seller's expected time $t'$ profit if he continues to price myopically. Therefore, we have that there is no profitable deviation at time $t$, from where it follows that $$V_{t,T}^{\boldsymbol{\alpha},\sigma,\boldsymbol{\pi}} \left(\langle\boldsymbol{\alpha},\sigma,\bar{h} \rangle\right) \geq V_{t,T}^{\boldsymbol{\alpha},\sigma,\boldsymbol{\pi}'} \left(\langle\boldsymbol{\alpha},\sigma,\bar{h} \rangle\right).$$
To establish that the platform's strategy is a best response to the seller's pricing policy,  suppose that the seller sets prices according to $\tilde\bP$, and consider the consumer surplus generated by a platform deviation to any~$\left(\bA',\sigma'\right)\neq \left(\tilde{\bA},\tilde{\sigma}\right)$. In this case, the seller sets $p^*$ every period.  Recall that $\sigma^T$ denotes the truthful signaling mechanism and note that, by definition of $\tilde \bP$, both $\left(\bA',\sigma'\right)$ and $\left(\tilde{\bA},\sigma^T\right)$ result in the seller setting~$p^*$ in every period, and thus are payoff-equivalent for the platform.\footnote{Unless we have that $\sigma^T= \tilde{\sigma}$, in which case we directly have that the welfare induced by $\left(\tilde{\bA},\sigma^T\right)$ is at least as large as when the seller sets $p^*$ in every period (by construction of $\bA^C$), and the result follows.} By construction of $\bA^C$ (see \eqref{eq: ConfoundingProblemSimple} in \S\ref{sec: OptSimplePolicies}), we have that the consumer surplus in every period is at least as large as the corresponding one when the seller sets $p^*$ if the seller prices myopically. Finally, since the truthful mechanism is a simple mechanism, it follows from the optimality of $\tilde{\sigma} = \sigma^C$ (as in the proof of Theorem~\ref{thm:LongRunAverageOptimalConsumerSurplus}) that
\begin{equation*}
	W_{T}^{\bA',\sigma',\tilde \bP}(\mu) \leq W_{T}^{\tilde{\bA},\sigma^T,\tilde \bP}(\mu) \leq W_{T}^{\tilde{\bA},\tilde{\sigma},\tilde \bP}(\mu)
\end{equation*}
Therefore, any deviation by the platform decreases the expected consumer surplus and  condition~\eqref{def:equilibrium:platform} holds. Thus, $\left(\tilde{\bA},\tilde{\sigma},\tilde \bP\right)$ is a Bayesian Nash equilibrium. Moreover, the seller prices myopically on the equilibrium path by definition of the policy $\tilde \bP$ and, on the other hand, the platform's promotion decisions are made according to the confounding policy $\bA^C$. Finally, by the proof of Theorem~\ref{thm:LongRunAverageOptimalConsumerSurplus} and by construction of $\tilde{\bA}$ we have, as desired, that
	\begin{equation*}
	\frac{1}{T}W_{T}^{\tilde{\bA},\tilde{\sigma},\tilde \bP}(\mu) = 	co(W^C)(\mu_0).
\end{equation*}
}
{
\subsection{Proof of Theorem \ref{thm:OptimalRobustEquilibria}} \label{app: proof_RobustEqm}


The proof consists of two parts. We first establish that the equilibrium defined in the proof of Theorem~\ref{thm:epsilonEq} is a Horizon-Maximin Equilibrium. In particular, we show that this equilibrium yields an average consumer surplus of $co(W^C)(\mu_0)$. Then, we show that the proposed equilibrium is long-run optimal for the platform with respect to its robust payoff.


\textbf{Part 1 (Verification of Conditions for Horizon Maximin Equilibrium).} Fix $T\geq 1$, $t\in \{1,\dots,T\}$, and define $(\tilde\bA,\tilde\sigma,\tilde \bP)$ as in the proof of Theorem~\ref{thm:epsilonEq} (see~\eqref{eq: PromoPolicyThm2} and \eqref{eq: SellerPolicyThm2} in \S\ref{app: proof_BayesianNashEqm}). We claim that $(\tilde\bA,\tilde\sigma,\tilde \bP)$ is a Horizon-Maximin Equilibrium. To prove this, we show that conditions  \eqref{def:equilibrium:platform:robust} and \eqref{def:equilibrium:seller:robust} hold.

Given a history $\bar{h}\in \bar{H}_t$, for the seller we have:
\begin{equation*}
	\begin{split}
	\max_{\bP\in \Pi} RV_{t,T}^{\tilde\bA,\tilde\sigma,\boldsymbol{\pi}} \left(\langle \tilde\bA,\tilde\sigma,\bar{h} \rangle\right) \stackrel{(a)}{\leq}\max_{\bP \in \Pi }V_{t,t}^{\tilde\bA,\tilde\sigma,\bP}\left(\langle \tilde\bA,\tilde\sigma,\bar{h} \rangle\right)\stackrel{(b)}{\leq}V_{t,t}^{\tilde\bA,\tilde\sigma,\tilde \bP}\left(\langle \tilde\bA,\tilde\sigma,\bar{h} \rangle\right) \stackrel{(c)}{\leq} RV_{t,T}^{\tilde\bA,\tilde\sigma,\tilde \bP} \left(\langle \tilde\bA,\tilde\sigma,\bar{h} \rangle\right),
	\end{split}
\end{equation*}

    where (a) follows by definition of $RV_{t,T}$, (b) follows since the price set by $\tilde \bP$ at time $t$ is myopically optimal given the platform's strategy $\left(\tilde{\bA},\tilde{\sigma}\right)$, and (c) follows because the seller's per-period payoff under $(\tilde\bA,\tilde\sigma,\tilde \bP)$ is, by construction, the same in every period. 
 Therefore, by employing $\tilde\bP$, the seller is best-responding to the platform policy $(\tilde\bA,\tilde\sigma)$ with respect to his robust payoffs (i.e., condition \eqref{def:equilibrium:seller:robust} holds). Moreover, by definition of $\tilde\bP$ (see \eqref{eq: SellerPolicyThm2} in \S\ref{app: proof_BayesianNashEqm}),  the seller prices  myopically on every history of the form $\left(\langle \tilde\bA,\tilde\sigma,\bar{h} \rangle\right)$ under $(\tilde\bA,\tilde\sigma,\tilde\bP)$.

    Now we establish that the platform is also best-responding.
    Consider the consumer surplus generated by the platform deviating to some alternative strategy $\left(\bA',\sigma'\right)$, so that by definition of $\tilde \bP$, the seller sets $p^*$ in every period. Denote the truthful signaling mechanism by $\sigma^T$  and the myopic promotion policy by $\hat{\bA}$ (see~\S\ref{subsec:InsufficiencyOfTruthfulDisclosure}). We have that:
    \begin{equation*}
    	RW_{T}^{\bA',\sigma',\tilde \bP}(\mu) \stackrel{(a)}{\leq} \frac{1}{T}W_{T}^{\bA',\sigma',\tilde \bP}(\mu) \stackrel{(b)}{\leq} \frac{1}{T} W_{T}^{\hat{\bA},\sigma^T,\tilde \bP}(\mu) \stackrel{(c)}{\leq}  \frac{1}{T} W_{T}^{\tilde\bA,\sigma^T,\tilde \bP}(\mu) \stackrel{(d)}{\leq}  \frac{1}{T} W_{T}^{\tilde\bA,\tilde\sigma,\tilde \bP}(\mu) \stackrel{(e)}{=} RW_{T}^{\tilde\bA,\tilde\sigma,\tilde \bP}(\mu),
    \end{equation*}

    where (a) holds by definition of $RW_T$, and (b) follows 
    because under the myopic policy $\hat{\bA}$, even with truthful revelation, the per-period consumer welfare is at least as large as the corresponding one when the seller sets price $p^*$ (regardless of the true value of $\phi$). Next, (c) follows because, given a truthful signal, an optimal confounding policy such as $\tilde \bA$ results in at least the same surplus as the myopic policy $\hat{\bA}$ (as, if  $\hat{\bA}\neq \tilde \bA$, the seller sets $p^*$ in every period), and (d) follows since $\tilde{\sigma}$ is an optimal signaling mechanism associated with the promotion policy $\tilde{\bA}$. Finally, (e) follows because the expected consumer surplus is the same in every period under a confounding promotion policy. Therefore, condition \eqref{def:equilibrium:platform:robust} holds, i.e.,  $RW_{T}^{\bA',\sigma',\tilde \bP}(\mu)\leq RW_{T}^{\tilde\bA,\tilde\sigma,\tilde \bP}(\mu)$.

    Thus, $(\tilde\bA,\tilde\sigma,\tilde \bP)$ is a Horizon-Maximin Equilibrium. Finally, note that by construction of $(\tilde\bA,\tilde\sigma,\tilde \bP)$, by Theorem~\ref{thm:LongRunAverageOptimalConsumerSurplus} and by (e) above, we have that
    \begin{equation}\label{eq: aux_long_run_optimal_robust_eqm}
    	RW_{T}^{\tilde\bA,\tilde\sigma,\tilde \bP}(\mu) = \frac{1}{T} W_{T}^{\tilde\bA,\tilde\sigma,\tilde \bP}(\mu) = co(W^C)(\mu).
    \end{equation}

	\textbf{Part 2 (Long-run Optimality).} We now establish that $(\tilde\bA,\tilde\sigma,\tilde \bP)$ is a long-run optimal equilibrium for the platform with respect to its robust payoffs. By \eqref{eq: aux_long_run_optimal_robust_eqm}, we have that $(\tilde\bA,\tilde\sigma,\tilde \bP)$ results in a robust payoff of~$co(W^C)(\mu_0)$ for the platform, so to establish that that condition \eqref{eq: long_run_platform_optimal_eqm} holds, it remains to show that 
	given~$\mu_0\in[0,1]$, 
	\begin{equation*}
		\lim_{T\rightarrow \infty}\sup_{(\bA,\sigma,\bP) \in \mathcal{E}(T)}  RW^{\bA,\sigma,\bP}_T(\mu_0) = co(W^C)(\mu_0),
	\end{equation*}
	where $\mathcal{E}(T)$ denotes the set of Horizon-Maximin equilibria with maximal horizon length $T$. To show this, we rely on the following auxiliary result.
	
	\begin{lemma}\label{lem: aux_robust_eqm}
		Fix $T\geq 1$, $\mu\in[0,1]$ and $(\bA,\sigma,\bP) \in \mathcal{E}(T)$. Then, there exists $\bA'\in \mathcal{A}$ such that
		$$RW_T^{\bA,\sigma,\bP} (\mu) \leq RW_{T}^{\bA',\sigma,\bP^*} (\mu)$$
	\end{lemma}

We defer the proof of Lemma~\ref{lem: aux_robust_eqm} to the end of this section. Note then that for fixed $T$, by Lemma~\ref{lem: aux_robust_eqm} and the definition of $RW_T$ we have that
\begin{equation*}
	\sup_{(\bA,\sigma,\bP) \in \mathcal{E}(T)}  RW^{\bA,\sigma,\bP}_T(\mu_0) \leq \sup_{(\bA',\sigma') \in \mathcal{A}\times \Sigma}  RW^{\bA',\sigma',\bP^*}_T(\mu_0) \leq \sup_{(\bA',\sigma') \in \mathcal{A}\times \Sigma} \frac{1}{T} W_{T}^{\bA',\sigma',\bP^*}(\mu_0).
\end{equation*}
By taking limits on this inequality and applying Theorem~\ref{thm:LongRunAverageOptimalConsumerSurplus}, we conclude that
\begin{equation*}
	\lim_{T\rightarrow \infty}\sup_{(\bA,\sigma,\bP) \in \mathcal{E}(T)}  RW^{\bA,\sigma,\bP}_T(\mu_0) \leq \lim_{T\rightarrow \infty} \sup_{(\bA',\sigma') \in \mathcal{A}\times \Sigma} \frac{1}{T} W_{T}^{\bA',\sigma',\bP^*}(\mu_0) = co(W^C)(\mu_0).
\end{equation*}

This concludes the proof of Theorem~\ref{thm:OptimalRobustEquilibria}.
To close this section, we provide the proof of Lemma~\ref{lem: aux_robust_eqm}, which was used in the argument above.

	\begin{proof}[Proof of Lemma~\ref{lem: aux_robust_eqm}]
		We use a series of five claims to prove this result. The first two claims show that under an Horizon-Maximin equilibria, the seller's per-period revenue at any history is at least the corresponding one to setting price $p^*$ and not being promoted by the platform.
		
		\textbf{Claim 1.} Let $(\bA,\sigma,\bP) \in \mathcal{E}(T)$, $t\in \{1,\dots,T\}$ and $\bar{h}\in\bar{H}_t$. Denote $\mu_t$ as the seller's belief at history $\langle \bA,\sigma,\bar{h} \rangle$. Then, $RV_{t,T}^{\bA,\sigma,\bP} \left(\langle \bA,\sigma,\bar{h} \rangle\right) \geq (1-\bar{\phi}\left(\mu_t \right))p^* \bar{\rho}_c(p^*).$
		
		\begin{proof}
			Let $\bar{\bP}$ be the pricing policy that sets $p^*$ on every history. Then, since $(\bA,\sigma,\bP) \in \mathcal{E}(T)$, it follows from \eqref{def:equilibrium:seller:robust} that
			\begin{equation*}
				RV_{t,T}^{\bA,\sigma,\bP} \left(\langle \bA,\sigma,\bar{h} \rangle\right) \geq RV_{t,T}^{\bA,\sigma,\bar{\bP}} \left(\langle \bA,\sigma,\bar{h} \rangle\right) = \left(\frac{1}{t'-t+1}\right)V_{t, t'}^{\bA,\sigma,\bar{\bP}} \left( \langle \bA,\sigma,\bar{h} \rangle \right),
			\end{equation*}
		for some $t'\in \{t,...,T\}$. Since Bayesian beliefs form a martingale, we have that
		\begin{equation}\label{eq: aux_claim1_theorem3}
			\left(\frac{1}{t'-t+1}\right)V_{t, t'}^{\bA,\sigma,\bar{\bP}} \left( \langle \bA,\sigma,\bar{h} \rangle \right) \geq \left(\frac{1}{t'-t+1}\right) \sum_{\tau =t}^{t'}\E \left[1-\bar{\phi}(\mu_\tau) | \mathcal{H}_t \right]p^* \bar{\rho}_c(p^*) = (1-\bar{\phi}\left(\mu_t \right))p^* \bar{\rho}_c(p^*),
		\end{equation}
	which implies that $RV_{t,T}^{\bA,\sigma,\bP} \left(\langle \bA,\sigma,\bar{h} \rangle\right) \geq (1-\bar{\phi}\left(\mu_t \right))p^* \bar{\rho}_c(p^*).$
		\end{proof}
	
	\textbf{Claim 2.} As in Claim 1, let $(\bA,\sigma,\bP) \in \mathcal{E}(T)$, $t\in \{1,\dots,T\}$, $\bar{h}\in\bar{H}_t$, and denote $\mu_t = \mu_t\left( \langle \bA,\sigma,\bar{h} \rangle \right)$. Then, the seller's expected revenue at time $t$ under $(\bA,\sigma,\bP)$ is at least $(1-\bar{\phi}\left(\mu_t \right))p^* \bar{\rho}_c(p^*)$.
	
	\begin{proof}
		By Claim 1 and the definition of $RV_{t,T}$, it follows that
		\begin{equation*}
			(1-\bar{\phi}\left(\mu_t \right))p^* \bar{\rho}_c(p^*)\leq RV_{t,T}^{\bA,\sigma,\bP} \left(\langle \bA,\sigma,\bar{h} \rangle\right) = \min_{t\leq \bar{t}\leq T}\left(\frac{1}{\bar{t}-t+1}\right)V_{t,\bar t}^{\bA,\sigma,\bP} \left( \langle \bA,\sigma,\bar{h} \rangle \right) \leq V_{t,t}^{\bA,\sigma,\bP} \left( \langle \bA,\sigma,\bar{h} \rangle \right),
		\end{equation*}
	as desired.
	\end{proof}

The next step is to construct an alternative pricing policy $\bP'$ that is deterministic and does not reduce the platform's robust payoff. To do so, let us denote the expected consumer surplus from periods $t$ and $t'\geq t$, given the strategies $(\bA,\sigma,\bP)$ and a history $h\in H_t$ by:
$$W_{t,t'}^{\bA,\sigma,\bP}(h) = \E\left(\sum_{\tau=t}^{t'} W(p_{\tau},a_{\tau},\psi_{\tau}) \middle|\boldsymbol{\alpha},\sigma,\bP,h_t = h\right),$$
where the expectation is taken with respect to any randomness in the pricing policy, promotion policy, customer types, purchase decisions, and the true value of $\phi$. Moreover, define the minimal time-average continuation consumer surplus at time $t$ given history $h \in H_t$ as:
\begin{equation*}
	RCS_{t,T}^{\bA,\sigma,\bP} (h) := \min_{t\leq \bar{t}\leq T}\left(\frac{1}{\bar{t}-t+1}\right)W_{t,\bar t}^{\bA,\sigma,\bP} (h).
\end{equation*}

Now, for the pricing policy $\bP$ and a history $h\in H_t$, let $S(h)$ denote the support of the price distribution induced by $\pi_t(h)$. To construct the desired policy $\bP'$, we can proceed recursively in a similar fashion to the proof of Lemma~\ref{lemma:PromotionAsFunctionOfBelief}.

\textbf{Period $\mathbf{T}$:} Given a history $h\in H_T$, define $p_T'(h)$ as the price that maximizes the minimal time-average continuation consumer surplus at $h$:
\begin{equation*}
	p_T'(h) = \arg\max_{p\in S(h)} \E \left[RCS_{T,T}^{\bA,\sigma,\bP} \left( h \right) \middle| p_T = p \right],
\end{equation*}
and define $\bP^{(T)}$ as the pricing policy that is equal to $\bP$ at periods $t=1,\dots,T-1$, but prices at $p_T'(h)$ at any time-$T$ history, i.e., that for $h\in H_T$,
\begin{equation*}
	\pi^{(T)}_T(h) = p_T'(h),\text{ with probability 1.}
\end{equation*}
Then, by definition of $\bP^{(T)}$, for any  $t=1,\dots,T$, and any history $h\in H_t$, we have that
\begin{equation*}
	RCS_{t,T}^{\bA,\sigma,\bP^{(T)}} (h)\geq RCS_{t,T}^{\bA,\sigma,\bP} (h).
\end{equation*}

\textbf{Recursive definition:} We proceed as follows. At period $1\leq t< T$, define the policy $\bP^{(t)}$ as being equal to~$\bP$ in periods before $t$, and equal to $\bP^{(t+1)}$ for periods after $t$. For period $t$ and a history $h\in H_t$, define~$p_t'(h)$ as the price that maximizes the minimal time-average continuation consumer surplus at $h$, among the prices in the support of $\pi_t(h)$, assuming that subsequent prices are set according to $\bP^{(t+1)}$, that is:\footnote{If there are multiple maximizers, we can choose one arbitrarily.}
\begin{equation*}
	p_t'(h) = \arg\max_{p\in S(h)} \E \left[RCS_{t,T}^{\bA,\sigma,\bP^{(t+1)}} \left( h \right) \middle| p_t = p \right],
\end{equation*}
where the expectation is taken with respect to the randomness of the platform's promotion choice, the consumer's purchase decision and her type at time $t$. Finally, given a time $t$ history $h\in H_t$, let us define the pricing policy $\bP^{(t)}$ by:
\begin{equation*}
	\pi^{(t)}_t(h) = p_t'(h),\text{ with probability 1.}
\end{equation*}
By construction of $\bP^{(t)}$ and applying induction on the subsequent time periods, we have that for any time period $t'$ and any history $h\in H_{t'}$, we have that 
\begin{equation} \label{eq: robust_policy_induction}
	RCS_{t',T}^{\bA,\sigma,\bP^{(t)}} (h)\geq RCS_{t',T}^{\bA,\sigma,\bP^{(t+1)}} (h)\geq RCS_{t',T}^{\bA,\sigma,\bP} (h).
\end{equation}
Finally, we define our desired policy as
\begin{equation}\label{eq: alternative_pricing_policy_robust}
	\bP' = \bP^{(1)},
\end{equation}
where $\bP^{(1)}$ is constructed recursively as described above.

\textbf{Claim 3.} Let $\bP'\in \Pi$ be the pricing policy defined by \eqref{eq: alternative_pricing_policy_robust}. Then, $RW_T^{\bA,\sigma,\bP} (\mu) \leq RW_{T}^{\bA,\sigma,\bP'} (\mu)$.

\begin{proof}
	By the construction described above and in particular \eqref{eq: robust_policy_induction}, we have that $RCS_{t,T}^{\bA,\sigma,\bP} (h)\leq RCS_{t,T}^{\bA,\sigma,\bP'} (h)$ for all $t=1,\dots,T$ and $h\in H_t$.  In particular, for $t=1$ this implies that $RW_T^{\bA,\sigma,\bP} (\mu) \leq RW_{T}^{\bA,\sigma,\bP'} (\mu)$.
\end{proof}

Now, we construct an alternative promotion policy that result in the same robust payoffs as $(\bA,\sigma,\bP')$. Specifically, for $\bar{h}\in \bar{H}_t$, let us define $\bA'$ by:
\begin{equation}\label{eq: aux_def_promo_policy_robust_eqm}
	\alpha_t'(p,\phi,\bar{h}) = \begin{cases}
		\alpha_t(p,\phi,\bar{h}) & \text{ if $p=p_t'\left(\langle \bA,\sigma,\bar{h} \rangle  \right)$,}\\
		0 & \text{ otherwise.}
	\end{cases}
\end{equation}

\textbf{Claim 4.} We have that $RW_T^{\bA,\sigma,\bP'} (\mu) = RW_{T}^{\bA',\sigma,\bP'} (\mu)$.
\begin{proof}
	Since $\bP'$ is deterministic by construction, $(\bA,\sigma,\bP')$ and $(\bA',\sigma,\bP')$ induce the same outcomes and the same payoffs.
\end{proof}

Finally, we show that the pricing policy $\bP'$ results in the same prices as the myopic Bayesian pricing policy~$\bP^*$ (see \S\ref{sec:LongRunAverageOptimalConsumerSurplus}) given the promotion policy $\bA'$. In particular, the platform's robust payoff is the same under these two policies.

\textbf{Claim 5.} Given the policy $\bA'$, the pricing policy $\bP'$ results in myopically optimal prices for the seller (i.e., it satisfies \eqref{eq:myopic}). In particular, we have that $RW_T^{\bA,\sigma,\bP'} (\mu) = RW_{T}^{\bA',\sigma,\bP^*} (\mu)$.

\begin{proof}
	Fix a history $\bar{h}\in \bar{H}_t$. Given that the promotion policy defined by \eqref{eq: aux_def_promo_policy_robust_eqm} is single-price, the price that maximizes the seller's period $t$ revenue is either $p^*$ or $p_t'\left(\langle \bA,\sigma,\bar{h} \rangle  \right)$. By Claim 2, $p_t'\left(\langle \bA,\sigma,\bar{h} \rangle  \right)$ results in (weakly) larger period-$t$ revenue than $p^*$. If this holds strictly, then $p_t'\left(\langle \bA,\sigma,\bar{h} \rangle  \right)$ is the unique myopically optimal price at time $t$. Otherwise, if both $p_t'\left(\langle \bA,\sigma,\bar{h} \rangle  \right)$ and $p^*$ result in the same expected revenue at time~$t$, we can assume without loss of generality  that $p_t'\left(\langle \bA,\sigma,\bar{h} \rangle  \right)$ results in higher present consumer surplus than~$p^*$ since, if that was not the case, the platform could redefine its promotion policy to incentivize the seller to set $p^*$ in the present period without altering future actions, and generate at least the same consumer surplus as under $\bA'$. Thus, $\pi_t'\left( \langle \bA,\sigma,\bar{h} \rangle \right)$ satisfies \eqref{eq:myopic}, as desired.
\end{proof}

To conclude the proof of the Lemma, we have from Claims 3, 4, and 5 that
\begin{equation*}
	RW_T^{\bA,\sigma,\bP} (\mu) \leq RW_T^{\bA,\sigma,\bP'} (\mu) = RW_T^{\bA',\sigma,\bP'} (\mu) = RW_{T}^{\bA',\sigma,\bP^*} (\mu),
\end{equation*}
whre $\bP'$ and $\bA'$ are defined on \eqref{eq: alternative_pricing_policy_robust} and \eqref{eq: aux_def_promo_policy_robust_eqm}, respectively.
	\end{proof}

}




\section{Two Competing Sellers}\label{app: TwoSellers}

In our baseline model, we focused on a single seller's pricing decisions and treated competitors' prices as exogenous. In what follows, we study a natural extension of our single-seller model, where we explicitly model competition among sellers. At the beginning of the horizon the platform sets $(i)$ a promotion policy that, given sellers' prices in each period, determines which seller will be promoted; and $(ii)$ a signaling mechanism that may provide sellers with information about the true value of $\phi$. Then, in every period, sellers set their prices taking into account the platform's promotion policy, observe their sales realizations and update their beliefs about the proportion of patient buyers $\phi$.

As sellers now set prices simultaneously, their corresponding pricing decisions must take into account their competitor's actions.
Concretely, we assume that in every period, sellers' posted prices form a Nash equilibrium in the pricing game defined by thir corresponding revenue at that period, considering the platform's promotion policy. Thus, while sellers make myopic pricing decisions, 
they best-respond to their competitor's price at each period. We introduce notation for the two-seller model and formally define the notion of myopic Bayesian pricing in this setting in~\S\ref{sec: model_two_sellers}.

We analyze three natural families of platform policies. The first one is \textit{truthful revelation,} under which the platform discloses the value of $\phi$ to both sellers at the beginning of the horizon and then optimally designs its promotion policy, taking into account the sellers' myopic pricing decisions. The second class involves \textit{confounding one seller only}, consisting of direct extensions of the confounding promotion policies defined in~\S\ref{sec:ConfoundingPromotionPolicies}; these policies prevent one of the sellers from updating his beliefs throughout the horizon while disclosing the value of $\phi$ to the other seller.
The third family consists of policies that \textit{confound both sellers;} these policies are designed to eliminate the informational content provided by sales observations to keep the beliefs of both sellers constant. We introduce these policies and discuss their properties in \S\ref{sec: policies_definitions_two_sellers}. In particular, we illustrate that these policies enable us to reduce the dynamics of the game to a static setting, where the equilibrium prices set by sellers and the platform's promotion policy are the same at every period.

Next, we compare these three classes of policies in terms of the expected consumer welfare they induce in~\S\ref{sec: example_analysis_two_sellers}. To this end, we focus on a demand model with consumers that have independent, uniformly distributed valuations for each seller's product, which extends the setting presented in Example~\ref{example:UniformDemand} to the case of two sellers. Then, we compute the policies described above and compare them in terms of the welfare they induce over a broad range of parameters for the model. Our findings split the parameter space in two broad regions. In the first one, consisting of 54.32\% of the parameter combinations, confounding promotion policies are not beneficial for the platform compared to truthful revelation. In the second region, consisting of 45.43\% of parameter combinations, confounding one seller outperforms truthfully disclosing $\phi$ to both sellers. Importantly, we observe that confounding both sellers almost never dominates truthful revelation (which occurs only in 0.25\% of parametric combinations). These observations demonstrate that promotion policies that confound a single seller are a valuable tool for the platform even when sellers compete in prices in every period. In that sense, these observations solidify the extension of our findings to a setting with more than one seller.

\subsection{The Model} \label{sec: model_two_sellers}

The timeline of the game follows the same structure as in the single-seller model presented in \S\ref{sec:Model} with the exception that we now consider two sellers, $1$ and $2$, that update their beliefs and then set prices simultaneously in every period. For each seller $i$, the set of possible prices is denoted by $\mathcal{P}_i$ and is assumed to be a compact interval of the real line that contains zero. In each period $t$ the platform's action is either to promote some seller $i$ ($a_t = i\in \{1,2\}$) or neither of them ($a_t = 0$), and is determined by the platform's promotion policy.

Following similar notation as in the single-seller model, we assume that the probability of purchasing from seller  $i\in \{1,2\}$ is captured by a commonly known function $\rho_i$ that depends on the consumer type $\psi\in \{P,I\}$ where $\psi = P$ ($\psi=I$) denotes that the consumer is patient (impatient), both sellers' prices $p_1\in \mathcal{P}_1,p_2\in \mathcal{P}_2$,  and the platform promotion decision $a\in\{0,1,2\}$:
\begin{equation}\label{eq:TwoSellerConsumerDemand}
	\rho_i(p_1,p_2,a, \psi) = \Prob(y_i = 1|p_1,p_2,a, \psi) =
	\begin{cases}
		\bar{\rho}_c^i(p_1,p_2), &\text{ if $\psi=P$,} \\
		\bar{\rho}_0^i(p_i), &\text{ if $\psi=I$} \text{ and } a = i, \\
		0, &\text{ if $\psi=I$} \text{ and } a \neq i.
	\end{cases}
\end{equation}

We make the following standard assumptions on the demand model.

\begin{manualassumption}{\ref{assump:TwoSellerDemand}}[Two-seller demand]
	For each $i=1,2$, the function $\bar{\rho}_0^i(p_i)$ is decreasing and Lipschitz continuous in $p_i$; and the function $\bar{\rho}_c^i(p_i,p_{-i})$ is decreasing in $p_i$, increasing in $p_{-i}$ and Lipschitz continuous in $(p_1,p_2)$. In addition,  $p_i\bar{\rho}_c^i(p_i,p_{-i})$ and $p_i\bar{\rho}_0^i(p_i)$ are strictly concave in $p_i$ for any fixed value of $p_{-i}$. Finally,~$\bar{\rho}_0^i(p_i) \geq \bar{\rho}_c^i(p_i,p_{-i})$ for all $p_i\in\mathcal{P}_i$, $p_{-i}\in\mathcal{P}_{-i}$.
\end{manualassumption}

In each period the platform's payoff is defined as the expected consumer surplus, and is denoted by a commonly known function $W$ of the sellers' prices $p_1$, $p_2$, the platform's promotion decision $a$, and the consumer's type:
\begin{equation}\label{eq:TwoSellerConsumerWelfare:W}
	W(p_1,p_2,a, \psi) =
	\begin{cases}
		\bar{W}_c(p_1,p_2),  &\text{ if $\psi=P$,} \\
		\bar{W}_1(p_1),  &\text{ if $\psi=I$ and }  a = 1, \\
		\bar{W}_2(p_2), &\text{ if $\psi=I$ and } a =2,\\
		0, &\text{ if $\psi=I$ and } a =0,
	\end{cases}
\end{equation}
where we normalize the welfare associated with not promoting either seller when the buyer is impatient as zero. In line with the single-seller model, we make the following assumption on the welfare functions.

\begin{assumption}[Two-seller consumer Surplus]\label{Assump:TwoSellerConsSurplus}
	The functions $\bar{W}_c(p_1,p_2)$, $\bar{W}_1(p_1)$, and $\bar{W}_2(p_2)$ are all decreasing and Lipschitz continuous in each of the prices $p_1\in\mathcal{P}_1$, $p_2\in\mathcal{P}_2$.
\end{assumption}


The definitions for histories, strategies and belief systems extend from the single-seller model with minor modifications which we now describe. First, the platform's signaling mechanism is a pair $\mathbf{\sigma} = \left(\sigma_1,\sigma_2\right)$ where each component $\sigma_i:\{\phi_L,\phi_H\}\rightarrow \Delta(\mathcal{S})$ represents the signaling mechanism for seller $i$, that is, we allow the platform to send different signals to each seller. Second, at each period the platform has access to the history of prices and sales realizations for both sellers in all past period, whereas each seller observes both prices but only their own sales realization. The history observed by the platform at the beginning of period $t$ consists of the true value of $\phi$ and the sequence of signals, prices and sales realizations for both sellers. We formally denote $\bar{h}_t = \left\langle (s_1,s_2), \left(p_1^{\tau},p_2^{\tau}, y_1^{\tau}, y_2^{\tau}\right)_{\tau=1}^{t-1} \right\rangle$, where $s_i$ is the realized signal sent to seller $i$, and $y_i^t\in\{0,1\}$ is the sales realization for seller $i$ at time $t$. We denote the set of all such possible histories by time $t$ as  $\bar{H}_t = S^2 \times  \left(\mathcal{P}_1\times \mathcal{P}_2 \times \{0,1\}^2\right)^{t-1}$. In line with the formulation in \S\ref{sec:Model}, the platform's promotion policy, $\boldsymbol{\alpha} = \{\alpha_t\}_{t=1}^T$, is a collection of functions  that map time-$t$ histories, seller's prices and the state $\phi$ to a probability distribution on the platform's action space $\{0,1,2\}$, i.e., $\bA_t: \mathcal{P}_1\times \mathcal{P}_2 \times \{\phi_L,\phi_H\} \times \bar{H}_t \rightarrow \Delta\left(\{0,1,2\}\right)$.

On the other hand, at the beginning of time $t$, each of the sellers has observed his own signal, price and sales realization, but only his competitor's price choices. In addition, each seller observes the signaling mechanism and promotion policy chosen by the platform. Formally, we denote the history observed by seller $i$ at time~$t$ by $h_t^{(i)} = \left\langle s_i,\bA,\mathbf{\sigma}, \left(p_1^{\tau},p_2^{\tau}, y_i^{\tau}\right)_{\tau=1}^{t-1}\right\rangle$, and the set of possible histories observed by seller $i$ at the beginning of period $t$ by $H_t^{(i)} = \mathcal{S}\times \mathcal{A}\times \Sigma \times \left(\mathcal{P}_1 \times \mathcal{P}_2 \times \{0,1\}\right)^{t-1}$. Finally, given a platform's strategy $(\bA,\mathbf{\sigma})$ and a platform history $\bar{h}_t \in \bar{H}_t$, we denote the corresponding history observed by seller $i$ (slightly abusing notation) by $h_t^{(i)}(\bar{h}_t,\bA,\mathbf{\sigma})$.


The expected payoffs are defined as in \S\ref{sec:Model}. In addition, we extend the definition of Myopic Bayesian Pricing policies to consider the presence of two sellers as follows.

\begin{definition}[Myopic Bayesian Nash Pricing Policy]\label{def:MyopicPricingPolicy_comp}
	Fix a platform's signaling mechaism $\mathbf{\sigma}$ and a promotion policy $\bA$. In every period $t$ and at every history $\bar{h}_t \in \bar{H}_t$, a myopic Bayesian Nash pricing policy $\boldsymbol{\pi^*}=\{\pi_{t}^*\}_{t=1}^T$ selects a pair of prices $\left(p_1^t,p_2^t\right) \in \mathcal{P}_1\times \mathcal{P}_2$ that form a Nash equilibrium in the game with payoff functions defined as each seller's expected revenue for period $t$. 
	Formally, $\pi_t^*$ determines prices $\left(p_1^t,p_2^t\right) \in \mathcal{P}_1\times \mathcal{P}_2$ such that for $i=1,2$:
	\begin{equation}\label{eq:myopic_Nash}
		p_i^t \in \arg \max_{p_i \in \mathcal{P}_i} ~~~\E \left[ p_{i} \rho_i(p_i,p_{-i}^t,a_t, \psi_t) |h_t^{(i)} = h_t^{(i)}(\bar{h}_t,\bA,\mathbf{\sigma})\right],
	\end{equation}
	where the expectation above is taken with respect to the distribution of histories $\bar{h}'_t \in H_t$ such that $h_t^{(i)}(\bar{h}'_t,\bA,\mathbf{\sigma}) = h_t^{(i)}(\bar{h}_t,\bA,\mathbf{\sigma})$ and the platform's promotion being drawn according to $\bA_t$. If there exist multiple price pairs that satisfy \eqref{eq:myopic_Nash}, we assume that $\boldsymbol{\pi^*}$ selects one that maximizes the present consumer surplus.\footnote{As in the single-seller model, this is in line with the concept of \emph{sender preferred equilibria}. Moreover, it is straightforward to see that this policy is well-defined within the class of single-price promotion policies, which is a superset of the ones we consider in the next section.}
\end{definition}

Finally, let us denote by $p_i^{BR}:\mathcal{P}_{-i}\to \mathcal{P}_i$ the best-response for seller $i$ as a function of his competitor's price conditional on the buyer being patient, i.e.,
\begin{equation} \label{eq: best_response_two_sellers}
	p_i^{BR}(p_{-i}):= \arg \max_{p_i\in \mathcal{P}_i} p_i \bar{\rho}^i_c(p_i,p_{-i}),
\end{equation}
which is well-defined and single-valued by Assumption~\ref{assump:TwoSellerDemand}.

\subsection{Promotion Policy Design under Confounding and Truthful Revelation} \label{sec: policies_definitions_two_sellers}

The rest of this Appendix is dedicated to studying the value of jointly designing signaling mechanisms and confounding promotion policies in the setting with two sellers whose prices are set according to the Myopic Bayesian Nash policy that we previously defined. Specifically, we consider three platform policies: (i) \textit{truthful revelation}, meaning that the platform truthfully discloses the true state $\phi$ to both sellers and employs an optimal promotion policy thereafter; (ii) \textit{confound one seller only}, in which the platform truthfully reveals the value of $\phi$ to one of the sellers and keeps the other one confounded (i.e.,  designs a promotion policy to keep the uninformed seller's beliefs constant after the initial signal is realized; and (iii) \textit{confound both sellers}, in which the promotiofn policy is designed to keep both sellers from learning the value of $\phi$. We formally define each of these benchmarks in \S\ref{sec: two_sellers_truthful}--\ref{sec: two_sellers_confound_both}.


In line with the discussion in \S\ref{sec: OptSimplePolicies} and Appendix~\ref{app:DesigningSimplePolicies}, we focus on the design of static single-price policies.\footnote{The definition of single-price policies extends straightforwardly from Definition~\ref{def: SinglePricePromoPolicies} as the policies that do not promote any seller unless both of them set a target price.} Furthermore, since each of the classes of policies defined below maintains sellers' beliefs constant throughout the horizon (for each seller, by either confounding their beliefs or disclosing the true value of $\phi$), and since sellers' price according to the Myopic Bayesian Nash policy defined in~\eqref{eq:myopic_Nash}, the resulting posted prices are static, forming a Nash equilibrium at every period. As in the case of a single seller, this allows for tractability by inducing the promotion design problem and its analysis to a static setting.

\subsubsection{Truthful Revelation} \label{sec: two_sellers_truthful}

To define the welfare under truthful revelation, denote  $W^\phi$ as the maximum welfare achievable when the true state is $\phi \in \{\phi_L,\phi_H\}$ and both sellers are informed of the true state. Once both sellers are informed of the value of $\phi$, their beliefs are constant (and equal to $0$ or~$1$, depending on the value of $\phi$), and thus we can define $W^\phi$ by the following optimization problem, where we denote the probability of promoting seller~$i$ if both sellers comply with the target prices by $\alpha_{\phi}^i$.
\begin{equation}\label{eq: max_welfare_cond_state}
	\begin{split}
		W^{\phi}&:= \max_{\substack{\alpha_{\phi}^1, \alpha_{\phi}^2 \in [0,1],\\p_1\in \mathcal{P}_1,p_2\in \mathcal{P}_2}} ~~ \phi \alpha_{\phi}^1 \bar{W}_1(p_1) +\phi\alpha_{\phi}^2\bar{W}_2(p_2)  + (1-\phi)\bar{W}_c(p_1,p_2) \\
		\text{s.t.}&~~\phi \alpha_{\phi}^i p_i \bar{\rho}^i_0(p_i) + (1-\phi) p_i \bar{\rho}^i_c(p_i,p_{-i}) \geq (1-\phi) p_i^{BR}(p_{-i}) \bar{\rho}^i_c(p_i^{BR}(p_{-i}),p_{-i}), \quad \text{for $i=1,2$,}\\
		&~~ \alpha_{\phi}^1 + \alpha_{\phi}^2 \leq 1.
	\end{split}
\end{equation}
The first set of constraints ensure incentive compatibility for each seller, by requiring that complying with the price targeted by the promotion policy yields a larger expected profit than deviating to the second best option. This problem is always feasible as we can take $(p_1,p_2)$ as Nash equilibrium prices of the game with payoff functions $\{p_i \bar{\rho}^i_c(p_i,p_{-i})\}_{i=1,2}$ (which exists by Assumption~\ref{assump:TwoSellerDemand}) and letting $\alpha_\phi^1 = \alpha_\phi^2 = 1/2$. It follows that given a prior belief $\mu\in[0,1]$, the expected welfare under truthful revelation is
\begin{equation*}
	W^{\text{truth}}(\mu) = \mu W^{\phi_H} +(1-\mu) W^{\phi_L}.
\end{equation*}

\subsubsection{Confound One Seller Only}\label{sec: two_sellers_confound_one}

Suppose now that the platform selects a seller to confound and truthfully reveals the value of $\phi$ to the other seller. For simplicity we assume that seller 1 is the one chosen to be confounded while seller 2 observes the state $\phi$ in the first period. The case where the platform decides to confound seller 2 instead is simply defined by reversing the roles of sellers 1 and 2 in the constraints of \eqref{eq: max_welfare_confound_seller_1} below.

In order to confound seller 1, i.e., to keep his beliefs constant over time, the platform's promotion policy must be designed in a manner that prevents seller 1 from learning from his own demand realizations and seller 2's price choices. In particular, seller 2 must also be incentivized to make pricing decisions independently of the value of $\phi$, as otherwise these decisions, which are observed by seller 1, would be informative about $\phi$. More concretely, if at a certain time period seller 2's price is $p_2^H$ if $\phi = \phi_H$ and $p_2^L$ otherwise, the following two conditions must hold in order to confound seller 1:
\begin{equation*}
	\phi_H\alpha_{\phi_H}^1\bar{\rho}^1_0(p_1) + (1-\phi_H)\bar{\rho}^1_c(p_1,p_2^H) = \phi_L\alpha_{\phi_L}^1\bar{\rho}^1_0(p_1) + (1-\phi_L)\bar{\rho}^1_c(p_1,p_2^L), \quad \text{and} \quad
	p_2^H  = p_2^L.
\end{equation*}
The first condition follows from requiring that sales realizations are not informative for seller 1 (following the same reasoning to derive equation \eqref{eq:EqualityOfSalesProb} in \S\ref{sec:LongRunAverageOptimalConsumerSurplus}), while the second condition ($p_2^H = p_2^L$) ensures that seller 2's price is uninformative of the state $\phi$ for seller 1, which is achieved by requiring these prices to be independent of the value of $\phi$. By denoting then $p_2 = p_2^H = p_2^L$ and following the analysis in \S\ref{sec: OptSimplePolicies}, we can write the maximum welfare achievable by confounding seller 1 and revealing the true state to seller 2 when seller 1's belief that the state is $\phi = \phi_H$ is $\mu \in(0,1)$ as:
\begin{equation}\label{eq: max_welfare_confound_seller_1}
	\begin{split}
		W^{C1}(\mu)&:= \max_{\substack{\alpha_{\phi_H}^i,\alpha_{\phi_L}^i \in [0,1], i = 1,2,\\
				p_1 \in \mathcal{P}_1,p_2\in \mathcal{P}_2}} ~~ \E_{\phi}\left(\phi \alpha_{\phi}^1 \bar{W}_1(p_1) +\phi\alpha_{\phi}^2\bar{W}_2(p_2)  + (1-\phi)\bar{W}_c(p_1,p_2)|\mu\right) \\
		\text{s.t.}&~~ p_1\bar{\rho}^1_0(p_1)(\phi_L\alpha_{\phi_L}^1(1-\mu) +\phi_H\alpha_{\phi_H}^1\mu) +  p_1\bar{\rho}^1_c(p_1,p_2)(1-\phi_L -\mu(\phi_H-\phi_L)) \geq \\
		&\quad \quad \quad (1-\phi_L -\mu(\phi_H-\phi_L))p_1^{BR}(p_2) \bar{\rho}^1_c(p_1^{BR}(p_2),p_2);  \\
		&~~  \phi\alpha_{\phi}^2 p_2 \bar{\rho}^2_0(p_2) + (1-\phi) p_2 \bar{\rho}^2_c(p_1,p_2) \geq (1-\phi) p_2^{BR}(p_1) \bar{\rho}^2_c(p_1,p_2^{BR}(p_1)), \text{for $\phi\in\{\phi_H,\phi_L\}$;}\\
		&~~ \phi_H\alpha_{\phi_H}^1\bar{\rho}^1_0(p_1) + (1-\phi_H)\bar{\rho}^1_c(p_1,p_2) = \phi_L\alpha_{\phi_L}^1\bar{\rho}^1_0(p_1) + (1-\phi_L)\bar{\rho}^1_c(p_1,p_2);\\
		&~~\alpha_{\phi}^1 + \alpha_{\phi}^2 \leq 1, \text{for $\phi\in\{\phi_H,\phi_L\}$.}
	\end{split}
\end{equation}
The first two constraints establish incentive compatibility for seller~1 (given a belief $\mu$) and seller~2 (given access to the true value of~$\phi$); the third constraint ensures confounding seller~1. In particular, by solving problem~\eqref{eq: max_welfare_confound_seller_1} we can establish the existence of promotion policies that confound seller 1, provided the seller~2 is informed. This allows us to prove Proposition~\ref{prop: existence_confound_one_policy} as presented in \S\ref{sec: TwoSellers}.

\begin{manualprop}{\ref{prop: existence_confound_one_policy}}
	Suppose that seller 2 observes the value of $\phi$ and that Assumption~\ref{assump:TwoSellerDemand} holds. Then,
	there exists a simple promotion policy that confounds seller 1 given any seller 1 belief $\mu\in[0,1]$.
\end{manualprop}

\begin{proof}
	If $\mu\in \{0,1\}$, seller 1 never updates his belief so any promotion policy is confounding. Suppose then that $\mu\in(0,1)$.
	By the previous argument, and following a similar reasoning as the proof of Proposition~\ref{prop: ConfoundingPolicyExistence} (see \S\ref{proof:prop: ConfoundingPolicyExistence}), it suffices to show that problem~\eqref{eq: max_welfare_confound_seller_1} is always feasible. To do so, let $(p_1,p_2)$ be Nash equilibrium prices of the game with payoff functions $\{p_i \bar{\rho}^i_c(p_i,p_{-i})\}_{i=1,2}$, and set
	\begin{equation*}
		\alpha_{\phi_H}^1 = \left(\frac{\phi_H-\phi_L}{\phi_H}\right)\left(\frac{\bar{\rho}_c^1(p_1,p_2)}{\bar{\rho}_0^1(p_1)}\right)+\alpha_{\phi_L}^1\left(\frac{\phi_L}{\phi_H}\right), \quad \alpha_{\phi_L}^2 = 1-\alpha_{\phi_L}^2, \quad \alpha_{\phi_H}^2 = 1-\alpha_{\phi_H}^2,
	\end{equation*}
	so that, by Assumption~\ref{assump:TwoSellerDemand}, for any selection of $\alpha_{\phi_L}^1\in[0,1]$, this combination of prices and promotion probabilities is feasible for the optimization problem in \eqref{eq: max_welfare_confound_seller_1}.
\end{proof}

 Finally, to complete the definition of $W^{C1}$, one may proceed in line with the single-seller model when seller~1's belief is $\mu \in \{0,1\}$. In that case, seller 1's belief remain constant throughout the horizon regardless of any observed information, so we simply define $W^{C1}$ to be the maximum welfare achievable when both sellers know the true state, i.e.,  $W^{C1}(\mu) = W^{\text{truth}}(\mu)$ for $\mu \in \{0,1\}$.  


Given these observations, $W^{C1}(\mu)$ defines the maximal expected consumer welfare that is achievable by optimally designing a promotion policy that confounds seller 1 when seller 2 is informed as a function of seller 1's belief $\mu$. However, the platform may generate additional expected consumer welfare through designing the signaling mechanism for seller 1. The following proposition establishes that,
in line with the analysis in the proof of Theorem~\ref{thm:LongRunAverageOptimalConsumerSurplus}, the platform can construct a signal that results in an expected welfare of $co(W^{C1}(\mu_0))$.

\begin{proposition}\label{prop: signaling_confound_one}
	Suppose that both sellers have a prior of $\mu_0$ and that the platform sets the promotion policy to confound seller 1 only by solving~\eqref{eq: max_welfare_confound_seller_1}. Then,
	there exists an optimal signaling mechanism $\sigma_1$ for seller~1, and the resulting expected surplus is $co(W^{C1})(\mu_0)$.
\end{proposition}

\begin{proof}
	We first claim that $W^{C1}(\mu)$ is an upper-semicontinuous function of $\mu$. To see this,
	first note that by Proposition~\ref{prop: existence_confound_one_policy}, problem~\eqref{eq: max_welfare_confound_seller_1} is non-empty for all $\mu\in(0,1)$. Moreover, the objective of~\eqref{eq: max_welfare_confound_seller_1} is continuous in~$\mu$ (as it is linear). Moreover, since all the constraints of~\eqref{eq: max_welfare_confound_seller_1} are defined by continuous functions,  the feasible set correspondence of~\eqref{eq: max_welfare_confound_seller_1} is compact-valued and upper-hemicontinuous in $\mu$. It then follows from Berge's Maximum Theorem that $W^{C1}(\mu)$ is upper-semicontinuous for $\mu\in(0,1)$.
	
	For $\mu \in \{0,1\}$, recall that we defined $W^{C1}(\mu) = W^{\text{truth}}(\mu)$. It is easy to see that at beliefs $\mu\in\{0,1\}$, problems \eqref{eq: max_welfare_cond_state} \eqref{eq: max_welfare_confound_seller_1} have the same objective, and problem \eqref{eq: max_welfare_confound_seller_1} contains the constraints of \eqref{eq: max_welfare_cond_state} as well as other constraints (e.g., the confounding constraint). Thus, it follows that
	\begin{equation*}
		\lim_{\mu \to 0^+} W^{C1}(\mu)\leq W^{C1}(0), \qquad \lim_{\mu \to 1^-} W^{C1}(\mu)\leq W^{C1}(1),
	\end{equation*}
	and therefore $W^{C1}(\mu)$ is upper-semicontinuous for $\mu\in\{0,1\}$ as well.
	
	Given that we assume that seller 2 is informed of the value of $\phi$, we only need to design the signaling mechanism for seller 1. By the same argument as in the proof of Theorem~\ref{thm:LongRunAverageOptimalConsumerSurplus} (see also Corollaries 1 and 2 in \cite{Kamenica2011} and related discussion there), there exists an optimal signaling mechanism that results in an expected payoff of $co(W^{C1})(\mu)$.
\end{proof}

Finally, we define the maximum welfare achievable when seller 2 is the one chosen to be confounded (instead of seller 1), $W^{C2}(\mu)$, by reversing the roles of sellers 1 and 2 in the constraints of \eqref{eq: max_welfare_confound_seller_1}. By the same logic as in Proposition~\ref{prop: signaling_confound_one}, an optimal signaling mechanism that results in an expected welfare of $co(W^{C2})(\mu_0)$ can be constructed in this case as well.

\subsubsection{Confound Both Sellers}\label{sec: two_sellers_confound_both}

Consider now the case where the platform confounds both sellers. Following a similar reasoning as in \S\ref{sec: two_sellers_truthful} and \ref{sec: two_sellers_confound_one}, we define the maximum welfare achievable when both sellers have a belief of $\mu$ as follows:
\begin{equation}\label{eq: max_welfare_confound_both}
	\begin{split}
		W^{C,\text{both}}(\mu)&:= \max_{\substack{\alpha_{\phi_H}^i,\alpha_{\phi_L}^i \in [0,1], i = 1,2,\\
				p_1 \in \mathcal{P}_1,p_2\in \mathcal{P}_2}} ~~ \E_{\phi}\left(\phi \alpha_{\phi}^1 \bar{W}_1(p_1) +\phi\alpha_{\phi}^2\bar{W}_2(p_2)  + (1-\phi)\bar{W}_c(p_1,p_2)|\mu\right) \\
		\text{s.t.}&~~ p_i\bar{\rho}^i_0(p_i)(\phi_L\alpha_{\phi_L}^i(1-\mu) +\phi_H\alpha_{\phi_H}^i\mu) +  p_i\bar{\rho}^i_c(p_i,p_{-i})(1-\phi_L -\mu(\phi_H-\phi_L)) \geq \\
		&\quad \quad \quad (1-\phi_L -\mu(\phi_H-\phi_L))p_i^{BR}(p_{-i}) \bar{\rho}^i_c(p_i^{BR}(p_{-i}),p_{-i}), \text{for $i\in \{1,2\}$;}  \\
		&~~ \left(\phi_H\alpha_{\phi_H}^i-\phi_L\alpha_{\phi_L}^i\right)\bar{\rho}^i_0(p_i)  =  (\phi_H-\phi_L)\bar{\rho}^i_c(p_i,p_{-i}), \text{for $i\in \{1,2\}$;}  \\
		& ~~ \alpha_{\phi}^1+\alpha_{\phi}^2 \leq 1,\text{for $\phi\in \{\phi_H,\phi_L\}$.}
	\end{split}
\end{equation}

\begin{remark}
	In principle, one could consider the case where sellers have different beliefs after the initial signaling is realized. However, in that case sellers could learn information about their competitor's posterior beliefs based on their price choices (due to the incentive compatibility constraint), which would limit the platform's ability to confound sellers. Thus, in order to maintain beliefs as constant throughout the horizon, the platform's signaling mechanism must ensure that the beliefs of both sellers are equal after the initial signaling is realized.
\end{remark}

In contrast with the two scenarios we have discussed in Appendices~\ref{sec: two_sellers_truthful} and \ref{sec: two_sellers_confound_one}, confounding both sellers is not always feasible. In fact, we next establish that under mild assumptions on the demand model, there always exist model parameters under which confounding both sellers is infeasible.

\begin{manualprop}{\ref{prop: impossbility_confounding_two_sellers}}
	Suppose that both sellers have a belief of $\mu\in(0,1)$ and that the demand model is such that for any prices $p_1,p_2$ one has that
	\begin{equation} \label{eq: impossibility_confounding}
		\frac{\bar{\rho}^1_c(p_1,p_2)}{\bar{\rho}^1_0(p_1)}+\frac{\bar{\rho}^2_c(p_1,p_2)}{\bar{\rho}^2_0(p_2)} > 1.
	\end{equation}
	Then, for all $\phi_L$ small enough, there exists no promotion policy that confounds both sellers.
\end{manualprop}


Condition \eqref{eq: impossibility_confounding} is a regularity property satisfied by a wide range of demand models. To discuss the intuition of this condition, consider a symmetric demand model (i.e. $\bar{\rho}_0^1 = \bar{\rho}_0^2$ and $\bar{\rho}^1_c(p_1,p_2) = \bar{\rho}^1_c(p_2,p_1)$), and suppose that both sellers set a price $p$. Then, evaluating the LHS of \eqref{eq: impossibility_confounding} yields
\begin{equation*}
	\frac{\bar{\rho}^1_c(p,p)}{\bar{\rho}^1_0(p)}+\frac{\bar{\rho}^2_c(p,p)}{\bar{\rho}^2_0(p)} = 2\frac{\bar{\rho}^1_c(p,p)}{\bar{\rho}^1_0(p)}.
\end{equation*}
If \eqref{eq: impossibility_confounding} holds, then we would have that
	$2\bar{\rho}^1_c(p,p) > \bar{\rho}^1_0(p)$.
In other words, one obtains that the aggregate demand when the consumer may buy from two sellers is larger than the demand when consumers are provided with only one seller to buy from, which is likely to hold in most settings of interest.

Moreover, condition \eqref{eq: impossibility_confounding} is also satisfied by a range of non-symmetric demand models. For example, it is straightforward to show that this condition holds for a logit demand model. Indeed, one has:
$$\frac{\bar{\rho}_c^1(p_1,p_2)}{\rho_0^1(p_1)} =
\frac{e^{q_1-p_1}}{1+e^{q_1-p_1}+e^{q_2-p_2}}\bigg/\frac{e^{q_1-p_1}}{1+e^{q_1-p_1}} =
\frac{1+e^{q_1-p_1}}{1+e^{q_1-p_1}+e^{q_2-p_2}} $$
Then, combining that with the corresponding expression for seller 2's demand, we have:
\begin{equation*}
	\frac{\bar{\rho}^1_c(p_1,p_2)}{\bar{\rho}^1_0(p_1)}+\frac{\bar{\rho}^2_c(p_1,p_2)}{\bar{\rho}^2_0(p_2)} = \frac{1+e^{q_1-p_1}}{1+e^{q_1-p_1}+e^{q_2-p_2}}+\frac{1+e^{q_2-p_2}}{1+e^{q_1-p_1}+e^{q_2-p_2}}
	= 1 + \frac{1}{1+e^{q_1-p_1}+e^{q_2-p_2}} >~1.
\end{equation*}
Finally, condition \eqref{eq: impossibility_confounding} can also be verified to hold for the model with uniform valuations that we study in Appendix~\ref{sec: example_analysis_two_sellers} as long as the parameters of the model are not boundary cases.

Nonetheless, even though one cannot always guarantee that confounding both sellers is feasible, we can still consider whether this policy performs well in parametric regimes where it is indeed feasible, which we do in Appendix~\ref{sec: example_analysis_two_sellers}. We now prove Proposition~\ref{prop: impossbility_confounding_two_sellers}.

\begin{proof}[Proof of Proposition~\ref{prop: impossbility_confounding_two_sellers}]
	To establish the Proposition, it suffices to show that ensuring that the confounding constraints of both sellers hold (i.e., the second set of constraints in problem \eqref{eq: max_welfare_confound_both}) is infeasible for all $\phi_L$ small enough.
	Suppose towards a contradiction that there exists a pair of prices $(p_1,p_2)\in \mathcal{P}_1\times \mathcal{P}_2$ such that it is possible to satisfy both sellers' confounding constraints:
	\begin{align*}
		\phi_H\alpha_{\phi_H}^1\bar{\rho}^1_0(p_1) + (1-\phi_H)\bar{\rho}^1_c(p_1,p_2) = \phi_L\alpha_{\phi_L}^1\bar{\rho}^1_0(p_1) + (1-\phi_L)\bar{\rho}^1_c(p_1,p_2),\quad   \text{Confound Seller 1}, \\
		\phi_H\alpha_{\phi_H}^2\bar{\rho}^2_0(p_2) + (1-\phi_H)\bar{\rho}^i_c(p_1,p_2) = \phi_L\alpha_{\phi_L}^2\bar{\rho}^2_0(p_2) + (1-\phi_L)\bar{\rho}^2_c(p_1,p_2),\quad   \text{Confound Seller 2},
	\end{align*}
	where we use short-hand notation for the promotion policy (i.e., $\alpha_{\phi} := \alpha(p_1,p_2,\phi,\bar{h})$). By simple algebra and combining both constraints we have that
	\begin{equation}\label{eq: aux_confound_both_infeasible}
		\alpha_{\phi_H}^1+\alpha_{\phi_H}^2 = \left(\frac{\phi_L}{\phi_H}\right) \left(\alpha_{\phi_L}^1+\alpha_{\phi_L}^2\right) + \left(1-\frac{\phi_L}{\phi_H}\right)\left(\frac{\bar{\rho}^1_c(p_1,p_2)}{\bar{\rho}^1_0(p_1)}+\frac{\bar{\rho}^2_c(p_1,p_2)}{\bar{\rho}^2_0(p_2)}\right).
	\end{equation}
	For the promotion policy to be feasible, it must satisfy  $\alpha_{\phi_H}^1+\alpha_{\phi_H}^2\leq 1$. If that is the case, it follows from \eqref{eq: aux_confound_both_infeasible} that:
	\begin{equation*}
		\left(1-\frac{\phi_L}{\phi_H}\right)\left(\frac{\bar{\rho}^1_c(p_1,p_2)}{\bar{\rho}^1_0(p_1)}+\frac{\bar{\rho}^2_c(p_1,p_2)}{\bar{\rho}^2_0(p_2)}\right)\leq~1.
	\end{equation*}
	It then follows from \eqref{eq: impossibility_confounding} that this condition fails to hold for all small enough $\phi_L$ given a fixed $\phi_H>~0$. This concludes the proof.
\end{proof}

\subsection{Comparisons of Platform Policies under Uniform Valuations} \label{sec: example_analysis_two_sellers}


We now analyze the benefit (in terms of welfare) of confounding one or both sellers relative to truthfully reporting the value of promotion. To demonstrate this benefit, we consider a demand model that extends the structure of Example~\ref{example:UniformDemand} to a settings with two sellers posting prices in each period. For this demand model, we 
compare the values of $W^{\text{truth}}$, $W^{C1}$, $W^{C2}$, and $W^{C,\text{both}}$, as defined in~\eqref{eq: max_welfare_cond_state}--\eqref{eq: max_welfare_confound_both} across a broad range of parameters. With this analysis, we aim to better understand the extent to which confounding promotion policies generate value in settings with competition. 


Our analysis illustrates that confounding one of the sellers often outperforms truthful revelation. At the same time, we observe that confounding both sellers is either infeasible or typically suboptimal. This leads to conclude that the platform may often restrict attention to confounding one of the sellers or neither. Motivated by this observation, we then illustrate that when confounding one seller is valuable, the platform should select to confound the seller with the higher demand of the two (and disclose the true value of $\phi$ to the seller with lower demand). In what follows, we define the demand model based on uniform valuations and two sellers and then provide comparisons to illustrate the insights we have described in \S\ref{sec: TwoSellers}.

\textbf{Uniform Valuations.} 
We consider a simple extension to the model defined in Example~\ref{example:UniformDemand}. Suppose that given $a,b \in [0,1]$, each arriving buyer has valuations for the products sold by sellers 1 and 2 that are independently distributed uniformly over a unit square: $v^1 \sim U[a-1,a]$ and $v^2 \sim U[b-1,b]$. Then, given prices $p_1$ and $p_2$, a utility-maximizing buyer derives a utility of $\max\{v_1-p_1,v_2-p_2,0\}$ by purchasing from the option that results in higher utility given her valuations, and where we normalize the outside option of refraining to purchase to yield a payoff of zero. Then, we can compute the demand functions for each seller in terms of sellers' prices  $p_1\in[0,a],p_2\in[0,b]$ as follows:
\begin{equation*}
	\bar{\rho}_0^1(p_1) = a-p_1, \qquad \bar{\rho}_0^2(p_2) = b-p_2,
\end{equation*}
\begin{equation} \label{eq: rho_c_uniform_two_sellers}
	\begin{split}
			\bar{\rho}_c^1(p_1,p_2) &= 	\begin{cases}
					(a-p_1)(1-b+p_2)+\frac{1}{2}(a-p_1)^2, &\text{ if $p_1-p_2 \geq a-b$,} \\
					a-p_1 -\frac{1}{2}(b-p_2)^2, &\text{ if $p_1-p_2 < a-b$.}
				\end{cases}\\
			\bar{\rho}_c^2(p_1,p_2) &= 	\begin{cases}
					b-p_2 -\frac{1}{2}(a-p_1)^2, &\text{ if $p_1-p_2 \geq a-b$,} \\
					(b-p_2)(1-a+p_1)+\frac{1}{2}(b-p_2)^2, &\text{ if $p_1-p_2 < a-b$.}
				\end{cases}
	\end{split}
\end{equation}

In addition, the expected welfare functions are given by
\begin{equation*}
	\bar{W}_1(p_1) = \frac{1}{2}(a-p_1)^2, \qquad \bar{W}_2(p_2) = \frac{1}{2}(b-p_2)^2,
\end{equation*}
\begin{equation*}
	\bar{W}_c(p_1,p_2) = 	\begin{cases}
		\frac{1}{6}(a-p_1)^3+\frac{1}{2}(a-p_1)^2(1-b+p_2)+\frac{1}{2}(b-p_2)^2, &\text{ if $p_1-p_2 \geq a-b$,} \\
		\frac{1}{6}(b-p_2)^3+\frac{1}{2}(b-p_2)^2(1-a+p_1)+\frac{1}{2}(a-p_1)^2, &\text{ if $p_1-p_2 < a-b$.}
	\end{cases}
\end{equation*}

With some algebra, one can show that the best response for seller 2 when taking only the demand of patient buyers (as defined in \eqref{eq: best_response_two_sellers}) is
\begin{equation*}
	p_2^{BR}(p_1) = 	\begin{cases}
		\frac{1}{2}\left(b-\frac{1}{2}(a-p_1)^2\right), &\text{ if $a-p_1 \leq 2-\sqrt{4-2b}$,} \\
		\frac{1}{3} \left(2\left(1-a+b+p_1\right)  -\sqrt{b^2 + 2b(1-a+p_1) + 4(1-a+p_1)^2}\right), &\text{ if $a-p_1 > 2-\sqrt{4-2b}$.}
	\end{cases}
\end{equation*}
The corresponding best response function for seller 1 is obtained by flipping $a$ with $b$ and $p_1$ with $p_2$ in this expression.



\textbf{Setup of Numerical Comparison.} Considering the demand model with uniform valuations described above, we numerically compare the welfare achieved by the policies defined in \S\ref{sec: policies_definitions_two_sellers} over a broad range of parameter values. 
For each combination of parameters, we compute the values of $W^{\text{truth}}$, $W^{C1}$, $W^{C2}$ and $W^{C,\text{both}}$ (as defined in \S\ref{sec: two_sellers_truthful}--\ref{sec: two_sellers_confound_both}) by
numerically solving problems~\eqref{eq: max_welfare_cond_state} and
\eqref{eq: max_welfare_confound_seller_1}, as well as problem \eqref{eq: max_welfare_confound_both} when it is feasible.\footnote{We used the Sequential Quadratic Programming (SQP) algorithm implemented in MATLAB R2021a for numerical optimization. We declared problem~\eqref{eq: max_welfare_confound_both} to be infeasible if this algorithm could not find a point that satisfies the problem's constraints, with a tolerance parameter of $10^{-10}$.} As discussed in \S\ref{sec: policies_definitions_two_sellers}, solving the corresponding problem for each class of policies yields an optimal static single-price promotion policy for the platform, as well as the associated myopic Nash equilibrium prices that  sellers are incentivized to set at every period. The resulting value of each problem is the corresponding expected welfare associated with the truthful, ``confound one'' and ``confound both'' promotion policies, together with a signaling mechanism that is uninformative for the seller/s that are chosen to be confounded. At the end of this section, we consider the additional gains of optimal signaling. In our analysis we considered the grid of parameter values consisting of $\mu \in \{0.01,0.02,\dots,0.99\}$, $a,b \in \{0.1,0.2,\dots,1\}$, and the pairs $(\phi_H,\phi_L) \in \Phi =  \{0.01,0.05,0.1,0.2,\dots,0.9,0.95,0.99\}$ with $\phi_H>\phi_L$, and pairs of the form $(\phi_H,\phi_H-\epsilon)$ with $\epsilon = 0.005$ and $\phi_H\in \Phi$.

\textbf{Results and Discussion.} We find that confounding both sellers is either infeasible or typically dominated by truthful revelation (i.e., for the vast majority of parameter combinations). Concretely, we observe that the optimal promotion designed to confound both sellers (defined in \eqref{eq: max_welfare_confound_both}) is infeasible for 44.73\% of the parameter combinations we tested. In 55.02\% of cases, confounding both sellers is feasible but outperformed by truthful revelation (i.e., $W^{\text{truth}}(\mu)\geq  W^{C,\text{both}}(\mu)$), with the opposite scenario taking place only on 0.25\% of the cases. This aligns with the intuition provided by Proposition~\ref{prop: impossbility_confounding_two_sellers}, namely that the constraints required to confound both sellers are quite restrictive in the promotion design optimization problem, usually rendering the value of this problem to be relatively low when even feasible. We illustrate these observations in Figure~\ref{fig:confoundbothvstruthful}.

\begin{figure}
	\centering
	\includegraphics[width=\linewidth]{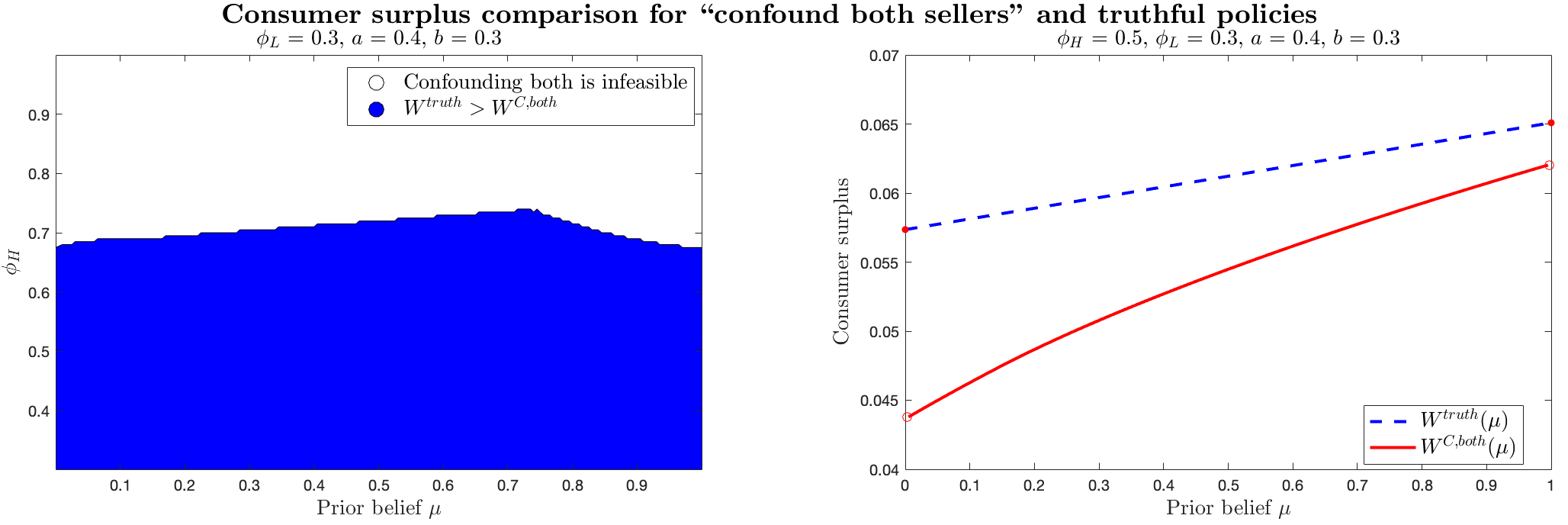}
	\caption{Comparison of welfare achievable by truthful revelation and ``confound both sellers'' policies. The left panel shows a (typical) case where  confounding both sellers is either unfeasible or outperformed by truthful revelation depending on the parameters of the model. The right panel displays the relationship between the welfare achievable by these policies and the prior belief $\mu$ for a fixed combination of parameters.}
	\label{fig:confoundbothvstruthful}
\end{figure}

In addition, we observe that although confounding both sellers is either infeasible or typically suboptimal, confounding only one of the sellers often outperforms truthful revelation. Specifically, in 45.46\% of the instances we tested the platform can achieve higher welfare under a policy that confounds only one of the sellers relative to the welfare achieved by truthful revelation. Interestingly, when this is the case, the choice of which seller to confound is relevant 
since only in 3.3\% of the instances analyzed we have that both $W^{C1}(\mu)> W^{\text{truth}}(\mu)$ and  $W^{C2}(\mu)> W^{\text{truth}}(\mu)$, and in the rest of the cases (42.16\%). 
selecting the appropriate seller to confound is essential to outperform truthful revelation. These results 
are illustrated in Figure~\ref{fig:confoundonevstruthful}, that displays the relationship between the prior belief $\mu$ and the maximum achievable welfare under the policies we described before for select parameters of the model.

\begin{figure}
	\centering
	\includegraphics[width=\linewidth]{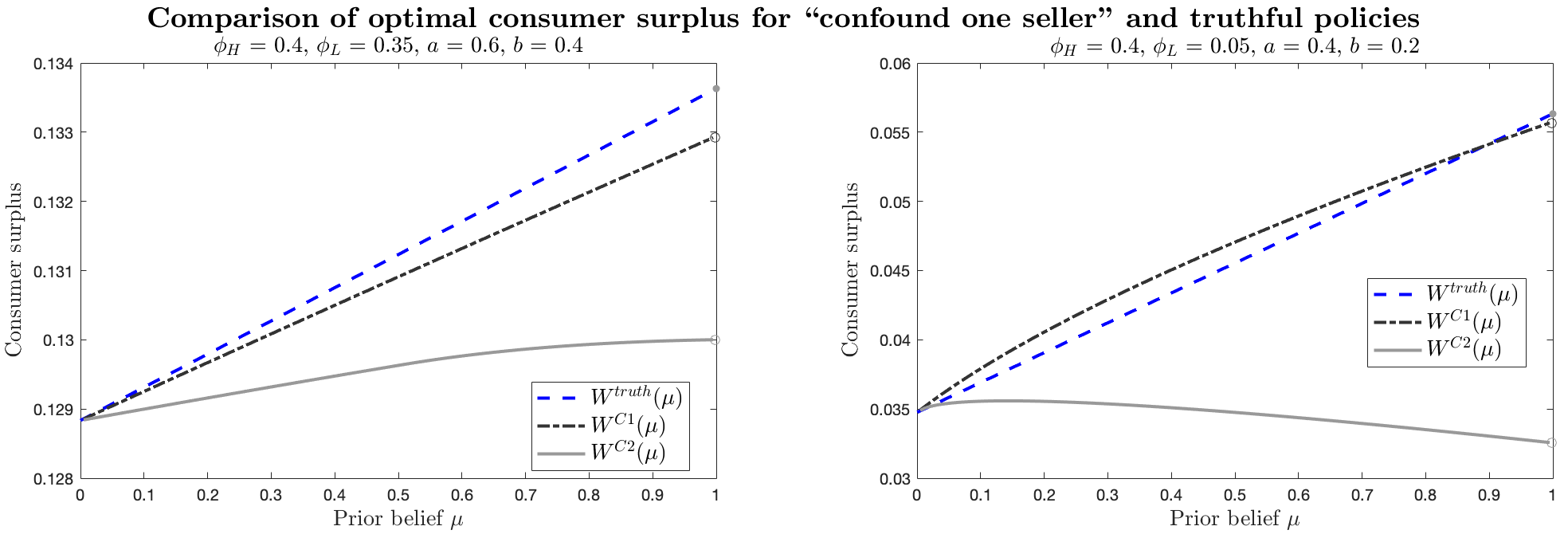}
	\caption{Comparison of welfare achievable by truthful revelation and ``confound only one'' policies as a function of $\mu$. While there are instances where truthful revelation outpeforms confounding either seller (left panel), it is often the case that confounding one of the sellers achieves higher welfare than truthful revelation (right panel). Typically, when confounding one of the sellers is valuable for the platform, confounding the other seller underperforms truthful revelation.}
	\label{fig:confoundonevstruthful}
\end{figure}

Motivated by the previous observation, we look to analyze which seller should the platform select to confound, if any. We find that when confounding one of the sellers outperforms truthful revelation, confounding the seller with ``high'' demand and revealing the state truthfully to the seller with ``low'' demand achieves the highest welfare. More precisely, if the parameters of the demand model defined in~\eqref{eq: rho_c_uniform_two_sellers} satisfy $a>b$, the platform may restrict its attention to consider whether to confound seller 1 or not, while truthfully revealing the value of $\phi$ to seller 2 without loss of optimality, while the opposite holds if $b>a$. This observation holds for all the parameter combinations we analyzed. We illustrate this finding in Figure~\ref{fig:whotoconfound}, which depicts the optimal choice of seller to confound as a function of the prior $\mu$ and seller 1's demand parameter $a$. Indeed, we observe that the regions that define the optimal choice of seller to confound are separated by the horizontal line defined by $a=b$.

The intuition driving this observation is as follows. In order for the platform to confound one of the sellers, it must set an incentive for the other (informed) seller to post the same price independently of the true value of $\phi$, which results in the informed seller offering a relatively high price (in order to maintain incentive compatibility independently of the state). This high price is offset by the fact that the platform is able to induce the confounded seller to offer a relatively low price, by leveraging the uncertainty on $\phi$ that remains for this seller. When $a>b$, seller 1 transactions have a larger impact on the consumer welfare function relative to seller 2, and therefore the platform aims to having seller 1 offer a relatively low price to maximize welfare. This is achieved by confounding seller 1, at the expense of having a higher price posted by seller 2 (which is required for confounding seller 1 in the presence of his competitor's prices).

\begin{figure}[!hbtp]
	\centering
	\includegraphics[width=\linewidth]{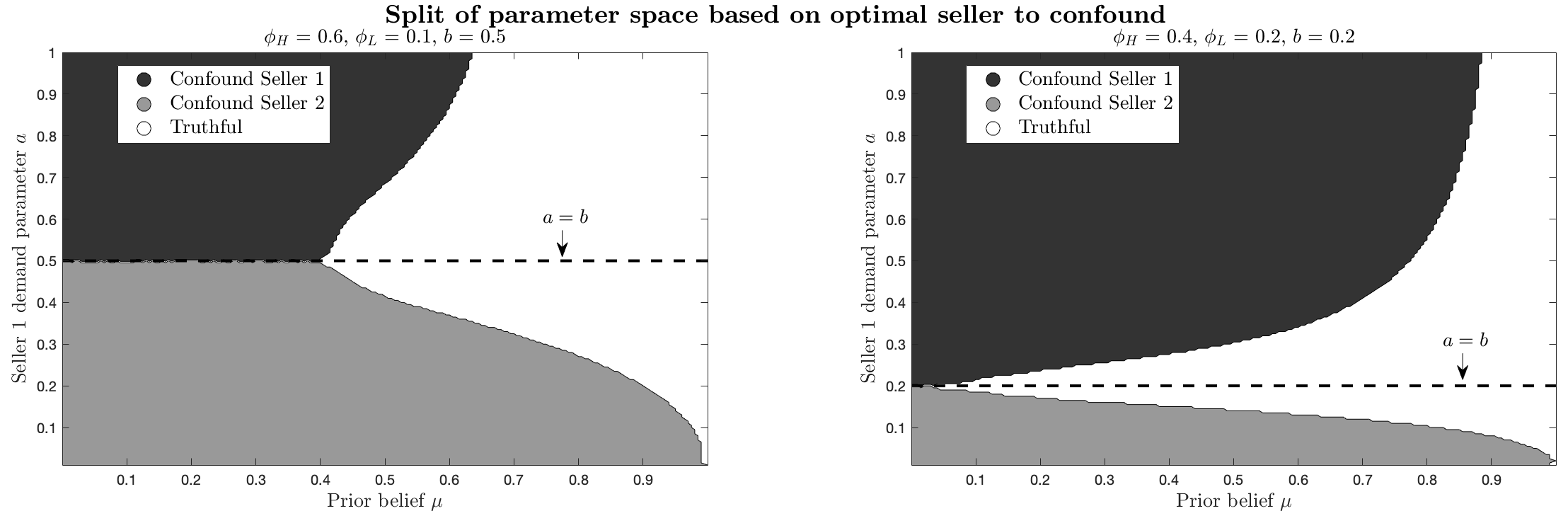}
	\caption{Split of the parameter space according to the optimal seller to confound, if any, as a function of $\mu$ and $a$ with the rest of parameters fixed. When $a>b$, the platform's optimal policy is either to only confound seller 1 or to truthfully reveal the state $\phi$ to both sellers, while the opposite holds when $b>a$. This indicates that in order to maximize consumer welfare, the platform should consider only whether to confound the seller with relatively high demand while revealing the true value of $\phi$ to the seller with lower demand, or to disclose the value of $\phi$ to both sellers.}
	\label{fig:whotoconfound}
\end{figure}

Notice that so far we have analyzed the value of confounding promotion policies without considering the additional effect of optimally designing the platform's signaling mechanism. Based on our previous observations, if the platform considers promotion policies that confound one of the sellers only, it follows from Proposition~\ref{prop: signaling_confound_one} and the discussion in \S\ref{sec: two_sellers_confound_one}
that by appropriately choosing which of the two sellers to confound and carefully designing the signaling mechanism for this seller (while employing a truthful signaling mechanism for the other seller), the platform can achieve an expected welfare of $\max \{co(W^{C1})(\mu),co(W^{C2})(\mu)\}$. 
Therefore, in line with our analysis from \S\ref{sec: NumericalEvals}, we define the relative gain in consumer welfare associated with confounding one of the sellers as follows:
$$
RG(\mu):=\max \left \lbrace \frac{co(W^{C1})(\mu)-W^{truth}(\mu)}{W^{truth}(\mu)} ,\frac{co(W^{C2})(\mu)-W^{truth}(\mu)}{W^{truth}(\mu)}\right \rbrace.
$$

Recall that by definition, $W^{C1}(\mu)=W^{C2}(\mu)=W^{truth}(\mu)$ for $\mu\in \{0,1\}$. Since $W^{truth}(\mu)$ is linear in $\mu$ and the concavification of the optimal welfare functions is itself a concave function, we have that for all $\mu \in [0,1]$,
\begin{equation*}
	co(W^{C1})(\mu) \geq \mu W^{C1}(1) + (1-\mu)W^{C1}(0) = W^{truth}(\mu),
\end{equation*}
implying that the relative gain $RG(\mu)$ is always non-negative. Furthermore, one has that $RG(\mu)~=~0$ if and only if confounding any seller is dominated by truthful revelation (i.e., if $co(W^{C1})(\mu) = co(W^{C2})(\mu) =W^{truth}(\mu)$, implying that revealing the value of $\phi$ to both sellers dominates confounding either of them).

Figure~\ref{fig:relativegains_twosellers} depicts the relationship between the relative gain associated with confounding one of the sellers and the parameters of the model. We observe that the relative gain is largest when the gap between $\phi_H$ and $\phi_L$ is relatively large, and when the prior belief $\mu$ is below $1/2$ but not exceedingly low. Intuitively, the difference $\phi_H - \phi_L$ is a measure of the uncertainty regarding the true value of $\phi$, and sellers are further incentivized to offer lower prices in order to be promoted when the mass of impatient consumers is larger. As the gap $\phi_H - \phi_L$ increases, confounding allows the platform to incentivize sellers to offer prices that are closer to those they would offer if the proportion of impatient buyers was high ($\phi = \phi_H$); and the payoff of establishing such incentives is higher when neither of the possible states of the world are very unlikely to occur so that sellers take both possibilities into account in their pricing decisions.

\begin{figure}
	\centering
	\includegraphics[width=\linewidth]{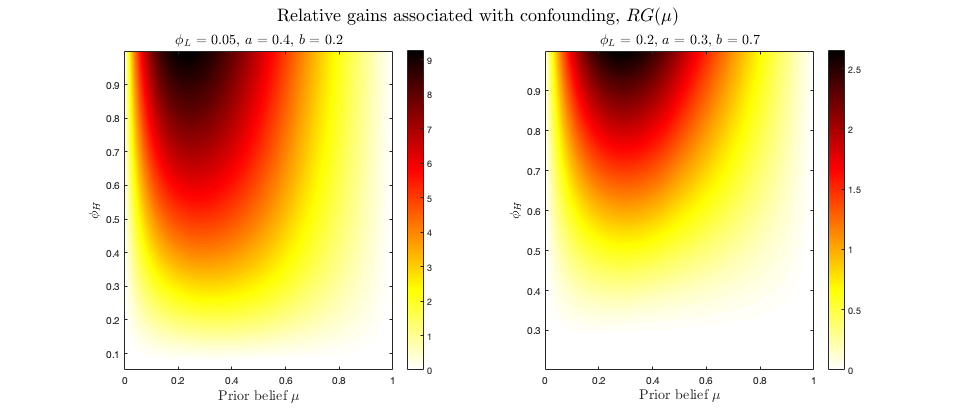}
	\caption{Relationship between the relative gain associated with confounding, $RG(\mu)$ (in percentage points) and the parameters of the model (keeping $\phi_L$, $a$ and $b$ fixed).}
	\label{fig:relativegains_twosellers}
\end{figure}

The previous discussion illustrates that in the presence of two competing sellers, the platform can often benefit from jointly designing signaling mechanisms and a promotion policy that confounds one of the sellers. Once the platform selects which seller to confound and the target price for the informed seller (who observes the value of $\phi$), the optimization problem for deriving the optimal confounding policy reduces to one where a single seller is confounded, with the additional constraint that the promotion probabilities must incentivize the informed seller to set the desired target price. 



\section{Additional Extensions}\label{app:AdditionalExtensions}

In this appendix we present two 
extensions to our baseline model.  
In \S\ref{sec:ObservedPromotions} we consider a setting where, in addition to observing the sales realization in every period, the seller also observes the platform's promotion decision. The main insights we derive from this extension are summarized in \S\ref{sec:Conclusion}. Then, in \S\ref{app:GeneralizedDemand} we extend the baseline model to a more general demand structure, in which both patient and impatient consumers may be affected by the platform's promotion decisions, and show that our main results continue to hold.


\subsection{Seller Observes Promotions}\label{sec:ObservedPromotions}

We now analyze a setting where the seller observes the promotion decision at each period. We show that, as in our baseline formulation, the achievable long-run average consumer surplus when the seller is myopic is determined by the optimal confounding payoff. 
Formally, we adjust the model of \S\ref{sec:Model} by denoting the information available to the seller at the \emph{beginning} of period $t$ as:\vspace{-2mm}
\begin{align*}
    h_1^a &= \left\langle s,\bA,\sigma \right\rangle, \quad \text{ and } \quad	h_t^a = \left\langle s,\bA,\sigma, \left(p_{t'},a_{t'}, y_{t'}\right)_{t'=1}^{t-1}\right\rangle, \text{ for } t>1.
\end{align*}
We denote by $\{\mathcal{H}_t^a = \sigma(h_t), t=1,...,T\}$ the filtration associated with the process $\{h_t^a\}_{t=1}^T$, and we denote the set of possible histories at the beginning of period $t$ as $\bar{H}_t^a = \{L,H\} \times \left(\mathcal{P} \times \{0,1\}^2\right)^{t-1}$. The seller's belief system, $\boldsymbol{\mu}$ is defined in terms of these histories. The payoffs and action spaces remain the same so the seller's and platform's myopic policies remain the same (with respect to the seller's beliefs). However, with new information revealed, the space of confounding promotion policies changes. Confounding promotion policies are defined in the same way (though in terms of the adjusted histories and belief structure).
\begin{definition}[Confounding Promotion Policies]\label{def:confounding:ObservedPromotion}
    Suppose the seller uses the myopic pricing policy, $\boldsymbol{\pi^*}$. For each belief $\mu \in [0,1]$, define the set of confounding promotion policies $\mathcal{A}^{C,a}(\mu) \subset \mathcal{A}^M $ as those which prevent the seller's belief from updating throughout periods $t=1,\ldots,T$. That is, $\bA \in \mathcal{A}^{C,a}(\mu)$, if and only if for all~$t=1,...,T$, one has $\Prob(\mu_{t+1}=\mu|\mu_t=\mu,\boldsymbol{\pi^*},\bA) = 1.$
\end{definition}
With access to promotion decisions, the seller can  learn the true value of $\phi$ based on sales observations and/or promotion decisions, so confounding the seller requires the platform to use policies that satisfy more stringent conditions.  To see this, suppose that the seller sets a price of $p$ at history $\bar{h}$ and consider the seller's belief update when it observes that he was promoted and made a sale. Proceeding as in \eqref{eq:EqualityOfSalesProb} in \S\ref{sec:ConfoundingPromotionPolicies}, we have that for the promotion policy to be confounding it must hold that:
\begin{equation}\label{eq: confounding_with_promo_observations_1}
	\alpha(p,\phi_H,\bar{h}) \left( \phi_H \bar{\rho}_0(p)+(1-\phi_H)\bar{\rho}_c(p)\right) = \alpha(p,\phi_L,\bar{h}) \left( \phi_L \bar{\rho}_0(p)+(1-\phi_L)\bar{\rho}_c(p)\right).
\end{equation}
On the other hand, if the seller is not promoted but still makes a sale, the beliefs will stay constant after updating only if
\begin{equation}\label{eq: confounding_with_promo_observations_2}
	(1-\phi_H)\bar{\rho}_c(p)(1-\alpha(p,\phi_H,\bar{h})) = (1-\phi_L)\bar{\rho}_c(p)(1-\alpha(p,\phi_L,\bar{h})).
\end{equation}
These two conditions make confounding quite restrictive in comparison to the setting where promotions are not observed. To illustrate this, suppose that both \eqref{eq: confounding_with_promo_observations_1} and \eqref{eq: confounding_with_promo_observations_2} hold. Note that condition \eqref{eq: confounding_with_promo_observations_2} holds if either $\bar{\rho}_c(p) = 0$ or $\alpha(p,\phi_H,\bar{h}) = \alpha(p,\phi_L,\bar{h}) = 1$, which gives us two separate cases to consider.

First, if $\bar{\rho}_c(p) = 0$, condition \eqref{eq: confounding_with_promo_observations_1} reduces to $\alpha(p,\phi_H,\bar{h}) \phi_H \bar{\rho}_0(p) = \alpha(p,\phi_L,\bar{h})  \phi_L \bar{\rho}_0(p)$. Since for the seller to be incentivized to set price $p$, we must have that $\bar{\rho}_0(p)>0$, we then have that the policy must satisfy:
\begin{equation*}
	\alpha(p,\phi_H,\bar{h}) \phi_H  = \alpha(p,\phi_L,\bar{h})  \phi_L  ,~~ \bar{\rho}_0(p)>0, ~~\text{ and }  \bar\rho_c(p)=0.
\end{equation*}
In particular, the second condition requires the seller to set a price that results in no demand from patient buyers, which is unlikely to be incentive compatible in many settings.

For the second case, note that if $\alpha(p,\phi_H,\bar{h}) = \alpha(p,\phi_L,\bar{h}) = 1$, conditions \eqref{eq: confounding_with_promo_observations_1} and \eqref{eq: confounding_with_promo_observations_2} reduce to:
\begin{equation*}
	\begin{split}
		\alpha(p,\phi_H,\bar{h}) = \alpha(p,\phi_L,\bar{h}) = 1, ~~\text{ and }  \bar\rho_0(p)=\bar\rho_c(p).   
	\end{split}
\end{equation*}
Again, these conditions impose considerable restrictions on the promotion design problem as, in particular, they impose restrictions on the demand model. In particular, for many demand models satisfying Assumption~\ref{assump:Demand} (including Example \ref{example:UniformDemand}), confounding is not possible. That is,~$A^{C,a}(\mu) = \emptyset.$

To formally show that revealing promotion decisions and employing confounding policies results in a loss of consumer surplus, let us define as in the baseline analysis, the optimal confounding payoff~$W^{C,a}(\mu)$ under the alternate histories $h_t^a$, i.e.,
\begin{equation}\label{PlatformProblem:SinglePeriodConfounding:ObservedPromotion}
    \begin{split}
        W^{C,a}(\mu):= \max_{\boldsymbol{\alpha} \in \mathcal{A}^{C,a}(\mu)}& ~~
        \frac{1}{T}\E\left(\sum_{t=1}^T  W(p_t,a_t,\psi_t)\middle| \boldsymbol{\alpha},\boldsymbol{\pi^*},\mu \right),
    \end{split}
\end{equation}
where we set $W^{C,a}(\mu) = \bar{W}_{\text{out}}$ if $\mathcal{A}^{C,a}(\mu) =\emptyset$. In particular, we can establish that this payoff is dominated by the one achievable by confounding policies in the setting where promotion decisions are not observed.


\begin{proposition}[Access to Promotion Decisions Decreases Consumer Surplus] \label{corr:DescreasesConsSurpl}
    For all $\mu \in [0,1]$, 
    \begin{equation}
        W^C(\mu) \geq W^{C,a}(\mu).
    \end{equation}
\end{proposition}
\begin{proof}
	From the two cases in the preceding discussion, note that for any $\mu \in[0,1]$, $\mathcal{A}^{C,a}(\mu) \subset \mathcal{A}^{C}(\mu)$, i.e., any confounding policy with observed promotions is also confounding when promotions are unobservable.
\end{proof}
In many cases the inequality in Proposition \ref{corr:DescreasesConsSurpl} is strict. For example, since $\mathcal{A}^{C,a}(\mu) = \emptyset $ for all $\mu \in (0,1)$ in Example \ref{example:UniformDemand}, $W^{C,a}(\mu) =\bar{W}_{\text{out}}$ for all $\mu \in (0,1)$. Finally, we can establish that in the setting with observable promotions, $co(W^C)(\mu)$  remains an upper bound on the maximum achievable long-run consumer welfare when the seller employs a myopic pricing policy.

\begin{proposition} \label{prop: upper_bound_observed_promotions}
	If the seller observes promotion decisions, then for all $\mu_0 \in [0,1]$,
	$$\lim_{T\rightarrow \infty}\sup_{\substack{\alpha \in\mathcal{A},\\\sigma \in \Sigma}} ~~\frac{1}{T}W^{\boldsymbol{\alpha},\sigma,\boldsymbol{\pi^*}}_T(\mu_0) \leq co(W^C)(\mu_0),$$
	where $W^C$ is given by \eqref{PlatformProblem:SinglePeriodConfounding} in the baseline setting.
\end{proposition}

Proposition~\ref{prop: upper_bound_observed_promotions} follows from the first part of the proof of Theorem~\ref{thm:LongRunAverageOptimalConsumerSurplus}, which can be directly adapted to this setting by modifying Lemmas~\ref{lemma:delta_epsilonRelationCoU} and~\ref{lemma:generalizedHarrison} (see Appendix~\ref{Sec:Proof:TheoremLongRun}) to consider that the seller has additional sources of information in this setting, and therefore his beliefs also converge  to the truth expontentially fast in the number of periods where some learning occurs.

Through Propositions~\ref{corr:DescreasesConsSurpl} and~\ref{prop: upper_bound_observed_promotions}, we have established that providing access to promotion decisions reduces the long-run average consumer surplus generated by the platform's policy when the seller makes myopic pricing decisions. Thus, one concrete policy recommendation for a platform seeking to maximize consumer surplus is to withhold access to individual promotion decisions.

\subsection{A Generalization of the Demand Model}\label{app:GeneralizedDemand}

In what follows, we provide a more general version of the consumer demand model presented in \S\ref{sec:Model}. In our baseline model, we assumed that patient consumers consider all available products while impatient consumers consider the promoted product only. We now relax this assumption by allowing all types of consumers to consider all products; however, everything else equal, we assume that impatient consumers are more likely 
to buy from a promoted 
seller than patient ones. Specifically, the present extension generalizes the baseline model by (i) allowing impatient consumers to buy from non-promoted sellers with nonzero probability, and (ii) allowing the demand of patient consumers to also depend on the platform's promotion decision. We present the extension in full detail next and establish that our main results continue to hold in this setting.

As in the baseline formulation, we assume that in every period a new consumer arrives and draws her patience type $\psi \in \{I,P\}$ independently. We focus on a single focal seller that in every period chooses a price $p$ from a compact interval $\mathcal{P}$  of the real line. As in \S\ref{sec:Model}, the probability of purchasing from the seller is determined by a demand function $D(\cdot)$ which depends on the seller's price $p\in\cal P$, the platform's promotion decision $a\in \{0,1\}$, the consumer's type $\psi \in \{I,P\}$ and the vector of competitors' prices $\vec{q}$. Thus, we can write the demand function as follows:
\begin{equation}\label{eq: GeneralDemand}
	D(p,a,\psi,\vec q) := D_{a,\psi,\vec q}(p) = \Prob(y = 1|p,a,\psi,\vec q).
\end{equation}
Throughout our analysis, we focus on our focal seller and keep the competitors'  prices $\vec q$ fixed. To simplify exposition we therefore remove the vector $\vec q$ from the notation unless necessary, and simply denote demand by $D_{a,\psi}(p)$; thus, depending on the consumer's type and the platform's promotion decision we have four scenarios that determine the demand as a function of the product's price. 

\begin{table}[!htbp]
	\centering
	\begin{tabular}{|c|c|c|c|}
		\hline
		\multicolumn{2}{|c|}{\multirow{2}{*}{$D(p,a,\psi,\vec q)$}} & \multicolumn{2}{c|}{Customer type} \\ \cline{3-4}
		\multicolumn{2}{|c|}{}                            & Patient         & Impatient        \\ \hline
		\multirow{2}{*}{Promotion}     & $a=0$     & $D_{0,P}(p)$              & $D_{0,I}(p)$               \\ \cline{2-4}
		& $a=1$     & $D_{1,P}(p)$              & $D_{1,I}(p)$               \\ \hline
	\end{tabular}
	\caption{Demand as a function of the customer's type and the platform's promotion decision.}
	\label{tab:ConsDemandGeneral}
\end{table}

In line with the baseline formulation, we make the following assumption on the demand function.
\begin{assumption}[Generalized demand]\label{assump:GeneralDemand}
	For any $a\in\{0,1\}$, $\psi \in \{I,P\}$, 
	the demand functions $D_{a,\psi}(p)$ are decreasing and Lipschitz continuous in $p$; the expected revenue functions $pD_{a,\psi}(p)$ are strictly concave in $p$. In addition, for all $p\in \mathcal{P}$, it holds that
	\begin{equation*}
		D_{1,I}(p)\geq D_{1,P}(p) \geq D_{0,P}(p) \geq D_{0,I}(p).
	\end{equation*}
\end{assumption}
The order of the demand functions given in Assumption~\ref{assump:GeneralDemand} represents a setting where, as in model of \S\ref{sec:Model}, impatient consumers are relatively more likely to buy from a promoted seller than patient consumers. Moreover, note that the baseline formulation presented in \S\ref{sec:Model} is captured here by assuming that for all $p\in \mathcal{P}$, $D_{1,P}(p)=D_{0,P}(p)$ and $D_{0,I}(p) = 0$. Thus, the present extension generalizes the baseline model by (i) allowing impatient consumers to buy from non-promoted sellers with nonzero probability, and (ii) allowing the demand of patient consumers to be influenced by the platform's promotion decision, although to a lesser degree than impatient consumers. Intuitively, we can interpret this generalization of the model in terms of the consumers' degree of impatience -- in the baseline model, consumers' patience degrees are extreme, i.e., impatient consumers are ``infinitely impatient'', while patient consumers were ``infinitely patient.'' The present setting relaxes this feature by allowing variable degrees of patience and impatience, while preserving the notion that impatient consumers are relatively more sensitive to promotion decisions. Finally, the remaining properties that we assume on the demand functions are as in Assumption~\ref{assump:Demand} in the baseline model.

In addition, we also allow consumer surplus to depend on the platform's promotion policy and the consumer's type. Thus, consumer surplus is given by
\begin{equation}\label{eq: GeneralWelfare}
	W(p,a,\psi) := W_{a,\psi}(p),
\end{equation}
where the promotion and type-dependent functions $W_{a,\psi}$ satisfy the following conditions, in line with Assumption~\ref{Assump:ConsSurplus} of the baseline formulation.

\begin{assumption}[Generalized Consumer Surplus]\label{Assump:ConsSurplusGeneral}
	The functions $W_{a,\psi}(p)$ are decreasing and Lipschitz continuous in $p$, for all fixed $a=0,1$ and $\psi = I,P$.
\end{assumption}
In what follows, the definition of the seller's payoff remains as in \S\ref{sec:Model}, i.e., $v(p,y)=py$. In addition, the players' total payoffs, histories, beliefs, signaling mechanisms and strategies are defined as in the baseline formulation of \S\ref{sec:Model}.

\subsubsection{Consumer Surplus Analysis with a Myopic Seller}

As with the baseline model, we analyze the case where the seller sets prices according to a myopic policy $\boldsymbol{\pi^*}$, which is defined as in \S \ref{sec:LongRunAverageOptimalConsumerSurplus}. In addition, Lemma~\ref{lemma:PromotionAsFunctionOfBelief} continues to hold\footnote{The proof presented on \S\ref{proof:lemma:PromotionAsFunctionOfBelief} applies to this formulation as well with minor notational changes. The only step that requires additional attention is establishing that $W_T^{cont}(\mu)$ as defined by \eqref{eq: W_cont_problem} is a continuous function of $\mu$, which can be established by noting that the feasible set correspondence defined by condition \eqref{eq: FirstConstraintGeneral} is continuous in $\mu$.} in the current setting (with the general demand and consumer welfare functions defined on \eqref{eq: GeneralDemand} and \eqref{eq: GeneralWelfare}), so we can focus on promotion policies based on the seller's belief (i.e., $\bA \in \mathcal{A}^M$) without loss of optimality.

Following a similar structure as the baseline model, we can define the promotion policy that maximizes instantaneous consumer welfare  as a function of the seller's belief $\mu_1$ by modifying the optimization problem defined by \eqref{eq:myopicPromotion} on \S\ref{subsec:InsufficiencyOfTruthfulDisclosure} to consider the generalized demand and consumer welfare functions defined on \eqref{eq: GeneralDemand} and \eqref{eq: GeneralWelfare} as follows:\footnote{Slightly abusing notation, for $\alpha\in[0,1]$ we denote $W(p,\alpha,\psi) = \alpha W_{1,\psi}(p)+(1-\alpha)W_{0,\psi}(p)$, and similarly for the demand function.}
\begin{equation} \label{eq:myopicPromotionGeneral}
	\begin{split}
		\max_{\substack{ p \in \mathcal{P},\\ \alpha:~P\times \{\phi_L,\phi_H\}\times[0,1]\rightarrow [0,1]}}& ~~
		\E_{\phi} \left[\E_{\psi}\left[ W(p,\alpha(p,\phi,\mu_1),\psi) \big| \phi \right]    \big|\mu_1\right]\\
		\text{s.t.}&~~ p\in\arg\max_{p'\in \mathcal{P}}~ \E_{\phi} \left[\E_{\psi}\left[ p' D(p',\alpha(p',\phi,\mu_1),\psi) \big| \phi \right]    \big|\mu_1\right],
	\end{split}
\end{equation}
where, as in \S\ref{subsec:InsufficiencyOfTruthfulDisclosure}, the constraint above ensures that price $p$ is selected according to the Myopic Bayesian policy induced by the platform's promotion policy. We refer to the solution of \eqref{eq:myopicPromotionGeneral} as the optimal myopic promotion policy.

As in the analysis presented in \S\ref{sec:LongRunAverageOptimalConsumerSurplus}, a key quantity to consider is the price that maximizes the seller's revenue if it were to assume that the platform will not promote him. Given a belief $\mu\in [0,1]$, we denote this quantity by $p^*(\mu)$, formally defined by
\begin{equation}\label{eq:pStarOutsideOptGeneral}
	p^*(\mu) := \arg \max_{p\in \mathcal{P}} ~~ p\E_{\phi} \left[\phi D_{0,I}(p) + (1-\phi) D_{0,P}(p) \big| \mu\right].
\end{equation}
The price $p^*(\mu)$ is uniquely defined due to Assumption~\ref{assump:GeneralDemand}, which requires $pD_{0,I}(p)$ and $pD_{0,P}(p)$ to be  strictly concave in $p$. By denoting $\bar{\phi}(\mu):= \phi_L + (\phi_H-\phi_L)\mu$ as in \S \ref{sec:LongRunAverageOptimalConsumerSurplus}, it follows that the seller can always achieve an expected revenue of at least $p^*(\mu)\left[\bar{\phi}(\mu) D_{0,I}(p^*(\mu)) + (1-\bar{\phi}(\mu)) D_{0,P}(p^*(\mu)) \right]$, given a belief of $\mu$.

Note that, in contrast with the baseline model, the price $p^*(\mu)$ is a function of the seller's belief~$\mu$ (rather than a constant). This is due to the fact that in the original model we assumed that impatient buyers do not buy from the seller unless it is promoted, i.e., $D_{0,I}(p)=0$. The present generalization relaxes this assumption and allow this demand to be positive.

\subsubsection{Confounding Promotion Policies}

We now turn our attention to the class of confounding policies, $\mathcal{A}^C(\mu)$, which is defined as in the baseline model (see Definition~\ref{def:confounding} on \S\ref{sec:ConfoundingPromotionPolicies}). Recall that promotion policies are those that prevent the seller from learning information about $\phi$ from sales observations, that is,
\begin{equation}\label{eq:EqualityOfSalesProbGeneral}
	\Prob\left(y_t=1|\phi=\phi_H,\mu_t=\mu,p_t = p,\alpha_t\right)  = \Prob\left(y_t=1|\phi=\phi_L,\mu_t=\mu,p_t = p,\alpha_t\right).
\end{equation}
We can write this condition in terms of the demand functions, price and promotion policy as follows. Note that given the true value of $\phi$, in any given period, a patient consumer arrives and purchases from the seller with probability  $(1-\phi)\left[\alpha_t(p,\phi,\mu)D_{1,P}(p)+ (1-\alpha_t(p,\phi,\mu))D_{0,P}(p) \right]$ while an impatient consumer does so with probability  $\phi\left[\alpha_t(p,\phi,\mu)D_{1,I}(p)+ (1-\alpha_t(p,\phi,\mu))D_{0,I}(p) \right].$ Then, we have that
\begin{equation*}
	\Prob\left(y_t=1|\phi,\mu_t=\mu,p_t = p,\alpha_t\right) = \phi \bar{D}_I(p,\alpha_t(p,\phi,\mu)) + (1-\phi) \bar{D}_P(p,\alpha_t(p,\phi,\mu)),
\end{equation*}
where, for $j=I,P$, we denote
\begin{equation*}
	\bar{D}_j(p,a) = aD_{1,j}(p) + (1-a)D_{0,j}(p).
\end{equation*}
The confounding constraint (equation~\eqref{eq:EqualityOfSalesProbGeneral}) can then be written as:
\begin{equation*}
	\phi_H \bar{D}_I(p,\alpha_t(p,\phi_H,\mu)) + (1-\phi_H) \bar{D}_P(p,\alpha_t(p,\phi_H,\mu)) = \phi_L \bar{D}_I(p,\alpha_t(p,\phi_L,\mu)) + (1-\phi_L) \bar{D}_P(p,\alpha_t(p,\phi_L,\mu)).
\end{equation*}

As in Proposition~\ref{prop: ConfoundingPolicyExistence}, it can be easily established that the class of confounding promotion policies  $A^C(\mu)$ is non-empty for all $\mu\in[0,1]$. This can be shown by considering the policy defined by
\begin{equation*}
	\begin{split}
		\bar{\alpha}_t(p,\phi,\mu) = ~\begin{cases}
			\frac{\left(\phi_H-\phi_L\right)\left[D_{0,P}(p^*(\mu))-D_{0,I}(p^*(\mu))\right]}{\phi_H\left[D_{1,I}(p^*(\mu))-D_{0,I}(p^*(\mu)\right]+(1-\phi_H)\left[D_{1,P}(p^*(\mu))-D_{0,P}(p^*(\mu)\right]}, &\text{ if } p = p^*(\mu) \text{ and } \phi = \phi_H\\
			0,& \text{otherwise}
		\end{cases}
	\end{split}
\end{equation*}
One can verify that, by Assumption~\ref{Assump:ConsSurplusGeneral}, it holds that $0\leq \bar{\alpha}_t(p,\phi,\mu)\leq 1$. In addition, $\bar{\alpha}_t(p,\phi,\mu)$ satisfies the confounding condition given in \eqref{eq:EqualityOfSalesProbGeneral} by construction. Finally, as in \S\ref{sec:ConfoundingPromotionPolicies} we define the maximum long-run average consumer surplus induced by a confounding policy as
\begin{equation*}
	\begin{split}
		W^C(\mu_1):= \max_{\boldsymbol{\alpha} \in \mathcal{A}^C(\mu)}& ~~
		\frac{1}{T}\E\left(\sum_{t=1}^T  W(p_t,a_t,\psi_t)\middle| \boldsymbol{\alpha},\boldsymbol{\pi^*},\mu_1 \right).
	\end{split}
\end{equation*}

We now establish that the characterization of Theorem~\ref{thm:LongRunAverageOptimalConsumerSurplus} continues to hold in this setting.

\begin{theorem}\label{thm:General_LongRunAverageOptimalConsumerSurplus}
	Consider the setting with the generalized demand and consumer welfare functions defined on \eqref{eq: GeneralDemand} and \eqref{eq: GeneralWelfare}. Then, for all $\mu \in [0,1]$,
	$$\lim_{T\rightarrow \infty}\sup_{\substack{\alpha \in\mathcal{A},\\\sigma \in \Sigma}} ~~\frac{1}{T}W^{\bA,\sigma,\boldsymbol{\pi^*}}_T(\mu) = co(W^C)(\mu).$$
	Furthermore, for any fixed $T$, there exists a signaling mechanism $\sigma$ and a confounding policy $\bA$ that generate an expected average consumer surplus of $co(W^C)(\mu_0)$.
\end{theorem}

\begin{proof}
	We first note that Proposition~\ref{lemma:RevenueEquivalence} holds in this generalized setting as well, as its proof does not depend on forms of the demand or consumer welfare functions (see \S\ref{app: proof_RevenueEquivalence}). Therefore, we can restrict out attention to the class of single-price policies $\mathcal{A}^P$ without loss of optimality.
	
	Provided that we verify that analogous results to Lemmas~\ref{lemma:delta_epsilonRelationCoU}, \ref{lemma:generalizedHarrison} and \ref{lemma:ConvergenceOfWC} hold in this setting as well, the result can be established by following the same steps as in the proof of Theorem~\ref{thm:LongRunAverageOptimalConsumerSurplus}. We now verify that we can indeed verify these analogous results. As in the proof of Theorem~\ref{thm:LongRunAverageOptimalConsumerSurplus} (see \S\ref{Sec:Proof:TheoremLongRun}), define
	$$M^{\alpha_t}(\epsilon):=\{\mu\in[0,1]: \E_{a_t,p_t,\custtype_t,\phi} \left(W(p_t,a_t,\custtype_t)\middle|\alpha_t,\boldsymbol{\pi}^*,\mu\right)>co(W^{C})(\mu)+\epsilon\}.$$
	Then, the following analogous result to Lemma~\ref{lemma:delta_epsilonRelationCoU} can be established:

	\begin{lemma}[Separation of Purchase Probabilities]\label{lemma:General_delta_epsilonRelationCoU}
		Fix $\epsilon>0$. There exists $\delta>0$ such that for all  $\boldsymbol{\alpha} \in \mathcal{A}^P$, if $\mu \in M^{\alpha_t}(\epsilon)$ and $p_t = \pi^*_t(\mu)$, then:
		$$|\phi_H \bar{D}_I(p,\alpha_t(p,\phi_H,\mu)) + (1-\phi_H) \bar{D}_P(p,\alpha_t(p,\phi_H,\mu)) - \phi_L \bar{D}_I(p,\alpha_t(p,\phi_L,\mu)) - (1-\phi_L) \bar{D}_P(p,\alpha_t(p,\phi_L,\mu))|>\delta.$$
	\end{lemma}
	
	\begin{proof}[Proof of Lemma~\ref{lemma:General_delta_epsilonRelationCoU}]
		This result can be established by adapting the optimization problem defined by \eqref{PlatformProblem:relaxedOptimizationProblem} in the proof of Lemma~\ref{lemma:delta_epsilonRelationCoU} (see \S\ref{sec:Theorem1Proof:AuxProofs}) to consider the general form of the demand and consumer welfare functions:
		\begin{equation}\label{eq: General_PlatformProblem:relaxedOptimizationProblem}
			W^C(\mu,\delta):= \max \left \lbrace \E_{\phi}\left[\E_{\psi}\left[ W(p,\alpha_\phi,\psi) \big| \phi \right]    \big|\mu\right]:\,  \left(\alpha_{\phi_H},\alpha_{\phi_L},p\right) \in F(\mu,\delta)\right \rbrace,
		\end{equation}
		where the feasible set $F(\mu,\delta)$ is defined as the collection of vectors $\left(\alpha_{\phi_H},\alpha_{\phi_L},p\right) \in [0,1]\times [0,1] \times \mathcal{P}$ such that the seller (weakly) prefers setting price $p$ over $p^*(\mu)$:
		\begin{equation} \label{eq: FirstConstraintGeneral}
			\begin{split}
				\E_{\phi} \left[\E_{\psi}\left[ p D(p,\alpha_\phi,\psi) \big| \phi \right]    \big|\mu\right] \geq
				p^*(\mu)\left[\bar{\phi}(\mu) D_{0,I}(p^*(\mu)) + (1-\bar{\phi}(\mu)) D_{0,P}(p^*(\mu)) \right],
			\end{split}
		\end{equation}
		and the ``$\delta$-confounding'' constraint holds:
		\begin{equation}\label{eq: SecondConstraintGeneral}
			|\phi_H \bar{D}_I(p,\alpha_{\phi_H}) + (1-\phi_H) \bar{D}_P(p,\alpha_{\phi_H}) - \phi_L \bar{D}_I(p,\alpha_{\phi_L}) - (1-\phi_L) \bar{D}_P(p,\alpha_{\phi_L}) |\leq  \delta.
		\end{equation}
		Lemma~\ref{lemma:General_delta_epsilonRelationCoU} can then be established by following the same three steps as in the proof of Lemma~\ref{lemma:delta_epsilonRelationCoU} (see \S\ref{sec:Theorem1Proof:AuxProofs}) to the modified problem $W^C(\mu,\delta)$ as defined on~\eqref{eq: General_PlatformProblem:relaxedOptimizationProblem}.
	\end{proof}
	
	The next step is to show that an analogous result to Lemma~\ref{lemma:generalizedHarrison} holds in this setting as well. 
	\begin{lemma}[Convergence of Seller Beliefs]\label{lemma:General_generalizedHarrison}
		Fix $\mu_0 \in [0,1]$ and let $\{t_n\}$ be defined according to \eqref{eq:stoppingTimes} in \S\ref{Sec:Proof:TheoremLongRun}. There exist constants $\chi,\beta>0$ such that for all $n\geq 1$,
		\begin{align*}
			\E(\mu_{t_n} |\phi=\phi_L )&\leq \chi \exp(-\beta n),&\E(1-\mu_{t_n} |\phi=\phi_H )&\leq \chi \exp(-\beta n)
		\end{align*}
	\end{lemma}
	
	\begin{proof}[Proof of Lemma~\ref{lemma:General_generalizedHarrison}]
		Fix a promotion policy $\bA\in \mathcal{A}^P$ and define, for $i\in\{L,H\}$:
		\begin{equation*}
			\rho_t^i :=\phi_i \bar{D}_I(p_t,\alpha_t(p_t,\phi_i,\mu_t)) + (1-\phi_i) \bar{D}_P(p_t,\alpha_t(p_t,\phi_i,\mu_t)),
		\end{equation*}
		where $p_t = \pi_t^*(\mu_t)$. 
		The proof of Lemma~\ref{lemma:generalizedHarrison} (see \S\ref{sec:Theorem1Proof:AuxProofs}) can then be applied verbatim to establish the result.
	\end{proof}
	
	Finally, we can show that Lemma~\ref{lemma:ConvergenceOfWC} continues to hold in this setting. As in \S\ref{sec:Theorem1Proof:PrelAux}, define $W^{\max}(\mu)$ as the maximum consumer surplus achievable by any promotion policy when $T=1$ and the seller has belief $\mu$. 
	
	\begin{lemma}[$W^C(\mu)$ Bounded by Linear Functions]\label{lemma:General_ConvergenceOfWC}
		Fix $\epsilon>0$. There exists $\bar{C}\geq 0$ such that for all $\mu \in [0,1]$:
		\begin{align*}
			co(W^{\max})(\mu) - co(W^{C})(\mu) &< \frac{\epsilon}{2} + \bar{C}\mu,\text{ and} & co(W^{\max})(\mu) - co(W^{C})(\mu) &< \frac{\epsilon}{2} +\bar{C}(1-\mu).&
		\end{align*}
	\end{lemma}
	
	\begin{proof}[Proof of Lemma~\ref{lemma:General_ConvergenceOfWC}]
		The result can be established by following the same argument as in the proof of Lemma~\ref{lemma:ConvergenceOfWC} (see \S\ref{sec:Theorem1Proof:AuxProofs}). The only step that requires additional justification is showing that the feasible set correspondence $F(\mu,\delta)$ is upper hemicontinuous in $\mu$, for fixed $\delta$. To see this, note that $p^*(\mu)$ is a continuous function of $\mu$ as the objective in \eqref{eq:pStarOutsideOptGeneral} is continuous in both $p$ and $\mu$. Therefore, both constraints \eqref{eq: FirstConstraintGeneral} and \eqref{eq: SecondConstraintGeneral} are defined by continuous functions, from where it follows that $F(\mu,\delta)$ satisfies the closed-graph property. In addition, $F(\mu,\delta)$ has compact range, which implies upper hemicontinuity.
	\end{proof}
	
	Finally, Theorem~\ref{thm:General_LongRunAverageOptimalConsumerSurplus} follows by applying the same argument as in the proof of Theorem~\ref{thm:LongRunAverageOptimalConsumerSurplus} (see \S\ref{sec:Theorem1Proof:Main}), but with Lemmas~\ref{lemma:General_delta_epsilonRelationCoU}, \ref{lemma:General_generalizedHarrison}, and \ref{lemma:General_ConvergenceOfWC} playing the roles of Lemmas~\ref{lemma:delta_epsilonRelationCoU}, \ref{lemma:generalizedHarrison}, and \ref{lemma:ConvergenceOfWC}, respectively.
\end{proof}

\subsubsection{Equilibrium Analysis}

To conclude, we illustrate that the analysis presented in \S\ref{sec:EquilibriumAnalysis} also extends to the present setting, that is that Theorems~\ref{thm:epsilonEq} and \ref{thm:OptimalRobustEquilibria} also hold with the genderalized demand and consumer welfare functions defined on \eqref{eq: GeneralDemand} and \eqref{eq: GeneralWelfare}. { To see this, let us define the optimization problem that the platform can use to design a simple confounding promotion policy, which extends the definition of problem \eqref{eq: ConfoundingProblemSimple} on \S\ref{sec: OptSimplePolicies} to the current setting:
	\begin{equation} \label{eq: General_ConfoundingProblemSimple}
		\begin{split}
			W^C(\mu)&:= \max_{\substack{\alpha_{\phi_H},\alpha_{\phi_L} \in [0,1],\\p \in \mathcal{P}}} ~~ \E_{\phi}\left[\E_{\psi}\left[ W(p,\alpha_\phi,\psi) \big| \phi \right]    \big|\mu\right] \\
			\text{s.t.}&~~ \E_{\phi} \left[\E_{\psi}\left[ p D(p,\alpha_\phi,\psi) \big| \phi \right]    \big|\mu\right] \geq
			p^*(\mu)\left[\bar{\phi}(\mu) D_{0,I}(p^*(\mu)) + (1-\bar{\phi}(\mu)) D_{0,P}(p^*(\mu)) \right],\\
			&~~ \phi_H \bar{D}_I(p,\alpha_{\phi_H}) + (1-\phi_H) \bar{D}_P(p,\alpha_{\phi_H}) = \phi_L \bar{D}_I(p,\alpha_{\phi_L}) + (1-\phi_L) \bar{D}_P(p,\alpha_{\phi_L}).
		\end{split}
	\end{equation}
	Letting $\bA^C$ be the simple confounding promotion policy that is constructed by solving \eqref{eq: General_ConfoundingProblemSimple} for each belief $\mu\in[0,1]$, we can then define the strategy profile $(\tilde\bA,\tilde\sigma,\tilde \bP)$ as in the proof of Theorem~\ref{thm:epsilonEq} (see equations \eqref{eq: PromoPolicyThm2} and \eqref{eq: SellerPolicyThm2} on \S\ref{app: proof_BayesianNashEqm}). With this strategy profile, the proofs of Theorems~\ref{thm:epsilonEq} and \ref{thm:OptimalRobustEquilibria} extends almost verbatim (see \S\ref{app: proof_BayesianNashEqm} and \S\ref{app: proof_RobustEqm}); with the slight modification that involves replacing the price that maximizes revenue from patient consumers exclusively, $p^*$, and its corresponding revenue, $(1-\bar{\phi}(\mu))p^* \bar{\rho}_c(p^*)$, with the corresponding belief-dependent price $p^*(\mu)$ as defined on \eqref{eq:pStarOutsideOptGeneral}, and the associated revenue corresponding to the seller not being promoted in this setting, $p^*(\mu)\left[\bar{\phi}(\mu) D_{0,I}(p^*(\mu)) + (1-\bar{\phi}(\mu)) D_{0,P}(p^*(\mu)) \right]$. The proof of Theorem~\ref{thm:OptimalRobustEquilibria} requires an additional minor modification when establishing Claim 1 of Lemma~\ref{lem: aux_robust_eqm}. Here, we define $\bar{\bP}$ as the policy that sets price $p^*(\mu)$ when the seller's belief at the time of choosing his price is $\mu$. Then, we need an additional step to establish condition \eqref{eq: aux_claim1_theorem3}:
	\begin{equation*}
		\begin{split}
			V_{t, t'}^{\bA,\sigma,\bar{\bP}} \left( \langle \bA,\sigma,\bar{h} \rangle \right) &\stackrel{(a)}{\geq}  \sum_{\tau =t}^{t'}\E \left[p^*(\mu_\tau)\left(\bar{\phi}(\mu_\tau) D_{0,I}(p^*(\mu_\tau)) + (1-\bar{\phi}(\mu_\tau)) D_{0,P}(p^*(\mu_\tau)) \right) | \mathcal{H}_t \right]\\
			& \stackrel{(b)}{\geq}  \sum_{\tau =t}^{t'}\E \left[p^*(\mu_t)\left(\bar{\phi}(\mu_\tau) D_{0,I}(p^*(\mu_t)) + (1-\bar{\phi}(\mu_\tau)) D_{0,P}(p^*(\mu_t)) \right) | \mathcal{H}_t \right]\\
			& \stackrel{(c)}{=} \left(t-t'+1\right) p^*(\mu_t)\left[\bar{\phi}(\mu_t) D_{0,I}(p^*(\mu_t)) + (1-\bar{\phi}(\mu_t)) D_{0,P}(p^*(\mu_t)) \right],
		\end{split}
	\end{equation*}

where (a) follows by definition of $\bar{\bP}$, (b) follows by definition of $p^*(\mu)$ (see \eqref{eq:pStarOutsideOptGeneral}) and, finally, (c) follows since the belief process $\{\mu_t\}$  is a martingale (as it is constructed by Bayesian updating) and since $\bar{\phi}(\mu)$ is an affine expression of $\mu$. The rest of the proof follows the same structure as in \S\ref{app: proof_RobustEqm}.
	
}

\end{document}